\documentclass[10pt,showpacs,letterpaper,aps,twocolumn,nofootinbib,pra]{revtex4-1}
\usepackage{relsize,natbib}

\usepackage{amsthm,amsmath,amssymb,graphicx,comment}
\usepackage{bbold}
\usepackage[none]{hyphenat}

\usepackage{abraces}

\usepackage{microtype}

\usepackage{units}
\usepackage{multirow}
\usepackage[normalem]{ulem}
\usepackage{mathtools}
\usepackage{fourier}

\usepackage[usenames,dvipsnames,table]{xcolor}
\definecolor{ultramarine}{RGB}{63, 0, 255}
\definecolor{medblue}{RGB}{0, 0, 100}
\definecolor{panblue}{RGB}{0,24,150}
\definecolor{carmine}{RGB}{150, 0, 24}
\definecolor{gray}{RGB}{150, 150, 150}
\definecolor{darkred}{RGB}{200, 0, 0}
\definecolor{darkgreen}{RGB}{0, 80, 0}
\definecolor{darkblue}{RGB}{0, 0, 200}
\definecolor{darkorange}{rgb}{1.0, 0.55, 0.0}
\definecolor{nred}{rgb}{0.9,0.1,0.1}
\definecolor{nblack}{rgb}{0,0,0}
\definecolor{nblue}{rgb}{0.2,0.2,0.8}
\definecolor{ngreen}{rgb}{0.2,0.6,0.2}
\definecolor{darkestblue}{rgb}{0, 0, 0.3}

\usepackage{tcolorbox}

\newcommand{\prob}{\ensuremath{\mathsf{Prob}}}

\newcommand{\op}[1]{\ensuremath{\mathsf{#1}}}

\newcommand{\crealist}{classical realist }
\newcommand{\Crealist}{Classical realist }

\usepackage{hyperref}
\hypersetup{colorlinks, breaklinks, linkcolor=darkred, urlcolor=darkblue, citecolor=darkgreen}

\usepackage{bm}

\makeatletter
\newcommand{\neutralize}[1]{\expandafter\let\csname c@#1\endcsname\count@}
\makeatother

\newtheorem{conj}{Conjecture}

\newenvironment{conjbis}[1]
  {\renewcommand{\theconj}{\ref*{#1}$'$}%
     \neutralize{conj}\phantomsection
   \begin{conj}}
  {\end{conj}}
\usepackage{dsfont}
\usepackage{yfonts}
\usepackage{rotating}
\usepackage{floatpag}
\rotfloatpagestyle{empty}

\usepackage{amsmath,amsthm,amssymb}
\usepackage{xspace,enumerate,color,epsfig}
\usepackage{graphicx}
\graphicspath{{.}{./figures/}}
\usepackage{marvosym}

\usepackage{tikzfig}
\usepackage{stmaryrd}
\usepackage{docmute}
\usepackage{keycommand}

\usepackage{pifont}

\usepackage{enumitem}

\newkeycommand{\pointmap}[style=point][1]{\,\tikz{\node[style=\commandkey{style}] (x) at (0,-0.3) {$#1$};\node [style=none] (3) at (0, 1.0) {};\node [style=none] (3) at (0, -1.0) {};\draw (x)--(0,0.7);}\,}

\newkeycommand{\typedkpoint}[style=kpoint,edgestyle=][2]{%
\,\begin{tikzpicture}
  \begin{pgfonlayer}{nodelayer}
    \node [style=none] (0) at (0, 1) {};
    \node [style=\commandkey{style}] (1) at (0, -0.75) {$#2$};
    \node [style=right label] (2) at (0.25, 0.5) {$#1$};
  \end{pgfonlayer}
  \begin{pgfonlayer}{edgelayer}
    \draw [\commandkey{edgestyle}] (1) to (0.center);
  \end{pgfonlayer}
\end{tikzpicture}\,}

\newkeycommand{\typedkpointdag}[style=kpoint adjoint,edgestyle=][2]{%
\,\begin{tikzpicture}
  \begin{pgfonlayer}{nodelayer}
    \node [style=none] (0) at (0, -1) {};
    \node [style=\commandkey{style}] (1) at (0, 0.75) {$#2$};
    \node [style=right label] (2) at (0.25, -0.5) {$#1$};
  \end{pgfonlayer}
  \begin{pgfonlayer}{edgelayer}
    \draw [\commandkey{edgestyle}] (0.center) to (1);
  \end{pgfonlayer}
\end{tikzpicture}\,}

\newcommand{\bistateadj}[1]{\,\begin{tikzpicture}[yshift=-3mm]
  \begin{pgfonlayer}{nodelayer}
    \node [style=kpointadj, minimum width=1 cm, inner sep=2pt] (0) at (0, 0.5) {$\psi$};
    \node [style=none] (1) at (-0.75, 0.25) {};
    \node [style=none] (2) at (0.75, 0.25) {};
    \node [style=none] (3) at (-0.75, -0.5) {};
    \node [style=none] (4) at (0.75, -0.5) {};
  \end{pgfonlayer}
  \begin{pgfonlayer}{edgelayer}
    \draw (1.center) to (3.center);
    \draw (2.center) to (4.center);
  \end{pgfonlayer}
\end{tikzpicture}\,}

\newkeycommand{\pointketbra}[style1=copoint,style2=point][2]{\,%
\begin{tikzpicture}
  \begin{pgfonlayer}{nodelayer}
    \node [style=none] (0) at (0, -1.5) {};
    \node [style=\commandkey{style1}] (1) at (0, -0.75) {$#1$};
    \node [style=\commandkey{style2}] (2) at (0, 0.75) {$#2$};
    \node [style=none] (3) at (0, 1.5) {};
  \end{pgfonlayer}
  \begin{pgfonlayer}{edgelayer}
    \draw (0.center) to (1);
    \draw (2) to (3.center);
  \end{pgfonlayer}
\end{tikzpicture}\,}

\newkeycommand{\twocopointketbra}[style1=copoint,style2=copoint,style3=point][3]{\,%
\begin{tikzpicture}
  \begin{pgfonlayer}{nodelayer}
    \node [style=none] (0) at (-0.75, -1.5) {};
    \node [style=none] (1) at (0.75, -1.5) {};
    \node [style=\commandkey{style1}] (2) at (-0.75, -0.75) {$#1$};
    \node [style=\commandkey{style2}] (3) at (0.75, -0.75) {$#2$};
    \node [style=\commandkey{style3}] (4) at (0, 0.75) {$#3$};
    \node [style=none] (5) at (0, 1.5) {};
  \end{pgfonlayer}
  \begin{pgfonlayer}{edgelayer}
    \draw (0.center) to (2);
    \draw (1.center) to (3);
    \draw (4) to (5.center);
  \end{pgfonlayer}
\end{tikzpicture}\,}

\newkeycommand{\twopointketbra}[style1=copoint,style2=point,style3=point][3]{\,%
\begin{tikzpicture}
  \begin{pgfonlayer}{nodelayer}
    \node [style=none] (0) at (0, -1.5) {};
    \node [style=none] (1) at (-0.75, 1.5) {};
    \node [style=\commandkey{style1}] (2) at (0, -0.75) {$#1$};
    \node [style=\commandkey{style2}] (3) at (-0.75, 0.75) {$#2$};
    \node [style=\commandkey{style3}] (4) at (0.75, 0.75) {$#3$};
    \node [style=none] (5) at (0.75, 1.5) {};
  \end{pgfonlayer}
  \begin{pgfonlayer}{edgelayer}
    \draw (0.center) to (2);
    \draw (1.center) to (3);
    \draw (4) to (5.center);
  \end{pgfonlayer}
\end{tikzpicture}\,}

\newkeycommand{\pointbraket}[style1=point,style2=copoint][2]{\,%
\begin{tikzpicture}
  \begin{pgfonlayer}{nodelayer}
    \node [style=\commandkey{style1}] (0) at (0, -0.5) {$#1$};
    \node [style=\commandkey{style2}] (1) at (0, 0.5) {$#2$};
  \end{pgfonlayer}
  \begin{pgfonlayer}{edgelayer}
    \draw (0) to (1);
  \end{pgfonlayer}
\end{tikzpicture}\,}

\newkeycommand{\kpointbraket}[style1=kpoint,style2=kpoint adjoint][2]{\,%
\begin{tikzpicture}
  \begin{pgfonlayer}{nodelayer}
    \node [style=\commandkey{style1}] (0) at (0, -0.75) {$#1$};
    \node [style=\commandkey{style2}] (1) at (0, 0.75) {$#2$};
  \end{pgfonlayer}
  \begin{pgfonlayer}{edgelayer}
    \draw (0) to (1);
  \end{pgfonlayer}
\end{tikzpicture}\,}

\newkeycommand{\kpointmap}[style1=kpoint,style2=map][2]{\,%
\begin{tikzpicture}
  \begin{pgfonlayer}{nodelayer}
    \node [style=\commandkey{style1}] (0) at (0, -1.1) {$#1$};
    \node [style=\commandkey{style2}] (1) at (0, 0.2) {$#2$};
    \node [style=none] (2) at (0, 1.5) {};
  \end{pgfonlayer}
  \begin{pgfonlayer}{edgelayer}
    \draw (0) to (1);
    \draw (1) to (2.center);
  \end{pgfonlayer}
\end{tikzpicture}%
\,}

\newkeycommand{\sandwichtwo}[style1=point,style2=map,style3=map,style4=copoint][4]{\,%
\begin{tikzpicture}
  \begin{pgfonlayer}{nodelayer}
    \node [style=\commandkey{style2}] (0) at (0, -0.75) {$#2$};
    \node [style=\commandkey{style3}] (1) at (0, 0.75) {$#3$};
    \node [style=\commandkey{style1}] (2) at (0, -2) {$#1$};
    \node [style=\commandkey{style4}] (3) at (0, 2) {$#4$};
  \end{pgfonlayer}
  \begin{pgfonlayer}{edgelayer}
    \draw (2) to (0);
    \draw (0) to (1);
    \draw (1) to (3);
  \end{pgfonlayer}
\end{tikzpicture}\,}

\newkeycommand{\longbraket}[style1=point,style2=copoint][2]{\,%
\begin{tikzpicture}
  \begin{pgfonlayer}{nodelayer}
    \node [style=\commandkey{style1}] (0) at (0, -2) {$#1$};
    \node [style=\commandkey{style2}] (1) at (0, 2) {$#2$};
  \end{pgfonlayer}
  \begin{pgfonlayer}{edgelayer}
    \draw (0) to (1);
  \end{pgfonlayer}
\end{tikzpicture}\,}

\newkeycommand{\onb}[style=point]{\ensuremath{\left\{\tikz{\node[style=\commandkey{style}] (x) at (0,-0.3) {$j$};\draw (x)--(0,0.7);}\right\}}\xspace}

\newkeycommand{\copointmap}[style=copoint][1]{\,\tikz{\node[style=\commandkey{style}] (x) at (0,0.3) {$#1$};\draw (x)--(0,-0.7);}\,}

\newcommand{\ket}[1]{\ensuremath{\left|  #1 \right\rangle}}

\tikzstyle{cdiag}=[matrix of math nodes, row sep=3em, column sep=3em, text height=1.5ex, text depth=0.25ex,inner sep=0.5em]
\tikzstyle{arrow above}=[transform canvas={yshift=0.5ex}]
\tikzstyle{arrow below}=[transform canvas={yshift=-0.5ex}]

\usetikzlibrary{shapes.multipart}

\tikzstyle{CqWire}=[color=gray,line width = .75pt,->-]
\tikzstyle{CcWire}=[cWire]
\tikzstyle{RqWire}=[line width = 1pt, color=black,-<-]
\tikzstyle{RcWire}=[cWire]
\tikzstyle{env}=[copoint,regular polygon rotate=0,minimum width=0.2cm, fill=black]

\tikzstyle{probs}=[shape=semicircle,fill=white,draw=black,shape border rotate=180,minimum width=1.2cm]

\tikzset{->-/.style={decoration={
  markings,
  mark=at position .5 with {\arrow{>}}},postaction={decorate}}}
\tikzset{-<-/.style={decoration={
  markings,
  mark=at position .5 with {\arrow{<}}},postaction={decorate}}}

\tikzstyle{bwSpider}=[
       rectangle split,
       rectangle split parts=2,
       rectangle split part fill={black,white},
 minimum size=3.6 mm, inner sep=-2mm, draw=black,scale=0.5,rounded corners=0.8 mm
       ]
 \tikzstyle{wbSpider}=[
       rectangle split,
       rectangle split parts=2,
       rectangle split part fill={white,black},
 minimum size=3.6 mm, inner sep=-2mm, draw=black,scale=0.5,rounded corners=0.8 mm
       ]

\tikzstyle{every picture}=[baseline=-0.25em,scale=0.5]
\tikzstyle{dotpic}=[] 
\tikzstyle{diredges}=[every to/.style={diredge}]
\tikzstyle{math matrix}=[matrix of math nodes,left delimiter=(,right delimiter=),inner sep=2pt,column sep=1em,row sep=0.5em,nodes={inner sep=0pt},text height=1.5ex, text depth=0.25ex]

\tikzstyle{inline text}=[text height=1.5ex, text depth=0.25ex,yshift=0.5mm]
\tikzstyle{label}=[font=\footnotesize,text height=1.5ex, text depth=0.25ex,yshift=0.5mm]
\tikzstyle{left label}=[label,anchor=east,xshift=1.5mm]
\tikzstyle{right label}=[label,anchor=west,xshift=-1mm]
\tikzstyle{up label}=[label,anchor=south,yshift=-1mm]

\tikzstyle{braceedge}=[decorate,decoration={brace,amplitude=2mm,raise=-1mm}]
\tikzstyle{small braceedge}=[decorate,decoration={brace,amplitude=1mm,raise=-1mm}]

\tikzstyle{doubled}=[line width=1.6pt] 
\tikzstyle{boldedge}=[doubled,shorten <=-0.17mm,shorten >=-0.17mm]
\tikzstyle{boldedgegray}=[doubled,gray,shorten <=-0.17mm,shorten >=-0.17mm]
\tikzstyle{singleedgegray}=[gray]

\tikzstyle{semidoubled}=[line width=1.4pt] 
\tikzstyle{semiboldedgegray}=[semidoubled,gray,shorten <=-0.17mm,shorten >=-0.17mm]

\tikzstyle{boxedge}=[semiboldedgegray]

\tikzstyle{boldedgedashed}=[very thick,dashed,shorten <=-0.17mm,shorten >=-0.17mm]
\tikzstyle{vboldedgedashed}=[doubled,dashed,shorten <=-0.17mm,shorten >=-0.17mm]
\tikzstyle{left hook arrow}=[left hook-latex]
\tikzstyle{right hook arrow}=[right hook-latex]
\tikzstyle{sembracket}=[line width=0.5pt,shorten <=-0.07mm,shorten >=-0.07mm]

\tikzstyle{causal edge}=[->,thick,gray]
\tikzstyle{causal nondir}=[thick,gray]
\tikzstyle{timeline}=[thick,gray, dashed]

\tikzstyle{cedge}=[<->,thick,gray!70!white]

\tikzstyle{empty diagram}=[draw=gray!40!white,dashed,shape=rectangle,minimum width=1cm,minimum height=1cm]
\tikzstyle{empty diagram small}=[draw=gray!50!white,dashed,shape=rectangle,minimum width=0.6cm,minimum height=0.5cm]

\tikzstyle{dot}=[inner sep=0mm,minimum width=2mm,minimum height=2mm,draw,shape=circle]
\tikzstyle{bigdot}=[inner sep=0mm,minimum width=5mm,minimum height=5mm,draw,shape=circle]
\tikzstyle{leak}=[white dot, shape=regular polygon, minimum size=3.3 mm, regular polygon sides=3, outer sep=-0.2mm, regular polygon rotate=270]
\tikzstyle{proj}=[regular polygon,regular polygon sides=4,draw,scale=0.75,inner sep=-0.5pt,minimum width=6mm,fill=white]
\tikzstyle{projOut}=[regular polygon,regular polygon sides=3,draw,scale=0.75,inner sep=-0.5pt,minimum width=7.5mm,fill=white,regular polygon rotate=180]
\tikzstyle{projIn}=[regular polygon,regular polygon sides=3,draw,scale=0.75,inner sep=-0.5pt,minimum width=7.5mm,fill=white]
\tikzstyle{Vleak}=[white dot, shape=regular polygon, minimum size=3.3 mm, regular polygon sides=3, outer sep=-0.2mm, regular polygon rotate=90]
\tikzstyle{dleak}=[white dot, line width=1.6pt, shape=regular polygon, minimum size=3.3 mm, regular polygon sides=3, outer sep=-0.2mm, regular polygon rotate=270]

\tikzstyle{Wsquare}=[white dot, shape=regular polygon, rounded corners=0.8 mm, minimum size=3.3 mm, regular polygon sides=3, outer sep=-0.2mm]
\tikzstyle{Wsquareadj}=[white dot, shape=regular polygon, rounded corners=0.8 mm, minimum size=3.3 mm, regular polygon sides=3, outer sep=-0.2mm, regular polygon rotate=180]
\tikzstyle{ddot}=[inner sep=0mm, doubled, minimum width=2.5mm,minimum height=2.5mm,draw,shape=circle]

\tikzstyle{black dot}=[dot,fill=black]
\tikzstyle{white Wsquare}=[Wsquare,fill=gray,text depth=-0.2mm]
\tikzstyle{white Wsquareadj}=[Wsquareadj,fill=white,text depth=-0.2mm]
\tikzstyle{green dot}=[white dot] 
\tikzstyle{gray dot}=[dot,fill=gray,text depth=-0.2mm]
\tikzstyle{red dot}=[gray dot] 

\tikzstyle{black ddot}=[ddot,fill=black]
\tikzstyle{white ddot}=[ddot,fill=white]
\tikzstyle{gray ddot}=[ddot,fill=gray!40!white]

\tikzstyle{gray edge}=[gray!60!white]

\tikzstyle{small dot}=[inner sep=0.2mm,minimum width=0pt,minimum height=0pt,draw,shape=circle]

\tikzstyle{small black dot}=[small dot,fill=black]
\tikzstyle{small white dot}=[small dot,fill=white]
\tikzstyle{small gray dot}=[small dot,fill=gray,draw=gray]

\tikzstyle{causal dot}=[inner sep=0.4mm,minimum width=0pt,minimum height=0pt,draw=white,shape=circle,fill=gray!40!white]

\tikzstyle{phase dimensions}=[minimum size=5mm,font=\footnotesize,rectangle,rounded corners=2.5mm,inner sep=0.2mm,outer sep=-2mm]
\tikzstyle{dphase dimensions}=[minimum size=5mm,font=\footnotesize,rectangle,rounded corners=2.5mm,inner sep=0.2mm,outer sep=-2mm]

\tikzstyle{white phase dot}=[dot,fill=white,phase dimensions]
\tikzstyle{white phase ddot}=[ddot,fill=white,dphase dimensions]

\tikzstyle{white rect ddot}=[draw=black,fill=white,doubled,minimum size=5mm,font=\footnotesize,rectangle,rounded corners=2.5mm,inner sep=0.2mm]
\tikzstyle{gray rect ddot}=[draw=black,fill=gray!40!white,doubled,minimum size=6mm,font=\footnotesize,rectangle,rounded corners=3mm]

\tikzstyle{gray phase dot}=[dot,fill=gray!40!white,phase dimensions]
\tikzstyle{gray phase ddot}=[ddot,fill=gray!40!white,dphase dimensions]
\tikzstyle{grey phase dot}=[gray phase dot]
\tikzstyle{grey phase ddot}=[gray phase ddot]

\tikzstyle{small phase dimensions}=[minimum size=4mm,font=\tiny,rectangle,rounded corners=2mm,inner sep=0.2mm,outer sep=-2mm]
\tikzstyle{small dphase dimensions}=[minimum size=4mm,font=\tiny,rectangle,rounded corners=2mm,inner sep=0.2mm,outer sep=-2mm]

\tikzstyle{small gray phase dot}=[dot,fill=gray!40!white,small phase dimensions]
\tikzstyle{small gray phase ddot}=[ddot,fill=gray!40!white,small dphase dimensions]

\tikzstyle{small map}=[draw,shape=rectangle,minimum height=4mm,minimum width=4mm,fill=white]

\tikzstyle{cnot}=[fill=white,shape=circle,inner sep=-1.4pt]

\tikzstyle{asym hadamard}=[fill=white,draw,shape=NEbox,inner sep=0.6mm,font=\footnotesize,minimum height=4mm]
\tikzstyle{asym hadamard conj}=[fill=white,draw,shape=NWbox,inner sep=0.6mm,font=\footnotesize,minimum height=4mm]
\tikzstyle{asym hadamard dag}=[fill=white,draw,shape=SEbox,inner sep=0.6mm,font=\footnotesize,minimum height=4mm]

\tikzstyle{hadamard}=[fill=white,draw,inner sep=0.6mm,font=\footnotesize,minimum height=4mm,minimum width=4mm]
\tikzstyle{small hadamard}=[fill=white,draw,inner sep=0.6mm,minimum height=1.5mm,minimum width=1.5mm]
\tikzstyle{small hadamard rotate}=[small hadamard,rotate=45]
\tikzstyle{dhadamard}=[hadamard,doubled]
\tikzstyle{small dhadamard}=[small hadamard,doubled]
\tikzstyle{small dhadamard rotate}=[small hadamard rotate,doubled]
\tikzstyle{antipode}=[white dot,inner sep=0.3mm,font=\footnotesize]

\tikzstyle{small gray box}=[small box,fill=gray!30]
\tikzstyle{medium box}=[rectangle,inline text,fill=white,draw,minimum height=5mm,yshift=-0.5mm,minimum width=10mm,font=\small]
\tikzstyle{square box}=[small box] 
\tikzstyle{medium gray box}=[small box,fill=gray!30]
\tikzstyle{semilarge box}=[rectangle,inline text,fill=white,draw,minimum height=5mm,yshift=-0.5mm,minimum width=12.5mm,font=\small]
\tikzstyle{large box}=[rectangle,inline text,fill=white,draw,minimum height=5mm,yshift=-0.5mm,minimum width=15mm,font=\small]
\tikzstyle{large gray box}=[small box,fill=gray!30]

\tikzstyle{Bayes box}=[rectangle,fill=black,draw, minimum height=3mm, minimum width=3mm]

\tikzstyle{gray square point}=[small box,fill=gray!50]

\tikzstyle{dphase box white}=[dhadamard]
\tikzstyle{dphase box gray}=[dhadamard,fill=gray!50!white]
\tikzstyle{phase box white}=[hadamard]
\tikzstyle{phase box gray}=[hadamard,fill=gray!50!white]

\tikzstyle{point nosep}=[regular polygon,regular polygon sides=3,draw,scale=0.75,inner sep=-2pt,minimum width=9mm,fill=white,regular polygon rotate=180]
\tikzstyle{dpoint}=[point,doubled]
\tikzstyle{dcopoint}=[copoint,doubled]

\tikzstyle{pointgrow}=[shape=cornerpoint,kpoint common,scale=0.75,inner sep=3pt]
\tikzstyle{pointgrow dag}=[shape=cornercopoint,kpoint common,scale=0.75,inner sep=3pt]

\tikzstyle{wide copoint}=[fill=white,draw,shape=isosceles triangle,shape border rotate=90,isosceles triangle stretches=true,inner sep=0pt,minimum width=1.5cm,minimum height=6.12mm]
\tikzstyle{wide point}=[fill=white,draw,shape=isosceles triangle,shape border rotate=-90,isosceles triangle stretches=true,inner sep=0pt,minimum width=1.5cm,minimum height=6.12mm,yshift=-0.0mm]
\tikzstyle{wide point plus}=[fill=white,draw,shape=isosceles triangle,shape border rotate=-90,isosceles triangle stretches=true,inner sep=0pt,minimum width=1.74cm,minimum height=7mm,yshift=-0.0mm]

\tikzstyle{wide dpoint}=[fill=white,doubled,draw,shape=isosceles triangle,shape border rotate=-90,isosceles triangle stretches=true,inner sep=0pt,minimum width=1.5cm,minimum height=6.12mm,yshift=-0.0mm]

\tikzstyle{tinypoint}=[regular polygon,regular polygon sides=3,draw,scale=0.55,inner sep=-0.15pt,minimum width=6mm,fill=white,regular polygon rotate=180]

\tikzstyle{white point}=[point]
\tikzstyle{white dpoint}=[dpoint]
\tikzstyle{green point}=[white point] 
\tikzstyle{white copoint}=[copoint]
\tikzstyle{gray point}=[point,fill=gray!40!white]
\tikzstyle{gray dpoint}=[gray point,doubled]
\tikzstyle{red point}=[gray point] 
\tikzstyle{gray copoint}=[copoint,fill=gray!40!white]
\tikzstyle{gray dcopoint}=[gray copoint,doubled]

\tikzstyle{white point guide}=[regular polygon,regular polygon sides=3,font=\scriptsize,draw,scale=0.65,inner sep=-0.5pt,minimum width=9mm,fill=white,regular polygon rotate=180]

\tikzstyle{black point}=[point,fill=black,font=\color{white}]
\tikzstyle{black copoint}=[copoint,fill=black,font=\color{white}]

\tikzstyle{tiny gray point}=[tinypoint,fill=gray!40!white]

\tikzstyle{diredge}=[->]
\tikzstyle{ddiredge}=[<->]
\tikzstyle{rdiredge}=[<-]
\tikzstyle{thickdiredge}=[->, very thick]
\tikzstyle{pointer edge}=[->,very thick,gray]
\tikzstyle{pointer edge part}=[very thick,gray]
\tikzstyle{dashed edge}=[dashed]
\tikzstyle{thick dashed edge}=[very thick,dashed]
\tikzstyle{thick map edge}=[very thick,|->]

\makeatletter
\newcommand{\boxshape}[3]{%
\pgfdeclareshape{#1}{
\inheritsavedanchors[from=rectangle] 
\inheritanchorborder[from=rectangle]
\inheritanchor[from=rectangle]{center}
\inheritanchor[from=rectangle]{north}
\inheritanchor[from=rectangle]{south}
\inheritanchor[from=rectangle]{west}
\inheritanchor[from=rectangle]{east}
\backgroundpath{
\southwest \pgf@xa=\pgf@x \pgf@ya=\pgf@y
\northeast \pgf@xb=\pgf@x \pgf@yb=\pgf@y

\@tempdima=#2
\@tempdimb=#3

\pgfpathmoveto{\pgfpoint{\pgf@xa - 5pt + \@tempdima}{\pgf@ya}}
\pgfpathlineto{\pgfpoint{\pgf@xa - 5pt - \@tempdima}{\pgf@yb}}
\pgfpathlineto{\pgfpoint{\pgf@xb + 5pt + \@tempdimb}{\pgf@yb}}
\pgfpathlineto{\pgfpoint{\pgf@xb + 5pt - \@tempdimb}{\pgf@ya}}
\pgfpathlineto{\pgfpoint{\pgf@xa - 5pt + \@tempdima}{\pgf@ya}}
\pgfpathclose
}
}}
\newcommand{\smallboxshape}[3]{%
\pgfdeclareshape{#1}{
\inheritsavedanchors[from=rectangle] 
\inheritanchorborder[from=rectangle]
\inheritanchor[from=rectangle]{center}
\inheritanchor[from=rectangle]{north}
\inheritanchor[from=rectangle]{south}
\inheritanchor[from=rectangle]{west}
\inheritanchor[from=rectangle]{east}
\backgroundpath{
\southwest \pgf@xa=\pgf@x \pgf@ya=\pgf@y
\northeast \pgf@xb=\pgf@x \pgf@yb=\pgf@y

\@tempdima=#2
\@tempdimb=#3

\pgfpathmoveto{\pgfpoint{\pgf@xa - 3pt + \@tempdima}{\pgf@ya}}
\pgfpathlineto{\pgfpoint{\pgf@xa - 3pt - \@tempdima}{\pgf@yb}}
\pgfpathlineto{\pgfpoint{\pgf@xb + 3pt + \@tempdimb}{\pgf@yb}}
\pgfpathlineto{\pgfpoint{\pgf@xb + 3pt - \@tempdimb}{\pgf@ya}}
\pgfpathlineto{\pgfpoint{\pgf@xa - 3pt + \@tempdima}{\pgf@ya}}
\pgfpathclose
}
}}

\boxshape{NEbox}{0pt}{3pt}
\boxshape{SEbox}{0pt}{-3pt}
\boxshape{NWbox}{5pt}{0pt}
\boxshape{SWbox}{-5pt}{0pt}
\smallboxshape{SWboxSmall}{-3pt}{0pt}
\boxshape{EBox}{-3pt}{3pt}
\boxshape{WBox}{3pt}{-3pt}
\makeatother

\tikzstyle{cloud}=[shape=cloud,draw,minimum width=1.5cm,minimum height=1.5cm]

\tikzstyle{map}=[draw,shape=NEbox,inner sep=1pt,minimum height=4mm,fill=white]
\tikzstyle{dashedmap}=[draw,dashed,shape=NEbox,inner sep=2pt,minimum height=6mm,fill=white]
\tikzstyle{mapdag}=[draw,shape=SEbox,inner sep=1pt,minimum height=4mm,fill=white]
\tikzstyle{mapadj}=[draw,shape=SEbox,inner sep=2pt,minimum height=6mm,fill=white]
\tikzstyle{maptrans}=[draw,shape=SWbox,inner sep=2pt,minimum height=6mm,fill=white]
\tikzstyle{mapconj}=[draw,shape=NWbox,inner sep=2pt,minimum height=6mm,fill=white]

\tikzstyle{medium map}=[draw,shape=NEbox,inner sep=2pt,minimum height=6mm,fill=white,minimum width=7mm]
\tikzstyle{medium map dag}=[draw,shape=SEbox,inner sep=2pt,minimum height=6mm,fill=white,minimum width=7mm]
\tikzstyle{medium map adj}=[draw,shape=SEbox,inner sep=2pt,minimum height=6mm,fill=white,minimum width=7mm]
\tikzstyle{medium map trans}=[draw,shape=SWbox,inner sep=2pt,minimum height=6mm,fill=white,minimum width=7mm]
\tikzstyle{medium map conj}=[draw,shape=NWbox,inner sep=2pt,minimum height=6mm,fill=white,minimum width=7mm]
\tikzstyle{semilarge map}=[draw,shape=NEbox,inner sep=2pt,minimum height=6mm,fill=white,minimum width=9.5mm]
\tikzstyle{semilarge map trans}=[draw,shape=SWbox,inner sep=2pt,minimum height=6mm,fill=white,minimum width=9.5mm]
\tikzstyle{semilarge map adj}=[draw,shape=SEbox,inner sep=2pt,minimum height=6mm,fill=white,minimum width=9.5mm]
\tikzstyle{semilarge map dag}=[draw,shape=SEbox,inner sep=2pt,minimum height=6mm,fill=white,minimum width=9.5mm]
\tikzstyle{semilarge map conj}=[draw,shape=NWbox,inner sep=2pt,minimum height=6mm,fill=white,minimum width=9.5mm]
\tikzstyle{large map}=[draw,shape=NEbox,inner sep=2pt,minimum height=6mm,fill=white,minimum width=12mm]
\tikzstyle{large map conj}=[draw,shape=NWbox,inner sep=2pt,minimum height=6mm,fill=white,minimum width=12mm]
\tikzstyle{very large map}=[draw,shape=NEbox,inner sep=2pt,minimum height=6mm,fill=white,minimum width=17mm]

\tikzstyle{medium dmap}=[draw,doubled,shape=NEbox,inner sep=2pt,minimum height=6mm,fill=white,minimum width=7mm]
\tikzstyle{medium dmap dag}=[draw,doubled,shape=SEbox,inner sep=2pt,minimum height=6mm,fill=white,minimum width=7mm]
\tikzstyle{medium dmap adj}=[draw,doubled,shape=SEbox,inner sep=2pt,minimum height=6mm,fill=white,minimum width=7mm]
\tikzstyle{medium dmap trans}=[draw,doubled,shape=SWbox,inner sep=2pt,minimum height=6mm,fill=white,minimum width=7mm]
\tikzstyle{medium dmap conj}=[draw,doubled,shape=NWbox,inner sep=2pt,minimum height=6mm,fill=white,minimum width=7mm]
\tikzstyle{semilarge dmap}=[draw,doubled,shape=NEbox,inner sep=2pt,minimum height=6mm,fill=white,minimum width=9.5mm]
\tikzstyle{semilarge dmap trans}=[draw,doubled,shape=SWbox,inner sep=2pt,minimum height=6mm,fill=white,minimum width=9.5mm]
\tikzstyle{semilarge dmap adj}=[draw,doubled,shape=SEbox,inner sep=2pt,minimum height=6mm,fill=white,minimum width=9.5mm]
\tikzstyle{semilarge dmap dag}=[draw,doubled,shape=SEbox,inner sep=2pt,minimum height=6mm,fill=white,minimum width=9.5mm]
\tikzstyle{semilarge dmap conj}=[draw,doubled,shape=NWbox,inner sep=2pt,minimum height=6mm,fill=white,minimum width=9.5mm]
\tikzstyle{large dmap}=[draw,doubled,shape=NEbox,inner sep=2pt,minimum height=6mm,fill=white,minimum width=12mm]
\tikzstyle{large dmap conj}=[draw,doubled,shape=NWbox,inner sep=2pt,minimum height=6mm,fill=white,minimum width=12mm]
\tikzstyle{large dmap trans}=[draw,doubled,shape=SWbox,inner sep=2pt,minimum height=6mm,fill=white,minimum width=12mm]
\tikzstyle{large dmap adj}=[draw,doubled,shape=SEbox,inner sep=2pt,minimum height=6mm,fill=white,minimum width=12mm]
\tikzstyle{large dmap dag}=[draw,doubled,shape=SEbox,inner sep=2pt,minimum height=6mm,fill=white,minimum width=12mm]
\tikzstyle{very large dmap}=[draw,doubled,shape=NEbox,inner sep=2pt,minimum height=6mm,fill=white,minimum width=19.5mm]

\tikzstyle{muxbox}=[draw,shape=rectangle,minimum height=3mm,minimum width=3mm,fill=white]
\tikzstyle{dmuxbox}=[muxbox,doubled]

\tikzstyle{box}=[draw,shape=rectangle,inner sep=2pt,minimum height=6mm,minimum width=6mm,fill=white]
\tikzstyle{dbox}=[draw,doubled,shape=rectangle,inner sep=2pt,minimum height=6mm,minimum width=6mm,fill=white]
\tikzstyle{dmap}=[draw,doubled,shape=NEbox,inner sep=2pt,minimum height=6mm,fill=white]
\tikzstyle{dmapdag}=[draw,doubled,shape=SEbox,inner sep=2pt,minimum height=6mm,fill=white]
\tikzstyle{dmapadj}=[draw,doubled,shape=SEbox,inner sep=2pt,minimum height=6mm,fill=white]
\tikzstyle{dmaptrans}=[draw,doubled,shape=SWbox,inner sep=2pt,minimum height=6mm,fill=white]
\tikzstyle{dmapconj}=[draw,doubled,shape=NWbox,inner sep=2pt,minimum height=6mm,fill=white]

\tikzstyle{ddmap}=[draw,doubled,dashed,shape=NEbox,inner sep=2pt,minimum height=6mm,fill=white]
\tikzstyle{ddmapdag}=[draw,doubled,dashed,shape=SEbox,inner sep=2pt,minimum height=6mm,fill=white]
\tikzstyle{ddmapadj}=[draw,doubled,dashed,shape=SEbox,inner sep=2pt,minimum height=6mm,fill=white]
\tikzstyle{ddmaptrans}=[draw,doubled,dashed,shape=SWbox,inner sep=2pt,minimum height=6mm,fill=white]
\tikzstyle{ddmapconj}=[draw,doubled,dashed,shape=NWbox,inner sep=2pt,minimum height=6mm,fill=white]

\boxshape{sNEbox}{0pt}{3pt}
\boxshape{sSEbox}{0pt}{-3pt}
\boxshape{sNWbox}{3pt}{0pt}
\boxshape{sSWbox}{-3pt}{0pt}
\tikzstyle{smap}=[draw,shape=sNEbox,fill=white]
\tikzstyle{smapdag}=[draw,shape=sSEbox,fill=white]
\tikzstyle{smapadj}=[draw,shape=sSEbox,fill=white]
\tikzstyle{smaptrans}=[draw,shape=sSWbox,fill=white]
\tikzstyle{smapconj}=[draw,shape=sNWbox,fill=white]

\tikzstyle{dsmap}=[draw,dashed,shape=sNEbox,fill=white]
\tikzstyle{dsmapdag}=[draw,dashed,shape=sSEbox,fill=white]
\tikzstyle{dsmaptrans}=[draw,dashed,shape=sSWbox,fill=white]
\tikzstyle{dsmapconj}=[draw,dashed,shape=sNWbox,fill=white]

\boxshape{mNEbox}{0pt}{10pt}
\boxshape{mSEbox}{0pt}{-10pt}
\boxshape{mNWbox}{10pt}{0pt}
\boxshape{mSWbox}{-10pt}{0pt}
\tikzstyle{mmap}=[draw,shape=mNEbox]
\tikzstyle{mmapdag}=[draw,shape=mSEbox]
\tikzstyle{mmaptrans}=[draw,shape=mSWbox]
\tikzstyle{mmapconj}=[draw,shape=mNWbox]

\tikzstyle{mmapgray}=[draw,fill=gray!40!white,shape=mNEbox]
\tikzstyle{smapgray}=[draw,fill=gray!40!white,shape=sNEbox]

\makeatletter

\pgfdeclareshape{cornerpoint}{
\inheritsavedanchors[from=rectangle] 
\inheritanchorborder[from=rectangle]
\inheritanchor[from=rectangle]{center}
\inheritanchor[from=rectangle]{north}
\inheritanchor[from=rectangle]{south}
\inheritanchor[from=rectangle]{west}
\inheritanchor[from=rectangle]{east}
\backgroundpath{
\southwest \pgf@xa=\pgf@x \pgf@ya=\pgf@y
\northeast \pgf@xb=\pgf@x \pgf@yb=\pgf@y

\pgfmathsetmacro{\pgf@shorten@left}{\pgfkeysvalueof{/tikz/shorten left}}
\pgfmathsetmacro{\pgf@shorten@right}{\pgfkeysvalueof{/tikz/shorten right}}

\pgfpathmoveto{\pgfpoint{0.5 * (\pgf@xa + \pgf@xb)}{\pgf@ya - 5pt}}
\pgfpathlineto{\pgfpoint{\pgf@xa - 8pt + \pgf@shorten@left}{\pgf@yb - 1.5 * \pgf@shorten@left}}
\pgfpathlineto{\pgfpoint{\pgf@xa - 8pt + \pgf@shorten@left}{\pgf@yb}}
\pgfpathlineto{\pgfpoint{\pgf@xb + 8pt - \pgf@shorten@right}{\pgf@yb}}
\pgfpathlineto{\pgfpoint{\pgf@xb + 8pt - \pgf@shorten@right}{\pgf@yb - 1.5 * \pgf@shorten@right}}
\pgfpathclose
}
}

\pgfdeclareshape{cornercopoint}{
\inheritsavedanchors[from=rectangle] 
\inheritanchorborder[from=rectangle]
\inheritanchor[from=rectangle]{center}
\inheritanchor[from=rectangle]{north}
\inheritanchor[from=rectangle]{south}
\inheritanchor[from=rectangle]{west}
\inheritanchor[from=rectangle]{east}
\backgroundpath{
\southwest \pgf@xa=\pgf@x \pgf@ya=\pgf@y
\northeast \pgf@xb=\pgf@x \pgf@yb=\pgf@y

\pgfmathsetmacro{\pgf@shorten@left}{\pgfkeysvalueof{/tikz/shorten left}}
\pgfmathsetmacro{\pgf@shorten@right}{\pgfkeysvalueof{/tikz/shorten right}}

\pgfpathmoveto{\pgfpoint{0.5 * (\pgf@xa + \pgf@xb)}{\pgf@yb + 5pt}}
\pgfpathlineto{\pgfpoint{\pgf@xa - 8pt + \pgf@shorten@left}{\pgf@ya + 1.5 * \pgf@shorten@left}}
\pgfpathlineto{\pgfpoint{\pgf@xa - 8pt + \pgf@shorten@left}{\pgf@ya}}
\pgfpathlineto{\pgfpoint{\pgf@xb + 8pt - \pgf@shorten@right}{\pgf@ya}}
\pgfpathlineto{\pgfpoint{\pgf@xb + 8pt - \pgf@shorten@right}{\pgf@ya + 1.5 * \pgf@shorten@right}}
\pgfpathclose
}
}

\makeatother

\pgfkeyssetvalue{/tikz/shorten left}{0pt}
\pgfkeyssetvalue{/tikz/shorten right}{0pt}

\tikzstyle{kpoint common}=[draw,fill=white,inner sep=1pt,minimum height=4mm]
\tikzstyle{kpoint sc}=[shape=cornerpoint,kpoint common]
\tikzstyle{kpoint adjoint sc}=[shape=cornercopoint,kpoint common]
\tikzstyle{kpoint}=[shape=cornerpoint,shorten left=5pt,kpoint common]
\tikzstyle{kpoint adjoint}=[shape=cornercopoint,shorten left=5pt,kpoint common]
\tikzstyle{kpoint conjugate}=[shape=cornerpoint,shorten right=5pt,kpoint common]
\tikzstyle{kpoint transpose}=[shape=cornercopoint,shorten right=5pt,kpoint common]
\tikzstyle{kpoint symm}=[shape=cornerpoint,shorten left=5pt,shorten right=5pt,kpoint common]

\tikzstyle{wide kpoint sc}=[shape=cornerpoint,kpoint common, minimum width=1 cm]
\tikzstyle{wide kpointdag sc}=[shape=cornercopoint,kpoint common, minimum width=1 cm]

\tikzstyle{black kpoint}=[shape=cornerpoint,shorten left=5pt,kpoint common,fill=black,font=\color{white}]

\tikzstyle{black kpoint sm}=[shape=cornerpoint,shorten left=5pt,kpoint common,fill=black,font=\color{white},scale=0.75]

\tikzstyle{black kpoint adjoint}=[shape=cornercopoint,shorten left=5pt,kpoint common,fill=black,font=\color{white}]
\tikzstyle{black kpointadj}=[shape=cornercopoint,shorten left=5pt,kpoint common,fill=black,font=\color{white}]

\tikzstyle{black kpointadj sm}=[shape=cornercopoint,shorten left=5pt,kpoint common,fill=black,font=\color{white},scale=0.75]

\tikzstyle{black dkpoint}=[shape=cornerpoint,shorten left=5pt,kpoint common,fill=black, doubled,font=\color{white}]
\tikzstyle{black dkpoint adjoint}=[shape=cornercopoint,shorten left=5pt,kpoint common,fill=black, doubled,font=\color{white}]
\tikzstyle{black dkpointadj}=[shape=cornercopoint,shorten left=5pt,kpoint common,fill=black, doubled,font=\color{white}]

\tikzstyle{black dkpoint sm}=[shape=cornerpoint,shorten left=5pt,kpoint common,fill=black, doubled,font=\color{white},scale=0.75]
\tikzstyle{black dkpointadj sm}=[shape=cornercopoint,shorten left=5pt,kpoint common,fill=black, doubled,font=\color{white},scale=0.75]

\tikzstyle{kpointdag}=[kpoint adjoint]
\tikzstyle{kpointadj}=[kpoint adjoint]
\tikzstyle{kpointconj}=[kpoint conjugate]
\tikzstyle{kpointtrans}=[kpoint transpose]

\tikzstyle{big kpoint}=[kpoint, minimum width=1.2 cm, minimum height=8mm, inner sep=4pt, text depth=3mm]

\tikzstyle{wide kpoint}=[kpoint, minimum width=1 cm, inner sep=2pt]
\tikzstyle{wide kpointdag}=[kpointdag, minimum width=1 cm, inner sep=2pt]
\tikzstyle{wide kpointconj}=[kpointconj, minimum width=1 cm, inner sep=2pt]
\tikzstyle{wide kpointtrans}=[kpointtrans, minimum width=1 cm, inner sep=2pt]

\tikzstyle{wider kpoint}=[kpoint, minimum width=1.25 cm, inner sep=2pt]
\tikzstyle{wider kpointdag}=[kpointdag, minimum width=1.25 cm, inner sep=2pt]
\tikzstyle{wider kpointconj}=[kpointconj, minimum width=1.25 cm, inner sep=2pt]
\tikzstyle{wider kpointtrans}=[kpointtrans, minimum width=1.25 cm, inner sep=2pt]

\tikzstyle{gray kpoint}=[kpoint,fill=gray!50!white]
\tikzstyle{gray kpointdag}=[kpointdag,fill=gray!50!white]
\tikzstyle{gray kpointadj}=[kpointadj,fill=gray!50!white]
\tikzstyle{gray kpointconj}=[kpointconj,fill=gray!50!white]
\tikzstyle{gray kpointtrans}=[kpointtrans,fill=gray!50!white]

\tikzstyle{gray dkpoint}=[kpoint,fill=gray!50!white,doubled]
\tikzstyle{gray dkpointdag}=[kpointdag,fill=gray!50!white,doubled]
\tikzstyle{gray dkpointadj}=[kpointadj,fill=gray!50!white,doubled]
\tikzstyle{gray dkpointconj}=[kpointconj,fill=gray!50!white,doubled]
\tikzstyle{gray dkpointtrans}=[kpointtrans,fill=gray!50!white,doubled]

\tikzstyle{white label}=[draw,fill=white,rectangle,inner sep=0.7 mm]
\tikzstyle{gray label}=[draw,fill=gray!50!white,rectangle,inner sep=0.7 mm]
\tikzstyle{black label}=[draw,fill=black,rectangle,inner sep=0.7 mm]

\tikzstyle{dkpoint}=[kpoint,doubled]
\tikzstyle{wide dkpoint}=[wide kpoint,doubled]
\tikzstyle{dkpointdag}=[kpoint adjoint,doubled]
\tikzstyle{wide dkpointdag}=[wide kpointdag,doubled]
\tikzstyle{dkcopoint}=[kpoint adjoint,doubled]
\tikzstyle{dkpointadj}=[kpoint adjoint,doubled]
\tikzstyle{dkpointconj}=[kpoint conjugate,doubled]
\tikzstyle{dkpointtrans}=[kpoint transpose,doubled]

\tikzstyle{kscalar}=[kpoint common, shape=EBox, inner xsep=-1pt, inner ysep=3pt,font=\small]
\tikzstyle{kscalarconj}=[kpoint common, shape=WBox, inner xsep=-1pt, inner ysep=3pt,font=\small]

\tikzstyle{spekpoint}=[kpoint sc,minimum height=5mm,inner sep=3pt]
\tikzstyle{spekcopoint}=[kpoint adjoint sc,minimum height=5mm,inner sep=3pt]

\tikzstyle{dspekpoint}=[spekpoint,doubled]
\tikzstyle{dspekcopoint}=[spekcopoint,doubled]

 \tikzstyle{bigground}=[regular polygon,regular polygon sides=3,draw=gray,scale=0.50,inner sep=-0.5pt,minimum width=10mm,fill=gray]

\tikzstyle{arrs}=[-latex,font=\small,auto]
\tikzstyle{arrow plain}=[arrs]
\tikzstyle{arrow dashed}=[dashed,arrs]
\tikzstyle{arrow bold}=[very thick,arrs]
\tikzstyle{arrow hide}=[draw=white!0,-]
\tikzstyle{arrow reverse}=[latex-]
\tikzstyle{cdnode}=[]
\tikzstyle{dscalar}=[diamond,doubled, draw,inner sep=0.5pt,font=\small]

\tikzstyle{epiCopoint}=[regular polygon,regular polygon sides=3,draw,scale=0.75,inner sep=-0.5pt,minimum width=5mm,fill=white,regular polygon rotate=0,line width=1pt]
\tikzstyle{epiPoint}=[regular polygon,regular polygon sides=3,draw,scale=0.75,inner sep=-0.5pt,minimum width=5mm,fill=white,regular polygon rotate=180,line width=1pt]
\tikzstyle{epiPointWide}=[regular polygon,regular polygon sides=3,draw,xscale=0.75,yscale=.5,inner sep=-0.5pt,minimum width=8mm,fill=white,regular polygon rotate=180,line width=1pt]
\tikzstyle{epiBox}=[fill=white,draw, line width = 1pt,inner sep=0.6mm,font=\footnotesize,minimum height=3mm,minimum width=3mm]
\tikzstyle{epiBoxWide}=[fill=white,draw, line width = 1pt,inner sep=0.6mm,font=\footnotesize,minimum height=3mm,minimum width=5mm]
\tikzstyle{epiBoxVeryWide}=[fill=white,draw, line width = 1pt,inner sep=0.6mm,font=\footnotesize,minimum height=3mm,minimum width=7mm]

\tikzstyle{clear dot}=[dot,fill=none,text depth=-0.2mm,draw=gray, line width = .75pt]
\tikzstyle{tall clear dot}=[dot,fill=none,text depth=-0.2mm,draw=gray, line width = .75pt,shape=ellipse, minimum height=5mm]
\tikzstyle{wide clear dot}=[dot,fill=none,text depth=-0.2mm,draw=gray, line width = .75pt, shape=ellipse, minimum width = 5mm]
\tikzstyle{very wide clear dot}=[dot,fill=none,text depth=-0.2mm,draw=gray, line width = .75pt, shape=ellipse, minimum width = 7mm ]

\tikzstyle{point}=[regular polygon,regular polygon sides=3,draw,scale=0.75,inner sep=-0.5pt,minimum width=9mm,fill=white,regular polygon rotate=180]
\tikzstyle{infpoint}=[regular polygon,regular polygon sides=3,draw,scale=0.75,inner sep=-0.5pt,minimum width=9mm,fill=white,regular polygon rotate=90]
\tikzstyle{infcopoint}=[regular polygon,regular polygon sides=3,draw,scale=0.75,inner sep=-0.5pt,minimum width=9mm,fill=white,regular polygon rotate=270]
\tikzstyle{copoint}=[regular polygon,regular polygon sides=3,draw,scale=0.75,inner sep=-0.5pt,minimum width=9mm,fill=white]
\tikzstyle{small box}=[rectangle,inline text,fill=white,draw,minimum height=5mm,yshift=-0.5mm,minimum width=5mm,font=\small]
\tikzstyle{scalar}=[diamond,draw,inner sep=0.5pt,font=\small]

\tikzstyle{white dot}=[dot,fill=white,text depth=-0.2mm]
\tikzstyle{uControl}= [draw, shape=diamond, aspect=.5,inner sep=0pt,minimum height=2mm,minimum width=3.5mm,fill=black]
\tikzstyle{funcApp} = [draw, shape=diamond, aspect=2,inner sep=0pt,minimum height=3.5mm,minimum width=2mm,fill=black]
\tikzstyle{funcSplit} = [draw,shape=regular polygon, regular polygon sides = 3,inner sep=0pt,minimum height=3mm,minimum  width=1mm,regular polygon rotate=210, fill=black]
\tikzstyle{copy} = [black dot]
\tikzstyle{infMerge} = [draw=gray,fill=gray,shape=regular polygon, regular polygon sides = 3,inner sep=0pt,minimum height=2.5mm,minimum  width=1mm,regular polygon rotate=30]
\tikzstyle{infSplit}= [draw=gray,fill=gray,shape=regular polygon, regular polygon sides = 3,inner sep=0pt,minimum height=2.5mm,minimum  width=1mm,regular polygon rotate=210]
\tikzstyle{seqComp}=[white dot]
\tikzstyle{parComp}=[circuit ee IEC, bulb,fill=white]
\tikzstyle{addStar}=[draw, shape=star, star points=5, star rotate=90,minimum size=2mm, inner sep=0pt,fill=black]
\tikzstyle{removeStar}=[draw, shape=star, star points=5, star rotate =270,minimum size=2mm, inner sep=0pt,fill=white]

\tikzset{XOR/.style={draw,fill=white,circle,append after command={
        [shorten >=\pgflinewidth, shorten <=\pgflinewidth,]
        (\tikzlastnode.north) edge (\tikzlastnode.south)
        (\tikzlastnode.east) edge (\tikzlastnode.west)
        }
    }
}

\tikzstyle{upground}=[circuit ee IEC,thick,ground,rotate=90,xscale=2.5,yscale=2]
 \tikzstyle{downground}=[circuit ee IEC,thick,ground,rotate=-90,xscale=2.5,yscale=2]
 \tikzstyle{infupground}=[circuit ee IEC,thick,ground,rotate=0,xscale=2.5,yscale=2]
  \tikzstyle{ignore}=[circuit ee IEC,thick,ground,rotate=0,xscale=1.5,yscale=1,xshift=5pt]
 \tikzstyle{infdownground}=[circuit ee IEC,thick,ground,rotate=180,xscale=2.5,yscale=2]

\tikzstyle{oWire}=[line width = .75pt, color=green!40!black!70!]
\tikzstyle{qWire}=[line width = 1pt, color=black]
\tikzstyle{cWire}=[color=gray,line width = .75pt]
\tikzstyle{thick gray dashed edge}=[thick dashed edge,gray!40]

\let\olddagger\dagger
\renewcommand{\dagger}{\ensuremath{\olddagger}\xspace}

\usepackage{makeidx}
\makeindex

\theoremstyle{plain}
\newtheorem*{main theorem}{Main Theorem}
\newtheorem{theorem}{Theorem}[section]
\newtheorem{corollary}[theorem]{Corollary}
\newtheorem{lemma}[theorem]{Lemma}
\newtheorem{proposition}[theorem]{Proposition}

\newtheorem{definition}[theorem]{Definition}

\newtheorem{example*}[theorem]{Example*}
\newtheorem{examples*}[theorem]{Examples*}
\newtheorem{remark}[theorem]{Category-theoretic remark}
\newtheorem{remark*}[theorem]{Remark*}

\newtheorem*{search problem}{Search Problem}

\usepackage{color}
\def\bR{\begin{color}{red}}
\def\bB{\begin{color}{blue}}
\def\bM{\begin{color}{magenta}}
\def\bC{\begin{color}{cyan}}
\def\bW{\begin{color}{white}}
\def\bBl{\begin{color}{black}}
\def\bG{\begin{color}{green}}
\def\bY{\begin{color}{yellow}}
\def\e{\end{color}\xspace}
\newcommand{\bit}{\begin{itemize}}
\newcommand{\eit}{\end{itemize}\par\noindent}
\newcommand{\ben}{\begin{enumerate}}
\newcommand{\een}{\end{enumerate}\par\noindent}
\newcommand{\beq}{\begin{equation}}
\newcommand{\eeq}{\end{equation}\par\noindent}
\newcommand{\beqa}{\begin{align*}}
\newcommand{\eeqa}{\end{align*}}
\newcommand{\beqn}{\begin{align}}
\newcommand{\eeqn}{\end{align}\par\noindent}

\def\jR{\begin{color}{black}}
\def\jB{\begin{color}{black}}
\def\jM{\begin{color}{magenta}}
\def\jC{\begin{color}{cyan}}
\def\jW{\begin{color}{white}}
\def\jBl{\begin{color}{black}}
\def\jG{\begin{color}{green}}
\def\jY{\begin{color}{yellow}}

\usepackage{stmaryrd}
\newcommand{\morph}[2]{\meas{#1}{#2}}
\newcommand{\meas}[2]{{\ensuremath{ \; \overline{\! \mathsf{#1} \!{\shortrightarrow} \! \mathsf{#2} \! }\;}}}

\newcommand{\FI}{\ensuremath{\textsc{F\!\text{-}\!S}}}
\newcommand{\PI}{\ensuremath{\textsc{P\!\text{-}\!S}}}
\newcommand{\PQI}{\ensuremath{\textsc{P}_\textsc{Q}\!\text{-}\!\textsc{S}}}
\newcommand{\CI}{\ensuremath{\textsc{C\!\text{-}\!I}}}
\newcommand{\Func}{\ensuremath{\textsc{Func} }}
\newcommand{\Proc}{\ensuremath{\textsc{Proc} }}
\newcommand{\Inf}{\ensuremath{\textsc{SubStoch} }}
\newcommand{\genInf}{\ensuremath{\textsc{Inf} }}
\newcommand{\Caus}{\ensuremath{\textsc{Caus} }}
\newcommand{\XFXI}{\ensuremath{\textsc{XF\!\text{-}\!XS}}}
\newcommand{\QFQI}{\ensuremath{\textsc{QF\!\text{-}\!QS}}}

\newcommand{\InfSplit}{%
\begin{tikzpicture}
	\begin{pgfonlayer}{nodelayer}
		\node [style=infSplit] (0) at (0, -0) {};
	\end{pgfonlayer}
\end{tikzpicture}}}
\newcommand{\InfMerge}{%
\begin{tikzpicture}
	\begin{pgfonlayer}{nodelayer}
		\node [style=infMerge] (0) at (0, -0) {};
	\end{pgfonlayer}
\end{tikzpicture}}}
\newcommand{\BlackDiamond}{%
\begin{tikzpicture}
	\begin{pgfonlayer}{nodelayer}
		\node [style=funcApp] (0) at (0, 0) {};
	\end{pgfonlayer}
\end{tikzpicture}
}}

\allowdisplaybreaks

\begin{document}

\title{Unscrambling the omelette of causation and inference: \\The framework of causal-inferential theories}

\author{David Schmid}
\email{dschmid@perimeterinstitute.ca}
\affiliation{Perimeter Institute for Theoretical Physics, 31 Caroline Street North, Waterloo, Ontario Canada N2L 2Y5}
\affiliation{Institute for Quantum Computing, University of Waterloo, Waterloo, Ontario N2L 3G1, Canada}
\author{John H. Selby}
\email{john.h.selby@gmail.com}
\affiliation{International Centre for Theory of Quantum Technologies,
University of Gda\'nsk, Wita Stwosza 63, 80-308 Gda\'nsk, Poland}
\author{Robert W. Spekkens}
\affiliation{Perimeter Institute for Theoretical Physics, 31 Caroline Street North, Waterloo, Ontario Canada N2L 2Y5}

\begin{abstract}
Using a process-theoretic formalism, we introduce the notion of a {\em causal-inferential theory}: a triple consisting of a theory of causal influences, a theory of inferences, and a specification of how these interact. Recasting the notions of operational and realist theories in this mold clarifies what a realist account of an experiment offers beyond an operational account. It also yields a novel characterization of the assumptions and implications of standard no-go theorems for realist representations of operational quantum theory, namely, those  based on Bell's notion of locality and those based on generalized noncontextuality.  Moreover, our process-theoretic characterization of generalized noncontextuality is shown to be implied by an even more natural principle which we term {\em Leibnizianity}. Most strikingly, our framework offers a way forward in a research program  that seeks to circumvent these no-go results.  Specifically, we argue that if one can identify axioms for a realist causal-inferential theory such that the notions of causation and inference can differ from their conventional (classical) interpretations, then one has the means of defining an intrinsically quantum notion of realism, and thereby a realist representation of operational quantum theory that salvages the spirit of locality and of noncontextuality.
\end{abstract}
\maketitle
\tableofcontents

\vspace{3mm}
\begin{tcolorbox}
We recommend that the reader watch \href{http://pirsa.org/displayFlash.php?id=20110051}{http://pirsa.org/displayFlash.php?id=20110051} while reading this paper.
\end{tcolorbox}

\section{Introduction}

One of the key disagreements among quantum researchers is the question of which elements of the quantum formalism refer to ontological concepts and which refer to epistemological concepts.  The importance of settling this issue was famously noted by E.T. Jaynes~\cite{Jaynesquote}:
\begin{quote}
[O]ur present [quantum mechanical] formalism is not purely
epistemological; it is a peculiar mixture
describing in part realities of Nature, in part
incomplete human information about Nature
--- all scrambled up by Heisenberg and Bohr
into an omelette that nobody has seen how
to unscramble. Yet we think that the
unscrambling is a prerequisite for any further
advance in basic physical theory. For, if we
cannot separate the subjective and objective
aspects of the formalism, we cannot know
what we are talking about; it is just that
simple.
\end{quote}

In our view, the most constructive way of defining `realities of Nature' is as causal mechanisms acting on causal relata.  
 That is, we here take an account of an operational phenomenon to be {\em realist} if it secures a {\em causal explanation} of that phenomenon. 
Hence, the particular omelette of  ontology and epistemology that we will be endeavouring to unscramble in this work is the one that results from the mixing up of the concepts of {\em epistemic inference} on the one hand, and of {\em causal influence} on the other.

Scrambling of this sort is not unique to the quantum formalism---it arises also in the standard formalism for classical statistics.  In that context, the difference can be characterized as follows: Bayesian inference stipulates how {\em learning} the value of one variable allows an agent to update their information about the value of another, while causal influence stipulates how the value of one variable {\em determines} the value of another (with a consequence being that an agent who {\em controls} the first variable can come to have some control over the second).  
Despite the apparent clarity of the distinction, it is often challenging to disentangle the two concepts.  The statistical phenomena known as `Simpson's paradox'~\cite{simpsonsparadox} and `Berkson's paradox'~\cite{berksonparadox}, for example, have the appearance of paradoxes
 {\em precisely because of} our tendency to inappropriately slide from statements about conditional probabilities (which merely support inferences) to statements about cause-effect relations.
A satisfactory understanding of these phenomena was only found after the development of the mathematical framework of causal modeling~\cite{spirtes2000causation,pearl2009causality}
 that incorporated certain
 formal distinctions between inference and influence which are absent in the standard framework for statistical reasoning.

The conceptual difficulty of disentangling influence and inference is only compounded in the quantum realm, where the interpretation of the elements of the mathematical formalism is even less clear than it is in classical statistics.

The current article takes up this challenge more broadly, by pursuing the unscrambling project for
  two mathematical frameworks that have been used extensively in attempts to understand the conceptual content of quantum theory.  The first is the framework of {\em operational theories}, which aims to clarify what is distinctive about quantum theory by situating it in a landscape of other possible theories, all characterized in a minimalist way in terms of their operational predictions.  The second is the framework of {\em realist theories} (including ontological models), which has been used to constrain the possibilities for
   causal explanations of the operational predictions of quantum theory (and other operational theories).

We aim to recast both types of theory within a new mathematical framework
that incorporates
a formal distinction between inference and influence---
a distinction which is lacking in previous frameworks.

A theory in our framework is termed a {\em causal-inferential theory}, and is constructed out of two components:
\begin{itemize}
\item a causal theory, which describes physical systems in the world and the causal mechanisms that relate them, and
\item an inferential theory, which describes an agent's beliefs about these systems and about the causal mechanisms that relate them, as well as how such beliefs are
 updated under the acquisition of new information.
\end{itemize}
The full causal-inferential theory is defined by the interplay between these two components, and allows one to describe a physical scenario in a manner which cleanly distinguishes causal and inferential aspects.

Different causal-inferential theories can be obtained by varying the causal theory, varying the inferential theory, or varying both simultaneously.
Note, however, that these two components are required to interact in a particular manner, so that the choice of one may be limited by the choice of the other. In this work, we explore, in detail, two particular choices of the causal theory and a single choice of inferential theory, as we now outline.

We take the inferential theory to be Bayesian probability theory combined with Boolean logic.
Although we do not explicitly construct any alternatives to this choice in this article, we will have much to say about the possiblity of {\em nonclassical} alternatives to this inferential theory.
  Given that such a putative nonclassical inferential theory is the primary contrast class for us, we will refer to the inferential theory consisting of Bayesian probability theory and Boolean logic as the {\em classical} theory of inference.

The two types of causal theory that we consider correspond to operational and realist theories, respectively.
In the first type, systems are conceptualized as the causal inputs and causal outputs of experimental procedures,
and the causal mechanisms holding between such systems are simply descriptions of these experimental procedures.
The causal-inferential theories that one can construct from this causal theory together with the classical inferential theory are called \emph{operational causal-inferential theories}, and can be viewed as a refinement of the notion of operational theory described in Ref.~\cite{Spe05} and as a competitor
 to the framework of  Operational Probabilistic Theories~\cite{chiribella2010probabilistic}.

  In the second type of causal theory we consider, systems are classical variables and the causal mechanisms holding between these are functions.
  The (unique) causal-inferential theory that we construct from this is termed a {\em classical realist causal-inferential theory} and is a refinement of the notion of a  {\em structural equation model} in the field of causal inference~\cite{pearl2009causality}.\footnote{Although such frameworks achieved significantly more unscrambling of the causal-inferential omelette than the statistical frameworks that preceded them---as noted above in our discussion of statistical `paradoxes'---our novel framework achieves some further unscrambling.}

In order to make a connection to other standard notions of operational and realist theories, it is useful to introduce a notion of {\em inferential equivalence}.  Two processes are said to be inferentially equivalent if they lead one to make the same inferences whatever scenario they might be embedded within. If one quotients a causal-inferential theory with respect to the congruence relation associated to inferential equivalence, one obtains a novel type of theory, which we term a {\em quotiented} causal-inferential theory.  Importantly, the latter sort of theory
blurs the distinction between causation and inference, and hence constitutes a partial rescrambling of the causal-inferential omelette.
 {\em Generalized probabilistic theories} (GPTs)~\cite{Hardy,barrett2007,hardy2011reformulating}, we argue,  are best understood as subtheories of quotiented operational causal-inference theories\footnote{This view of GPTs as quotiented operational causal-inferential theories is closely related to the quotienting of operational probabilistic theories of \cite{chiribella2010probabilistic}.}
   and consequently, unlike {\em unquotiented} operational causal-inferential theories,
   they necessarily involve some scrambling of causal and inferential concepts.
We also show that {\em ontological models}~\cite{Spe05} (or, more precisely, the ontological theories that are the codomain of ontological modelling maps) are best understood as subtheories of quotiented classical realist causal-inferential theories,
 and hence are {\em also} guilty of such scrambling.

Our framework leverages the mathematics of process theories~\cite{coecke2018picturing,gogioso2017fantastic,selby2018reconstructing}, which allows it to be manifestly compositional, and consequently to apply to operational or realist scenarios with arbitrarily complex causal and inferential structure.
Many previous frameworks for operational theories~\cite{chiribella2010probabilistic,hardy2011reformulating} have also availed themselves of the mathematics of process theories to allow the representation of arbitrarily complex structures. These did not, however, explicitly distinguish the structures that are causal and those that are inferential, as we do here.  Our use of the mathematics of process theories represents more of an innovation on the realist  side, since previous frameworks for realist theories focused almost exclusively on the simple structures that arise when describing prepare-measure scenarios, sometimes with an intervening transformation or sequence of transformations~\cite{lillystone2019single,PP1,PP2,AWV,AWVrobust,Lostaglio2020contextualadvantage}. (These frameworks also did not distinguish structures that are causal from those that are inferential.)

One of the motivations for the standard framework for ontological models was to answer the question of whether the predictions of a given operational theory admit of an explanation in terms of an underlying ontology.
The counterpart of this question in our new framework is whether the predictions of a given {\em operational causal-inferential theory} admit of an explanation in terms of an underlying {\em \crealist  causal-inferential} theory.
Such an explanation is deemed possible if the former can be {\em represented} in terms of the latter.  We refer to this as a {\em \crealist representation} of an operational causal-inferential theory.

The key constraint we impose on such representations is that they preserve the causal and inferential structures encoded in the diagrams, a property that is formalized by demanding that the map between the two process theories (operational and  \crealist) is {\em diagram-preserving}~\cite{schmid2020structure}.  
We show that this constraint
 involves no loss of generality in terms of the sorts of realist models of experimental phenomena one can describe in the framework. 
 Moreover, in concert with standard hypotheses about the causal and inferential structure, it acts as an umbrella principle
 which subsumes many principles that have previously been used to derive no-go theorems for  \crealist  representations of operational quantum theory~\cite{schmid2020structure}.

In particular, in the case of a Bell scenario, the assumption of diagram preservation subsumes the causal and inferential assumptions that go into deriving Bell inequalities (when this derivation is conceptualized in terms of causal modeling~\cite{wood2015lesson}).  However, it is much more general than this, and subsumes the causal and inferential assumptions that go into deriving Bell-like inequalities (also known as {\em causal compatibility inequalities}) for scenarios that have a causal structure distinct from the Bell scenario~\cite{pearl2009causality,fritz2012beyond,wood2015lesson,inflation,triangle,Chaves2018}.
 Our framework therefore also constitutes a refinement of (or alternative to) recently proposed
 frameworks~\cite{pearl2009causality,inflation} for identifying causal compatibility constraints in such scenarios.

We also demonstrate how a principle proposed by Leibniz and used extensively by Einstein~\cite{SpekLeibniz19} can be  generalized in a natural way to apply to theories incorporating epistemological claims in addition to ontological claims and that this principle implies a formal 
  constraint on realist representations of an operational causal-inferential theory  that we term {\em Leibnizianity}: the representation must  {\em preserve inferential equivalences}.
We also demonstrate that the principle of Leibnizianity implies a rehabilitated version of the principle of generalized noncontextuality~\cite{Spe05}, such that no-go theorems for generalized-noncontextual classical realist representations of operational quantum theory 
 imply no-go theorems for Leibnizian classical realist representations.  The question of whether the reverse implication holds remains open.
We also discuss the connections between this principle and the old version of generalized noncontextuality.

Bell's no-go theorem is understood in our framework as follows. If quantum theory is cast as an operational causal-inferential theory, then it predicts distributions for certain 
 causal-inferential 
 structures that cannot be realized by a classical realist causal-inferential theory with the same 
  causal-inferential 
  structure.  Meanwhile, noncontextuality no-go theorems are understood in our framework similarly, but where one demands that inferential equivalences {\em as well as} the 
   causal-inferential structure are preserved.

The conventional realist responses to the standard no-go theorems are unsatisfactory in various ways, such as requiring superluminal causal influences, requiring fine-tuning, and running afoul of Leibniz's principle.  In light of this, it has been suggested that a more satisfactory way out of these no-go theorems may be achieved by modifying the notion of a realist representation (see, e.g., Sec. 1.C of Ref.~\cite{spekkens2016quasi}).  This has been described in past work as `going beyond the standard ontological models framework', but here is understood as seeking a nonclassical generalization of the notion of a classical realist representation.   Our process-theoretic framework
 provides the formal means of achieving this because it allows the interpretation of causal and inferential concepts to be determined by the axioms of the process theories that describe them and hence to differ from the conventional, classical interpretations of these concepts.  This is analogous to how, in nonEuclidean geometries, the concepts of point and line acquire novel meanings distinct from their conventional ones.  Success in such a research program consists in finding an intrinsically quantum notion of a realist causal-inferential theory which can provide a Leibnizian representation of operational quantum theory.   We propose natural constraints on the axioms describing a theory of causal influences, a theory of epistemic inferences, and their interactions.  We also point to pre-existing work that offer clues for how to proceed.

 Thus, the work we present here provides a significant step forward in this research program.
  On the one hand, it provides, for the first time, a concrete proposal for the mathematical form of the sought-after theory,  and, on the other hand, it provides a set of ideas for the form of its axioms, thereby providing  a road map for future research.

\subsection{Process theories and diagram-preserving maps} \label{sec:PTs}

We formalize the ideas discussed in the introduction using the mathematical language of process theories and diagram-preserving maps. This section serves as a brief introduction to these concepts.
Process theories provide a mathematical framework for describing an extremely broad class of theories, finding utility both in physics~\cite{gogioso2017categorical,selby2018reconstructing} and beyond~\cite{CSC,bankova2016graded,MarthaC}.
They can ultimately be given a category-theoretic foundation, but in this paper we will require only the diagrammatic approach. We do provide some ``category-theoretic remarks'' when a given definition or result can be expressed concisely in this language, but these remarks can be skipped without impacting the comprehensibility of the rest of the article. 

As the name would suggest, the basic building blocks of a process theory are processes. These could correspond to physical processes in the world, but equally well could apply to computational processes, mathematical processes, etc. In our work, we will focus on causal and inferential processes.

\begin{definition}[Processes]
A \emph{process} is defined as a labeled box with labeled input and output \emph{systems}, e.g.:
\beq%
\InputIfFileExists{Diagrams/process.tikz}{}{\input{./figures/Diagrams/process.tikz}}.\eeq
The label of a system, e.g., $A$, is known as the \emph{type} of the system while the label on the box, e.g., $u$, is simply the \emph{name} of the process.

Note that it is allowed for a process to have no inputs systems and/or no output systems. Processes with no inputs are called \emph{states}, those with no outputs are called \emph{effects}, and those with neither inputs nor outputs are  simply called \emph{closed diagrams} (or sometimes \emph{scalars}). 
That is,
\beq
\InputIfFileExists{Diagrams/bipartiteState.tikz}{}{\input{./figures/Diagrams/bipartiteState.tikz}}\ \text{,}\quad %
\InputIfFileExists{Diagrams/tripartiteEffect.tikz}{}{\input{./figures/Diagrams/tripartiteEffect.tikz}} \text{ , and}\quad %
\begin{tikzpicture}
	\begin{pgfonlayer}{nodelayer}
		\node [style=scalar] (0) at (0, -0) {{\color{white}*}$r${\color{white}*}};
	\end{pgfonlayer}
\end{tikzpicture}
}
\eeq
are examples of a state, an effect and a closed diagram 
 respectively.
\end{definition}

Before we can define a process theory we must introduce the idea of a diagram of processes.
\begin{definition}[Diagrams]
A \emph{diagram} is defined as a `wiring together' of a finite set of processes---that is, an output of one process is connected to the input of another, such that the system types match and no cycles are created. For example,
\beq%
\InputIfFileExists{Diagrams/diagram.tikz}{}{\input{./figures/Diagrams/diagram.tikz}}.\eeq

Only the \emph{connectivity} of the diagram---
which systems are wired together and the ordering of the input and output systems---is relevant.
That is, two diagrams are the same if one can be deformed into the other while preserving this connectivity. For example,
\beq%
\InputIfFileExists{Diagrams/DiagrammaticEquality1.tikz}{}{\input{./figures/Diagrams/DiagrammaticEquality1.tikz}}\quad\equiv\quad%
\InputIfFileExists{Diagrams/DiagrammaticEquality2.tikz}{}{\input{./figures/Diagrams/DiagrammaticEquality2.tikz}}.\eeq
\end{definition}

We will now formally define what we mean by a process theory.

\begin{definition}[Process theories]
A process theory  is defined as a collection of processes, $\textsc{T}$, which is closed under forming diagrams. For example, we can draw a box around the above diagram and view it as another process in the theory. That is,
\beq%
\InputIfFileExists{Diagrams/diagram2.tikz}{}{\input{./figures/Diagrams/diagram2.tikz}}\quad\in\ \textsc{T}\eeq
for $u,v,w \in \textsc{T}$.
\end{definition}

This completes the definition of a process theory. However, it is sometimes useful to introduce elementary notions from which any diagram can be built up. To start, we highlight certain elements of the diagrams by picking them out with dashed boxes below:
\beq \label{swapetc}
\InputIfFileExists{Diagrams/extraElements.tikz}{}{\input{./figures/Diagrams/extraElements.tikz}}\quad =:\quad %
\InputIfFileExists{Diagrams/extraElementsDef.tikz}{}{\input{./figures/Diagrams/extraElementsDef.tikz}}.
\eeq
That is: i) one can view the empty box on the left as a  special closed diagram 
(as it has no input and no output) which we refer to as the scalar $1$; ii) one can view the box with the $A$ wire running through it as an identity process $\mathds{1}_A$; iii) one can view the box with the crossed $A$ and $B$ wires as a swap process $\mathds{S}_{BA}$; and iv) one can view the output system of $v$ as the trivial system $I$. Clearly our diagrammatic notation implies constraints on these extra elements. In particular, wiring identity processes onto any other process leaves that process invariant; the composite of a trivial system with any other system is just that system; swapping twice is the identity on the two systems; and finally, composing a process with the scalar $1$ leaves that process invariant. Together with the elements just introduced, one can introduce two elementary notions of composition from which any diagram can be built up: sequential composition of processes, denoted
\beq
\InputIfFileExists{Diagrams/sequentialComp2.tikz}{}{\input{./figures/Diagrams/sequentialComp2.tikz}}\ \circ\ %
\InputIfFileExists{Diagrams/sequentialComp1.tikz}{}{\input{./figures/Diagrams/sequentialComp1.tikz}} \quad := \quad %
\InputIfFileExists{Diagrams/sequentialComp3.tikz}{}{\input{./figures/Diagrams/sequentialComp3.tikz}} \ ,
\eeq
and parallel composition of processes, denoted
\beq
\InputIfFileExists{Diagrams/parallelComp2.tikz}{}{\input{./figures/Diagrams/parallelComp2.tikz}}\ \otimes \ %
\InputIfFileExists{Diagrams/parallelComp1.tikz}{}{\input{./figures/Diagrams/parallelComp1.tikz}} \quad := \quad %
\InputIfFileExists{Diagrams/parallelComp3.tikz}{}{\input{./figures/Diagrams/parallelComp3.tikz}}\ .
\eeq
Note that, given any diagram, there are a (generally infinite) number of ways in which it can be expressed in terms of these primitive notions of composition, and yet these are all the same diagram. Hence, we view the diagrammatic representation as being the fundamental description, and we view the elementary notions from which they can be built as an (at times) convenient mathematical representation of them.

\begin{remark}
Having defined these extra structures implicitly in the diagrammatic notation, it should be clear how to identify the structure of a process theory with that of a \emph{strict symmetric monoidal category} (SMC). In short, we take processes to be morphisms and systems to be objects, with sequential and parallel composition providing morphism composition and the monoidal product, respectively. For a more formal treatment, see Ref.~\cite{patterson2021wiring}. 
\end{remark}

We will often consider higher-order processes such as
\beq
\begin{tikzpicture}
	\begin{pgfonlayer}{nodelayer}
		\node [style=none] (0) at (0, 1.25) {};
		\node [style=none] (1) at (0, -1.25) {};
		\node [style=none] (2) at (-0.5, 1.25) {};
		\node [style=none] (3) at (-0.5, 1.75) {};
		\node [style=none] (4) at (1.75, 1.75) {};
		\node [style=none] (5) at (1.75, -1.75) {};
		\node [style=none] (6) at (-0.5, -1.75) {};
		\node [style=none] (7) at (-0.5, -1.25) {};
		\node [style=none] (8) at (0.75, 1.25) {};
		\node [style=none] (9) at (0.75, -1.25) {};
		\node [style=none] (10) at (1.25, 0) {$\tau$};
		\node [style={right label}] (11) at (0, -1) {$A$};
		\node [style={right label}] (12) at (0, 0.75) {$B$};
		\node [style=none] (13) at (0, 0.5000002) {};
		\node [style=none] (14) at (0, -0.7499998) {};
		\node [style={right label}] (15) at (1.25, -2) {$A'$};
		\node [style=none] (16) at (1.25, -1.75) {};
		\node [style=none] (17) at (1.25, 1.75) {};
		\node [style={right label}] (18) at (1.25, 2) {$B'$};
		\node [style=none] (19) at (1.25, 2.5) {};
		\node [style=none] (20) at (1.25, -2.25) {};
	\end{pgfonlayer}
	\begin{pgfonlayer}{edgelayer}
		\draw (6.center) to (5.center);
		\draw (5.center) to (4.center);
		\draw (4.center) to (3.center);
		\draw (3.center) to (2.center);
		\draw (2.center) to (8.center);
		\draw (8.center) to (9.center);
		\draw (9.center) to (7.center);
		\draw (7.center) to (6.center);
		\draw [qWire] (0.center) to (13.center);
		\draw [qWire] (14.center) to (1.center);
		\draw [qWire] (19.center) to (17.center);
		\draw [qWire] (16.center) to (20.center);
	\end{pgfonlayer}
\end{tikzpicture},
\eeq
which we will call \emph{clamps}. These can be thought of as objects which map a process from $A$ to $B$ to a process from $A'$ to $B'$  via
\beq
\begin{tikzpicture}
	\begin{pgfonlayer}{nodelayer}
		\node [style=small box] (0) at (0, 0) {$T$};
		\node [style=none] (1) at (0, 1.25) {};
		\node [style=none] (2) at (0, -1.25) {};
		\node [style=right label] (12) at (0, -1) {$A$};
		\node [style=right label] (13) at (0, 0.75) {$B$};
	\end{pgfonlayer}
	\begin{pgfonlayer}{edgelayer}
		\draw [style=qWire] (0) to (1.center);
		\draw [style=qWire] (0) to (2.center);
	\end{pgfonlayer}
\end{tikzpicture}
\quad\mapsto\quad
\begin{tikzpicture}
	\begin{pgfonlayer}{nodelayer}
		\node [style=none] (0) at (0, 1.25) {};
		\node [style=none] (1) at (0, -1.25) {};
		\node [style=none] (2) at (-0.5, 1.25) {};
		\node [style=none] (3) at (-0.5, 1.75) {};
		\node [style=none] (4) at (1.75, 1.75) {};
		\node [style=none] (5) at (1.75, -1.75) {};
		\node [style=none] (6) at (-0.5, -1.75) {};
		\node [style=none] (7) at (-0.5, -1.25) {};
		\node [style=none] (8) at (0.75, 1.25) {};
		\node [style=none] (9) at (0.75, -1.25) {};
		\node [style=none] (10) at (1.25, 0) {$\tau$};
		\node [style={right label}] (11) at (0, -1) {$A$};
		\node [style={right label}] (12) at (0, 0.75) {$B$};
		\node [style=none] (13) at (0, -0) {};
		\node [style=none] (14) at (0, -0) {};
		\node [style={right label}] (15) at (1.25, -2) {$A'$};
		\node [style=none] (16) at (1.25, -1.75) {};
		\node [style=none] (17) at (1.25, 1.75) {};
		\node [style={right label}] (18) at (1.25, 2) {$B'$};
		\node [style=none] (19) at (1.25, 2.5) {};
		\node [style=none] (20) at (1.25, -2.25) {};
		\node [style={small box}] (21) at (0, -0) {$T$};
	\end{pgfonlayer}
	\begin{pgfonlayer}{edgelayer}
		\draw (6.center) to (5.center);
		\draw (5.center) to (4.center);
		\draw (4.center) to (3.center);
		\draw (3.center) to (2.center);
		\draw (2.center) to (8.center);
		\draw (8.center) to (9.center);
		\draw (9.center) to (7.center);
		\draw (7.center) to (6.center);
		\draw [qWire] (0.center) to (13.center);
		\draw [qWire] (14.center) to (1.center);
		\draw [qWire] (19.center) to (17.center);
		\draw [qWire] (16.center) to (20.center);
	\end{pgfonlayer}
\end{tikzpicture}.
\eeq
These are not primitive notions within the framework of process theories, but rather are constructed out of two processes $x_\tau$ and $y_\tau$ connected together with an auxiliary system $W_\tau$, as
\beq\label{eq:tester}
\begin{tikzpicture}
	\begin{pgfonlayer}{nodelayer}
		\node [style=none] (0) at (0, 1.25) {};
		\node [style=none] (1) at (0, -1.25) {};
		\node [style=none] (2) at (-0.5, 1.25) {};
		\node [style=none] (3) at (-0.5, 1.75) {};
		\node [style=none] (4) at (1.75, 1.75) {};
		\node [style=none] (5) at (1.75, -1.75) {};
		\node [style=none] (6) at (-0.5, -1.75) {};
		\node [style=none] (7) at (-0.5, -1.25) {};
		\node [style=none] (8) at (0.75, 1.25) {};
		\node [style=none] (9) at (0.75, -1.25) {};
		\node [style=none] (10) at (1.25, 0) {$\tau$};
		\node [style={right label}] (11) at (0, -1) {$A$};
		\node [style={right label}] (12) at (0, 0.75) {$B$};
		\node [style=none] (13) at (0, 0.5000002) {};
		\node [style=none] (14) at (0, -0.7499998) {};
		\node [style={right label}] (15) at (1.25, -2) {$A'$};
		\node [style=none] (16) at (1.25, -1.75) {};
		\node [style=none] (17) at (1.25, 1.75) {};
		\node [style={right label}] (18) at (1.25, 2) {$B'$};
		\node [style=none] (19) at (1.25, 2.5) {};
		\node [style=none] (20) at (1.25, -2.25) {};
	\end{pgfonlayer}
	\begin{pgfonlayer}{edgelayer}
		\draw (6.center) to (5.center);
		\draw (5.center) to (4.center);
		\draw (4.center) to (3.center);
		\draw (3.center) to (2.center);
		\draw (2.center) to (8.center);
		\draw (8.center) to (9.center);
		\draw (9.center) to (7.center);
		\draw (7.center) to (6.center);
		\draw [qWire] (0.center) to (13.center);
		\draw [qWire] (14.center) to (1.center);
		\draw [qWire] (19.center) to (17.center);
		\draw [qWire] (16.center) to (20.center);
	\end{pgfonlayer}
\end{tikzpicture}
\ = \
\begin{tikzpicture}
	\begin{pgfonlayer}{nodelayer}
		\node [style=none] (0) at (0.2500001, 1.25) {};
		\node [style=none] (1) at (0.2500001, -1.25) {};
		\node [style={right label}] (2) at (0.2500001, -0.9999999) {$A$};
		\node [style={right label}] (3) at (0.2500001, 0.7499998) {$B$};
		\node [style=none] (4) at (0.2500001, 0.5000002) {};
		\node [style=none] (5) at (0.2500001, -0.7499998) {};
		\node [style=none] (6) at (-0.5000002, -1.25) {};
		\node [style=none] (7) at (2.5, -1.25) {};
		\node [style=none] (8) at (2.5, 1.25) {};
		\node [style=none] (9) at (-0.5000002, 1.25) {};
		\node [style=none] (10) at (2, 1.25) {};
		\node [style=none] (11) at (2, -1.25) {};
		\node [style=none] (12) at (1.25, -1.75) {$x_\tau$};
		\node [style=none] (13) at (1.25, 1.75) {$y_\tau$};
		\node [style={right label}] (14) at (2, -0) {$W_\tau$};
		\node [style=none] (15) at (1.5, 2.25) {};
		\node [style={right label}] (16) at (1.5, -2.75) {$A'$};
		\node [style=none] (17) at (1.5, -3) {};
		\node [style=none] (18) at (1.5, 3) {};
		\node [style={right label}] (19) at (1.5, 2.75) {$B'$};
		\node [style=none] (20) at (1.5, -2.25) {};
		\node [style=none] (21) at (0.5000001, 2.25) {};
		\node [style=none] (22) at (2, 2.25) {};
		\node [style=none] (23) at (2, -2.25) {};
		\node [style=none] (24) at (0.5000001, -2.25) {};
	\end{pgfonlayer}
	\begin{pgfonlayer}{edgelayer}
		\draw [qWire] (0.center) to (4.center);
		\draw [qWire] (5.center) to (1.center);
		\draw (9.center) to (8.center);
		\draw (7.center) to (6.center);
		\draw [qWire] (10.center) to (11.center);
		\draw [qWire] (18.center) to (15.center);
		\draw [qWire] (20.center) to (17.center);
		\draw (9.center) to (21.center);
		\draw (21.center) to (22.center);
		\draw (22.center) to (8.center);
		\draw (6.center) to (24.center);
		\draw (24.center) to (23.center);
		\draw (23.center) to (7.center);
	\end{pgfonlayer}
\end{tikzpicture}.
\eeq

 We will consider a number of different process theories, some of which are sub-process-theories of others.
\begin{definition}
A sub-process-theory $T \subseteq T'$ is a process theory where the processes are a subset of the processes in $T'$ and
composition of processes in $T$ is given by composition in $T'$.
Note that since a sub-process-theory is a process theory, $T$
must be closed under composition.
\end{definition}
\begin{remark}
In terms of the associated SMCs, this is simply defining $T$ as a subSMC of $T'$.
\end{remark}

As well as these process theories and sub-process-theories, we will also consider
structure-preserving maps between these. The relevant structure which we demand be preserved is the composition of processes as described by diagrams.

\begin{definition}[Diagram-preserving maps]
A diagram-preserving map, \colorbox{green!40}{$\mathbf{m}:\textsc{T}\to \textsc{T}'$}, is a map from processes in $\textsc{T}$ to processes in $\textsc{T}'$
 such that wiring together processes in $\textsc{T}$ to form a diagram and then applying the map $\mathbf{m}$ is the same as applying $\mathbf{m}$ to each of the component processes and then wiring them together in $\textsc{T}'$.
We depict these maps as shaded boxes, e.g.
\beq%
\InputIfFileExists{Diagrams/diagramMap1.tikz}{}{\input{./figures/Diagrams/diagramMap1.tikz}},\eeq
where the diagram in the green box is a diagram in $\textsc{T}$ (with input $A$ and outputs $B$ and $A$) which is mapped by the green box $\mathbf{m}$ to a process in $\textsc{T}'$ with input $\mathbf{m}_A$ and outputs $\mathbf{m}_A$ and $\mathbf{m}_B$.
In this example, the constraint of diagram preservation is simply that
\beq%
\InputIfFileExists{Diagrams/diagramMap1.tikz}{}{\input{./figures/Diagrams/diagramMap1.tikz}}\quad=\quad%
\InputIfFileExists{Diagrams/diagramMap2.tikz}{}{\input{./figures/Diagrams/diagramMap2.tikz}}.\eeq
\end{definition}

\begin{remark}
If one views each process theory as an SMC, then such diagram-preserving maps are simply \emph{strict symmetric monoidal functors} between the SMCs. The diagrammatic notation as shaded regions was introduced to us by Ref.~\cite{fritz2018bimonoidal}, which was itself was based on Ref.~\cite{mellies2006functorial}. 
\end{remark}

It will also be useful to consider partial diagram preserving (DP) maps where the domain is limited in scope.

\begin{definition}[Partial diagram-preserving maps]
 A partial diagram-preserving map \colorbox{green!40}{$\mathbf{m}:\textsc{T}'\to \textsc{T}''$} is a diagram-preserving map from some sub-process-theory $T\subseteq \textsc{T}'$ to $\textsc{T}''$.
\end{definition}

\begin{remark}
Categorically, such a map is a partial strict symmetric monoidal functor between the relevant SMCs.
\end{remark}

\begin{remark}[Category of process theories, $\textsc{ProcessTheory}$]
The category of process theories is defined as follows:
The objects of $\textsc{ProcessTheory}$ are process theories and the morphisms are diagram-preserving maps. It is simple to see that this is indeed a category as one can easily define identity morphisms and morphism composition satisfying the relevant conditions.
\end{remark}

\section{Causal primitives}\label{sec:causal}

We denote a generic process theory of causal relations 
 by $\Caus$.
The primitive elements of such a theory 
 are systems and the causal mechanisms that relate them.
Systems correspond to physical degrees of freedom in the conventional sense of being the loci of causal relations, i.e., the causal relata. 
Causal mechanisms are
autonomous physical relationships between these systems, relationships governed by the laws of nature and by the arrangement of relevant physical systems and apparatuses.

In classical theories, systems are often represented by sets, and causal mechanisms by functions between these (e.g., in structural equation models~\cite{pearl2009causality}). In quantum theory, systems are typically represented by Hilbert spaces, and causal mechanisms by unitaries between these~\cite{allen2017quantum}\footnote{ Although how to decompose a given unitary with multiple outputs and a given internal causal structure remains an open problem~\cite{allen2017quantum,Lorenz2020}.}. In  operational causal-inferential (CI) theories, one typically does not have a direct description of the causal mechanisms, but rather only a very coarse-grained description of them
 in terms of laboratory procedures that are implemented on the relevant systems; these systems are represented only as an abstract label, and typically represent the physical systems one imagines are input and output from the apparatuses.

In the next two sections, we will consider two distinct classes of causal primitives in more detail, first those relevant for operational CI theories, and then those relevant for \crealist CI theories. In Section~\ref{realconstr}, we return to the question of what properties any process theory must satisfy for it to be considered a good theory of causal relations.

\subsection{Process theories of laboratory procedures}

We now define the sort of process theory which will ultimately  constitute the causal component
of an operational CI theory. We denote it $\Proc$ (as shorthand for `procedure' not `process').
The systems in $\Proc$ label the primitive causal relata, while the processes, which describe the potential causal relations between them, are labeled by
laboratory procedures, conceptualized as a list of instructions of what to do in the lab, and presumed to be individuated by the system they act on (the input system) and the system they prepare (the output system).

We will label general systems by $\op{A}$, $\op{B}$, etc. Some systems (which could, for instance, represent setting or outcome variables) will be deemed {\em classical}. A classical system $\op{X}$ will be associated with a set $X$ which represents the set of distinct states of the classical system. Diagrammatically, a general laboratory procedure $\op{t}$ with input system $\op{A}$ and output system $\op{B}$ will be drawn as
\beq%
\InputIfFileExists{Diagrams/procedureAB.tikz}{}{\input{./figures/Diagrams/procedureAB.tikz}}.\eeq
We define a measurement $\op{m}$ on system $\op{A}$ as a process with a generic input system $\op{A}$ and a classical output system $\op{X}$:
\beq%
\InputIfFileExists{Diagrams/procedureAn.tikz}{}{\input{./figures/Diagrams/procedureAn.tikz}}.\eeq
Classical systems in $\Proc$ will be drawn with a light grey wire, as was done in this diagram.
We denote the set of operations with input system $\op{A}$ and output system $\op{B}$ as $\morph{\op{A}}{\op{B}}$. The set of measurements on a system $\op{A}$ having outcome space $X$ is therefore denoted $\morph{\op{A}}{X}$. (Classical systems allow one to describe more than just measurement outcomes; for instance, they can also represent classical control systems.)

One can compose these operations to describe experiments. For example, a preparation procedure on system $\op{A}$ followed by a measurement on $\op{A}$ with outcome space $X$ is described by the diagram
\beq%
\InputIfFileExists{Diagrams/procedurePrepMeasScenario.tikz}{}{\input{./figures/Diagrams/procedurePrepMeasScenario.tikz}} \ , \eeq
while controlling a transformation from $\op{B}$ to $\op{C}$  on the output $\op{X}$ of a measurement on $\op{A}$
would be described by the diagram
\beq%
\InputIfFileExists{Diagrams/measurementControlledProcedure.tikz}{}{\input{./figures/Diagrams/measurementControlledProcedure.tikz}}.\eeq
An example of a more general diagram is
\beq%
\InputIfFileExists{Diagrams/diagramSmall.tikz}{}{\input{./figures/Diagrams/diagramSmall.tikz}}.\eeq
In our formalism, such a diagram represents a hypothesis about the fundamental causal structure. In Appendix~\ref{manifestorcausal2} we discuss the consequences of this choice and how it differs from the choice typically made in operational frameworks.
Here, two wires in parallel should be interpreted as independent subsystems (e.g., independent degrees of freedom), where one can talk independently about either subsystem as a potential causal influence on other systems.

 Note that we have here defined a {\em class} of process theories, insofar as we have not specified the {\em particular} set of systems and procedures that define $\Proc$.  Perhaps the most common operational theory to consider is that containing all known physical systems and laboratory procedures on them. One might also consider a restriction of this set, for example, the set of two-level systems and laboratory procedures on them.  Finally, one might consider a foil operational theory~\cite{chiribella2016quantum}, with a set of hypothetical systems and procedures on them.  Each possible choice for $\Proc$ defines a different operational CI theory.

\begin{remark}
Unlike the other process theories that we deal with in this paper, $\Proc$ is a \emph{free} process theory. This means that there are no equalities other than those defined by the framework of process theories---two diagrams are equal if and only if they can be transformed into one another by sliding the processes around on the page while preserving the wiring.
\end{remark}

\subsection{Process theory of classical functional dynamics} \label{funcdyn}

We now define the process theory which will ultimately act as the causal component of our notion of a \crealist causal-inferential theory. We denote it $\Func$.

The systems in $\Func$ again label the primitive causal relata. However, what distinguishes them from the systems in $\Proc$ is that we assume that these relata are described by ontic state spaces, that is, some finite\footnote{This assumption of finiteness is made for simplicity of presentation, but could be removed in future work.
} sets $\Lambda, \Lambda',...$. The processes in $\Func$ are functions $f:\Lambda \to \Lambda'$ describing dynamics on these ontic state spaces.\footnote{Note that we will allow arbitrary functions in this work, although in some cases one might wish to restrict the dynamics, e.g., to reversible functions or to symplectomorphisms.}

Diagrammatically, a function $f$ with input $\Lambda$ and output $\Lambda'$ will be drawn as
\beq%
\InputIfFileExists{Diagrams/funcAB.tikz}{}{\input{./figures/Diagrams/funcAB.tikz}} \ .\eeq

We denote the set of functions with input $\Lambda$ and output $\Lambda'$ as $\morph{\Lambda}{\Lambda'}$.  We take the trivial system to be the singleton set $\star = \{*\}$. Hence, states correspond to functions $s:\star \to \Lambda$ which are in one-to-one correspondence with the elements of $\Lambda$; there is a unique effect $u:\Lambda \to \star$ for each system, defined by $u(\lambda)=*$ for all $\lambda\in \Lambda$; and, there is a unique scalar $1:\star \to \star$ corresponding to the identity function on the singleton set.

We can compose these functions to describe ontological scenarios. A function which prepares some ontic state of system $\Lambda$ followed by a function describing the functional dynamics of the system
is described by the diagram
\beq\label{CompositionExample1}%
\InputIfFileExists{Diagrams/funcPrepMeasScenario.tikz}{}{\input{./figures/Diagrams/funcPrepMeasScenario.tikz}},\eeq
while some more general ontological scenario could be described by the diagram
\beq\label{CompositionExample2}%
\InputIfFileExists{Diagrams/diagramSmallOnt.tikz}{}{\input{./figures/Diagrams/diagramSmallOnt.tikz}} \ .\eeq

The key formal distinction between $\Func$ and $\Proc$ is that $\Func$ is {\em not} a free process theory. Indeed, there are many nontrivial equalities provided by composition of functions. An example is provided by Eq.~\eqref{CompositionExample1}, the diagram, which is made up of two functions $g:\star \to \Lambda$ and $m:\Lambda\to\Lambda'$, is {\em strictly equal} to the diagram
\beq
\begin{tikzpicture}
	\begin{pgfonlayer}{nodelayer}
		\node [style=point] (0) at (0, -0) {$h$};
		\node [style=none] (1) at (0, 1.5) {};
		\node [style={right label}] (2) at (0, 1) {$\Lambda'$};
	\end{pgfonlayer}
	\begin{pgfonlayer}{edgelayer}
		\draw [style=cWire] (1.center) to (0);
	\end{pgfonlayer}
\end{tikzpicture} \ ,
\eeq
where $h$ is the sequential composition of $m$ and $g$, i.e., $h:\star\to \Lambda'::*\to m(g(*))$.

Composite systems are given by the Cartesian product of the associated sets:
\beq
\begin{tikzpicture}
	\begin{pgfonlayer}{nodelayer}
		\node [style=none] (0) at (0, 0.7500001) {};
		\node [style=none] (1) at (0, -0.7500001) {};
		\node [style={right label}] (2) at (0, -0.5000001) {$\Lambda$};
	\end{pgfonlayer}
	\begin{pgfonlayer}{edgelayer}
		\draw [oWire] (0.center) to (1.center);
	\end{pgfonlayer}
\end{tikzpicture}}%
\begin{tikzpicture}
	\begin{pgfonlayer}{nodelayer}
		\node [style=none] (0) at (0, 0.7500001) {};
		\node [style=none] (1) at (0, -0.7500001) {};
		\node [style={right label}] (2) at (0, -0.5000001) {$\Lambda'$};
	\end{pgfonlayer}
	\begin{pgfonlayer}{edgelayer}
		\draw [oWire] (0.center) to (1.center);
	\end{pgfonlayer}
\end{tikzpicture}}\quad := \quad %
\begin{tikzpicture}
	\begin{pgfonlayer}{nodelayer}
		\node [style=none] (0) at (0, 0.7500001) {};
		\node [style=none] (1) at (0, -0.7500001) {};
		\node [style={right label}] (2) at (0, -0.5000001) {$\Lambda\times\Lambda'$};
	\end{pgfonlayer}
	\begin{pgfonlayer}{edgelayer}
		\draw [oWire] (0.center) to (1.center);
	\end{pgfonlayer}
\end{tikzpicture}}.
\eeq
 Parallel composition of functions is given by their Cartesian product, 
\beq
\InputIfFileExists{Diagrams/onticFunction2.tikz}{}{\input{./figures/Diagrams/onticFunction2.tikz}}%
\InputIfFileExists{Diagrams/onticFunction1.tikz}{}{\input{./figures/Diagrams/onticFunction1.tikz}}\quad := \quad %
\InputIfFileExists{Diagrams/onticFunction.tikz}{}{\input{./figures/Diagrams/onticFunction.tikz}},
\eeq
and sequential composition is given by composition of functions,
\beq
\InputIfFileExists{Diagrams/onticFunction3.tikz}{}{\input{./figures/Diagrams/onticFunction3.tikz}}\quad := \quad %
\InputIfFileExists{Diagrams/onticFunction4.tikz}{}{\input{./figures/Diagrams/onticFunction4.tikz}}.
\eeq
It follows that any diagram is equal to the function from its inputs to its outputs, and that one can compute this effective function directly from the specific functions which comprise the diagram.

\begin{remark}
Categorically we are simply defining the symmetric monoidal category $\textsc{FinSet}$ whose objects are finite sets, morphisms are all functions between them, and the monoidal structure is given by the cartesian product and the singleton set.
\end{remark}

\section{Inferential primitives}\label{sec:inferential}

 We denote a generic process theory of inference by $\genInf$. The primitive inferential notions in our framework are systems and the processes which specify what an agent knows about them and how the agent reasons---for example, how they update what they know about one system given new information about 
  another.  When the inferential theory is defined intrinsically, the systems within the theory are simply understood as the entities about which one has states of knowledge and about which one asks questions, regardless of what these entities are: one could be making inferences about physical systems, or about mathematical truths, and so on.  When an inferential theory is considered as a part of a causal-inferential theory, however, the entities about which one makes inferences are causal mechanisms and the causal relata that these act on. 

The processes in $\genInf$ are particular inferences (i.e., updates of the knowledge one has about one system given new information about another), while the rewrite rules in $\genInf$ encode the laws of inference that an 
agent should follow if they are to be rational.  The standard laws of inference are those of Bayesian probability theory and Boolean logic. We formalize the laws of inference that will be relevant in this article 
within a single process theory, namely $\Inf$, the process theory of substochastic maps. 
$\Inf$ is the only explicit example of an inferential process theory that we will consider in this paper\footnote{ However, another good example that one could consider is the category $\textsc{Rel}$ of finite sets and relations. This would describe possibilistic reasoning as opposed to the probabilistic reasoning of $\Inf$.}. 
In Section~\ref{realconstr}, we return to the question of what properties any process theory must satisfy for it to be considered a good process theory of inference.

To diagrammatically distinguish the causal structure encoded in a given diagram of $\Caus$ from the inferential structure encoded in a given diagram of $\genInf$, we draw diagrams in  the former vertically, and diagrams in the latter horizontally. We will term the systems in the former {\em causal systems}, and systems in the latter {\em inferential systems}. 

\subsection{Bayesian probability theory}

The first component of $\Inf$ is (classical) Bayesian probability theory, describing an agent's states of knowledge and the updating thereof.   We denote this process theory $\textsc{Bayes}$.

Systems in  $\textsc{Bayes}$ are represented by finite sets $X, Y, \dots$.
A process from $X$ to $Y$ in this theory is interpreted as the propagation of an agent's knowledge about $X$ to her knowledge about $Y$, and is represented by a stochastic map.
    Such a process will be depicted diagrammatically as
\beq
\InputIfFileExists{Diagrams/inferenceProcess.tikz}{}{\input{./figures/Diagrams/inferenceProcess.tikz}}.
\eeq
The trivial system is the singleton set, $\star:=\{*\}$, and a map $q$ from the trivial system to system $X$, depicted as
\beq
\begin{tikzpicture}
	\begin{pgfonlayer}{nodelayer}
		\node [style=infpoint] (0) at (0, -0) {$\sigma$};
		\node [style=none] (1) at (1, -0) {};
		\node [style={up label}] (2) at (0.75, -0) {$X$};
	\end{pgfonlayer}
	\begin{pgfonlayer}{edgelayer}
		\draw [style=CcWire] (0) to (1.center);
	\end{pgfonlayer}
\end{tikzpicture}
},
\eeq
corresponds to a probability distribution over $X$.  We will denote the point distribution $\delta_{X,x}$ on some element $x\in X$ as $[x]$.
There is a unique effect for each system which corresponds to marginalisation over the variable. This is drawn as
\beq
\begin{tikzpicture}
	\begin{pgfonlayer}{nodelayer}
		\node [style=none] (0) at (0,0) {};
		\node [style=infupground] (1) at (1.5, 0) {};
		\node [style={up label}] (2) at (0.75, 0) {$X$};
	\end{pgfonlayer}
	\begin{pgfonlayer}{edgelayer}
		\draw [style=CcWire] (0) to (1);
	\end{pgfonlayer}
\end{tikzpicture}
}\ .
\eeq
As this is the unique effect, it is clear that any closed diagram is associated with the number 1, which is the unique scalar in $\textsc{Bayes}$.
It follows that given a state of knowledge $\sigma$ about a pair of variables, $X$ and $Y$, one can define the marginal distribution on the variable $X$ as
\beq
\InputIfFileExists{Diagrams/marginalisationNew.tikz}{}{\input{./figures/Diagrams/marginalisationNew.tikz}}.
\eeq

Note that it is possible to take convex combinations of stochastic processes provided that they have matching system types. We denote a convex combination of stochastic processes $\{s_i:X\to Y\}$ with weights $\{p_i\}$ as
\beq\label{eq:ConvexMixing}
\InputIfFileExists{Diagrams/ConvexMixing.tikz}{}{\input{./figures/Diagrams/ConvexMixing.tikz}},
\eeq
where $s:X\to Y$ is another stochastic process in the theory, namely, $s=\sum_i p_i s_i$.

\begin{remark}
Categorically, $\textsc{Bayes}$ is simply the symmetric monoidal category $\textsc{FinStoch}$  where objects are finite sets and morphisms are stochastic maps between them.
\end{remark}

\subsection{Boolean propositional logic}

The second component of $\Inf$ is the description of propositions, as governed by (classical) Boolean propositional logic. We denote the process theory describing this as $\textsc{Boole}$.  The systems in $\textsc{Boole}$ are finite sets labeled by $X,Y,...$, just as in $\textsc{Bayes}$. However, the processes in $\textsc{Boole}$ are partial functions; that is, functions that may only be defined on a subset of their domain.

Many of the key processes in $\textsc{Boole}$ are simply functions; we introduce these first, and only later discuss the more general processes in $\textsc{Boole}$ for which partial functions are required.

First, we consider states in $\textsc{Boole}$: functions from the trivial system $\star$ to a generic system $X$. These are in one-to-one correspondence with the elements of $X$, and so we can simply label each function by the element $x$ which is the image of $*$ under it. Hence, we can write $x:\star \to X::* \mapsto x$, depicted diagrammatically as
\beq
\begin{tikzpicture}
	\begin{pgfonlayer}{nodelayer}
		\node [style=infpoint] (0) at (0, -0) {$x$};
		\node [style=none] (1) at (1, -0) {};
		\node [style={up label}] (2) at (0.75, -0) {$X$};
	\end{pgfonlayer}
	\begin{pgfonlayer}{edgelayer}
		\draw [style=CcWire] (0) to (1.center);
	\end{pgfonlayer}
\end{tikzpicture}.
\eeq
We will refer to states in $\textsc{Boole}$ as {\bf value assignments}, because they can naturally be viewed as assigning a value $x$ to the variable ranging over the set $X$.

 Yes-no questions about a system $X$ can be represented as functions from $X$ to the answer set $\textsc{b}:=\{\textsc{y},\textsc{n}\}$, we refer to this answer set $\textsc{b}$ as the {\bf Boolean system}. Diagrammatically, these yes-no questions are denoted by
  \beq
\InputIfFileExists{Diagrams/yesno.tikz}{}{\input{./figures/Diagrams/yesno.tikz}}.
\eeq
A value assignment $x\in X$ assigns the answer `yes' or `no' to such questions via composition:
\beq
\begin{tikzpicture}
	\begin{pgfonlayer}{nodelayer}
		\node [style=small box] (0) at (-0.5, 0) {$\pi$};
		\node [style=infpoint] (1) at (-2.25, 0) {$x$};
		\node [style=up label] (2) at (-1.5, 0) {$X$};
		\node [style=none] (3) at (0.75, 0) {};
		\node [style=up label] (4) at (0.5, 0) {$\textsc{b}$};
	\end{pgfonlayer}
	\begin{pgfonlayer}{edgelayer}
		\draw [cWire] (1) to (0);
		\draw [cWire] (0) to (3.center);
	\end{pgfonlayer}
\end{tikzpicture}
\in \left\{\begin{tikzpicture}
	\begin{pgfonlayer}{nodelayer}
		\node [style=infpoint] (0) at (0, -0) {$\textsc{y}$};
		\node [style=none] (1) at (1, -0) {};
		\node [style={up label}] (2) at (0.75, -0) {$\textsc{b}$};
	\end{pgfonlayer}
	\begin{pgfonlayer}{edgelayer}
		\draw [style=CcWire] (0) to (1.center);
	\end{pgfonlayer}
\end{tikzpicture}, \begin{tikzpicture}
	\begin{pgfonlayer}{nodelayer}
		\node [style=infpoint] (0) at (0, -0) {$\textsc{n}$};
		\node [style=none] (1) at (1, -0) {};
		\node [style={up label}] (2) at (0.75, -0) {$\textsc{b}$};
	\end{pgfonlayer}
	\begin{pgfonlayer}{edgelayer}
		\draw [style=CcWire] (0) to (1.center);
	\end{pgfonlayer}
\end{tikzpicture}\right\}.
\eeq
 Now, each of these yes-no questions can be uniquely characterized by the subset of $X$ for which the answer is `yes', that is, $\{x \in X\ \text{s.t.}\ \pi::x\mapsto \textsc{y}\}$. This means that they are in one-to-one correspondence with the elements of the powerset of $X$ which we will view as a Boolean algebra, and, hence which we will denote by $\mathcal{B}(X)$.  Due to this correspondence, we will use the symbol $\pi$ both to denote the function $\pi:X\to\textsc{B}$ and the element of the Boolean algebra $\pi\in\mathcal{B}(X)$. We therefore refer to such functions as {\bf propositional questions}\footnote{These could also have been termed `predicates'.}.

We now show how the structure of the Boolean algebra $\mathcal{B}(X)$ can be diagrammatically represented using such propositional questions. To begin, there are two distinguished propositions in the Boolean algebra, the tautological proposition $\top$ (corresponding to the subset $X\subseteq X$) and the contradictory proposition $\perp$ (corresponding to the empty subset $\emptyset \subseteq X$). As propositional questions these can be defined diagrammatically via
\begin{align}
\forall x\in X\qquad
\begin{tikzpicture}
	\begin{pgfonlayer}{nodelayer}
		\node [style=small box] (0) at (-0.5, 0) {$\top$};
		\node [style=infpoint] (1) at (-2.25, 0) {$x$};
		\node [style=up label] (2) at (-1.5, 0) {$X$};
		\node [style=none] (3) at (0.75, 0) {};
		\node [style=up label] (4) at (0.5, 0) {$\textsc{b}$};
	\end{pgfonlayer}
	\begin{pgfonlayer}{edgelayer}
		\draw [cWire] (1) to (0);
		\draw [cWire] (0) to (3.center);
	\end{pgfonlayer}
\end{tikzpicture}
 \quad &= \quad  \begin{tikzpicture}
	\begin{pgfonlayer}{nodelayer}
		\node [style=infpoint] (0) at (0, -0) {$\textsc{y}$};
		\node [style=none] (1) at (1, -0) {};
		\node [style={up label}] (2) at (0.75, -0) {$\textsc{b}$};
	\end{pgfonlayer}
	\begin{pgfonlayer}{edgelayer}
		\draw [style=CcWire] (0) to (1.center);
	\end{pgfonlayer}
\end{tikzpicture} \\
\forall x\in X\qquad \begin{tikzpicture}
	\begin{pgfonlayer}{nodelayer}
		\node [style=small box] (0) at (-0.5, 0) {$\perp$};
		\node [style=infpoint] (1) at (-2.25, 0) {$x$};
		\node [style=up label] (2) at (-1.5, 0) {$X$};
		\node [style=none] (3) at (0.75, 0) {};
		\node [style=up label] (4) at (0.5, 0) {$\textsc{b}$};
	\end{pgfonlayer}
	\begin{pgfonlayer}{edgelayer}
		\draw [cWire] (1) to (0);
		\draw [cWire] (0) to (3.center);
	\end{pgfonlayer}
\end{tikzpicture}
\quad &=\quad \begin{tikzpicture}
	\begin{pgfonlayer}{nodelayer}
		\node [style=infpoint] (0) at (0, -0) {$\textsc{n}$};
		\node [style=none] (1) at (1, -0) {};
		\node [style={up label}] (2) at (0.75, -0) {$\textsc{b}$};
	\end{pgfonlayer}
	\begin{pgfonlayer}{edgelayer}
		\draw [style=CcWire] (0) to (1.center);
	\end{pgfonlayer}
\end{tikzpicture}.
\end{align}

Moreover, our representation of propositions as functions in $\textsc{Boole}$ allows us to diagrammatically represent unary and binary logical operations on the propositions, by defining them in terms of their action on Boolean systems.
For example, the NOT operation on an arbitrary proposition can be represented as
\beq
\InputIfFileExists{Diagrams/logicalConnective2.tikz}{}{\input{./figures/Diagrams/logicalConnective2.tikz}}\quad = \quad %
\InputIfFileExists{Diagrams/logicalConnective3.tikz}{}{\input{./figures/Diagrams/logicalConnective3.tikz}},
\eeq
where the dot decorated by the $\neg$ symbol is the function $\neg:\textsc{b}\to \textsc{b}$ whose action on the Boolean system reflects the truth table of the logical NOT, namely $\neg(\textsc{y})=\textsc{n}$ and $\neg(\textsc{n})=\textsc{y}$ (implying that $\neg$ is self-inverse).
Similarly, one can represent the logical OR operation (disjunction), denoted $\lor$, as
\beq \label{logicalconnect1}
\InputIfFileExists{Diagrams/logicalConnective.tikz}{}{\input{./figures/Diagrams/logicalConnective.tikz}}\quad = \quad %
\InputIfFileExists{Diagrams/logicalConnective1.tikz}{}{\input{./figures/Diagrams/logicalConnective1.tikz}} \ ,
\eeq
where the black dot is the copy function $\bullet:X\to X\times X$ defined by $\bullet(x)=(x,x)$, and the dot decorated by the $\lor$ symbol is the function  $\lor:\textsc{b}\times \textsc{b} \to \textsc{b}$ whose action on the Boolean system reflects the truth table of the logical OR, namely, $\textsc{y}\lor \textsc{y} = \textsc{y}\lor \textsc{n} = \textsc{n}\lor \textsc{y} = \textsc{y}$ and $\textsc{n}\lor \textsc{n} = \textsc{n}$.
In a similar manner, one can construct representations of the logical AND, logical implication, exclusive OR, etc.

 We have therefore diagrammatically represented the basic operations required to define a Boolean algebra $\mathcal{B}(X)$. Moreover,
the basic properties of a Boolean algebra (associativity, absorption, commutativity, identity, annihilation, idempotence, complements, and distributivity) can also be shown to hold. These are defined and proven in Appendix~\ref{Boolpropproof}.  

One can also express propositional questions about composite systems, for example
\beq \label{compprop}
\InputIfFileExists{Diagrams/compositePropositionalQuestion.tikz}{}{\input{./figures/Diagrams/compositePropositionalQuestion.tikz}},
\eeq
where $\pi \in \mathcal{B}(X\times Y)$.
We discuss in Appendix~\ref{compositepropns} how these can be constructed out of single-system propositions.

Arbitrary functions are also valid processes in $\textsc{Boole}$, since (as we now show) each is a valid \emph{Boolean algebra homomorphism}, that is, a map between Boolean algebras that preserves the logical connectives $\land$ and $\lor$ as well as the top and bottom elements $\top$ and $\perp$.
Consider a general function
\beq
\InputIfFileExists{Diagrams/propositionalMap.tikz}{}{\input{./figures/Diagrams/propositionalMap.tikz}}.
\eeq
In the process theory $\textsc{Bayes}$, this process would be viewed
as a stochastic map acting on the right to take
 states of knowledge about $X$ to states of knowledge about $Y$.
Within $\textsc{Boole}$, however, this process is viewed
as a map acting on the left, taking
a propositional question about $Y$ to a propositional question about $X$ via
\beq\label{eq:propositionalquestionmap}
\InputIfFileExists{Diagrams/propositionalMapYN1.tikz}{}{\input{./figures/Diagrams/propositionalMapYN1.tikz}} \quad =: \quad %
\InputIfFileExists{Diagrams/propositionalMapYN2.tikz}{}{\input{./figures/Diagrams/propositionalMapYN2.tikz}}.
\eeq
We now show that $f(\_)$ defines a Boolean Algebra homomorphism from $\mathcal{B}(Y)$ to $\mathcal{B}(X)$, where each subset of $Y$ is mapped to the subset of $X$ which is the preimage under $f(\_)$ of the subset of $Y$. We will sometimes refer to such generic functions simply as {\bf propositional maps}.
It is easy to see that propositional maps do indeed preserve $\land$, $\lor$, $\top$, and $\perp$. The fact that the propositional map $\top: Y \to \textsc{b}:: y \mapsto \textsc{y}$ is preserved follows immediately from the fact that $f(\_)$ maps every $x$ to {\em some} $y$, which $\top$ then necessarily maps to $\textsc{y}$.
Preservation of $\perp$ is analogous.
To see that $\land$ is preserved---namely, that
$f(\pi \land \pi') = f(\pi) \land f(\pi')$---note that
\begin{align}
\begin{tikzpicture}
	\begin{pgfonlayer}{nodelayer}
		\node [style={small box}] (0) at (0, -0) {$f(\pi\land\pi')$};
		\node [style=none] (1) at (-1.75, -0) {};
		\node [style={up label}] (2) at (-1.75, -0) {$X$};
		\node [style=none] (3) at (1.75, -0) {};
		\node [style={up label}] (4) at (1.75, -0) {$\textsc{b}$};
	\end{pgfonlayer}
	\begin{pgfonlayer}{edgelayer}
		\draw [cWire] (1.center) to (0);
		\draw [cWire](0) to (3.center);
	\end{pgfonlayer}
\end{tikzpicture}
\quad&=\quad
\InputIfFileExists{Diagrams/logicalConnectivePres.tikz}{}{\input{./figures/Diagrams/logicalConnectivePres.tikz}}\\
 &= \quad %
\InputIfFileExists{Diagrams/logicalConnectivePres2.tikz}{}{\input{./figures/Diagrams/logicalConnectivePres2.tikz}}\label{lcpres2}\\
 &=\quad %
\InputIfFileExists{Diagrams/logicalConnectivePres1.tikz}{}{\input{./figures/Diagrams/logicalConnectivePres1.tikz}} \label{lcpres1}\\
 &=\quad %
\InputIfFileExists{Diagrams/logicalConnectivePres3.tikz}{}{\input{./figures/Diagrams/logicalConnectivePres3.tikz}} \\
 &=\quad \begin{tikzpicture}
	\begin{pgfonlayer}{nodelayer}
		\node [style={small box}] (0) at (0, -0) {$f(\pi)\land f(\pi')$};
		\node [style=none] (1) at (-2.25, -0) {};
		\node [style={up label}] (2) at (-2.25, -0) {$X$};
		\node [style=none] (3) at (2.5, -0) {};
		\node [style={up label}] (4) at (2.5, -0) {$\textsc{b}$};
	\end{pgfonlayer}
	\begin{pgfonlayer}{edgelayer}
		\draw [cWire] (1.center) to (0);
		\draw [cWire](0) to (3.center);
	\end{pgfonlayer}
\end{tikzpicture},
\end{align}
where the equality between Eq.~\eqref{lcpres2} and Eq.~\eqref{lcpres1} simply states that copying the output of a function is the same as copying the input and applying the function to the two copies of the input. The proof of preservation of $\lor$ is analogous.

\begin{remark}
Categorically, this dual picture of $f(\_)$ as a function from $X$ to $Y$ and as a Boolean Algebra homomorphism from $\mathcal{B}(Y)$ to $\mathcal{B}(X)$ is the duality between the categories $\textsc{FinSet}$ and $\textsc{FinBoolAlg}$ where $B=\{\textsc{y},\textsc{n}\}$ is the dualising object.
\end{remark}

To express the truth value assigned to a proposition, we must introduce scalars and effects, which moreover requires us to go beyond functions and consider partial functions\footnote{This is because in the process theory of functions, there is a unique scalar and a unique effect. The unique scalar is the function taking the singleton set to itself, while the unique effect is the function from $X$ to the singleton set which maps every element $x\in X$ to $*$.}.
A partial function $\hat{f}:X \to Y$ is a function from some (possibly empty) subset $\chi_{\hat{f}}\subseteq X$ to $Y$; the partial function is simply undefined on the elements of $X$ outside of this subset.

There are two scalars in $\textsc{Boole}$, which we identify with true and false; namely, the function $\mathsf{True}: \star \to \star:: * \mapsto *$ and the partial function $\mathsf{False}:\star \to \star$, which is defined only on the empty set $\emptyset$, that is, $\chi_{\mathsf{False}} =\emptyset$.
The scalar $\mathsf{True}$ is depicted by the empty diagram, since composing it with any other process leaves that process invariant:
\beq
\InputIfFileExists{Diagrams/emptyDiagram1.tikz}{}{\input{./figures/Diagrams/emptyDiagram1.tikz}}\quad = \quad %
\InputIfFileExists{Diagrams/emptyDigram2.tikz}{}{\input{./figures/Diagrams/emptyDigram2.tikz}}.
\eeq
The scalar $\mathsf{False}$ behaves as a `zero scalar', in the following sense. Defining the `zero process' for a pair of systems $(X,Y)$ as the unique partial function ${\bf{0}}:X\to Y$ such that $\chi_{\mathbf{0}} = \emptyset$, it follows that composing any other partial function $\hat{f}:X\to Y$ with $\mathsf{False}$ will give this zero process:
\beq
\InputIfFileExists{Diagrams/zeroDiagram1.tikz}{}{\input{./figures/Diagrams/zeroDiagram1.tikz}}\quad = \quad %
\InputIfFileExists{Diagrams/zeroDiagram2.tikz}{}{\input{./figures/Diagrams/zeroDiagram2.tikz}} \ .
\eeq

Now, effects within $\textsc{Boole}$ are partial functions taking $X\to \star$, diagrammatically denoted by
\beq
\begin{tikzpicture}
	\begin{pgfonlayer}{nodelayer}
		\node [style=infcopoint] (0) at (0, -0) {$\pi$};
		\node [style=none] (1) at (-1.25, -0) {};
		\node [style={up label}] (2) at (-1.25, -0) {$X$};
	\end{pgfonlayer}
	\begin{pgfonlayer}{edgelayer}
		\draw [cWire] (1.center) to (0);
	\end{pgfonlayer}
\end{tikzpicture} \ .
\eeq
We will see that these are also in one-to-one correspondence with the elements of the Boolean algebra $\mathcal{B}(X)$, which justifies our labelling them by propositions $\pi$. To see this, first note that value assignments $x\in X$ assign a truth-value to effects within $\textsc{Boole}$ when the two are composed together:
\beq
\begin{tikzpicture}
	\begin{pgfonlayer}{nodelayer}
		\node [style=infcopoint] (0) at (0.7499998, -0) {$\pi$};
		\node [style=none] (1) at (-0.7500001, -0) {};
		\node [style={up label}] (2) at (0, -0) {$X$};
		\node [style=infpoint] (3) at (-0.7500001, -0) {$x$};
	\end{pgfonlayer}
	\begin{pgfonlayer}{edgelayer}
		\draw [cWire] (1.center) to (0);
	\end{pgfonlayer}
\end{tikzpicture}
\quad \in \quad \left\{\begin{tikzpicture}
	\begin{pgfonlayer}{nodelayer}
		\node [style=scalar] (0) at (0, -0) {$\textsf{True}$};
	\end{pgfonlayer}
\end{tikzpicture}\ \  ,\ \begin{tikzpicture}
	\begin{pgfonlayer}{nodelayer}
		\node [style=scalar] (0) at (0, -0) {$\textsf{False}$};
	\end{pgfonlayer}
\end{tikzpicture}\right\}.
\eeq
Hence, one can uniquely associate a partial function with the subset of $X$ for which we obtain $\textsf{True}$; indeed, this subset is the domain $\chi_\pi$ of the effect, viewed as a partial function.
We will call such partial functions {\bf propositional effects}.\

At this point, we have three uses of the symbol $\pi$: we have $\pi\in\mathcal{B}(X)$ as a subset (an element of a Boolean Algebra), $\pi:X\to\textsc{b}$ as a propositional question, and now $\pi:X\to\star$ as a propositional effect.  The distinction between these should be clear from context.

To more explicitly see the connections between propositional effects and propositional questions, let us consider the special case of propositional effects for the Boolean system $\textsc{b}$.  There are four of these, corresponding to the four subsets of $\textsc{b}$ on which the partial function from $\textsc{b}$ to $\star$ can be defined, namely $\{\textsc{y}\}$, $\{\textsc{n} \}$, $\{\textsc{y}, \textsc{n} \}$, and $\emptyset$.

We denote these effects, respectively, as
\beq
\begin{tikzpicture}
	\begin{pgfonlayer}{nodelayer}
		\node [style=infcopoint] (0) at (0, -0) {$\textsc{y}$};
		\node [style=none] (1) at (-1.25, -0) {};
		\node [style={up label}] (2) at (-1.25, -0) {$\textsc{B}$};
	\end{pgfonlayer}
	\begin{pgfonlayer}{edgelayer}
		\draw [cWire] (1.center) to (0);
	\end{pgfonlayer}
\end{tikzpicture}
\ \ ,\
\begin{tikzpicture}
	\begin{pgfonlayer}{nodelayer}
		\node [style=infcopoint] (0) at (0, -0) {$\textsc{n}$};
		\node [style=none] (1) at (-1.25, -0) {};
		\node [style={up label}] (2) at (-1.25, -0) {$\textsc{b}$};
	\end{pgfonlayer}
	\begin{pgfonlayer}{edgelayer}
		\draw [cWire] (1.center) to (0);
	\end{pgfonlayer}
\end{tikzpicture}
\ \ ,\
\begin{tikzpicture}
	\begin{pgfonlayer}{nodelayer}
		\node [style=infcopoint] (0) at (0, -0) {$\top$};
		\node [style=none] (1) at (-1.25, -0) {};
		\node [style={up label}] (2) at (-1.25, -0) {$\textsc{b}$};
	\end{pgfonlayer}
	\begin{pgfonlayer}{edgelayer}
		\draw [cWire] (1.center) to (0);
	\end{pgfonlayer}
\end{tikzpicture}
\ \ ,\ \text{and}\
\begin{tikzpicture}
	\begin{pgfonlayer}{nodelayer}
		\node [style=infcopoint] (0) at (0, -0) {$\perp$};
		\node [style=none] (1) at (-1.25, -0) {};
		\node [style={up label}] (2) at (-1.25, -0) {$\textsc{b}$};
	\end{pgfonlayer}
	\begin{pgfonlayer}{edgelayer}
		\draw [cWire] (1.center) to (0);
	\end{pgfonlayer}
\end{tikzpicture}
\ \ .
\eeq
 Then, we can write a given propositional effect in terms of the associated propositional question via
\beq
\begin{tikzpicture}
	\begin{pgfonlayer}{nodelayer}
		\node [style=infcopoint] (0) at (0, -0) {$\pi$};
		\node [style=none] (1) at (-1.25, -0) {};
		\node [style={up label}] (2) at (-1.25, -0) {$X$};
	\end{pgfonlayer}
	\begin{pgfonlayer}{edgelayer}
		\draw [cWire] (1.center) to (0);
	\end{pgfonlayer}
\end{tikzpicture}
\ \ = \
\begin{tikzpicture}
	\begin{pgfonlayer}{nodelayer}
		\node [style=infcopoint] (0) at (2.25, -0) {$\textsc{y}$};
		\node [style={up label}] (1) at (1, -0) {$\textsc{b}$};
		\node [style={small box}] (2) at (0, -0) {$\pi$};
		\node [style=none] (3) at (-2, -0) {};
		\node [style={up label}] (4) at (-1, -0) {$X$};
	\end{pgfonlayer}
	\begin{pgfonlayer}{edgelayer}
		\draw [cWire] (3) to (2);
		\draw [cWire] (2) to (0);
	\end{pgfonlayer}
\end{tikzpicture}.
\eeq
Value assignments to propositional effects are then consistent with value assignment to propositional questions, in the sense that\footnote{ It is worth noting that one could dispense with the notion of propositional questions and express all claims in  terms of propositional effects.  We include the notion of a propositional question here because it helps to clarify certain conceptual distinctions.}
\beq
\begin{tikzpicture}
	\begin{pgfonlayer}{nodelayer}
		\node [style=infcopoint] (0) at (0.7499998, -0) {$\pi$};
		\node [style=none] (1) at (-0.7500001, -0) {};
		\node [style={up label}] (2) at (0, -0) {$X$};
		\node [style=infpoint] (3) at (-0.7500001, -0) {$x$};
	\end{pgfonlayer}
	\begin{pgfonlayer}{edgelayer}
		\draw [cWire] (1.center) to (0);
	\end{pgfonlayer}
\end{tikzpicture}
=\begin{tikzpicture}
	\begin{pgfonlayer}{nodelayer}
		\node [style=scalar] (0) at (0, -0) {$\textsf{True}$};
	\end{pgfonlayer}
\end{tikzpicture}
\quad \iff \quad
\begin{tikzpicture}
	\begin{pgfonlayer}{nodelayer}
		\node [style=small box] (0) at (-0.5, 0) {$\pi$};
		\node [style=infpoint] (1) at (-2.25, 0) {$x$};
		\node [style=up label] (2) at (-1.5, 0) {$X$};
		\node [style=none] (3) at (0.75, 0) {};
		\node [style=up label] (4) at (0.5, 0) {$\textsc{b}$};
	\end{pgfonlayer}
	\begin{pgfonlayer}{edgelayer}
		\draw [cWire] (1) to (0);
		\draw [cWire] (0) to (3.center);
	\end{pgfonlayer}
\end{tikzpicture}
= \begin{tikzpicture}
	\begin{pgfonlayer}{nodelayer}
		\node [style=infpoint] (0) at (0, -0) {$\textsc{y}$};
		\node [style=none] (1) at (1, -0) {};
		\node [style={up label}] (2) at (0.75, -0) {$\textsc{b}$};
	\end{pgfonlayer}
	\begin{pgfonlayer}{edgelayer}
		\draw [style=CcWire] (0) to (1.center);
	\end{pgfonlayer}
\end{tikzpicture}.
\eeq

It turns out that {\em all} partial functions can be generated from the elements we have introduced so far, and hence all of these are in $\textsc{Boole}$. An arbitrary partial function $\hat{f}$ can be written as
\beq \label{anypartfunc1}
\InputIfFileExists{Diagrams/arbitraryPartial.tikz}{}{\input{./figures/Diagrams/arbitraryPartial.tikz}},
\eeq
where $\chi_{\hat{f}}$ is an arbitrary propositional effect, $F:X\to Y$ is an arbitrary propositional map, and the black dot is the copy function. Here, the top part of the diagram defines the subset of the domain on which the partial function is defined, and then $F$ defines the action of the partial function on that domain.
Hence, we have that
\begin{remark}
Categorically, the process theory $\textsc{Boole}$ is the symmetric monoidal category $\textsc{FinSet}_{\textsc{part}}$ where objects are finite sets and morphisms are partial functions.
\end{remark}

We discussed how the functions in $\textsc{Boole}$ correspond to Boolean algebra homomorphisms. In contrast, partial functions in $\textsc{Boole}$ are more general. In general, partial functions do not map propositional questions to propositional questions (via Eq.~\eqref{eq:propositionalquestionmap}), but rather take them to other partial functions.
However, partial functions do map propositional effects to propositional effects via
\beq
\begin{tikzpicture}
	\begin{pgfonlayer}{nodelayer}
		\node [style=infcopoint] (0) at (2.25, -0) {$\pi$};
		\node [style={up label}] (1) at (1, -0) {$Y$};
		\node [style={small box}] (2) at (0, -0) {$\hat{f}$};
		\node [style=none] (3) at (-2, -0) {};
		\node [style={up label}] (4) at (-1, -0) {$X$};
	\end{pgfonlayer}
	\begin{pgfonlayer}{edgelayer}
		\draw [cWire] (3) to (2);
		\draw [cWire] (2) to (0);
	\end{pgfonlayer}
\end{tikzpicture}\ \ =:\
\begin{tikzpicture}
	\begin{pgfonlayer}{nodelayer}
		\node [style=infcopoint] (0) at (0, -0) {$\hat{f}(\pi)$};
		\node [style=none] (1) at (-1.25, -0) {};
		\node [style={up label}] (2) at (-1.25, -0) {$X$};
	\end{pgfonlayer}
	\begin{pgfonlayer}{edgelayer}
		\draw [cWire] (1.center) to (0);
	\end{pgfonlayer}
\end{tikzpicture} \ ,
\eeq
and so we can ask which structures of the Boolean algebra of propositional effects are preserved by such a map.  We show in Appendix~\ref{partfuncapp} that $\perp$, $\lor$ and $\land$ are preserved, but $\top$ and $\neg$ are not. We then show that a partial function from $X$ to $Y$ corresponds to a Boolean Algebra homomorphism from $\mathcal{B}(Y)$ to $\mathcal{B}(\chi_{\hat{f}})$.

\subsection{The full inferential process theory}\label{inftheory}

We now show how the probabilistic and the propositional parts of the inferential theory interact---for example, allowing one to compute the probability one should assign to a propositional effect on a system $X$,
given an arbitrary state of knowledge about $X$.
 This interaction is possible because $\textsc{Bayes}$ and $\textsc{Boole}$ define a collection of processes on the same types of systems. However, the processes in $\textsc{Bayes}$ are stochastic maps, while those in $\textsc{Boole}$ are partial functions, and so it remains to define how these interact with one another.

We proceed by showing that both $\textsc{Bayes}$ and $\textsc{Boole}$ can be faithfully represented within the process theory of  substochastic linear maps, $\Inf$ (see Remark~\ref{remsubstoch}). 
More formally, there is a diagram-preserving inclusion map from $\textsc{Bayes}$ into $\Inf$, and there is a diagram-preserving map (given by Eq.~\eqref{parfuncrepn} below) from $\textsc{Boole}$ into $\Inf$.
Hence, we have
\beq
\begin{tikzpicture}
	\begin{pgfonlayer}{nodelayer}
		\node [style=none] (0) at (-4.25, 0) {$\textsc{Bayes}$};
		\node [style=none] (1) at (0, 0) {$\Inf$};
		\node [style=none] (2) at (4.25, 0) {$\textsc{Boole}$};
		\node [style=none] (3) at (-3.25, 0) {};
		\node [style=none] (4) at (-1.75, 0) {};
		\node [style=none] (5) at (3.25, 0) {};
		\node [style=none] (6) at (1.75, 0) {};
		\node [style=none] (7) at (3.75, -0.25) {};
		\node [style=none] (8) at (1.75, -0.25) {};
		\node [style=up label] (10) at (-2.5, 0) {};
		\node [style=up label] (11) at (2.5, 0) {};
	\end{pgfonlayer}
	\begin{pgfonlayer}{edgelayer}
		\draw [style=arrow plain] (3.center) to (4.center);
		\draw [style=arrow plain] (5.center) to (6.center);
	\end{pgfonlayer}
\end{tikzpicture}.
\eeq
Moreover, we show that any substochastic map can be realised by composing the processes from these two representations.

First, let us consider $\textsc{Boole}$. Note that any function $f:X\to Y$ can be represented by an associated stochastic map,
$(\mathbf{f}_x^y )_{x\in X}^{y\in Y}$, via:
\beq
\mathbf{f}_x^y =
\begin{cases}
    1 & \text{if } f(x)=y\\
    0 & \text{otherwise}.
\end{cases}
\eeq
 Any such stochastic map is {\em deterministic}, meaning that 
 there is precisely a single $1$ in each column, with the rest of the elements are $0$.
It is not difficult to check that these stochastic maps compose in the same way as the underlying functions, and so this gives us a representation of the process theory of functions within the process theory of stochastic maps.

More generally, a partial function $\hat{f}$ is also associated with a stochastic map via
\beq \label{parfuncrepn}
\hat{\mathbf{f}}_x^y =
\begin{cases}
    1 & \text{if } \hat{f}(x)=y\\
    0 & \text{otherwise},
\end{cases}
\eeq
with the only difference being that for functions, the $0$ case occurs only when $f(x)\neq y$, whereas for partial functions it will also occur when $\hat{f}$ is not defined on $x$. These are more general than deterministic stochastic maps in that some of the $1$s may be replaced by $0$s---that is, they are \emph{substochastic} maps. Again, one can check that the representative substochastic maps compose in the same way as the underlying partial functions.

Clearly, the probabilistic part of the theory, which is described by stochastic maps, can also be represented as substochastic maps, as the former are simply a special case of the latter.

Hence, both $\textsc{Bayes}$ and $\textsc{Boole}$ can be represented within $\textsc{SubStoch}$.

Within this representation, certain processes in $\textsc{Bayes}$ and certain processes in $\textsc{Boole}$ correspond to the same substochastic map, and hence are identified. For example, we have the identification
\beq \label{topequalsmarg}
\begin{tikzpicture}
	\begin{pgfonlayer}{nodelayer}
		\node [style=infcopoint] (0) at (0, -0) {$\top$};
		\node [style=none] (1) at (-1.25, -0) {};
		\node [style={up label}] (2) at (-1.25, -0) {$X$};
	\end{pgfonlayer}
	\begin{pgfonlayer}{edgelayer}
		\draw [cWire] (1.center) to (0);
	\end{pgfonlayer}
\end{tikzpicture}
\ \ \ = \
\begin{tikzpicture}
	\begin{pgfonlayer}{nodelayer}
		\node [style=none] (0) at (-1.25, 0) {};
		\node [style=infupground] (1) at (0, 0) {};
		\node [style={up label}] (2) at (-1.25, 0) {$X$};
	\end{pgfonlayer}
	\begin{pgfonlayer}{edgelayer}
		\draw [style=CcWire] (0) to (1);
	\end{pgfonlayer}
\end{tikzpicture} \ ,
\eeq
since both are represented by the all-ones column vector. As another simple example, the representation of any function in the propositional theory will coincide with some deterministic stochastic map in $\textsc{Bayes}$; hence, such processes in $\Inf$ can be viewed either
acting on the left as propositional maps, or acting on the right as stochastic maps.
As a final example, the representation of a delta function probability distribution $[x]$ from $\textsc{Bayes}$ coincides with the representation of the value assignment asserting $X=x$ from $\textsc{Boole}$.

Consider now the representation of the scalars in $\textsc{Bayes}$ and $\textsc{Boole}$. The unique scalar $1$ in $\textsc{Bayes}$ remains the same in this representation, while the pair of scalars in $\textsc{Boole}$, namely $\mathsf{True}$ and $\mathsf{False}$, are represented respectively by the scalars $1$ and $0$ within $\textsc{SubStoch}$.

These two representations interact in the obvious way. For example, we expect the diagram
\beq
\begin{tikzpicture}
	\begin{pgfonlayer}{nodelayer}
		\node [style=infpoint] (0) at (-1, -0) {$\sigma$};
		\node [style=infcopoint] (1) at (1, -0) {$\pi$};
		\node [style={up label}] (2) at (0, -0) {$X$};
	\end{pgfonlayer}
	\begin{pgfonlayer}{edgelayer}
		\draw [cWire] (0) to (1);
	\end{pgfonlayer}
\end{tikzpicture}\eeq
to give the probability $\mathsf{Prob}(\pi:\sigma)$
that the proposition $\pi$ about $X$ is true, given a state of knowledge $\sigma$ about $X$. Indeed, this can be computed within the theory of substochastic maps:
\begin{align}\label{eq:InfConsistency}
\InputIfFileExists{Diagrams/yesTruthProb.tikz}{}{\input{./figures/Diagrams/yesTruthProb.tikz}}\ &= \sum_{x\in X} \sigma(x)\ %
\InputIfFileExists{Diagrams/yesTruthProb1.tikz}{}{\input{./figures/Diagrams/yesTruthProb1.tikz}}\\ &=\sum_{x\in X} \sigma(x) \delta_{\pi(x), \textsc{y}}\\ &= \sum_{x\in \pi} \sigma(x)\\ &= \mathsf{Prob}(\pi:\sigma).
\end{align}
It turns out that arbitrary substochastic maps can be realized by the interaction of stochastic maps and partial functions. An arbitrary substochastic map can be represented as
\beq
\InputIfFileExists{Diagrams/substochRealisation.tikz}{}{\input{./figures/Diagrams/substochRealisation.tikz}} \ ,
\eeq
where $s$ and $w$ are arbitrary stochastic maps. Here, $w$ specifies the normalization of each column of the substochastic map, while $s$ specifies the action of each column (apart from the normalization factor).
Hence, we have that
\begin{remark} \label{remsubstoch}
Categorically, $\Inf$ is the symmetric monoidal category $\textsc{FinSubStoch}$ where objects are finite sets and morphisms are substochastic maps between them, together with a selected subobject classifier  $\{\textsc{y},\textsc{n}\}$. It is, however, useful to see how this structure arises from the interaction between the probabilistic part (as described by $\textsc{FinStoch}$) and the propositional part (as described by $\textsc{FinSet}_{\textsc{part}}$).
\end{remark}

\begin{remark}
Strictly speaking, in what follows, we will have certain inferential systems labeled by sets of infinite cardinality.
 Recalling that $\Proc$ is a free process theory, its hom-set will typically be of infinite cardinality.
(In contrast, note that the hom-sets in $\Func$ are finite, so the issue will not arise there). To formally deal with this, rather than working with $\textsc{FinSubStoch}$, we should work with $\textsc{SubStoch}$---defined as the Kleisli category of the subdistribution monad on $\textsc{Set}$.\footnote{Thanks are due to Martti Karvonen for recognizing this issue, giving the resolution to it, and then explaining the resolution to us.} In this process theory, states are subnormalised probability distributions with finite support, and general processes are substochastic maps which do not generate distributions with infinite support. In future work, it will be important to consider more sophisticated measure-theoretic approaches to infinite sets---in particular, to allow for ontic state spaces of infinite cardinality.
\end{remark}

\section{Causal-inferential theories}\label{sec:CI}

Having introduced theories of causal primitives and of inferential primitives, we can now describe how the two interact to define causal-inferential theories. 
We first (non-exhaustively) describe features that we expect of a generic causal-inferential theory.
We then develop two special cases that instantiate these features, namely causal-inferential theories of lab procedures, or {\em operational CI theories}, and the causal-inferential theory of functional dynamics, namely the {\em \crealist CI theory}.

%


A causal-inferential theory consists of a triple of process theories and a triple of DP (partial) maps between them:
\beq\label{eq:CIFull}
\InputIfFileExists{Diagrams/causalInferentialFull.tikz}{}{\input{./figures/Diagrams/causalInferentialFull.tikz}} \ ,
\eeq
where the process theory $\CI$ includes all of the causal and inferential systems coming from $\Caus$ and $\genInf$, respectively. (That is, $\mathbf{i}$ and $\mathbf{e}$ are injective on objects.)
\begin{remark}
  Note that we will continue to draw the causal systems (i.e., those in the image of $\mathbf{e}$) vertically and the inferential systems (i.e., those in the image of $\mathbf{i}$) horizontally (i.e., just as in the respective domains of $\mathbf{e}$ and $\mathbf{i}$). This choice is merely a convention; we could alternatively have just used a different style of wire or labelling system to keep track of this information.  That is, on a formal level, $\CI$, like $\Caus$ and $\genInf$, is simply a symmetric monoidal category.  
\end{remark}

We define how the causal and inferential systems interact by introducing three special processes that involve both causal and inferential systems, and that are subject to a collection of rewrite rules. These three processes allow one to (i) specify a state of knowledge about a particular causal dynamics, (ii) gain information about a classical causal system, and (iii) ignore causal systems. These three generators are denoted, respectively, by
\beq
\begin{tikzpicture}
	\begin{pgfonlayer}{nodelayer}
		\node [style=none] (0) at (-1, 0) {};
		\node [style=epiBox] (1) at (0, 0) {};
		\node [style=none] (2) at (0, 1) {};
		\node [style=none] (3) at (0, -1) {};
	\end{pgfonlayer}
	\begin{pgfonlayer}{edgelayer}
		\draw [style=CcWire] (0.center) to (1);
		\draw [style=qWire] (2.center) to (1);
		\draw [style=qWire] (1) to (3.center);
	\end{pgfonlayer}
\end{tikzpicture}
\ \ , \quad \begin{tikzpicture}
	\begin{pgfonlayer}{nodelayer}
		\node [style=clear dot] (0) at (0, 0) {};
		\node [style=none] (1) at (0, -1) {};
		\node [style=none] (3) at (0, 1) {};
		\node [style=none] (5) at (1, 0) {};
		\node [style=none] (7) at (0, 0.25) {};
		\node [style=none] (8) at (0, -0.25) {};
	\end{pgfonlayer}
	\begin{pgfonlayer}{edgelayer}
		\draw [style=qWire] (0) to (1.center);
		\draw [style=qWire] (3.center) to (0);
		\draw [style=CcWire] (5.center) to (0);
		\draw [style=qWire] (7.center) to (8.center);
	\end{pgfonlayer}
\end{tikzpicture}
\ \ \text{and} \quad  
\begin{tikzpicture}
	\begin{pgfonlayer}{nodelayer}
		\node [style=none] (0) at (0, -1) {};
		\node [style=small black dot] (1) at (0, 0.25) {};
		\node [style=ignore] (3) at (0, 0.25) {};
	\end{pgfonlayer}
	\begin{pgfonlayer}{edgelayer}
		\draw [qWire] (0.center) to (1);
	\end{pgfonlayer}
\end{tikzpicture}.
\eeq
In our examples, the theory corresponding to $\CI$ in Eq.~\eqref{eq:CIFull} 
 is defined as the process theory constructed out of arbitrary composition of these three `generators', together with the processes in the image of $\mathbf{i}$.\footnote{
Note that, unlike the situation with the $\mathbf{i}$ map, the construction of $\CI$ does not require any of the processes in the image of $\mathbf{e}$ to be generators, since (as we will see) the set of all such processes can be obtained by combining elements of $\Inf$ (the image of $\mathbf{i}$) with the first generator via Eq.~\eqref{gen1}.  Indeed, this is how the map $\mathbf{e}$ will be {\em defined}; see Eq.~\eqref{defnemap}.   
}

The rewrite rules that these satisfy will be (to some extent) dependent on the exact causal-inferential theory that one is interested in -- indeed, there will be an important distinction between the rewrite rules for our two key examples. 
An important direction for future research is therefore  to determine which aspects of these rewrite rules 
are in fact generic to {\em all} CI theories. We return to this question in Section~\ref{realconstr}. 
 
Finally, we introduce a partial map $\mathbf{p}$ that 
allows one to make inferential predictions, by mapping diagrams in $\CI$ with only inferential inputs and outputs to a particular inferential process within $\Inf$. That is, in our examples, the map $\mathbf{p}$ 
 makes a probabilistic prediction given one's knowledge about a particular causal scenario. 

We will for simplicity sometimes refer to $\CI$ as `the causal-inferential theory' or use the symbol $\CI$ to refer to the full causal-inferential theory, since $\CI$ is the primary process theory of relevance. Strictly speaking, however, a causal-inferential theory is given by a triple as in Eq.~\eqref{eq:CIFull} (including, in particular, the prediction map).

\subsection{Operational causal-inferential theories}\label{sec:Op}

We will use the term `operational causal-inferential (CI) theory' to refer to a causal-inferential theory of lab procedures; that is,  taking $\Caus=\Proc$ and $\genInf=\Inf$ in Eq.~\eqref{eq:CIFull}, i.e.,
\beq %
\InputIfFileExists{Diagrams/OpInferentialEquiv.tikz}{}{\input{./figures/Diagrams/OpInferentialEquiv.tikz}} \label{OpInferentialEquiv} \ . \eeq
We have already defined $\Proc$ and $\Inf$, but it remains to explicitly define $\PI$ and the diagram-preserving maps  $\mathbf{e}$, $\mathbf{i}$ and $\mathbf{p}$ between these three process theories.

$\Inf$ is a subprocess theory of $\PI$,
explicitly represented by the inclusion of $\Inf$ into $\PI$ via a DP map \colorbox{green!30}{$\mathbf{i}:\Inf \to \PI$}. Diagrammatically, we denote this as a green map, e.g.
\beq
\InputIfFileExists{Diagrams/infInclusion.tikz}{}{\input{./figures/Diagrams/infInclusion.tikz}}.
\eeq
That is, $\mathbf{i}$ denotes that some process in $\PI$ is a member of the subprocess theory $\Inf$. For example, in the equation
\beq\label{Eq:ImageOfi}
\begin{tikzpicture}
	\begin{pgfonlayer}{nodelayer}
		\node [style=small box] (0) at (0, -0) {$s$};
		\node [style=none] (1) at (2, -0) {};
		\node [style=up label] (2) at (1.75, -0) {$Y$};
		\node [style=none] (3) at (-2, -0) {};
		\node [style=up label] (4) at (-1.75, -0) {$X$};
	\end{pgfonlayer}
	\begin{pgfonlayer}{edgelayer}
		\draw [style=CcWire] (0) to (1.center);
		\draw [style=CcWire] (3.center) to (0);
	\end{pgfonlayer}
\end{tikzpicture}\quad = \quad %
\InputIfFileExists{Diagrams/infInclusion.tikz}{}{\input{./figures/Diagrams/infInclusion.tikz}},
\eeq
 the process $s$ on the RHS is a process in $\Inf$, shown being mapped by $\mathbf{i}$ to the process $s$ in $\PI$, shown on the LHS.
In this case, $s$ as a process in $\PI$ (on the LHS) is in the image under $\mathbf{i}$ of a process $s$ in $\Inf$ (on the RHS).

However, $\Proc$ is not a sub-process-theory of $\PI$; that is, the DP map \colorbox{blue!30}{$\mathbf{e}:\Proc \to \PI$} is a more complicated embedding. All of the systems from $\Proc$ are directly included as systems within $\PI$. In order to fully define the embedding map $\mathbf{e}$, we must define how the causal and inferential systems in $\PI$ interact.

 To proceed, we discuss the interpretation within $\PI$ of the three fundamental generators of interactions between the causal and inferential systems.

The first generator allows us to specify our state of knowledge about which procedure occurs. There is one such generator for each pair $(\op{A},\op{B})$ of systems,
 depicted as
\beq%
\InputIfFileExists{Diagrams/interactionEpistemicProcedural.tikz}{}{\input{./figures/Diagrams/interactionEpistemicProcedural.tikz}}\label{keygen1} \ .\eeq
We then interpret
\beq \label{gen1}
\begin{tikzpicture}
	\begin{pgfonlayer}{nodelayer}
		\node [style=epiBox] (0) at (0, -0) {};
		\node [style=none] (1) at (0, 1.5) {};
		\node [style=none] (2) at (0, -1.5) {};
		\node [style={right label}] (3) at (0, -1.25) {$\op{A}$};
		\node [style={right label}] (4) at (0, 1.25) {$\op{B}$};
		\node [style=infpoint] (5) at (-2.75, -0) {$\sigma$};
		\node [style={up label}] (6) at (-1.25, -0) {$\morph{\op{A}}{\op{B}}$};
	\end{pgfonlayer}
	\begin{pgfonlayer}{edgelayer}
		\draw [style=qWire] (2.center) to (0);
		\draw [style=qWire] (0) to (1.center);
		\draw [style=CcWire] (5) to (0);
	\end{pgfonlayer}
\end{tikzpicture}
\eeq
as describing that we have state of knowledge $\sigma$ about the transformation procedure taking $\op{A}$ to $\op{B}$. Indeed, $\sigma$ is here a probability distribution over the set $\morph{\op{A}}{\op{B}}$ of transformation procedures. We will denote the delta function state of knowledge on $t\in \Proc$ by $[t]$, so that
\beq
\begin{tikzpicture}
	\begin{pgfonlayer}{nodelayer}
		\node [style=epiBox] (0) at (0, -0) {};
		\node [style=none] (1) at (0, 1.5) {};
		\node [style=none] (2) at (0, -1.5) {};
		\node [style={right label}] (3) at (0, -1.25) {$\op{A}$};
		\node [style={right label}] (4) at (0, 1.25) {$\op{B}$};
		\node [style=infpoint] (5) at (-2.75, -0) {$[t]$};
		\node [style={up label}] (6) at (-1.25, -0) {$\morph{\op{A}}{\op{B}}$};
	\end{pgfonlayer}
	\begin{pgfonlayer}{edgelayer}
		\draw [style=qWire] (2.center) to (0);
		\draw [style=qWire] (0) to (1.center);
		\draw [style=CcWire] (5) to (0);
	\end{pgfonlayer}
\end{tikzpicture}
\eeq
represents certainty that $\op{A}$ transforms into $\op{B}$ via the procedure $t$.

Now, suppose that we have states of knowledge about each of two transformation procedures where the output of one is the only input of the other (so that they are purely cause-effect related), then there is a stochastic map which represents how to update knowledge about each of the two individual transformations to a state of  knowledge of the composite transformation procedure. We denote this stochastic map as
\beq
\begin{tikzpicture} \label{seqrule}
	\begin{pgfonlayer}{nodelayer}
		\node [style=up label] (0) at (0.75, -1) {$\morph{\op{A}}{\op{B}}$};
		\node [style=none] (1) at (0, 1) {};
		\node [style=up label] (2) at (0.75, 1) {$\morph{\op{B}}{\op{C}}$};
		\node [style=none] (3) at (2, -0) {};
		\node [style=up label] (4) at (3, -0) {$\morph{\op{A}}{\op{C}}$};
		\node [style=none] (5) at (3.5, -0) {};
		\node [style=none] (6) at (2, -0) {};
		\node [style=seqComp] (7) at (2, -0) {};
		\node [style=none] (8) at (0, -1) {};
		\node [style=none] (9) at (2, -0) {};
	\end{pgfonlayer}
	\begin{pgfonlayer}{edgelayer}
		\draw [style=CcWire, in=105, out=0, looseness=1.00] (1.center) to (3.center);
		\draw [style=CcWire] (6.center) to (5.center);
		\draw [style=CcWire, in=-105, out=0, looseness=1.00] (8.center) to (9.center);
	\end{pgfonlayer}
\end{tikzpicture} \ ,
\eeq
defined by linearity and its action on delta-function states of knowledge, namely
\beq \forall t, t' \qquad %
\InputIfFileExists{Diagrams/EpistemicSequentialConstraint3.tikz}{}{\input{./figures/Diagrams/EpistemicSequentialConstraint3.tikz}}\quad=\quad%
\InputIfFileExists{Diagrams/EpistemicSequentialConstraint4.tikz}{}{\input{./figures/Diagrams/EpistemicSequentialConstraint4.tikz}}\label{parallelconstr1}\eeq
where $\circ$ denotes sequential composition in $\Proc$.
To reproduce the intuitive notion of composition, we demand that
\beq\label{constraint1FIXED}%
\InputIfFileExists{Diagrams/EpistemicSequentialConstraint.tikz}{}{\input{./figures/Diagrams/EpistemicSequentialConstraint.tikz}}\quad=\quad%
\InputIfFileExists{Diagrams/EpistemicSequentialConstraint2.tikz}{}{\input{./figures/Diagrams/EpistemicSequentialConstraint2.tikz}}.\eeq

This rewrite rule can be understood as the equality of two different methods of specifying one's knowledge that the causal structure is a chain  $A \rightarrow B \rightarrow C$. The fact that $\op{B}$ is a complete causal mediary between $\op{A}$ and $\op{C}$ can be encoded in the causal structure of a diagram (as on the LHS), but the RHS encodes it in the inferential structure, as a state of knowledge about the transformation from $\op{A}$ to $\op{C}$ that is specified in terms of one's state of knowledge about a transformation from $\op{A}$ to $\op{B}$ and about a transformation from $\op{B}$ to $\op{C}$.

Similarly, we can define a stochastic map which represents how one combines states of knowledge about transformations that are causally disconnected.
Specifically, suppose that in a transformation from $\op{A}\op{C}$ to $\op{B}\op{D}$, $\op{B}$ is influenced only by $\op{A}$ and $\op{D}$ is influenced only by $\op{C}$.
The relevant stochastic map,
\beq
\begin{tikzpicture} \label{parrule}
	\begin{pgfonlayer}{nodelayer}
		\node [style=up label] (0) at (0.5, -1.25) {$\morph{\op{C}}{\op{D}}$};
		\node [style=none] (1) at (0, 1.25) {};
		\node [style=up label] (2) at (0.5, 1.25) {$\morph{\op{A}}{\op{B}}$};
		\node [style=none] (3) at (2, -0) {};
		\node [style=up label] (4) at (3.5, -0) {$\morph{\op{AC}}{\op{BD}}$};
		\node [style=none] (5) at (4.75, -0) {};
		\node [style=none] (6) at (2, -0) {};
		\node [style=parComp] (7) at (2, -0) {};
		\node [style=none] (8) at (0, -1.25) {};
		\node [style=none] (9) at (2, -0) {};
	\end{pgfonlayer}
	\begin{pgfonlayer}{edgelayer}
		\draw [style=CcWire, in=105, out=0, looseness=1.00] (1.center) to (3.center);
		\draw [style=CcWire] (6.center) to (5.center);
		\draw [style=CcWire, in=-105, out=0, looseness=1.00] (8.center) to (9.center);
	\end{pgfonlayer}
\end{tikzpicture},
\eeq
 can be defined by linearity and its action on delta-function states of knowledge, namely
\beq\forall t, t' \qquad %
\InputIfFileExists{Diagrams/EpistemicParallelConstraint3.tikz}{}{\input{./figures/Diagrams/EpistemicParallelConstraint3.tikz}}\quad=\quad%
\InputIfFileExists{Diagrams/EpistemicParallelConstraint4.tikz}{}{\input{./figures/Diagrams/EpistemicParallelConstraint4.tikz}}\ , \eeq
where $\otimes$ denotes the parallel composition of processes within $\Proc$.
Then, in analogy to Eq.~\eqref{constraint1FIXED}, we demand that
\beq%
\InputIfFileExists{Diagrams/EpistemicParallelConstraint.tikz}{}{\input{./figures/Diagrams/EpistemicParallelConstraint.tikz}}\quad=\quad%
\InputIfFileExists{Diagrams/EpistemicParallelConstraint2.tikz}{}{\input{./figures/Diagrams/EpistemicParallelConstraint2.tikz}}\label{constraint2}.\eeq

Finally, it will often be useful to be able to interpret some bits of wiring as themselves being processes in $\PI$, namely, the identity procedure $\mathds{1}$ and swap procedure $\mathds{S}$ respectively. We therefore impose that
\beq \label{identityembedding}
\InputIfFileExists{Diagrams/IdentityDelta.tikz}{}{\input{./figures/Diagrams/IdentityDelta.tikz}}\quad = \quad %
\begin{tikzpicture}
	\begin{pgfonlayer}{nodelayer}
		\node [style=none] (0) at (0, 1) {};
		\node [style=none] (1) at (0, -1) {};
	\end{pgfonlayer}
	\begin{pgfonlayer}{edgelayer}
		\draw [qWire] (1.center) to (0.center);
	\end{pgfonlayer}
\end{tikzpicture}}
\eeq
\beq
\InputIfFileExists{Diagrams/SwapDelta.tikz}{}{\input{./figures/Diagrams/SwapDelta.tikz}}\quad = \quad %
\InputIfFileExists{Diagrams/SwapDelta2.tikz}{}{\input{./figures/Diagrams/SwapDelta2.tikz}} \ .
\eeq

Essentially, these constraints allow us to lift the compositional, that is, causal, structure of $\Proc$ {\em into our theory of states of knowledge about $\Proc$}, that is, into $\PI$.  The choice about how to do this -- by focusing on $\circ$ and $\otimes$ -- is somewhat artificial and so we discuss an alternative approach in Appendix~\ref{app:firstgenerator}.  In particular, given the basic generators and constraints we have just introduced, one can construct a DP map \colorbox{blue!30}{$\mathbf{e}:\Proc\to\PI$} which embeds procedures into $\PI$ as delta-function states of knowledge:
\beq\label{defnemap}
\InputIfFileExists{Diagrams/EmbeddingProcToPI2.tikz}{}{\input{./figures/Diagrams/EmbeddingProcToPI2.tikz}}\quad :=\quad %
\InputIfFileExists{Diagrams/EmbeddingProcToPI.tikz}{}{\input{./figures/Diagrams/EmbeddingProcToPI.tikz}} \ .
\eeq
It is simple to check that our constraints on this generator imply that this map is indeed diagram-preserving.
For example,
\begin{align}
\InputIfFileExists{Diagrams/sequenceProof1.tikz}{}{\input{./figures/Diagrams/sequenceProof1.tikz}}\quad &= \quad %
\InputIfFileExists{Diagrams/sequenceProof2.tikz}{}{\input{./figures/Diagrams/sequenceProof2.tikz}}\\
&= \quad %
\InputIfFileExists{Diagrams/sequenceProof3.tikz}{}{\input{./figures/Diagrams/sequenceProof3.tikz}}\\
&= \quad %
\InputIfFileExists{Diagrams/sequenceProof4.tikz}{}{\input{./figures/Diagrams/sequenceProof4.tikz}}\\
&= \quad %
\InputIfFileExists{Diagrams/sequenceProof5.tikz}{}{\input{./figures/Diagrams/sequenceProof5.tikz}}\\
&= \quad %
\InputIfFileExists{Diagrams/sequenceProof6.tikz}{}{\input{./figures/Diagrams/sequenceProof6.tikz}} \ .
\end{align}

The second generator allows us to directly gain knowledge from a classical causal system. There is one such generator for each classical system $\op{X}$:
\beq%
\InputIfFileExists{Diagrams/ProcPropInteraction.tikz}{}{\input{./figures/Diagrams/ProcPropInteraction.tikz}}. \label{genattachprop} \eeq
This can equivalently be interpreted as a generator which allows us to ask a question about a classical system by attaching a proposition to it.
For example, a proposition $\pi$ about the outcome of a measurement $m$ is depicted as
\beq%
\InputIfFileExists{Diagrams/ProcPropInteractionExample.tikz}{}{\input{./figures/Diagrams/ProcPropInteractionExample.tikz}}.\eeq
Note that there is no such generator for systems that are not classical, since for these,  there is no way to directly gain information about the system; rather, one can only probe them indirectly via their interaction with classical systems.

As with the previous generator, this generator must satisfy certain constraints.
First, it must satisfy
\beq \label{trueprop}
\InputIfFileExists{Diagrams/TrivialProposition.tikz}{}{\input{./figures/Diagrams/TrivialProposition.tikz}} \quad = \quad %
\begin{tikzpicture}
	\begin{pgfonlayer}{nodelayer}
		\node [style=none] (0) at (0, 1) {};
		\node [style=none] (1) at (0, -1) {};
		\node [style=scalar] (2) at (1.5, 0) {$\textsf{True}$};
	\end{pgfonlayer}
	\begin{pgfonlayer}{edgelayer}
		\draw [cWire] (1.center) to (0.center);
	\end{pgfonlayer}
\end{tikzpicture}
} \quad = \quad %
\begin{tikzpicture}
	\begin{pgfonlayer}{nodelayer}
		\node [style=none] (0) at (0, 1) {};
		\node [style=none] (1) at (0, -1) {};
	\end{pgfonlayer}
	\begin{pgfonlayer}{edgelayer}
		\draw [cWire] (1.center) to (0.center);
	\end{pgfonlayer}
\end{tikzpicture}}\ \ .
\eeq
 That is, asking about the tautological proposition on a system is the same as not asking anything at all about the system.

Additionally, under sequential composition we demand
\beq%
\InputIfFileExists{Diagrams/ProcPropInteractionDef.tikz}{}{\input{./figures/Diagrams/ProcPropInteractionDef.tikz}} \label{constraint3} \ , \eeq
where $\bullet$ is the stochastic broadcasting map which can be defined by linearity and its action on delta-function states of knowledge, namely,
\beq
\InputIfFileExists{Diagrams/copyDotDef1.tikz}{}{\input{./figures/Diagrams/copyDotDef1.tikz}}\quad := \quad %
\InputIfFileExists{Diagrams/copyDotDef2.tikz}{}{\input{./figures/Diagrams/copyDotDef2.tikz}} \ .
\eeq
Eq.~\eqref{constraint3} states that directly gaining knowledge about the same system twice is the same as copying the knowledge gained from the system.

We also have a constraint for parallel composition:
\beq%
\InputIfFileExists{Diagrams/PropositionalParallelConstraint.tikz}{}{\input{./figures/Diagrams/PropositionalParallelConstraint.tikz}}\label{constraint4} \ ,\eeq
where $\InfSplit$ is really just the identity stochastic map, but where one is changing from diagrammatically denoting
a pair of systems as a single wire to denoting them by a pair of wires. Similarly, we have $\InfMerge$, which merges a pair of wires into a single wire:
\beq
\InputIfFileExists{Diagrams/infMerge.tikz}{}{\input{./figures/Diagrams/infMerge.tikz}} \ .
\eeq

Finally, we introduce our last generator, which represents the ignoring of a causal system as
\beq%
\InputIfFileExists{Diagrams/ignore.tikz}{}{\input{./figures/Diagrams/ignore.tikz}}. \eeq
We depict the process of ignoring a system by the same symbol (albeit smaller) as marginalisation in the inferential theory
to make clear that this is not a physical discarding process (such as physically annihilating a system somehow). It merely represents the fact that one is no longer interested in this system. That is, this {\em ignoring process} is applied whenever an agent decides that they will consider no further propositions about a system or its causal descendents.

The ignoring process satisfies the constraint
\beq%
\InputIfFileExists{Diagrams/ignoreConstraint1.tikz}{}{\input{./figures/Diagrams/ignoreConstraint1.tikz}}\quad =\quad %
\InputIfFileExists{Diagrams/ignoreConstraint2.tikz}{}{\input{./figures/Diagrams/ignoreConstraint2.tikz}}\label{constraint5}\ \ ,\eeq
stating that ignoring a composite system is the same as ignoring each of its components.
Moreover, it has a nontrivial interaction with the generator of Eq.~\eqref{keygen1}, as we demand that
\beq
\InputIfFileExists{Diagrams/ignoreConstraint3.tikz}{}{\input{./figures/Diagrams/ignoreConstraint3.tikz}}\quad =\quad %
\InputIfFileExists{Diagrams/ignoreConstraint4.tikz}{}{\input{./figures/Diagrams/ignoreConstraint4.tikz}}.\label{constraint6}
\eeq
That is, if one is not going to ask any propositional questions about $\op{B}$, then one can ignore the identity of the transformation from $\op{A}$ to $\op{B}$, as well as $\op{A}$ itself.
We term this the constraint of {\em ignorability}.

We will assume that the ignoring process for the trivial system, $\op{I}$, is simply given by the empty diagram and hence we obtain two special cases of Eq.~\eqref{constraint6}:
\beq \label{ignoreempty}
\InputIfFileExists{Diagrams/ignoreConstraintNew1.tikz}{}{\input{./figures/Diagrams/ignoreConstraintNew1.tikz}}\ =\ %
\InputIfFileExists{Diagrams/ignoreConstraintNew2.tikz}{}{\input{./figures/Diagrams/ignoreConstraintNew2.tikz}}\ =\ %
\InputIfFileExists{Diagrams/ignoreConstraintNew3.tikz}{}{\input{./figures/Diagrams/ignoreConstraintNew3.tikz}}\ =\ %
\InputIfFileExists{Diagrams/ignoreConstraintNew4.tikz}{}{\input{./figures/Diagrams/ignoreConstraintNew4.tikz}}
\eeq
and
\beq
\InputIfFileExists{Diagrams/ignoreConstraintNew5.tikz}{}{\input{./figures/Diagrams/ignoreConstraintNew5.tikz}}\ =\ %
\InputIfFileExists{Diagrams/ignoreConstraintNew6.tikz}{}{\input{./figures/Diagrams/ignoreConstraintNew6.tikz}}\ =\ %
\InputIfFileExists{Diagrams/ignoreConstraintNew7.tikz}{}{\input{./figures/Diagrams/ignoreConstraintNew7.tikz}}\ = \ %
\begin{tikzpicture}
	\begin{pgfonlayer}{nodelayer}
		\node [style=none] (0) at (-0.75, -0) {};
		\node [style=none] (1) at (-1.5, -0) {};
		\node [style=infupground] (2) at (-0.5, -0) {};
		\node [style={up label}] (3) at (-1.25, -0) {$\morph{I}{B}$};
	\end{pgfonlayer}
	\begin{pgfonlayer}{edgelayer}
		\draw [CcWire] (1.center) to (0.center);
	\end{pgfonlayer}
\end{tikzpicture}
}\ .
\eeq

Due to its compositional nature, this framework is clearly able to express scenarios far more general than the well-studied prepare-measure scenario. Even within a simple prepare-measure scenario, our framework allows us to express generality that is typically neglected. In a prepare-measure scenario, the conventional states of knowledge one has and propositions one considers (i.e., the conventional inferential structure) are represented in our framework by
\beq%
\InputIfFileExists{Diagrams/NormalPM.tikz}{}{\input{./figures/Diagrams/NormalPM.tikz}},\label{notgenpmopdiag}\eeq
where $\sigma$ is a state of knowledge about the preparation of $\op{A}$, $\tau$ is a state of knowledge about the measurement on $\op{A}$, and $\pi$ is a proposition about the measurement outcome $\op{X}$. One can change the inferential structure of the scenario without changing the causal structure of the scenario, e.g., to
\beq%
\InputIfFileExists{Diagrams/GeneralPM.tikz}{}{\input{./figures/Diagrams/GeneralPM.tikz}}\label{genpmopdiag}\ , \eeq
where $\sigma$ is a joint state of knowledge about the preparation and measurement procedures, $\tau$ is a state of knowledge about some auxiliary inferential system, and $\pi$ is a joint proposition about which 
 preparation and measurement 
 procedures were performed, the outcome of the measurement, and the auxiliary inferential system.
This extra generality is useful as it allows one to model scenarios where the choice of preparation and the choice of measurement are correlated.  This occurs, for instance, in two-party cryptography, wherein one party prepares a system and the other measures it, and they correlate their actions based on private randomness which they share.  Another example arises for a pair of communicating parties when the preparations and measurements are done relative to local reference frames in the labs of the parties, and where these are correlated with one another, but uncorrelated with a background reference frame.

It will be helpful to introduce a notation which represents a generic diagram in $\PI$ while not displaying all the internal structure---that is, how the diagram is built up out of the generators---but rather only shows its open inputs and outputs. We draw such a generic diagram as
\beq\label{mostgenproc}
\InputIfFileExists{Diagrams/genericProcess.tikz}{}{\input{./figures/Diagrams/genericProcess.tikz}}.
\eeq
 In a process theory, diagrams without any inputs and outputs are termed {\em closed diagrams}.
 Analogously, diagrams
 whose only inputs and outputs are inferential will be termed {\em causally closed} diagrams and diagrams
   whose only inputs and outputs are causal will be termed {\em inferentially closed} diagrams.

 At this point, we can make a useful observation: any diagram in $\PI$ which can be written using generators that do not involve causal systems can be considered as the image of some process in $\Inf$ under $\mathbf{i}$,
 as in Eq.~\eqref{Eq:ImageOfi}.

We have so far introduced the inferential and causal components of $\PI$ and the maps from these into 
 $\PI$, namely,
\beq%
\InputIfFileExists{Diagrams/OperationalTheory.tikz}{}{\input{./figures/Diagrams/OperationalTheory.tikz}}.\eeq

We now introduce the prediction map  $\mathbf{p}$, which describes the inferences one can make in a given scenario.\footnote{The idea of separating out the descriptive and the probabilistic components of one's notion of an operational theory can be found in earlier works, notably Ref.~\cite{chiribella2010probabilistic}. }  
$\mathbf{p}$ is a partial DP map whose domain is the set of
causally closed processes in $\PI$, and whose co-domain is $\Inf$.

 An example of the sort of causally closed processes on which $\mathbf{p}$ is defined is:
\beq%
\InputIfFileExists{Diagrams/GeneralHorizontal.tikz}{}{\input{./figures/Diagrams/GeneralHorizontal.tikz}},\eeq
A diagram which has open causal inputs or outputs is {\em not} in the domain of the prediction map, because the open causal wires correspond to systems about which either no state of knowledge or no propositional question has been specified.
For example, the inferences one should make in a situation described by the diagram
\beq %
\InputIfFileExists{Diagrams/FreeInput.tikz}{}{\input{./figures/Diagrams/FreeInput.tikz}}\eeq
depend on what one knows about previous procedures on $\op{A}$ as well as any propositions one considers about $\op{X}$.

Consider first the simple case of closed diagrams
 in $\PI$.   These are mapped to closed diagrams (scalars)
  in $\Inf$---i.e., elements of $[0,1]$. In the following example, 
\beq%
\InputIfFileExists{Diagrams/predictionMapScalar.tikz}{}{\input{./figures/Diagrams/predictionMapScalar.tikz}}\quad =\ \prob(\pi :\sigma),\eeq
 the prediction map specifies the probability that one should assign to the proposition $\pi$ being true given that one's state of knowledge is $\sigma$.
Meanwhile,
\beq%
\InputIfFileExists{Diagrams/etaNotFullExampleNew.tikz}{}{\input{./figures/Diagrams/etaNotFullExampleNew.tikz}}\eeq
is a stochastic map in $\Inf$,
 which takes a state of knowledge about the preparation $\op{X}$ as input and returns a state of knowledge about $X$.

There is an obvious consistency constraint on processes in $\PI$ which are also processes in the sub-process-theory $\Inf$, that is, those that are in the image of $\mathbf{i}$. If one maps a process in $\Inf$ to $\PI$ by the inclusion map $\mathbf{i}$ and then back to $\Inf$ via the prediction map $\mathbf{p}$, one should clearly obtain the process itself back again, so that
\beq\label{eq:constraintOnPredictions}%
\InputIfFileExists{Diagrams/inferentialProcessInclusion.tikz}{}{\input{./figures/Diagrams/inferentialProcessInclusion.tikz}}\quad =\quad %
\InputIfFileExists{Diagrams/inferentialProcessCopy.tikz}{}{\input{./figures/Diagrams/inferentialProcessCopy.tikz}} \ . \eeq
Hence it is a partial left inverse of $\mathbf{i}$, that is, $\mathbf{p}\circ \mathbf{i} = \mathds{1}_{\Inf}$.

Although $\mathbf{p}$ is only a partial map, it is still diagram-preserving on its domain; e.g., one can write
\begin{align}
\InputIfFileExists{Diagrams/predictionMapScalar.tikz}{}{\input{./figures/Diagrams/predictionMapScalar.tikz}}\quad &=\quad %
\InputIfFileExists{Diagrams/predictionMapDP.tikz}{}{\input{./figures/Diagrams/predictionMapDP.tikz}}.
\end{align}

In summary, an operational CI theory is specified by a triple of process theories and a triple of DP maps between them, succinctly drawn as
\beq%
\InputIfFileExists{Diagrams/OpInferentialEquiv.tikz}{}{\input{./figures/Diagrams/OpInferentialEquiv.tikz}},\eeq
where we use a dashed line to denote the fact that $\mathbf{p}$ is partial.

\subsubsection{Properties of the prediction map}

Our constraint of ignorability, Eq.~\eqref{constraint6}, implies that the probabilities assigned to propositions about systems are independent of what is known about the future processes applied to the system. For example, the probability
\begin{align}
\InputIfFileExists{Diagrams/CausalPredictions1.tikz}{}{\input{./figures/Diagrams/CausalPredictions1.tikz}}\quad &=\quad %
\InputIfFileExists{Diagrams/CausalPredictions2.tikz}{}{\input{./figures/Diagrams/CausalPredictions2.tikz}}\\
&=\quad%
\InputIfFileExists{Diagrams/CausalPredictions3.tikz}{}{\input{./figures/Diagrams/CausalPredictions3.tikz}} \label{causpred3}
 \end{align}
 is seen to be independent of the state of knowledge $\tau$.
In Ref.~\cite{chiribella2010probabilistic}, this constraint is termed `causality', and taken to be of central importance. In our framework, however, it does not express any notion of causality. As we discuss further in Appendix~\ref{manifestorcausal2}, the causal structure in our framework is primitive, and {\em cannot} be defined in terms of any probabilistic facts such as those expressed by Eq.~\eqref{causpred3}.
In our framework, the condition of ignorability, Eq.~\eqref{constraint6}, does not play a particularly special role; it is simply a fact about the way one makes inferences. In addition to implying Eq.~\eqref{causpred3}, it implies many similar independence relations. For example, it implies that a state of knowledge $\tau$ about a causal process occurring on one subsystem of a composite 
whose output is ignored is irrelevant for making inferences about the other  subsystem: 
 \begin{align}
\InputIfFileExists{Diagrams/BellNoSignalling.tikz}{}{\input{./figures/Diagrams/BellNoSignalling.tikz}}\quad &= \quad %
\InputIfFileExists{Diagrams/BellNoSignalling1.tikz}{}{\input{./figures/Diagrams/BellNoSignalling1.tikz}}\\
 &= \quad %
\InputIfFileExists{Diagrams/BellNoSignalling2.tikz}{}{\input{./figures/Diagrams/BellNoSignalling2.tikz}} \ .
 \end{align}

Conveniently, $\mathbf{p}$ can be fully specified by a relatively simple set of data: the probabilities assigned to point-distributed states of knowledge and atomic propositions. This is exactly the form of data provided in traditional approaches to operational theories.

\begin{theorem}\label{Thm:PointAndAtomicSuffice}
 For every process $\mathcal{D}\in\PI$ in the domain of \colorbox{red!30}{$\mathbf{p}:\PI\to\Inf$}, i.e., which is causally closed, the image $\mathbf{p}(\mathcal{D})$ of $\mathcal{D}$ under $\mathbf{p}$ is fully specified by the probabilities assigned to atomic propositions on its inferential output and point distributions on its inferential input.
\end{theorem}
\proof
Consider an arbitrary causally closed process $\mathcal{D}$
and imagine mapping it into $\Inf$ via
\beq\label{processforproof}
%
\InputIfFileExists{Diagrams/Rep1.tikz}{}{\input{./figures/Diagrams/Rep1.tikz}} \ .
\eeq
This is simply a substochastic  map 
 and hence is fully characterized by the set of scalars
\beq
\left\{%
\InputIfFileExists{Diagrams/Rep2New.tikz}{}{\input{./figures/Diagrams/Rep2New.tikz}} \right\}_{x\in X, y\in Y}\ ,
\eeq
 where $\llbracket y \rrbracket$ is an atomic proposition on $Y$. 
Such scalars can be rewritten as
\begin{align}\label{atomicproof}
\InputIfFileExists{Diagrams/Rep2New.tikz}{}{\input{./figures/Diagrams/Rep2New.tikz}} \quad &=  \quad  %
\InputIfFileExists{Diagrams/Rep3New.tikz}{}{\input{./figures/Diagrams/Rep3New.tikz}}\\
&=\quad  %
\InputIfFileExists{Diagrams/Rep4New.tikz}{}{\input{./figures/Diagrams/Rep4New.tikz}};
\end{align}
 that is, they are the probabilities assigned to atomic propositions $\llbracket y \rrbracket$ on $Y$ given point distributions $[x]$ on $X$, which is what we set out to prove.  
\endproof

\subsubsection{Quantum theory as an operational CI theory} \label{QTquaopCI}

The most straightforward way to cast quantum theory as an operational CI theory is as follows.
 The causal subtheory for quantum theory, which we denote $\Proc_{Q}$, contains laboratory procedures whose inputs and outputs are classical and quantum systems. The inferential subtheory is the classical one, $\Inf$. The full theory, $\PQI$, is constructed as
\beq  \label{PQIfull}
\begin{tikzpicture}
	\begin{pgfonlayer}{nodelayer}
		\node [style=none] (0) at (5, 0) {$\Inf$};
		\node [style=none] (1) at (0.25, 0) {$\PQI$};
		\node [style=none] (2) at (3.25, 0) {};
		\node [style=none] (3) at (1.25, 0) {};
		\node [style=none] (4) at (3.25, -0.25) {};
		\node [style=none] (5) at (1.25, -0.25) {};
		\node [style=none] (6) at (-0.75, 0) {};
		\node [style=none] (7) at (-2.75, 0) {};
		\node [style=none] (8) at (-3.75, 0) {$\Proc_{Q}$};
		\node [style=up label] (9) at (-1.75, 0) {$\mathbf{e}_Q$};
		\node [style=up label] (10) at (2.25, 0) {$\mathbf{i}_Q$};
		\node [style=up label] (11) at (2.25, -0.85) {$\mathbf{p}_Q$};
	\end{pgfonlayer}
	\begin{pgfonlayer}{edgelayer}
		\draw [style=arrow plain] (2.center) to (3.center);
		\draw [style=arrow dashed] (5.center) to (4.center);
		\draw [style=arrow plain] (7.center) to (6.center);
	\end{pgfonlayer}
\end{tikzpicture}
 \ .
\eeq
Here, the specific prediction map $\mathbf{p}_Q$ singles out quantum theory, and is defined as follows. To every quantum system is associated an algebra of operators on a complex Hilbert space of some dimension, and to every classical system is associated an algebra of {\em commuting} operators on such a Hilbert space; to every diagram is associated a completely-positive~\cite{Schmidcausal} trace-nonincreasing map between these; then, the joint probability distributions (on any set of propositions attached to the classical systems) can be computed by composition of these completely-positive trace-preserving maps.

\subsection{ \Crealist causal-inferential theories}\label{sec:Ont}

Next, we turn our attention to the second class of causal-inferential theories that we will consider, namely \crealist CI theories. These are very similar to the operational CI theories just introduced,
but the causal theory is not taken to be a process theory $\Proc$ of laboratory procedures, but rather  a process theory representing fundamental dynamics of ontic states of systems. In our case, we will take this to be the process theory $\Func$ of functional dynamics, introduced in Section~\ref{funcdyn}. The extra structure in $\Func$ relative to $\Proc$ accounts for all the differences within our framework between an operational CI theory and a \crealist CI theory, and implies that there is essentially a unique classical realist CI theory, insofar as there is a unique prediction map.

We will use the term `\crealist CI theory' to refer to the following causal-inferential theory of functional dynamics, namely
\beq%
\InputIfFileExists{Diagrams/OntInferentialEquiv.tikz}{}{\input{./figures/Diagrams/OntInferentialEquiv.tikz}}\label{OntInferentialEquiv}.\eeq
We have labeled the diagram-preserving maps here by $\mathbf{e'}$, $\mathbf{i'}$ and $\mathbf{p^*}$ to distinguish them from those in an operational CI theory.

The construction proceeds much like that in the previous section. $\Inf$ is a subprocess theory of $\FI$, explicitly represented by the inclusion of $\Inf$ into $\FI$ via a DP map \colorbox{darkgreen!30}{$\mathbf{i'}:\Inf\to\FI$}, diagrammatically represented as
\beq%
\InputIfFileExists{Diagrams/infInclusionToFunc.tikz}{}{\input{./figures/Diagrams/infInclusionToFunc.tikz}}.\eeq
$\Func$ is not a sub-process-theory of $\FI$, but rather embeds into $\FI$ via a map \colorbox{darkblue!30}{$\mathbf{e'}:\Func\to\FI$}, which we will define after introducing some relevant generators.

The first generator again allows one to specify a state of knowledge about the functional dynamics. There is one such generator for each pair of systems $(\Lambda,\Lambda')$, depicted as
\beq%
\InputIfFileExists{Diagrams/interactionEpistemicFunctional.tikz}{}{\input{./figures/Diagrams/interactionEpistemicFunctional.tikz}}.\eeq
Then, the diagram
\beq
\begin{tikzpicture}
	\begin{pgfonlayer}{nodelayer}
		\node [style=epiBox] (0) at (0, -0) {};
		\node [style=none] (1) at (0, 1.5) {};
		\node [style=none] (2) at (0, -1.5) {};
		\node [style={right label}] (3) at (0, -1.25) {$\Lambda$};
		\node [style={right label}] (4) at (0, 1.25) {$\Lambda'$};
		\node [style=infpoint] (5) at (-2.75, -0) {$\sigma$};
		\node [style={up label}] (6) at (-1.25, -0) {$\morph{\Lambda}{\Lambda'}$};
	\end{pgfonlayer}
	\begin{pgfonlayer}{edgelayer}
		\draw [style=oWire] (2.center) to (0);
		\draw [style=oWire] (0) to (1.center);
		\draw [style=CcWire] (5) to (0);
	\end{pgfonlayer}
\end{tikzpicture}
\eeq
represents the state of knowledge $\sigma$ about the function from $\Lambda$ to $\Lambda'$ describing the dynamics. Naturally, we demand that constraints analogous to those in Eq.~\eqref{constraint1FIXED} and Eq.~\eqref{constraint2} are satisfied, which then implies that we can construct a DP map \colorbox{darkblue!30}{$\mathbf{e'}:\Func\to\FI$} defined as
\beq \label{defneprime}
\InputIfFileExists{Diagrams/EmbeddintFuncToFI2.tikz}{}{\input{./figures/Diagrams/EmbeddintFuncToFI2.tikz}}\quad :=\quad %
\InputIfFileExists{Diagrams/EmbeddingFuncToFI.tikz}{}{\input{./figures/Diagrams/EmbeddingFuncToFI.tikz}}.
\eeq

The second generator allows us to directly gain knowledge from an ontological system, 
or equivalently, to ask a question about a system by attaching a proposition to it. Here, we see the first key distinction between ontological and operational CI theories---for operational CI theories, we could only define such a generator for classical systems;
however, 
because all systems in $\Func$ are sets $\Lambda$, this generator
\beq \label{propont}
\InputIfFileExists{Diagrams/FuncPropInteraction.tikz}{}{\input{./figures/Diagrams/FuncPropInteraction.tikz}}\eeq
 can be defined for any system in $\FI$.
Naturally, we demand that each such generator satisfies 
 the constraints stipulated in Eqs.~\eqref{constraint3} and ~\eqref{constraint4}.

Finally, we introduce a generator
\beq%
\begin{tikzpicture}
	\begin{pgfonlayer}{nodelayer}
		\node [style=none] (0) at (0, -0.5000001) {};
		\node [style={small black dot}] (1) at (0, 0.7500002) {};
		\node [style={right label}] (2) at (0, -0.2500001) {$\Lambda$};
		\node [style=ignore] (3) at (0, 0.7500001) {};
	\end{pgfonlayer}
	\begin{pgfonlayer}{edgelayer}
		\draw [oWire] (0.center) to (1);
	\end{pgfonlayer}
\end{tikzpicture}
} \eeq
which represents ignoring the system $\Lambda$ and which satisfies constraints analogous to Eq.~\eqref{constraint5} and Eq.~\eqref{constraint6}.

We now have the tools to describe a wide range of scenarios. For example, the scenario
\beq%
\InputIfFileExists{Diagrams/NormalPMOntolMod.tikz}{}{\input{./figures/Diagrams/NormalPMOntolMod.tikz}}\eeq
might arise as a \crealist model of the operational scenario in Diagram~\eqref{notgenpmopdiag}.
This is analogous to a prepare-measure scenario.
Even in this simple causal structure, however, we can also describe more general inferential structures; for example, an analogue of Diagram~\eqref{genpmopdiag}, namely
\beq%
\InputIfFileExists{Diagrams/GeneralPMOnt.tikz}{}{\input{./figures/Diagrams/GeneralPMOnt.tikz}}.\eeq
In fact, this is even more general than Diagram~\eqref{genpmopdiag}, since in $\FI$ (unlike in $\PI$), one can consider propositions about {\em arbitrary} systems.

Perhaps the central distinction between $\PI$ and $\FI$ is that in $\FI$, there is a constraint on the interactions between the first two generators we introduced. This constraint is a consequence of the fact that one can attach propositions to any system in $\FI$, together with our assumption that the causal mechanisms are described by functions.   This means that we can, in certain situations, propagate what we know about one physical system to knowledge about another.  It is due to this single extra constraint that \crealist CI theories have so much more structure than operational CI theories. 

This interaction between the two generators is governed by the equality: 
\beq\label{Axiom:PropKnowGenerators}
\InputIfFileExists{Diagrams/OntConstraint1.tikz}{}{\input{./figures/Diagrams/OntConstraint1.tikz}} \quad = \quad  %
\InputIfFileExists{Diagrams/OntConstraint2.tikz}{}{\input{./figures/Diagrams/OntConstraint2.tikz}},
\eeq
where the black diamond converts a state of knowledge about $\Lambda$ and a state of knowledge about $\morph{\Lambda}{\Lambda'}$ into a state of knowledge about $\Lambda'$, and is defined by linearity and its action on delta-function states of knowledge, namely,
\beq
\begin{tikzpicture}
	\begin{pgfonlayer}{nodelayer}
		\node [style=infpoint] (0) at (-2.5, 1.75) {$[f]$};
		\node [style=none] (1) at (0, 0) {};
		\node [style=none] (2) at (0, 0) {};
		\node [style=infpoint] (3) at (-2.5, -0) {$[\lambda]$};
		\node [style=funcApp] (4) at (0, 0) {};
		\node [style=none] (5) at (1.5, 0) {};
		\node [style={up label}] (6) at (1.5, 0) {$\Lambda'$};
		\node [style={up label}] (7) at (-1.5, -0) {$\Lambda$};
		\node [style={up label}] (8) at (-1.25, 1.75) {$\morph{\Lambda}{\Lambda'}$};
	\end{pgfonlayer}
	\begin{pgfonlayer}{edgelayer}
		\draw [CcWire, in=0, out=120, looseness=1.25] (1.center) to (0);
		\draw [CcWire] (2.center) to (3);
		\draw [CcWire] (5.center) to (1.center);
	\end{pgfonlayer}
\end{tikzpicture}
 \quad = \quad
\begin{tikzpicture}
	\begin{pgfonlayer}{nodelayer}
		\node [style=none] (1) at (-0.5, 0) {};
		\node [style=none] (2) at (-0.5, 0) {};
		\node [style=infpoint] (4) at (-0.5, 0) {$[f(\lambda)]$};
		\node [style=none] (5) at (2, 0) {};
		\node [style=up label] (6) at (1.75, 0) {$\Lambda'$};
	\end{pgfonlayer}
	\begin{pgfonlayer}{edgelayer}
		\draw [CcWire] (5.center) to (1.center);
	\end{pgfonlayer}
\end{tikzpicture},
\eeq
or equivalently,
\beq
\begin{tikzpicture} 
	\begin{pgfonlayer}{nodelayer}
		\node [style=infpoint] (0) at (-2.5, 1.25) {$[f]$};
		\node [style=none] (1) at (0, 0) {};
		\node [style=none] (2) at (0, 0) {};
		\node [style=none] (3) at (-3, -0) {};
		\node [style=funcApp] (4) at (0, 0) {};
		\node [style=none] (5) at (1.5, 0) {};
		\node [style={up label}] (6) at (1.5, 0) {$\Lambda'$};
		\node [style={up label}] (7) at (-3, -0) {$\Lambda$};
		\node [style={up label}] (8) at (-1.25, 1.25) {$\morph{\Lambda}{\Lambda'}$};
	\end{pgfonlayer}
	\begin{pgfonlayer}{edgelayer}
		\draw [CcWire, in=0, out=120, looseness=1.25] (1.center) to (0);
		\draw [CcWire] (2.center) to (3);
		\draw [CcWire] (5.center) to (1.center);
	\end{pgfonlayer}
\end{tikzpicture} \quad =\quad
\begin{tikzpicture}
	\begin{pgfonlayer}{nodelayer}
		\node [style=none] (0) at (0, 0) {};
		\node [style=none] (1) at (0, 0) {};
		\node [style=none] (2) at (-1.5, -0) {};
		\node [style=small box] (3) at (0, 0) {$f$};
		\node [style=none] (4) at (1.5, 0) {};
		\node [style={up label}] (5) at (1.5, 0) {$\Lambda'$};
		\node [style={up label}] (6) at (-1.5, -0) {$\Lambda$};
	\end{pgfonlayer}
	\begin{pgfonlayer}{edgelayer}
		\draw [CcWire] (1.center) to (2.center);
		\draw [CcWire] (4.center) to (0.center);
	\end{pgfonlayer}
\end{tikzpicture}.
\eeq

Eq.~\eqref{Axiom:PropKnowGenerators} ensures that one can 
specify what one knows about the output $\Lambda'$ of some dynamical process in one of two equivalent ways: either by knowing about it directly, or by taking what is jointly known about the dynamics and the state that is input to the dynamics, and then propagating one's beliefs accordingly (i.e., according to the stochastic map $\BlackDiamond$). To give a simple example, suppose we have a delta function state of knowledge that the ingoing system $\Lambda$ is prepared in state $\lambda$, and that the functional dynamics are given by $f$, then this rewrite rule allows us to reason as follows: 
\begin{align}
%
\InputIfFileExists{Diagrams/ontInt1.tikz}{}{\input{./figures/Diagrams/ontInt1.tikz}} 
&=%
\InputIfFileExists{Diagrams/ontInt2.tikz}{}{\input{./figures/Diagrams/ontInt2.tikz}} \label{ontInt2}\\
&=%
\InputIfFileExists{Diagrams/ontInt3.tikz}{}{\input{./figures/Diagrams/ontInt3.tikz}} \label{ontInt3} \\
&=%
\InputIfFileExists{Diagrams/ontInt4.tikz}{}{\input{./figures/Diagrams/ontInt4.tikz}} \label{ontInt4}\\
&=%
\InputIfFileExists{Diagrams/ontInt5.tikz}{}{\input{./figures/Diagrams/ontInt5.tikz}} \label{ontInt5}\\
&=%
\InputIfFileExists{Diagrams/ontInt6.tikz}{}{\input{./figures/Diagrams/ontInt6.tikz}} \label{ontInt6}\\
&=%
\InputIfFileExists{Diagrams/ontInt7.tikz}{}{\input{./figures/Diagrams/ontInt7.tikz}} \\
&=%
\InputIfFileExists{Diagrams/ontInt8.tikz}{}{\input{./figures/Diagrams/ontInt8.tikz}} \label{ontInt8}\\
&=%
\InputIfFileExists{Diagrams/ontInt9.tikz}{}{\input{./figures/Diagrams/ontInt9.tikz}}.\label{ontInt9} 
\end{align}
From Eq.~\eqref{ontInt4} to Eq.~\eqref{ontInt5}, the black diamond turns a delta function state of knowledge about functional dynamics from $\Lambda$ to $\Lambda'$ into a (functional) propagation of one's state of knowledge about $\Lambda$ to one's state of knowledge about $\Lambda'$. 
The rewrites from Eq.~\eqref{ontInt5}-Eq.~\eqref{ontInt9} are analogous to Eq.~\eqref{ontInt2}-Eq.~\eqref{ontInt5}, but slightly more subtle insofar as they involve the special case where an input is trivial.

For this special case with a trivial input system, Eq.~\eqref{Axiom:PropKnowGenerators} becomes
\beq
\InputIfFileExists{Diagrams/PropsDistinguishStates1.tikz}{}{\input{./figures/Diagrams/PropsDistinguishStates1.tikz}} \quad = \quad %
\InputIfFileExists{Diagrams/PropsDistinguishStates2.tikz}{}{\input{./figures/Diagrams/PropsDistinguishStates2.tikz}} \ .
\eeq
where, in this special case, $\BlackDiamond$ is simply the isormorphism between $\morph{\star}{\Lambda}$ and $\Lambda$.
For convenience, we will henceforth denote this isomorphism by
\beq
\InputIfFileExists{Diagrams/addStar.tikz}{}{\input{./figures/Diagrams/addStar.tikz}}\quad \text{and}\quad%
\InputIfFileExists{Diagrams/removeStar.tikz}{}{\input{./figures/Diagrams/removeStar.tikz}},
\eeq
so that
\beq
\InputIfFileExists{Diagrams/removeStar2.tikz}{}{\input{./figures/Diagrams/removeStar2.tikz}}\quad=\quad%
\InputIfFileExists{Diagrams/removeStar.tikz}{}{\input{./figures/Diagrams/removeStar.tikz}}
\eeq
and
\beq \label{blackdotstar}
\InputIfFileExists{Diagrams/PropsDistinguishStates1.tikz}{}{\input{./figures/Diagrams/PropsDistinguishStates1.tikz}} \quad = \quad \begin{tikzpicture}
	\begin{pgfonlayer}{nodelayer}
		\node [style=epiPoint] (0) at (0, 0) {};
		\node [style=none] (1) at (0, 1) {};
		\node [style=copy] (2) at (-1, -0.5000001) {};
		\node [style=none] (3) at (0, -1) {};
		\node [style={right label}] (4) at (0, 0.75) {$\Lambda$};
		\node [style=none] (5) at (0, -1) {};
		\node [style=removeStar] (6) at (0, -1) {};
		\node [style=none] (7) at (1.25, -1) {};
		\node [style={up label}] (8) at (1.25, -1) {$\Lambda$};
		\node [style=none] (9) at (-2.5, -0.5000001) {};
		\node [style={up label}] (10) at (-2.5, -0.5000001) {$\morph{\star}{\Lambda}$};
	\end{pgfonlayer}
	\begin{pgfonlayer}{edgelayer}
		\draw [oWire] (0) to (1.center);
		\draw [CcWire, in=-60, out=180, looseness=1.25] (3.center) to (2);
		\draw [CcWire] (7.center) to (3.center);
		\draw [cWire] (9.center) to (2);
		\draw [cWire, in=180, out=60, looseness=1.25] (2) to (0);
	\end{pgfonlayer}
\end{tikzpicture}.
\eeq

Predictions are made in a \crealist CI theory in a manner analogous to how predictions are made in an operational CI theory. They are represented by a partial diagram-preserving map, \colorbox{darkred!30}{$\mathbf{p^*}:\FI\to\Inf$}, whose domain is given by the set of causally closed processes in $\FI$. 
As before, the prediction map is a partial left inverse of $\mathbf{i'}$, so that $\mathbf{p^*}\circ \mathbf{i'} = \mathds{1}_{\Inf}$.
For example, 
closed diagrams in $\FI$ are mapped to closed diagrams (scalars) in $\Inf$---elements of $[0,1]$; e.g.,
\beq%
\InputIfFileExists{Diagrams/predictionMapScalarOnt.tikz}{}{\input{./figures/Diagrams/predictionMapScalarOnt.tikz}}\quad = \mathsf{Prob}(\pi:\sigma) \ .
\eeq
Meanwhile,
\beq%
\InputIfFileExists{Diagrams/etaNotFullExampleOnt.tikz}{}{\input{./figures/Diagrams/etaNotFullExampleOnt.tikz}}\eeq
is a stochastic map in $\Inf$,
as before.
Analogous to Theorem~\ref{Thm:PointAndAtomicSuffice}, one has that for every process $\mathcal{D}\in\FI$ in the domain of the prediction map $\mathbf{p^*}$, its image under $\mathbf{p^*}$ is fully specified by the probabilities assigned to atomic propositions on its output given point distributions on its input.

There is a key difference between the prediction map in a \crealist and in an operational CI theory: for \crealist CI theories, this map is unique.
To show this, we first prove a normal form for general diagrams in $\FI$.

\begin{theorem} \label{thmnormalform}  Any diagram in the \crealist CI theory $\FI$ can be rewritten (using rewrite rules in $\FI$) into the  form
\beq%
\InputIfFileExists{Diagrams/NF1.tikz}{}{\input{./figures/Diagrams/NF1.tikz}} \ ,\eeq
where $S$ is a substochastic map in $\Inf$.
\end{theorem}
\proof
See Appendix \ref{App:NF}.
\endproof

We conjecture that this normal form is unique, or equivalently, that the substochastic map $S$ is unique. (To prove this, it would suffice to prove that the normal form description of each generator is unique, since the composition of two diagrams in normal forms has a unique normal form description.)

Note that there is {\em not} an equivalent normal form for diagrams in operational CI theories, as Theorem~\ref{thmnormalform}  strongly relies on the constraint of Eq.~\eqref{Axiom:PropKnowGenerators}. This normal form then allows us to prove that the interactions between $\Inf$ and $\Func$ single out a unique prediction map for the full theory $\FI$.

\begin{theorem} \label{uniquerule}
The prediction map $\mathbf{p^*}$ is unique.
\end{theorem}
\proof
Consider an arbitrary process in the domain of $\mathbf{p^*}$---that is, an arbitrary causally closed process $\mathcal{D}$.
Writing it in normal form, we have
\beq
\begin{tikzpicture}
	\begin{pgfonlayer}{nodelayer}
		\node [style=none] (0) at (-0.5, 0.5) {};
		\node [style=none] (1) at (-0.5, -0.5) {};
		\node [style=none] (2) at (0.5, -0.5) {};
		\node [style=none] (3) at (0.5, 0.5) {};
		\node [style=none] (4) at (0, 0) {$\mathcal{D}$};
		\node [style=none] (9) at (1.5, 0) {};
		\node [style=none] (10) at (0.5, 0) {};
		\node [style=none] (11) at (-0.5, 0) {};
		\node [style=none] (12) at (-1.5, 0) {};
	\end{pgfonlayer}
	\begin{pgfonlayer}{edgelayer}
			\filldraw [fill=white,draw] (0.center) to (1.center) to (2.center) to (3.center) to cycle;
		\draw [CcWire] (9.center) to (10.center);
		\draw [CcWire] (11.center) to (12.center);
	\end{pgfonlayer}
\end{tikzpicture}
= \quad
\begin{tikzpicture}
	\begin{pgfonlayer}{nodelayer}
		\node [style=none] (0) at (-0.5, 0.5) {};
		\node [style=none] (1) at (-0.5, -0.5) {};
		\node [style=none] (2) at (0.5, -0.5) {};
		\node [style=none] (3) at (0.5, 0.5) {};
		\node [style=none] (4) at (0, 0) {$S$};
		\node [style=none] (9) at (1.75, 0) {};
		\node [style=none] (10) at (0.5, 0) {};
		\node [style=none] (11) at (-0.5, 0) {};
		\node [style=none] (12) at (-1.75, 0) {};
		\node [style=none] (13) at (1, -1) {\footnotesize $\mathbf{i'}$};
		\node [style=none] (14) at (-1, 1) {};
		\node [style=none] (15) at (1.25, 1) {};
		\node [style=none] (16) at (1.25, -1.25) {};
		\node [style=none] (17) at (-1, -1.25) {};
	\end{pgfonlayer}
	\begin{pgfonlayer}{edgelayer}
			\filldraw [fill=darkgreen!30,draw=darkgreen!60] (14.center) to (15.center) to (16.center) to (17.center) to cycle;
		\filldraw [fill=white,draw] (0.center) to (1.center) to (2.center) to (3.center) to cycle;
		\draw [CcWire] (9.center) to (10.center);
		\draw [CcWire] (11.center) to (12.center);
	\end{pgfonlayer}
\end{tikzpicture} \ ,
\eeq
for some substochastic map $S$. Furthermore, $S$ is unique since $\mathbf{i'}$ is an inclusion map and hence injective.
Applying the prediction map, then, one has
\beq
\begin{tikzpicture}
	\begin{pgfonlayer}{nodelayer}
		\node [style=none] (0) at (-0.5, 0.5) {};
		\node [style=none] (1) at (-0.5, -0.5) {};
		\node [style=none] (2) at (0.5, -0.5) {};
		\node [style=none] (3) at (0.5, 0.5) {};
		\node [style=none] (4) at (0, 0) {$\mathcal{D}$};
		\node [style=none] (9) at (1.75, 0) {};
		\node [style=none] (10) at (0.5, 0) {};
		\node [style=none] (11) at (-0.5, 0) {};
		\node [style=none] (12) at (-1.75, 0) {};
		\node [style=none] (13) at (1, -1) {\footnotesize $\mathbf{p^*}$};
		\node [style=none] (14) at (-1, 1) {};
		\node [style=none] (15) at (1.25, 1) {};
		\node [style=none] (16) at (1.25, -1.25) {};
		\node [style=none] (17) at (-1, -1.25) {};
	\end{pgfonlayer}
	\begin{pgfonlayer}{edgelayer}
			\filldraw [fill=darkred!30,draw=darkred!60] (14.center) to (15.center) to (16.center) to (17.center) to cycle;
		\filldraw [fill=white,draw] (0.center) to (1.center) to (2.center) to (3.center) to cycle;
		\draw [CcWire] (9.center) to (10.center);
		\draw [CcWire] (11.center) to (12.center);
	\end{pgfonlayer}
\end{tikzpicture}
\quad = \quad
\begin{tikzpicture}
	\begin{pgfonlayer}{nodelayer}
		\node [style=none] (0) at (-0.5, 0.5) {};
		\node [style=none] (1) at (-0.5, -0.5) {};
		\node [style=none] (2) at (0.5, -0.5) {};
		\node [style=none] (3) at (0.5, 0.5) {};
		\node [style=none] (4) at (0, 0) {$S$};
		\node [style=none] (9) at (2.5, 0) {};
		\node [style=none] (10) at (0.5, 0) {};
		\node [style=none] (11) at (-0.5, 0) {};
		\node [style=none] (12) at (-2, 0) {};
		\node [style=none] (13) at (1.5, -1.5) {\footnotesize $\mathbf{p^*}$};
		\node [style=none] (14) at (-1.25, 1.25) {};
		\node [style=none] (15) at (1.75, 1.25) {};
		\node [style=none] (16) at (1.75, -1.75) {};
		\node [style=none] (17) at (-1.25, -1.75) {};
		\node [style=none] (18) at (-1, 1) {};
		\node [style=none] (19) at (1.25, 1) {};
		\node [style=none] (20) at (1.25, -1.25) {};
		\node [style=none] (21) at (-1, -1.25) {};
		\node [style=none] (22) at (1, -1) {\footnotesize $\mathbf{i'}$};
	\end{pgfonlayer}
	\begin{pgfonlayer}{edgelayer}
				\filldraw [fill=darkred!30,draw=darkred!60] (14.center) to (15.center) to (16.center) to (17.center) to cycle;
								\filldraw [fill=darkgreen!30,draw=darkgreen!60] (18.center) to (19.center) to (20.center) to (21.center) to cycle;
		\filldraw [fill=white,draw] (0.center) to (1.center) to (2.center) to (3.center) to cycle;
		\draw [CcWire] (9.center) to (10.center);
		\draw [CcWire] (11.center) to (12.center);
	\end{pgfonlayer}
\end{tikzpicture}
= \quad
\begin{tikzpicture}
	\begin{pgfonlayer}{nodelayer}
		\node [style=none] (0) at (-0.5, 0.5) {};
		\node [style=none] (1) at (-0.5, -0.5) {};
		\node [style=none] (2) at (0.5, -0.5) {};
		\node [style=none] (3) at (0.5, 0.5) {};
		\node [style=none] (4) at (0, 0) {$S$};
		\node [style=none] (9) at (1.75, 0) {};
		\node [style=none] (10) at (0.5, 0) {};
		\node [style=none] (11) at (-0.5, 0) {};
		\node [style=none] (12) at (-1.75, 0) {};
		\node [style=none] (14) at (-1, 1) {};
		\node [style=none] (15) at (1.25, 1) {};
		\node [style=none] (16) at (1.25, -1.25) {};
		\node [style=none] (17) at (-1, -1.25) {};
	\end{pgfonlayer}
	\begin{pgfonlayer}{edgelayer}
		\filldraw [fill=white,draw] (0.center) to (1.center) to (2.center) to (3.center) to cycle;
		\draw [CcWire] (9.center) to (10.center);
		\draw [CcWire] (11.center) to (12.center);
	\end{pgfonlayer}
\end{tikzpicture} \ ,
\eeq
where the last line follows from the fact that $\mathbf{p^*}\circ \mathbf{i'} = \mathds{1}_{\Inf}$.
Hence, the prediction map applied to any process in its domain is associated to a unique real matrix, and so $\mathbf{p^*}$ is unique.
\endproof

The full picture of a \crealist CI theory is therefore given by a triple of process theories and a triple of DP maps between them:
\beq%
\InputIfFileExists{Diagrams/OntInferentialEquiv.tikz}{}{\input{./figures/Diagrams/OntInferentialEquiv.tikz}},\eeq
where we use a dashed line to denote that $\mathbf{p^*}$ is partial.

We close this section by noting that it remains to determine the scope of classical realist CI theories.  For instance, it is unclear whether Bohmian mechanics can {\em formally} be cast as such a theory.  (Note that this is not specific to our framework; it is also unclear whether it can be formalized within the standard framework of ontological models.)
In any case, we note that the central aim of our framework is not to capture the diversity of interpretational views, but rather to make progress on the questions posed in the introduction.

\section{Inferential equivalence} \label{sec:InfEquiv}

We now define
a notion of inferential equivalence between processes in a causal-inferential theory. This definition can clearly be made in any causal-inferential theory, but we will focus here only on operational CI theories and then on \crealist CI theories.
This will let us define quotiented operational CI theories and quotiented \crealist CI theories. We will discuss how the former relates to the notion of a generalized probabilistic theory, while the latter subsumes the traditional notion of an ontological model.

\subsection{Inferential equivalence in operational CI theories } \label{opquotinf}

Two elements of $\PI$ are inferentially equivalent if and only if they lead to exactly the same predictions, no matter which causally closed diagram they are embedded in. To make such statements diagrammatically, it is useful to introduce the notion of a {\em tester} for a given process---that is, a special case of a clamp (introduced in Section~\ref{sec:PTs})
whose composition with a given process yields a causally closed diagram.
As a simple example, we say that two states of knowledge $\sigma_{\morph{A}{B}}$ and $\sigma'_{\morph{A}{B}}$ about a transformation procedure from $\op{A}$ to $\op{B}$ are inferentially equivalent with respect to the prediction map $\mathbf{p}$, denoted
\beq
\label{infequiv}
\InputIfFileExists{Diagrams/InfEquiv1.tikz}{}{\input{./figures/Diagrams/InfEquiv1.tikz}} \quad \sim_\mathbf{p} \quad %
\InputIfFileExists{Diagrams/InfEquiv2.tikz}{}{\input{./figures/Diagrams/InfEquiv2.tikz}},
\eeq
if and only if they make the same predictions for all testers, $\mathcal{T}$, so that
\beq
\label{infequivdef}
\InputIfFileExists{Diagrams/InfEquiv3.tikz}{}{\input{./figures/Diagrams/InfEquiv3.tikz}}\ =\ %
\InputIfFileExists{Diagrams/InfEquiv4.tikz}{}{\input{./figures/Diagrams/InfEquiv4.tikz}}\quad \forall \mathcal{T}\in \PI.
\eeq

As an explicit example from within quantum theory, consider four lists of laboratory instructions, denoted $P_1$ to $P_4$, that are designed to prepare the quantum states $\ket{0}$, $\ket{1}$, $\ket{+}$, and $\ket{-}$, respectively. Then, the states of knowledge
\begin{equation} \label{twoinfstatesofknow}
\frac{1}{2}[P_1]+\frac{1}{2}[P_2]\ \ \ \text{   and   }\ \ \ \frac{1}{2}[P_3]+\frac{1}{2}[P_4],
\end{equation}
 although clearly distinct, are nonetheless inferentially equivalent, as they 
correspond to the same quantum state (namely the maximally mixed state).

More generally, the notion of inferential equivalence for any type of process in $\PI$ is defined as follows:
\begin{definition}[Inferential equivalence for operational CI theories]\label{Def:InfEquiv}
Two processes in  $\PI$, $\mathcal{D}$ and $\mathcal{E}$, are inferentially equivalent with respect to the prediction map $\mathbf{p}$, denoted $\mathcal{D}\sim_{\mathbf{p}}\mathcal{E}$, if and only if
\beq
\InputIfFileExists{Diagrams/InfEquiv5.tikz}{}{\input{./figures/Diagrams/InfEquiv5.tikz}}  =  %
\InputIfFileExists{Diagrams/InfEquiv6.tikz}{}{\input{./figures/Diagrams/InfEquiv6.tikz}}\quad \forall \mathcal{T} \in \PI.
\eeq
\end{definition}

In fact, one can test for inferential equivalence purely in terms of probabilities (as opposed to stochastic maps).

\begin{lemma} \label{simplerinfeq}
One has inferential equivalence $\mathcal{D}\sim_{\mathbf{p}}\mathcal{E}$ if and only if
\beq
\begin{tikzpicture}
	\begin{pgfonlayer}{nodelayer}
		\node [style=none] (0) at (-0.5, 0.5) {};
		\node [style=none] (1) at (-0.5, -0.5) {};
		\node [style=none] (2) at (0.5, -0.5) {};
		\node [style=none] (3) at (0.5, 0.5) {};
		\node [style=none] (4) at (-1, -1.25) {};
		\node [style=none] (5) at (-1, -2.25) {};
		\node [style=none] (6) at (1.5, -2.25) {};
		\node [style=none] (7) at (1.5, 1.5) {};
		\node [style=none] (8) at (-1.5, 1.5) {};
		\node [style=none] (9) at (-1.5, -0.5) {};
		\node [style=none] (10) at (-1.25, -0.5) {};
		\node [style=none] (11) at (-1.25, 1.25) {};
		\node [style=none] (12) at (1.25, 1.25) {};
		\node [style=none] (13) at (1.25, -1.25) {};
		\node [style=none] (14) at (-2, 2.25) {};
		\node [style=none] (15) at (2, 2.25) {};
		\node [style=none] (16) at (2, -2.75) {};
		\node [style=none] (17) at (-2, -2.75) {};
		\node [style=none] (18) at (0, 0) {$\mathcal{D}$};
		\node [style=none] (19) at (0, -1.75) {$\mathcal{T}$};
		\node [style=none] (20) at (0, 0.5) {};
		\node [style=none] (21) at (0, 1.25) {};
		\node [style=none] (22) at (0, -0.5) {};
		\node [style=none] (23) at (0, -1.25) {};
		\node [style=none] (24) at (1.25, 0) {};
		\node [style=none] (25) at (0.5, 0) {};
		\node [style=none] (26) at (-0.5, 0) {};
		\node [style=none] (27) at (-1.25, 0) {};
		\node [style=none] (28) at (1.75, -2.5) {\footnotesize $\mathbf{p}$};
	\end{pgfonlayer}
	\begin{pgfonlayer}{edgelayer}
	\filldraw [fill=red!30,draw=red!60] (14.center) to (15.center) to (16.center) to (17.center) to cycle;
		\filldraw [fill=white,draw] (0.center) to (1.center) to (2.center) to (3.center) to cycle;
		\filldraw [fill=white,draw] (4.center) to (5.center) to (6.center) to (7.center) to (8.center)to (9.center)to (10.center)to (11.center)to (12.center)to (13.center) to cycle;
		\draw [qWire] (20.center) to (21.center);
		\draw [qWire] (22.center) to (23.center);
		\draw [CcWire] (24.center) to (25.center);
		\draw [CcWire] (26.center) to (27.center);
	\end{pgfonlayer}
\end{tikzpicture}
 \quad = \quad \begin{tikzpicture}
	\begin{pgfonlayer}{nodelayer}
		\node [style=none] (0) at (-0.5, 0.5) {};
		\node [style=none] (1) at (-0.5, -0.5) {};
		\node [style=none] (2) at (0.5, -0.5) {};
		\node [style=none] (3) at (0.5, 0.5) {};
		\node [style=none] (4) at (-1, -1.25) {};
		\node [style=none] (5) at (-1, -2.25) {};
		\node [style=none] (6) at (1.5, -2.25) {};
		\node [style=none] (7) at (1.5, 1.5) {};
		\node [style=none] (8) at (-1.5, 1.5) {};
		\node [style=none] (9) at (-1.5, -0.5) {};
		\node [style=none] (10) at (-1.25, -0.5) {};
		\node [style=none] (11) at (-1.25, 1.25) {};
		\node [style=none] (12) at (1.25, 1.25) {};
		\node [style=none] (13) at (1.25, -1.25) {};
		\node [style=none] (14) at (-2, 2.25) {};
		\node [style=none] (15) at (2, 2.25) {};
		\node [style=none] (16) at (2, -2.75) {};
		\node [style=none] (17) at (-2, -2.75) {};
		\node [style=none] (18) at (0, 0) {$\mathcal{E}$};
		\node [style=none] (19) at (0, -1.75) {$\mathcal{T}$};
		\node [style=none] (20) at (0, 0.5) {};
		\node [style=none] (21) at (0, 1.25) {};
		\node [style=none] (22) at (0, -0.5) {};
		\node [style=none] (23) at (0, -1.25) {};
		\node [style=none] (24) at (1.25, 0) {};
		\node [style=none] (25) at (0.5, 0) {};
		\node [style=none] (26) at (-0.5, 0) {};
		\node [style=none] (27) at (-1.25, 0) {};
		\node [style=none] (28) at (1.75, -2.5) {\footnotesize $\mathbf{p}$};
	\end{pgfonlayer}
	\begin{pgfonlayer}{edgelayer}
	\filldraw [fill=red!30,draw=red!60] (14.center) to (15.center) to (16.center) to (17.center) to cycle;
		\filldraw [fill=white,draw] (0.center) to (1.center) to (2.center) to (3.center) to cycle;
		\filldraw [fill=white,draw] (4.center) to (5.center) to (6.center) to (7.center) to (8.center)to (9.center)to (10.center)to (11.center)to (12.center)to (13.center) to cycle;
		\draw [qWire] (20.center) to (21.center);
		\draw [qWire] (22.center) to (23.center);
		\draw [CcWire] (24.center) to (25.center);
		\draw [CcWire] (26.center) to (27.center);
	\end{pgfonlayer}
\end{tikzpicture}\quad \forall \mathcal{T} \in \PI.
\eeq
\end{lemma}
\proof This follows immediately from Definition~\ref{Def:InfEquiv} and Theorem~\ref{Thm:PointAndAtomicSuffice}.\endproof

For processes that are causally closed, this condition greatly simplifies:
\begin{lemma}\label{lem:simpleEquivTest}
Two causally closed processes are inferentially equivalent if and only if they are equal as stochastic maps under the application of the prediction map $\mathbf{p}$:
\beq
\begin{tikzpicture}
	\begin{pgfonlayer}{nodelayer}
		\node [style=small box] (0) at (0, -0) {$\mathcal{D}$};
		\node [style=none] (1) at (1, -0) {};
		\node [style=none] (3) at (-1, -0) {};
	\end{pgfonlayer}
	\begin{pgfonlayer}{edgelayer}
		\draw [style=CcWire] (0) to (1.center);
		\draw [style=CcWire] (3.center) to (0);
	\end{pgfonlayer}
\end{tikzpicture}
}\sim_\mathbf{p}%
\begin{tikzpicture}
	\begin{pgfonlayer}{nodelayer}
		\node [style=small box] (0) at (0, -0) {$\mathcal{E}$};
		\node [style=none] (1) at (1, -0) {};
		\node [style=none] (3) at (-1, -0) {};
	\end{pgfonlayer}
	\begin{pgfonlayer}{edgelayer}
		\draw [style=CcWire] (0) to (1.center);
		\draw [style=CcWire] (3.center) to (0);
	\end{pgfonlayer}
\end{tikzpicture}
} \  \iff \ %
\InputIfFileExists{Diagrams/InfEquiv7.tikz}{}{\input{./figures/Diagrams/InfEquiv7.tikz}}=%
\InputIfFileExists{Diagrams/InfEquiv8.tikz}{}{\input{./figures/Diagrams/InfEquiv8.tikz}} \ .
\eeq
\end{lemma}
This is a much simpler condition to check, since one need not quantify over all possible testers.
\proof
By definition,
\beq
\begin{tikzpicture}
	\begin{pgfonlayer}{nodelayer}
		\node [style=small box] (0) at (0, -0) {$\mathcal{D}$};
		\node [style=none] (1) at (1, -0) {};
		\node [style=none] (3) at (-1, -0) {};
	\end{pgfonlayer}
	\begin{pgfonlayer}{edgelayer}
		\draw [style=CcWire] (0) to (1.center);
		\draw [style=CcWire] (3.center) to (0);
	\end{pgfonlayer}
\end{tikzpicture}
}\sim_\mathbf{p}%
\begin{tikzpicture}
	\begin{pgfonlayer}{nodelayer}
		\node [style=small box] (0) at (0, -0) {$\mathcal{E}$};
		\node [style=none] (1) at (1, -0) {};
		\node [style=none] (3) at (-1, -0) {};
	\end{pgfonlayer}
	\begin{pgfonlayer}{edgelayer}
		\draw [style=CcWire] (0) to (1.center);
		\draw [style=CcWire] (3.center) to (0);
	\end{pgfonlayer}
\end{tikzpicture}
}
\eeq
is equivalent to
\beq
\forall \mathcal{T}\quad %
\InputIfFileExists{Diagrams/simpleEquiv1.tikz}{}{\input{./figures/Diagrams/simpleEquiv1.tikz}}=%
\InputIfFileExists{Diagrams/simpleEquiv2.tikz}{}{\input{./figures/Diagrams/simpleEquiv2.tikz}},
\eeq
which, by diagram-preservation, is equivalent to
\beq
\forall \mathcal{T}\quad %
\InputIfFileExists{Diagrams/simpleEquiv3.tikz}{}{\input{./figures/Diagrams/simpleEquiv3.tikz}}=%
\InputIfFileExists{Diagrams/simpleEquiv4.tikz}{}{\input{./figures/Diagrams/simpleEquiv4.tikz}}.
\eeq
Finally, this is equivalent to
\beq
\InputIfFileExists{Diagrams/InfEquiv7.tikz}{}{\input{./figures/Diagrams/InfEquiv7.tikz}}=%
\InputIfFileExists{Diagrams/InfEquiv8.tikz}{}{\input{./figures/Diagrams/InfEquiv8.tikz}},
\eeq
where the $\Rightarrow$ direction follows from the special case where $\mathcal{T}$ is simply the identity on the two inferential systems, and where the $\Leftarrow$ direction follows from the fact that equality is preserved by composition (in this case, with $\mathbf{p}(\mathcal{T})$).
\endproof

For the still more restricted set of processes in the image of $\mathbf{i}$ the condition simplifies even further:

\begin{corollary}\label{cor:uniqueSubStoch}
Two processes in the image of \colorbox{green!30}{$\mathbf{i}:\Inf \to \PI$} are inferentially equivalent if and only if they are equal as substochastic maps in $\Inf$:
\beq
\begin{tikzpicture}
	\begin{pgfonlayer}{nodelayer}
		\node [style=small box] (0) at (0, -0) {$\sigma$};
		\node [style=none] (1) at (2, -0) {};
		\node [style=none] (3) at (-2, -0) {};
		\node [style=none] (5) at (-1, 1) {};
		\node [style=none] (6) at (1, 1) {};
		\node [style=none] (7) at (1, -1) {};
		\node [style=none] (8) at (-1, -1) {};
		\node [style=none] (9) at (-2, 2) {};
		\node [style=none] (10) at (2, 2) {};
		\node [style=none] (11) at (2, -2) {};
		\node [style=none] (12) at (-2, -2) {};
		\node [style=none] (13) at (0.75, -0.75) {\footnotesize $\mathbf{i}$};
	\end{pgfonlayer}
	\begin{pgfonlayer}{edgelayer}
		\filldraw [fill=green!30, draw=green!60] (5.center) to (6.center) to (7.center) to (8.center) to cycle;
		\draw [style=CcWire] (0) to (1.center);
		\draw [style=CcWire] (3.center) to (0);
	\end{pgfonlayer}
\end{tikzpicture}
\sim_\mathbf{p} \begin{tikzpicture}
	\begin{pgfonlayer}{nodelayer}
		\node [style=small box] (0) at (0, -0) {$\sigma'$};
		\node [style=none] (1) at (2, -0) {};
		\node [style=none] (3) at (-2, -0) {};
		\node [style=none] (5) at (-1, 1) {};
		\node [style=none] (6) at (1, 1) {};
		\node [style=none] (7) at (1, -1) {};
		\node [style=none] (8) at (-1, -1) {};
		\node [style=none] (9) at (-2, 2) {};
		\node [style=none] (10) at (2, 2) {};
		\node [style=none] (11) at (2, -2) {};
		\node [style=none] (12) at (-2, -2) {};
		\node [style=none] (13) at (0.75, -0.75) {\footnotesize $\mathbf{i}$};
	\end{pgfonlayer}
	\begin{pgfonlayer}{edgelayer}
		\filldraw [fill=green!30, draw=green!60] (5.center) to (6.center) to (7.center) to (8.center) to cycle;
		\draw [style=CcWire] (0) to (1.center);
		\draw [style=CcWire] (3.center) to (0);
	\end{pgfonlayer}
\end{tikzpicture}
 \ \iff \
 \begin{tikzpicture}
	\begin{pgfonlayer}{nodelayer}
		\node [style=small box] (0) at (0, -0) {$\sigma$};
		\node [style=none] (1) at (1, -0) {};
		\node [style=none] (3) at (-1, -0) {};
	\end{pgfonlayer}
	\begin{pgfonlayer}{edgelayer}
		\draw [style=CcWire] (0) to (1.center);
		\draw [style=CcWire] (3.center) to (0);
	\end{pgfonlayer}
\end{tikzpicture}
=
\begin{tikzpicture}
	\begin{pgfonlayer}{nodelayer}
		\node [style=small box] (0) at (0, -0) {$\sigma'$};
		\node [style=none] (1) at (1, -0) {};
		\node [style=none] (3) at (-1, -0) {};
	\end{pgfonlayer}
	\begin{pgfonlayer}{edgelayer}
		\draw [style=CcWire] (0) to (1.center);
		\draw [style=CcWire] (3.center) to (0);
	\end{pgfonlayer}
\end{tikzpicture} \ .
 \eeq
\end{corollary}
\proof
By Lemma~\ref{lem:simpleEquivTest}, we have that the LHS of the implication in the corollary is equivalent to the equality
\beq
\begin{tikzpicture}
	\begin{pgfonlayer}{nodelayer}
		\node [style=none] (0) at (0, 0) {$\sigma$};
		\node [style=none] (1) at (-0.5, 0.5) {};
		\node [style=none] (2) at (0.5, 0.5) {};
		\node [style=none] (3) at (0.5, -0.5) {};
		\node [style=none] (4) at (-0.5, -0.5) {};
		\node [style=none] (5) at (-0.5, 0) {};
		\node [style=none] (6) at (-2, 0) {};
		\node [style=none] (7) at (2, 0) {};
		\node [style=none] (8) at (0.5, 0) {};
		\node [style=none] (13) at (1, -1) {};
		\node [style=none] (14) at (1, 1) {};
		\node [style=none] (15) at (-1, 1) {};
		\node [style=none] (16) at (-1, -1) {};
		\node [style=none] (17) at (1.5, -1.5) {};
		\node [style=none] (18) at (1.5, 1.5) {};
		\node [style=none] (19) at (-1.5, 1.5) {};
		\node [style=none] (20) at (-1.5, -1.5) {};
		\node [style=none] (21) at (0.75, -0.75) {\footnotesize $\mathbf{i}$};
		\node [style=none] (22) at (1.25, -1.25) {\footnotesize $\mathbf{p}$};
	\end{pgfonlayer}
	\begin{pgfonlayer}{edgelayer}
			\filldraw [fill=red!30,draw=red!60] (17.center) to (18.center) to (19.center) to (20.center) to cycle;
		\filldraw [fill=green!30,draw=green!60] (13.center) to (14.center) to (15.center) to (16.center) to cycle;
		\filldraw[fill=white,draw] (1.center) to (2.center) to (3.center) to (4.center) to cycle;
		\draw [style=CcWire] (6.center) to (5.center);
		\draw [style=CcWire] (8.center) to (7.center);
	\end{pgfonlayer}
\end{tikzpicture}
=
\begin{tikzpicture}
	\begin{pgfonlayer}{nodelayer}
		\node [style=none] (0) at (0, 0) {$\sigma'$};
		\node [style=none] (1) at (-0.5, 0.5) {};
		\node [style=none] (2) at (0.5, 0.5) {};
		\node [style=none] (3) at (0.5, -0.5) {};
		\node [style=none] (4) at (-0.5, -0.5) {};
		\node [style=none] (5) at (-0.5, 0) {};
		\node [style=none] (6) at (-2, 0) {};
		\node [style=none] (7) at (2, 0) {};
		\node [style=none] (8) at (0.5, 0) {};
		\node [style=none] (13) at (1, -1) {};
		\node [style=none] (14) at (1, 1) {};
		\node [style=none] (15) at (-1, 1) {};
		\node [style=none] (16) at (-1, -1) {};
		\node [style=none] (17) at (1.5, -1.5) {};
		\node [style=none] (18) at (1.5, 1.5) {};
		\node [style=none] (19) at (-1.5, 1.5) {};
		\node [style=none] (20) at (-1.5, -1.5) {};
		\node [style=none] (21) at (0.75, -0.75) {\footnotesize $\mathbf{i}$};
		\node [style=none] (22) at (1.25, -1.25) {\footnotesize $\mathbf{p}$};
	\end{pgfonlayer}
	\begin{pgfonlayer}{edgelayer}
			\filldraw [fill=red!30,draw=red!60] (17.center) to (18.center) to (19.center) to (20.center) to cycle;
		\filldraw [fill=green!30,draw=green!60] (13.center) to (14.center) to (15.center) to (16.center) to cycle;
		\filldraw[fill=white,draw] (1.center) to (2.center) to (3.center) to (4.center) to cycle;
		\draw [style=CcWire] (6.center) to (5.center);
		\draw [style=CcWire] (8.center) to (7.center);
	\end{pgfonlayer}
\end{tikzpicture} \ ,
\eeq
which gives the RHS of the implication in the corollary by Eq.~\eqref{eq:constraintOnPredictions}, namely $\mathbf{p}\circ \mathbf{i}= \mathds{1}_{\Inf}$.
\endproof

These results imply that every causally closed process
is associated to a unique stochastic map.
\begin{lemma} \label{infimage}
Every causally closed process $\mathcal{D}\in\PI$ is inferentially equivalent to a unique process in the image of $\mathbf{i}$, namely,
\beq\label{imageofi}
\InputIfFileExists{Diagrams/CalculationTheory1.tikz}{}{\input{./figures/Diagrams/CalculationTheory1.tikz}}\quad \sim_{\mathbf{p}} \quad%
\InputIfFileExists{Diagrams/CalculationTheory2.tikz}{}{\input{./figures/Diagrams/CalculationTheory2.tikz}}.
\eeq
\end{lemma}
\begin{proof}
The constraint on the prediction map $\mathbf{p}$ of Eq.~\eqref{eq:constraintOnPredictions} immediately implies that
\beq
\begin{tikzpicture}
	\begin{pgfonlayer}{nodelayer}
		\node [style=none] (0) at (0, 0) {$\mathcal{D}$};
		\node [style=none] (1) at (-0.5, 0.5) {};
		\node [style=none] (2) at (0.5, 0.5) {};
		\node [style=none] (3) at (0.5, -0.5) {};
		\node [style=none] (4) at (-0.5, -0.5) {};
		\node [style=none] (5) at (-0.5, 0) {};
		\node [style=none] (6) at (-2, 0) {};
		\node [style=none] (7) at (2, 0) {};
		\node [style=none] (8) at (0.5, 0) {};
		\node [style=none] (13) at (1, -1) {};
		\node [style=none] (14) at (1, 1) {};
		\node [style=none] (15) at (-1, 1) {};
		\node [style=none] (16) at (-1, -1) {};
		\node [style=none] (17) at (1.5, -1.5) {};
		\node [style=none] (18) at (1.5, 1.5) {};
		\node [style=none] (19) at (-1.5, 1.5) {};
		\node [style=none] (20) at (-1.5, -1.5) {};
		\node [style=none] (21) at (0.75, -0.75) {\footnotesize $\mathbf{p}$};
	\end{pgfonlayer}
	\begin{pgfonlayer}{edgelayer}
		\filldraw [fill=red!30,draw=red!60] (13.center) to (14.center) to (15.center) to (16.center) to cycle;
		\filldraw[fill=white,draw] (1.center) to (2.center) to (3.center) to (4.center) to cycle;
		\draw [style=CcWire] (6.center) to (5.center);
		\draw [style=CcWire] (8.center) to (7.center);
	\end{pgfonlayer}
\end{tikzpicture}
=
\begin{tikzpicture}
	\begin{pgfonlayer}{nodelayer}
		\node [style=none] (0) at (0, 0) {$\mathcal{D}$};
		\node [style=none] (1) at (-0.5, 0.5) {};
		\node [style=none] (2) at (0.5, 0.5) {};
		\node [style=none] (3) at (0.5, -0.5) {};
		\node [style=none] (4) at (-0.5, -0.5) {};
		\node [style=none] (5) at (-0.5, 0) {};
		\node [style=none] (6) at (-2.5, -0) {};
		\node [style=none] (7) at (2.5, -0) {};
		\node [style=none] (8) at (0.5, 0) {};
		\node [style=none] (9) at (1, -1) {};
		\node [style=none] (10) at (1, 1) {};
		\node [style=none] (11) at (-1, 1) {};
		\node [style=none] (12) at (-1, -1) {};
		\node [style=none] (13) at (1.5, -1.5) {};
		\node [style=none] (14) at (1.5, 1.5) {};
		\node [style=none] (15) at (-1.5, 1.5) {};
		\node [style=none] (16) at (-1.5, -1.5) {};
		\node [style=none] (17) at (0.75, -0.75) {\footnotesize $\mathbf{p}$};
		\node [style=none] (18) at (1.25, -1.25) {\footnotesize $\mathbf{i}$};
		\node [style=none] (19) at (-2, 2) {};
		\node [style=none] (20) at (-2, -2) {};
		\node [style=none] (21) at (2, -2) {};
		\node [style=none] (22) at (2, 2) {};
		\node [style=none] (23) at (1.75, -1.75) {\footnotesize $\mathbf{p}$};
	\end{pgfonlayer}
	\begin{pgfonlayer}{edgelayer}
		\filldraw [fill=red!30,draw=red!60] (19.center) to (20.center) to (21.center) to (22.center) to cycle;
		\filldraw [fill=green!30,draw=green!60] (13.center) to (14.center) to (15.center) to (16.center) to cycle;
		\filldraw [fill=red!30,draw=red!60] (9.center) to (10.center) to (11.center) to (12.center) to cycle;
		\filldraw[fill=white,draw] (1.center) to (2.center) to (3.center) to (4.center) to cycle;
		\draw [style=CcWire] (6.center) to (5.center);
		\draw [style=CcWire] (8.center) to (7.center);
	\end{pgfonlayer}
\end{tikzpicture} \ ,
\eeq
after which Lemma~\ref{lem:simpleEquivTest} implies that
\beq
\InputIfFileExists{Diagrams/CalculationTheory1.tikz}{}{\input{./figures/Diagrams/CalculationTheory1.tikz}}\quad \sim_{\mathbf{p}} \quad%
\InputIfFileExists{Diagrams/CalculationTheory2.tikz}{}{\input{./figures/Diagrams/CalculationTheory2.tikz}} \,
\eeq
and then Corollary~\ref{cor:uniqueSubStoch} implies that this is the {\em unique} process in the image of $\mathbf{i}$ in the equivalence class of $\mathcal{D}$.
\end{proof}

\subsubsection{Quotiented operational CI theories}

In many cases, one is only interested in the inferential equivalence class of processes
in a causal-inferential theory.
In such cases, it is useful to define a new type of theory, wherein one has quotiented\footnote{ Notions of quotiented operational theories
 can be found in earlier works, notably including Ref.~\cite{chiribella2010probabilistic}. } with respect to inferential equivalence.
 We now show how this is done for operational CI theories.

First, we note that the relation $\sim_\mathbf{p}$ is preserved under composition: 

\begin{lemma} \label{congrelopLem}
If $\mathcal{D} \sim_\mathbf{p} \mathcal{D}'$, then
\beq \label{congrelation}
\InputIfFileExists{Diagrams/congruence1.tikz}{}{\input{./figures/Diagrams/congruence1.tikz}}\quad\sim_\mathbf{p}\quad%
\InputIfFileExists{Diagrams/congruence2.tikz}{}{\input{./figures/Diagrams/congruence2.tikz}}
\eeq
for all clamps $\mathcal{C}$ in $\PI$.
\end{lemma}
\begin{proof}
Consider, for the sake of contradiction, that there exists some $\mathcal{C}^*$ such that Eq.~\eqref{congrelation} fails. Then, there exists some tester $\mathcal{T}^*$ such that
\beq
\InputIfFileExists{Diagrams/congruence3.tikz}{}{\input{./figures/Diagrams/congruence3.tikz}}\quad \neq\quad %
\InputIfFileExists{Diagrams/congruence4.tikz}{}{\input{./figures/Diagrams/congruence4.tikz}}.
 \eeq
This, however, would imply that the tester
\beq
\InputIfFileExists{Diagrams/congruence5.tikz}{}{\input{./figures/Diagrams/congruence5.tikz}}
\eeq
generates different inferences for $\mathcal{D}$ and $\mathcal{D}'$,
in contradiction with our initial assumption that  $\mathcal{D} \sim_\mathbf{p} \mathcal{D}'$.
\end{proof}

It follows that $\sim_\mathbf{p}$ is a process-theory congruence relation for $\PI$. That is, if $\mathcal{D}\sim_\mathbf{p}\mathcal{D}'$ and $\mathcal{E}\sim_\mathbf{p} \mathcal{E}'$ then any valid composite of $\mathcal{D}$ and $\mathcal{E}$ will be inferentially equivalent to the same composite of $\mathcal{D}'$ and $\mathcal{E}'$. 
\begin{lemma}\label{congrelop} 
The inferential equivalence relation $\sim_\mathbf{p}$ defines a process theory congruence relation on $\PI$.
\end{lemma}
\proof
Take $\mathcal{D}\sim_\mathbf{p}\mathcal{D}'$ and $\mathcal{E}\sim_\mathbf{p} \mathcal{E}'$, and consider some arbitrary composition of the non-primed versions. As a particular illustrative example, take
\beq
\InputIfFileExists{Diagrams/congExtra1.tikz}{}{\input{./figures/Diagrams/congExtra1.tikz}}.
\eeq
Using Lemma~\ref{congrelopLem}, the fact that $\mathcal{D}\sim_\mathbf{p}\mathcal{D}'$, and the fact that the (inferentially) serial composition of $\mathcal{D}$ with $\mathcal{E}$ is a special case of the clamp $\mathcal{C}$ from Eq.~\eqref{congrelation}, implies that 
\beq
\InputIfFileExists{Diagrams/congExtra2.tikz}{}{\input{./figures/Diagrams/congExtra2.tikz}}\quad \sim_\mathbf{p}\quad %
\InputIfFileExists{Diagrams/congExtra3.tikz}{}{\input{./figures/Diagrams/congExtra3.tikz}} \ .
\eeq
Then, by the same lemma, but now viewing $\mathcal{D}'$ as the clamp and using the fact that $\mathcal{E}\sim_\mathbf{p} \mathcal{E}'$, we have:
\beq
\InputIfFileExists{Diagrams/congExtra4.tikz}{}{\input{./figures/Diagrams/congExtra4.tikz}}\quad \sim_\mathbf{p}\quad %
\InputIfFileExists{Diagrams/congExtra5.tikz}{}{\input{./figures/Diagrams/congExtra5.tikz}} \ .
\eeq
Putting these two together (by transitivity of $\sim_\mathbf{p}$) we immediately have:
\beq
\InputIfFileExists{Diagrams/congExtra1.tikz}{}{\input{./figures/Diagrams/congExtra1.tikz}} \quad \sim_\mathbf{p} \quad %
\InputIfFileExists{Diagrams/congExtra6.tikz}{}{\input{./figures/Diagrams/congExtra6.tikz}}  \ .
\eeq
as we require. It is easy to see that identically structured proofs hold for any other way of composing $\mathcal{D}$ and $\mathcal{E}$.
\endproof


 This lemma is important because it is necessary that inferential equivalence defines a congruence relation in order for quotienting with respect to it to yield a valid process theory.

\begin{definition} We define a quotiented operational CI theory $\widetilde{\PI}$ as the process theory $\PI$ quotiented by the congruence relation $\sim_\mathbf{p}$. That is, it has the same systems as $\PI$, but its processes correspond to equivalence classes of processes in $\PI$, that is, to maximal sets of inferentially equivalent processes. We can moreover define a diagram-preserving map  \colorbox{orange!30}{$\sim_\mathbf{p}:\PI \to \widetilde{\PI}$}, as
\beq
\begin{tikzpicture}
	\begin{pgfonlayer}{nodelayer}
		\node [style=none] (0) at (-0.5, 0.5) {};
		\node [style=none] (1) at (-0.5, -0.5) {};
		\node [style=none] (2) at (0.5, -0.5) {};
		\node [style=none] (3) at (0.5, 0.5) {};
		\node [style=none] (4) at (0, 0) {$\mathcal{D}$};
		\node [style=none] (5) at (2, -0) {};
		\node [style=none] (6) at (0.5, 0) {};
		\node [style=none] (7) at (-0.5, 0) {};
		\node [style=none] (8) at (-2, -0) {};
		\node [style=none] (9) at (-1.25, 1.25) {};
		\node [style=none] (10) at (-1.25, -1.25) {};
		\node [style=none] (11) at (1.25, -1.25) {};
		\node [style=none] (12) at (1.25, 1.25) {};
		\node [style=none] (13) at (1, -1) {\footnotesize $\sim_\mathbf{p}$};
		\node [style=none] (14) at (0, 0.5000001) {};
		\node [style=none] (15) at (0, 2) {};
		\node [style=none] (16) at (0, -0.5000001) {};
		\node [style=none] (17) at (0, -2) {};
	\end{pgfonlayer}
	\begin{pgfonlayer}{edgelayer}
\filldraw [fill=orange!30, draw = orange!60] (9.center) to (10.center) to (11.center) to (12.center) to cycle;
\filldraw [fill=white,draw] (0.center) to (1.center) to (2.center) to (3.center) to cycle;
		\draw [CcWire] (5.center) to (6.center);
		\draw [CcWire] (7.center) to (8.center);
		\draw[qWire] (14.center) to (15.center);
		\draw[qWire] (17.center) to (16.center);
	\end{pgfonlayer}
\end{tikzpicture}
\quad :=
\quad
\begin{tikzpicture}
	\begin{pgfonlayer}{nodelayer}
		\node [style=none] (0) at (-0.5, 0.5) {};
		\node [style=none] (1) at (-0.5, -0.5) {};
		\node [style=none] (2) at (0.5, -0.5) {};
		\node [style=none] (3) at (0.5, 0.5) {};
		\node [style=none] (4) at (0, 0) {$\tilde{\mathcal{D}}$};
		\node [style=none] (5) at (1, -0) {};
		\node [style=none] (6) at (0.5, 0) {};
		\node [style=none] (7) at (-0.5, 0) {};
		\node [style=none] (8) at (-1, -0) {};
		\node [style=none] (9) at (-1.25, 1.25) {};
		\node [style=none] (10) at (-1.25, -1.25) {};
		\node [style=none] (11) at (1.25, -1.25) {};
		\node [style=none] (12) at (1.25, 1.25) {};
		\node [style=none] (14) at (0, 0.5000001) {};
		\node [style=none] (15) at (0, 1) {};
		\node [style=none] (16) at (0, -0.5000001) {};
		\node [style=none] (17) at (0, -1) {};
	\end{pgfonlayer}
	\begin{pgfonlayer}{edgelayer}
\filldraw [fill=white,draw] (0.center) to (1.center) to (2.center) to (3.center) to cycle;
		\draw [CcWire] (5.center) to (6.center);
		\draw [CcWire] (7.center) to (8.center);
		\draw[qWire] (14.center) to (15.center);
		\draw[qWire] (17.center) to (16.center);
	\end{pgfonlayer}
\end{tikzpicture},
\eeq
 where $\widetilde{\mathcal{D}}$ is the equivalence class which contains $\mathcal{D}$.  Composition of equivalence classes is defined by the equivalence class of the composite of an arbitrary choice of representative element for each. 
\end{definition}
That this notion of composition is well defined (i.e., independent of the choice of representative elements) follows from Lemma~\ref{congrelop}.
It then follows that the quotienting map \colorbox{orange!30}{$\sim_\mathbf{p}:\PI \to \widetilde{\PI}$} is indeed diagram-preserving.

From the above definition, one clearly has that
\beq
\begin{tikzpicture}
	\begin{pgfonlayer}{nodelayer}
		\node [style=none] (0) at (-0.5, 0.5) {};
		\node [style=none] (1) at (-0.5, -0.5) {};
		\node [style=none] (2) at (0.5, -0.5) {};
		\node [style=none] (3) at (0.5, 0.5) {};
		\node [style=none] (4) at (0, 0) {$\mathcal{D}$};
		\node [style=none] (5) at (1, -0) {};
		\node [style=none] (6) at (0.5, 0) {};
		\node [style=none] (7) at (-0.5, 0) {};
		\node [style=none] (8) at (-1, -0) {};
		\node [style=none] (9) at (-1.25, 1.25) {};
		\node [style=none] (10) at (-1.25, -1.25) {};
		\node [style=none] (11) at (1.25, -1.25) {};
		\node [style=none] (12) at (1.25, 1.25) {};
		\node [style=none] (14) at (0, 0.5000001) {};
		\node [style=none] (15) at (0, 1) {};
		\node [style=none] (16) at (0, -0.5000001) {};
		\node [style=none] (17) at (0, -1) {};
	\end{pgfonlayer}
	\begin{pgfonlayer}{edgelayer}
\filldraw [fill=white,draw] (0.center) to (1.center) to (2.center) to (3.center) to cycle;
		\draw [CcWire] (5.center) to (6.center);
		\draw [CcWire] (7.center) to (8.center);
		\draw[qWire] (14.center) to (15.center);
		\draw[qWire] (17.center) to (16.center);
	\end{pgfonlayer}
\end{tikzpicture}
\sim_{\mathbf{p}}
\begin{tikzpicture}
	\begin{pgfonlayer}{nodelayer}
		\node [style=none] (0) at (-0.5, 0.5) {};
		\node [style=none] (1) at (-0.5, -0.5) {};
		\node [style=none] (2) at (0.5, -0.5) {};
		\node [style=none] (3) at (0.5, 0.5) {};
		\node [style=none] (4) at (0, 0) {$\mathcal{E}$};
		\node [style=none] (5) at (1, -0) {};
		\node [style=none] (6) at (0.5, 0) {};
		\node [style=none] (7) at (-0.5, 0) {};
		\node [style=none] (8) at (-1, -0) {};
		\node [style=none] (9) at (-1.25, 1.25) {};
		\node [style=none] (10) at (-1.25, -1.25) {};
		\node [style=none] (11) at (1.25, -1.25) {};
		\node [style=none] (12) at (1.25, 1.25) {};
		\node [style=none] (14) at (0, 0.5000001) {};
		\node [style=none] (15) at (0, 1) {};
		\node [style=none] (16) at (0, -0.5000001) {};
		\node [style=none] (17) at (0, -1) {};
	\end{pgfonlayer}
	\begin{pgfonlayer}{edgelayer}
\filldraw [fill=white,draw] (0.center) to (1.center) to (2.center) to (3.center) to cycle;
		\draw [CcWire] (5.center) to (6.center);
		\draw [CcWire] (7.center) to (8.center);
		\draw[qWire] (14.center) to (15.center);
		\draw[qWire] (17.center) to (16.center);
	\end{pgfonlayer}
\end{tikzpicture}
\ \iff \
\begin{tikzpicture}
	\begin{pgfonlayer}{nodelayer}
		\node [style=none] (0) at (-0.5, 0.5) {};
		\node [style=none] (1) at (-0.5, -0.5) {};
		\node [style=none] (2) at (0.5, -0.5) {};
		\node [style=none] (3) at (0.5, 0.5) {};
		\node [style=none] (4) at (0, 0) {$\mathcal{D}$};
		\node [style=none] (5) at (2, -0) {};
		\node [style=none] (6) at (0.5, 0) {};
		\node [style=none] (7) at (-0.5, 0) {};
		\node [style=none] (8) at (-2, -0) {};
		\node [style=none] (9) at (-1.25, 1.25) {};
		\node [style=none] (10) at (-1.25, -1.25) {};
		\node [style=none] (11) at (1.25, -1.25) {};
		\node [style=none] (12) at (1.25, 1.25) {};
		\node [style=none] (13) at (1, -1) {\footnotesize $\sim_{\mathbf{p}}$};
		\node [style=none] (14) at (0, 0.5000001) {};
		\node [style=none] (15) at (0, 2) {};
		\node [style=none] (16) at (0, -0.5000001) {};
		\node [style=none] (17) at (0, -2) {};
	\end{pgfonlayer}
	\begin{pgfonlayer}{edgelayer}
\filldraw [fill=orange!30, draw = orange!60] (9.center) to (10.center) to (11.center) to (12.center) to cycle;
\filldraw [fill=white,draw] (0.center) to (1.center) to (2.center) to (3.center) to cycle;
		\draw [CcWire] (5.center) to (6.center);
		\draw [CcWire] (7.center) to (8.center);
		\draw[qWire] (14.center) to (15.center);
		\draw[qWire] (17.center) to (16.center);
	\end{pgfonlayer}
\end{tikzpicture}
=
\begin{tikzpicture}
	\begin{pgfonlayer}{nodelayer}
		\node [style=none] (0) at (-0.5, 0.5) {};
		\node [style=none] (1) at (-0.5, -0.5) {};
		\node [style=none] (2) at (0.5, -0.5) {};
		\node [style=none] (3) at (0.5, 0.5) {};
		\node [style=none] (4) at (0, 0) {$\mathcal{E}$};
		\node [style=none] (5) at (2, -0) {};
		\node [style=none] (6) at (0.5, 0) {};
		\node [style=none] (7) at (-0.5, 0) {};
		\node [style=none] (8) at (-2, -0) {};
		\node [style=none] (9) at (-1.25, 1.25) {};
		\node [style=none] (10) at (-1.25, -1.25) {};
		\node [style=none] (11) at (1.25, -1.25) {};
		\node [style=none] (12) at (1.25, 1.25) {};
		\node [style=none] (13) at (1, -1) {\footnotesize $\sim_{\mathbf{p}}$};
		\node [style=none] (14) at (0, 0.5000001) {};
		\node [style=none] (15) at (0, 2) {};
		\node [style=none] (16) at (0, -0.5000001) {};
		\node [style=none] (17) at (0, -2) {};
	\end{pgfonlayer}
	\begin{pgfonlayer}{edgelayer}
\filldraw [fill=orange!30, draw = orange!60] (9.center) to (10.center) to (11.center) to (12.center) to cycle;
\filldraw [fill=white,draw] (0.center) to (1.center) to (2.center) to (3.center) to cycle;
		\draw [CcWire] (5.center) to (6.center);
		\draw [CcWire] (7.center) to (8.center);
		\draw[qWire] (14.center) to (15.center);
		\draw[qWire] (17.center) to (16.center);
	\end{pgfonlayer}
\end{tikzpicture}.
\eeq

It is worth noting that in our framework, a quotiented operational CI theory is {\em not} an example of an operational CI theory (unless they are both trivial).
This is because if the operational theory is nontrivial, then the quotienting operation {\em necessarily} loses information: equivalence classes of states of knowledge about procedures are not themselves states of knowledge about procedures. 
 The following example proves the claim.  Recall that 
the closed diagrams  in the quotiented theory are (isomorphic to) probabilities, while in the operational CI theory, they 
 constitute a complete description of what one knows and what one is asking in the scenario under consideration. 
In other words, in any quotiented operational  CI theory, the only way for two closed diagrams to be distinct is if they are inferentially inequivalent whereas in any nontrivial unquotiented operational  CI theory, 
there will exist pairs of closed diagrams that are inferentially equivalent and yet still distinct.

Using the quotienting map, the probability associated with a closed diagram can
always be decomposed into a sequence of stochastic maps representing one's inferences, by grouping together processes into diagrams which are causally closed, e.g.
\beq
\InputIfFileExists{Diagrams/infGrouping.tikz}{}{\input{./figures/Diagrams/infGrouping.tikz}}.
\eeq

Clearly, composition with the quotienting map  $\sim_\mathbf{p}$ can be used to define two new DP maps.   The map $\tilde{\mathbf{e}}: \Proc \to \widetilde{\PI}$ is defined as $\tilde{\mathbf{e}} = \sim_\mathbf{p} \circ \mathbf{e}$, and the map $\tilde{\mathbf{i}}: \Inf \to \widetilde{\PI}$ is defined as $\tilde{\mathbf{i}} =  \sim_\mathbf{p} \circ \mathbf{i}$. 
We can also introduce a partial diagram-preserving prediction map $\tilde{\mathbf{p}}$ for the quotiented operational CI theory, whose action is given by mapping each process in $\widetilde{\PI}$ to an element (any element) of $\PI$ in its equivalence class, and then mapping that element to $\Inf$ via $\mathbf{p}$.
All of this can be concisely represented in the following commuting diagram:
\beq
\InputIfFileExists{Diagrams/FullOpDiag.tikz}{}{\input{./figures/Diagrams/FullOpDiag.tikz}} \ \ .
 \eeq

\subsubsection{Subsuming the framework of generalized probabilistic theories} \label{subsumeGPT}

At this point, one can see the relationship between our framework and another well-known framework for operational theories, namely, that of generalized probabilistic theories (GPTs).

A GPT is a minimal framework in which processes are wired together to form circuits which describe an operational scenario and predict the probabilities of the outcomes that one might observe.
Perhaps the key feature of a GPT is {\em tomographic completeness}, which implies that processes within a GPT are taken to represent {\em equivalence classes} of procedures or events with respect to the operational predictions. That is, two transformations in a GPT are represented distinctly if and only if there exists a circuit in which they can be embedded to give different probabilities for the outcome of some measurement in that circuit. The set of processes is also assumed to be convex and representable in a (typically finite dimensional) real vector space, to have a unique deterministic effect, and for composition of processes to be bilinear.

In forthcoming work~\cite{GPTsAsCITheories}, we prove that these properties are satisfied for a natural subset of processes in any quotiented operational CI theory (i.e., the inferentially closed processes), and that consequently the latter can be identified with GPT processes in the traditional sense.  Tomographic completeness follows naturally from the quotienting which defines $\widetilde{\PI}$, and convexity of these processes is inherited from convexity of the inferential theory.

 However, the quotiented operational CI theories in our framework are not equivalent to GPTs.
 There remain important formal and conceptual differences between the two.
For example, a quotiented operational CI theory contains both causal and inferential systems and processes, while GPTs contain only a single type of system.  It is also worth noting that GPT processes are conventionally viewed as representing equivalence classes of laboratory procedures, while processes in a quotiented operational CI theory have a different interpretation---they represent equivalence classes of {\em states of knowledge} about laboratory procedures.  

\subsection{Inferential equivalence in classical realist CI theories} \label{ontquotinf}

Analogously, two elements of $\FI$ are inferentially equivalent if and only if they lead to exactly the same predictions, no matter what causal diagram they are embedded in.
\begin{definition}[Inferential equivalence for classical realist CI theories]\label{Def:InfEquivOnt}
Two general elements of $\FI$, $\mathcal{D}$ and $\mathcal{E}$, are inferentially equivalent with respect to the prediction map $\mathbf{p}^*$, denoted $\mathcal{D}\sim_{\mathbf{p}^*}\mathcal{E}$, if and only if
\beq
\InputIfFileExists{Diagrams/InfEquiv5Ont.tikz}{}{\input{./figures/Diagrams/InfEquiv5Ont.tikz}} \quad = \quad %
\InputIfFileExists{Diagrams/InfEquiv6Ont.tikz}{}{\input{./figures/Diagrams/InfEquiv6Ont.tikz}}\quad \forall \mathcal{T} \in \FI.
\eeq
\end{definition}

 The fact that this is a nontrivial relationship may be somewhat surprising. Once it is recognized, however, it is not difficult to come up with examples to illustrate it; we give a simple example below.

Every process in $\FI$ can be associated with a stochastic map, via
\beq \label{assocstoc}
\begin{tikzpicture}
	\begin{pgfonlayer}{nodelayer}
		\node [style=none] (0) at (-0.5, 0.5) {};
		\node [style=none] (1) at (-0.5, -0.5) {};
		\node [style=none] (2) at (0.5, -0.5) {};
		\node [style=none] (3) at (0.5, 0.5) {};
		\node [style=none] (4) at (0, 0) {$\mathcal{D}$};
		\node [style=none] (5) at (1, -0) {};
		\node [style=none] (6) at (0.5, 0) {};
		\node [style=none] (7) at (-0.5, 0) {};
		\node [style=none] (8) at (-1, -0) {};
		\node [style=none] (9) at (-1.25, 1.25) {};
		\node [style=none] (10) at (-1.25, -1.25) {};
		\node [style=none] (11) at (1.25, -1.25) {};
		\node [style=none] (12) at (1.25, 1.25) {};
		\node [style=none] (14) at (0, 0.5000001) {};
		\node [style=none] (15) at (0, 1) {};
		\node [style=none] (16) at (0, -0.5000001) {};
		\node [style=none] (17) at (0, -1) {};
	\end{pgfonlayer}
	\begin{pgfonlayer}{edgelayer}
\filldraw [fill=white,draw] (0.center) to (1.center) to (2.center) to (3.center) to cycle;
		\draw [CcWire] (5.center) to (6.center);
		\draw [CcWire] (7.center) to (8.center);
		\draw[oWire] (14.center) to (15.center);
		\draw[oWire] (17.center) to (16.center);
	\end{pgfonlayer}
\end{tikzpicture}
\quad\mapsto\quad
\InputIfFileExists{Diagrams/sufficientTester1.tikz}{}{\input{./figures/Diagrams/sufficientTester1.tikz}}.
\eeq
Using Lemma~\ref{eq:newrewrite} (stated and proven in Appendix~\ref{prflemnewrewrite}), we can prove the following result, which
is an analogue of Lemma~\ref{lem:simpleEquivTest}, but strengthened to include processes with open causal systems.
\begin{lemma} \label{prfinfstochastic}
Two processes in $\FI$ are inferentially equivalent if and only if they are associated with the same substochastic map. That is,
\beq
\begin{tikzpicture}\label{infstochimp}
	\begin{pgfonlayer}{nodelayer}
		\node [style=none] (0) at (-0.5, 0.5) {};
		\node [style=none] (1) at (-0.5, -0.5) {};
		\node [style=none] (2) at (0.5, -0.5) {};
		\node [style=none] (3) at (0.5, 0.5) {};
		\node [style=none] (4) at (0, 0) {$\mathcal{D}$};
		\node [style=none] (5) at (1, -0) {};
		\node [style=none] (6) at (0.5, 0) {};
		\node [style=none] (7) at (-0.5, 0) {};
		\node [style=none] (8) at (-1, -0) {};
		\node [style=none] (9) at (-1.25, 1.25) {};
		\node [style=none] (10) at (-1.25, -1.25) {};
		\node [style=none] (11) at (1.25, -1.25) {};
		\node [style=none] (12) at (1.25, 1.25) {};
		\node [style=none] (14) at (0, 0.5000001) {};
		\node [style=none] (15) at (0, 1) {};
		\node [style=none] (16) at (0, -0.5000001) {};
		\node [style=none] (17) at (0, -1) {};
	\end{pgfonlayer}
	\begin{pgfonlayer}{edgelayer}
\filldraw [fill=white,draw] (0.center) to (1.center) to (2.center) to (3.center) to cycle;
		\draw [CcWire] (5.center) to (6.center);
		\draw [CcWire] (7.center) to (8.center);
		\draw[oWire] (14.center) to (15.center);
		\draw[oWire] (17.center) to (16.center);
	\end{pgfonlayer}
\end{tikzpicture}
\sim_{\mathbf{p^*}}
\begin{tikzpicture}
	\begin{pgfonlayer}{nodelayer}
		\node [style=none] (0) at (-0.5, 0.5) {};
		\node [style=none] (1) at (-0.5, -0.5) {};
		\node [style=none] (2) at (0.5, -0.5) {};
		\node [style=none] (3) at (0.5, 0.5) {};
		\node [style=none] (4) at (0, 0) {$\mathcal{E}$};
		\node [style=none] (5) at (1, -0) {};
		\node [style=none] (6) at (0.5, 0) {};
		\node [style=none] (7) at (-0.5, 0) {};
		\node [style=none] (8) at (-1, -0) {};
		\node [style=none] (9) at (-1.25, 1.25) {};
		\node [style=none] (10) at (-1.25, -1.25) {};
		\node [style=none] (11) at (1.25, -1.25) {};
		\node [style=none] (12) at (1.25, 1.25) {};
		\node [style=none] (14) at (0, 0.5000001) {};
		\node [style=none] (15) at (0, 1) {};
		\node [style=none] (16) at (0, -0.5000001) {};
		\node [style=none] (17) at (0, -1) {};
	\end{pgfonlayer}
	\begin{pgfonlayer}{edgelayer}
\filldraw [fill=white,draw] (0.center) to (1.center) to (2.center) to (3.center) to cycle;
		\draw [CcWire] (5.center) to (6.center);
		\draw [CcWire] (7.center) to (8.center);
		\draw[oWire] (14.center) to (15.center);
		\draw[oWire] (17.center) to (16.center);
	\end{pgfonlayer}
\end{tikzpicture}
\quad \iff \quad %
\InputIfFileExists{Diagrams/sufficientTester1.tikz}{}{\input{./figures/Diagrams/sufficientTester1.tikz}}=%
\InputIfFileExists{Diagrams/sufficientTester2.tikz}{}{\input{./figures/Diagrams/sufficientTester2.tikz}}.
\eeq
\end{lemma}
The proof is given in Appendix~\ref{prfinfstoch}.

 As an explicit example,
 consider the four bit-to-bit functions $\{f_0,f_1,f_{\text{id}},f_{\text{flip}}\}$. They are defined by their action on a bit $a \in \{0,1\}$, namely, $f_0(a) = 0$, $f_1(a)=1$, $f_{\text{id}}(a)=a$, and $f_{\text{flip}}(a)=a\oplus 1$, where $\oplus$ denotes summation modulo 2.
Then, the states of knowledge 
\beq \label{prerandommat}
 \sigma_c = \frac{1}{2}[f_0]+\frac{1}{2}[f_1]\; \textrm{ and}\; \sigma_d=\frac{1}{2}[f_{\text{id}}]+\frac{1}{2}[f_{\text{flip}}] 
\eeq
  are distinct but inferentially equivalent. This is easily seen by the fact that both states of knowledge correspond to the same stochastic map, namely,
 the completely randomizing bit-to-bit channel: 
\beq \label{randommat}
\InputIfFileExists{Diagrams/scrambling1.tikz}{}{\input{./figures/Diagrams/scrambling1.tikz}}\quad =\ \left(
\begin{array}{cc}
 \frac{1}{2} & \frac{1}{2}   \\
 \frac{1}{2} & \frac{1}{2}
\end{array}
\right)\ =\quad %
\InputIfFileExists{Diagrams/scrambling2.tikz}{}{\input{./figures/Diagrams/scrambling2.tikz}}.
\eeq

\subsubsection{Quotiented classical realist CI theories}

In direct analogy with Lemma~\ref{congrelop} and its proof, one can show that $\sim_{\mathbf{p}^*}$ defines a congruence relation, and it follows that one can 
quotient the \crealist CI theory with respect to this relation.

\begin{definition} 
We define a quotiented \crealist CI theory $\widetilde{\FI}$ as the process theory $\FI$ quotiented by the congruence $\sim_{\mathbf{p}^*}$. That is, it has the same systems as $\FI$, but its processes correspond to equivalence classes of processes in $\FI$, that is, to maximal sets of inferentially equivalent processes. We can moreover define a diagram-preserving map \colorbox{darkorange!30}{$\sim_{\mathbf{p}^*}:\FI \to \widetilde{\FI}$}, as
\beq
\begin{tikzpicture}
	\begin{pgfonlayer}{nodelayer}
		\node [style=none] (0) at (-0.5, 0.5) {};
		\node [style=none] (1) at (-0.5, -0.5) {};
		\node [style=none] (2) at (0.5, -0.5) {};
		\node [style=none] (3) at (0.5, 0.5) {};
		\node [style=none] (4) at (0, 0) {$\mathcal{D}$};
		\node [style=none] (5) at (2, 0) {};
		\node [style=none] (6) at (0.5, 0) {};
		\node [style=none] (7) at (-0.5, 0) {};
		\node [style=none] (8) at (-2, 0) {};
		\node [style=none] (9) at (-1.25, 1.25) {};
		\node [style=none] (10) at (-1.25, -1.25) {};
		\node [style=none] (11) at (1.5, -1.25) {};
		\node [style=none] (12) at (1.5, 1.25) {};
		\node [style=none] (13) at (1, -1) {\footnotesize $\sim_{\mathbf{p}^*}$};
		\node [style=none] (14) at (0, 0.5) {};
		\node [style=none] (15) at (0, 2) {};
		\node [style=none] (16) at (0, -0.5) {};
		\node [style=none] (17) at (0, -2) {};
	\end{pgfonlayer}
	\begin{pgfonlayer}{edgelayer}
		\draw [fill={darkorange!30}, draw={darkorange!60}] (9.center)
			 to (10.center)
			 to (11.center)
			 to (12.center)
			 to cycle;
		\draw [fill=white, draw] (0.center)
			 to (1.center)
			 to (2.center)
			 to (3.center)
			 to cycle;
		\draw [CcWire] (5.center) to (6.center);
		\draw [CcWire] (7.center) to (8.center);
		\draw [oWire] (14.center) to (15.center);
		\draw [oWire] (17.center) to (16.center);
	\end{pgfonlayer}
\end{tikzpicture}
\quad =
\quad
\begin{tikzpicture}
	\begin{pgfonlayer}{nodelayer}
		\node [style=none] (0) at (-0.5, 0.5) {};
		\node [style=none] (1) at (-0.5, -0.5) {};
		\node [style=none] (2) at (0.5, -0.5) {};
		\node [style=none] (3) at (0.5, 0.5) {};
		\node [style=none] (4) at (0, 0) {$\tilde{\mathcal{D}}$};
		\node [style=none] (5) at (1, -0) {};
		\node [style=none] (6) at (0.5, 0) {};
		\node [style=none] (7) at (-0.5, 0) {};
		\node [style=none] (8) at (-1, -0) {};
		\node [style=none] (9) at (-1.25, 1.25) {};
		\node [style=none] (10) at (-1.25, -1.25) {};
		\node [style=none] (11) at (1.25, -1.25) {};
		\node [style=none] (12) at (1.25, 1.25) {};
		\node [style=none] (14) at (0, 0.5000001) {};
		\node [style=none] (15) at (0, 1) {};
		\node [style=none] (16) at (0, -0.5000001) {};
		\node [style=none] (17) at (0, -1) {};
	\end{pgfonlayer}
	\begin{pgfonlayer}{edgelayer}
\filldraw [fill=white,draw] (0.center) to (1.center) to (2.center) to (3.center) to cycle;
		\draw [CcWire] (5.center) to (6.center);
		\draw [CcWire] (7.center) to (8.center);
		\draw[oWire] (14.center) to (15.center);
		\draw[oWire] (17.center) to (16.center);
	\end{pgfonlayer}
\end{tikzpicture} \ ,
\eeq
 where $\widetilde{\mathcal{D}}$ is the equivalence class which contains $\mathcal{D}$.
Composition of equivalence classes is defined by the equivalence class of the composite of an arbitrary choice of representative element for each.
\end{definition}
That this notion of composition is well defined (i.e. independent of the choice of representative elements) follows from the natural analogue of Lemma~\ref{congrelop}. 
It then follows that the map  \colorbox{darkorange!30}{$\sim_{\mathbf{p}^*}:\FI \to \widetilde{\FI}$} is indeed diagram-preserving.

From the above definition, it clearly follows that 
\beq
\begin{tikzpicture}
	\begin{pgfonlayer}{nodelayer}
		\node [style=none] (0) at (-0.5, 0.5) {};
		\node [style=none] (1) at (-0.5, -0.5) {};
		\node [style=none] (2) at (0.5, -0.5) {};
		\node [style=none] (3) at (0.5, 0.5) {};
		\node [style=none] (4) at (0, 0) {$\mathcal{D}$};
		\node [style=none] (5) at (1, -0) {};
		\node [style=none] (6) at (0.5, 0) {};
		\node [style=none] (7) at (-0.5, 0) {};
		\node [style=none] (8) at (-1, -0) {};
		\node [style=none] (9) at (-1.25, 1.25) {};
		\node [style=none] (10) at (-1.25, -1.25) {};
		\node [style=none] (11) at (1.25, -1.25) {};
		\node [style=none] (12) at (1.25, 1.25) {};
		\node [style=none] (14) at (0, 0.5000001) {};
		\node [style=none] (15) at (0, 1) {};
		\node [style=none] (16) at (0, -0.5000001) {};
		\node [style=none] (17) at (0, -1) {};
	\end{pgfonlayer}
	\begin{pgfonlayer}{edgelayer}
\filldraw [fill=white,draw] (0.center) to (1.center) to (2.center) to (3.center) to cycle;
		\draw [CcWire] (5.center) to (6.center);
		\draw [CcWire] (7.center) to (8.center);
		\draw[oWire] (14.center) to (15.center);
		\draw[oWire] (17.center) to (16.center);
	\end{pgfonlayer}
\end{tikzpicture}
\sim_{\mathbf{p^*}}
\begin{tikzpicture}
	\begin{pgfonlayer}{nodelayer}
		\node [style=none] (0) at (-0.5, 0.5) {};
		\node [style=none] (1) at (-0.5, -0.5) {};
		\node [style=none] (2) at (0.5, -0.5) {};
		\node [style=none] (3) at (0.5, 0.5) {};
		\node [style=none] (4) at (0, 0) {$\mathcal{E}$};
		\node [style=none] (5) at (1, -0) {};
		\node [style=none] (6) at (0.5, 0) {};
		\node [style=none] (7) at (-0.5, 0) {};
		\node [style=none] (8) at (-1, -0) {};
		\node [style=none] (9) at (-1.25, 1.25) {};
		\node [style=none] (10) at (-1.25, -1.25) {};
		\node [style=none] (11) at (1.25, -1.25) {};
		\node [style=none] (12) at (1.25, 1.25) {};
		\node [style=none] (14) at (0, 0.5000001) {};
		\node [style=none] (15) at (0, 1) {};
		\node [style=none] (16) at (0, -0.5000001) {};
		\node [style=none] (17) at (0, -1) {};
	\end{pgfonlayer}
	\begin{pgfonlayer}{edgelayer}
\filldraw [fill=white,draw] (0.center) to (1.center) to (2.center) to (3.center) to cycle;
		\draw [CcWire] (5.center) to (6.center);
		\draw [CcWire] (7.center) to (8.center);
		\draw[oWire] (14.center) to (15.center);
		\draw[oWire] (17.center) to (16.center);
	\end{pgfonlayer}
\end{tikzpicture}
\quad \iff \quad
\begin{tikzpicture}
	\begin{pgfonlayer}{nodelayer}
		\node [style=none] (0) at (-0.5, 0.5) {};
		\node [style=none] (1) at (-0.5, -0.5) {};
		\node [style=none] (2) at (0.5, -0.5) {};
		\node [style=none] (3) at (0.5, 0.5) {};
		\node [style=none] (4) at (0, 0) {$\mathcal{D}$};
		\node [style=none] (5) at (2, 0) {};
		\node [style=none] (6) at (0.5, 0) {};
		\node [style=none] (7) at (-0.5, 0) {};
		\node [style=none] (8) at (-2, 0) {};
		\node [style=none] (9) at (-1.25, 1.25) {};
		\node [style=none] (10) at (-1.25, -1.25) {};
		\node [style=none] (11) at (1.5, -1.25) {};
		\node [style=none] (12) at (1.5, 1.25) {};
		\node [style=none] (13) at (1, -1) {\footnotesize$\sim_{\mathbf{p}^*}$};
		\node [style=none] (14) at (0, 0.5) {};
		\node [style=none] (15) at (0, 2) {};
		\node [style=none] (16) at (0, -0.5) {};
		\node [style=none] (17) at (0, -2) {};
	\end{pgfonlayer}
	\begin{pgfonlayer}{edgelayer}
		\draw [fill={darkorange!30}, draw={darkorange!60}] (9.center)
			 to (10.center)
			 to (11.center)
			 to (12.center)
			 to cycle;
		\draw [fill=white, draw] (0.center)
			 to (1.center)
			 to (2.center)
			 to (3.center)
			 to cycle;
		\draw [CcWire] (5.center) to (6.center);
		\draw [CcWire] (7.center) to (8.center);
		\draw [oWire] (14.center) to (15.center);
		\draw [oWire] (17.center) to (16.center);
	\end{pgfonlayer}
\end{tikzpicture}
=
\begin{tikzpicture}
	\begin{pgfonlayer}{nodelayer}
		\node [style=none] (0) at (-0.5, 0.5) {};
		\node [style=none] (1) at (-0.5, -0.5) {};
		\node [style=none] (2) at (0.5, -0.5) {};
		\node [style=none] (3) at (0.5, 0.5) {};
		\node [style=none] (4) at (0, 0) {$\mathcal{E}$};
		\node [style=none] (5) at (2, 0) {};
		\node [style=none] (6) at (0.5, 0) {};
		\node [style=none] (7) at (-0.5, 0) {};
		\node [style=none] (8) at (-2, 0) {};
		\node [style=none] (9) at (-1.25, 1.25) {};
		\node [style=none] (10) at (-1.25, -1.25) {};
		\node [style=none] (11) at (1.5, -1.25) {};
		\node [style=none] (12) at (1.5, 1.25) {};
		\node [style=none] (13) at (1, -1) {\footnotesize$\sim_{\mathbf{p}^*}$};
		\node [style=none] (14) at (0, 0.5) {};
		\node [style=none] (15) at (0, 2) {};
		\node [style=none] (16) at (0, -0.5) {};
		\node [style=none] (17) at (0, -2) {};
	\end{pgfonlayer}
	\begin{pgfonlayer}{edgelayer}
		\draw [fill={darkorange!30}, draw={darkorange!60}] (9.center)
			 to (10.center)
			 to (11.center)
			 to (12.center)
			 to cycle;
		\draw [fill=white, draw] (0.center)
			 to (1.center)
			 to (2.center)
			 to (3.center)
			 to cycle;
		\draw [CcWire] (5.center) to (6.center);
		\draw [CcWire] (7.center) to (8.center);
		\draw [oWire] (14.center) to (15.center);
		\draw [oWire] (17.center) to (16.center);
	\end{pgfonlayer}
\end{tikzpicture}
 \ .
\eeq

Clearly $\sim_{\mathbf{p}^*}$ can be used to define two new DP maps $\tilde{\mathbf{e}}': \Func \to \widetilde{\FI}$ and $\tilde{\mathbf{i}}': \Inf \to \widetilde{\FI}$, where $\tilde{\mathbf{e}}' = \sim_{\mathbf{p}^*} \circ \mathbf{e}' $ and $\tilde{\mathbf{i}}' =  \sim_{\mathbf{p}^*} \circ \mathbf{i}'$.
We can also introduce a prediction map $\tilde{\mathbf{p^*}}$ for the quotiented \crealist CI theory, whose action is given by mapping each process in $\widetilde{\FI}$ to an element (any element) of $\FI$ in its equivalence class, and then mapping that element to $\Inf$ via $\mathbf{p}^*$.
All of this can be concisely represented in the following commuting diagram:
\beq
\InputIfFileExists{Diagrams/FullOntDiag.tikz}{}{\input{./figures/Diagrams/FullOntDiag.tikz}} \ \ .
\eeq

Finally, we derive a simplified normal form for $\widetilde{\FI}$. First, we show (in Appendix~\ref{appqnf}, using some useful identities proven in Appendix~\ref{usefulidentitiesapp}) that
\begin{theorem}\label{thm:QNF}
Any diagram in $\FI$ is always inferentially equivalent to one of the form
\beq%
\InputIfFileExists{Diagrams/QNF25.tikz}{}{\input{./figures/Diagrams/QNF25.tikz}} \ ,\eeq
where $\Sigma$ is a stochastic map and $\Pi$ is a propositional map.
\end{theorem}
Applying the quotienting map (and recalling that it is diagram-preserving and that it leaves processes in $\Inf$ invariant), this implies that
\begin{corollary} \label{normalformFI}
Any diagram in $\widetilde{\FI}$ can be rewritten into the following normal form:
\beq%
\InputIfFileExists{Diagrams/QNF26.tikz}{}{\input{./figures/Diagrams/QNF26.tikz}} \ . \eeq
\end{corollary}

\subsubsection{Subsuming the traditional notion of an ontological theory}\label{sec:subsumingOntological}

We can now point out a connection between the notion of a quotiented classical realist CI theory 
 and the traditional notion of an ontological model~\cite{Spe05, Harrigan}.
 Specifically, the notion of a quotiented classical reality CI theory subsumes the type of ontological theory that is presumed as the codomain of the traditional ontological modelling map.  In particular, the stochastic processes in the codomain of this map 
 (such as probabilisty distributions and response functions over the ontic state space) are recovered as the substochastic maps defined by our Eq.~\eqref{assocstoc}.

 Note, however,  that the notion of a quotiented \crealist CI theory contains both causal and inferential systems,
   while traditional ontological models
 concern only a single type of system. We will say much more about representing operational theories in Section~\ref{MAINsec:reps} and onward.

\subsection{To make an omelette...} \label{makingomelette}

A causal-inferential theory allows one to describe a physical scenario while maintaining the distinction between the causal and inferential components of the theory.   
On the classical realist side, highlighting this distinction helps to identify the root of the puzzlement associated to phenomena such as  `Simpson's paradox'~\cite{simpsonsparadox} and `Berkson's paradox'~\cite{berksonparadox},  by making plain the sense in which previous frameworks have scrambled causal and inferential notions.  Moreover, the conceptual and formal tools it provides can help to break the habits of mind that lead to such confusions, namely, the tendency to slide from statements about conditional probabilities to claims about cause-effect relationships. 

On the operational side, the framework we have developed here helps to highlight a type of scrambling of causal and inferential concepts that seems intrinsic to any operational theory and which is not so apparent in the conventional frameworks.
 It concerns the nature of a process in $\Proc$, the causal component of an operational theory. Recall that these processes are descriptions of laboratory procedures.  Now note that although specifying a laboratory procedure may serve to completely specify
 {\em some} degrees of freedom of the devices (usually macroscopic ones), it also generally involves expressing {\em incomplete}  knowledge of the vast majority of its degrees of freedom (the microscopic ones).\footnote{For operational CI theories that admit of a classical realist representation, this fact is reflected in our assumptions about the realist representation map $\xi$, namely, that a point distribution over causal processes in the operational CI theory need not be mapped by $\xi$ to a point distribution over causal processes in the realist CI theory.}
As such, a process in $\Proc$ generally does not stipulate an {\em actual} causal relation between its inputs and outputs, given that it is consistent with many possibilities for this causal relation. 
It is for this reason that we have stipulated that a process in $\Proc$ describes only {\em potential} causal influences.  The reason that the processes in the causal component of an {\em operational} CI theory---unlike those in the causal component of a {\em realist} CI theory---must in part stipulate an agent's uncertainty is because {\em such uncertainty is inherent in all descriptions of phenomena at the operational level.}   At the operational level, therefore, 
the best one can hope to do is to unscramble the notion of {\em potential causal influence} from that of inference.  
  


In the rest of this section, we describe a different---and much more significant---example of how our framework reveals a  type of scrambling of causal and inferential notions that previously went unnoticed.  Specifically, we argue that quotienting with respect to inferential equivalence necessarily erases information about causal relations, so that quotiented theories inevitably incorporate a scrambling of causal and inferential notions.

We hope to make this point in more detail 
in a subsequent paper~\cite{GPTsAsCITheories} whose purpose is to describe 
 the sense in which GPTs are recovered from quotientied operational CI theories. 
 Here, 
 For now, we will a concrete example involving processes in $\FI$ in order to clarify how causation and inference get scrambled under quotienting.

It is based on the example introduced in Section~\ref{ontquotinf}, involving the two states of knowledge about bit-to-bit functional dynamics described in Eq.~\eqref{prerandommat} and which can be recast diagramatically as follows: 
\beq\label{examb1}
\sigma_c:= \frac{1}{2}\left[%
\InputIfFileExists{Diagrams/bitIdentity.tikz}{}{\input{./figures/Diagrams/bitIdentity.tikz}} \right] +\frac{1}{2}\left[%
\InputIfFileExists{Diagrams/bitNot.tikz}{}{\input{./figures/Diagrams/bitNot.tikz}}\right]\quad \text{and}\quad \sigma_d:= \frac{1}{2}\left[%
\InputIfFileExists{Diagrams/bitReset0.tikz}{}{\input{./figures/Diagrams/bitReset0.tikz}} \right] +\frac{1}{2}\left[%
\InputIfFileExists{Diagrams/bitReset1.tikz}{}{\input{./figures/Diagrams/bitReset1.tikz}}\right]
\eeq
 These two states of knowledge refer to completely distinct causal
  relations between the input and output bit: 
$\sigma_c$  has support only on processes with a causal connection between the input and the output, 
 while $\sigma_d$ 
  has support only on processes that are causally disconnected.
   Nonetheless, because the stochastic map associated with each of these states of knowledge is given by the completely randomizing bit-to-bit channel, as noted in Eq.~\eqref{randommat}, it follows that they are inferentially equivalent:
\beq\label{examb2}
\InputIfFileExists{Diagrams/omeletteConnected1.tikz}{}{\input{./figures/Diagrams/omeletteConnected1.tikz}} \ \ \sim_{\mathbf{p^*}} \ \ \begin{tikzpicture}
	\begin{pgfonlayer}{nodelayer}
		\node [style=none] (0) at (0, -1.5) {};
		\node [style=none] (1) at (0, -0) {};
		\node [style=epiBox] (2) at (0, -0) {};
		\node [style=none] (3) at (0, 1.5) {};
		\node [style={right label}] (4) at (0, -1.25) {$B$};
		\node [style={right label}] (5) at (0, 1.25) {$B$};
		\node [style=infpoint] (6) at (-1, -0) {$\sigma_d$};
	\end{pgfonlayer}
	\begin{pgfonlayer}{edgelayer}
		\draw [oWire] (0.center) to (1.center);
		\draw [oWire] (3.center) to (2);
		\draw [cWire] (6) to (2);
	\end{pgfonlayer}
\end{tikzpicture}.
\eeq

We see, therefore, that two radically different causal structures (causal connection and causal disconnection)  lead one to make all of the same inferences about the relevant systems. Similar examples can be constructed in $\PI$.
 It is in this sense that the processes in quotiented causal-inferential theories, such as $\widetilde{\FI}$ and $\widetilde{\PI}$,  exhibit a scrambling of causation and inference.

\section{\Crealist representations}  \label{MAINsec:reps}

\subsection{\Crealist representations of operational CI theories}\label{sec:reps}

A \crealist representation of an operational theory is an attempt to provide an underlying realist explanation of the operational statistics. It posits that each system is characterized by an ontic state, which constitutes a complete characterization of its physical attributes and mediates causal influences between the laboratory procedures. 
 Something akin to such a representation in earlier work is the notion of an {\em ontological model of an operational theory}.  
(Strictly speaking, the closest analogue to the notion of an ontological model  in our framework is not the notion of a classical realist representation that we consider in this section, but the variant thereof corresponding to the $\zeta$ map given below in Definition~\ref{repnsthatmaybeNC}.  The notion of ontological modelling is mentioned here only to help the reader broadly situate the notion of a classical realist representation.)\footnote{Note that a classical realist representation of an operational scenario describes both its causal and inferential aspects, hence both ontological and epistemological aspects thereof.  In this sense, it is clear that the term `ontological modelling' would not be ideal for this sort of representation because the term suggests that one is concerned with modelling {\em only} the ontological aspects.  Insofar as the sort of representation that was referred to as `an ontological model of an operational theory' in prior work~\cite{Harrigan} {\em also} described both ontological and epistemological aspects (even if these were scrambled somewhat), the term was not ideal for those representations either.  In retrospect, a better terminology would have been one that signaled that both ontological {\em and} epistemological aspects were being described therein, just as the term `causal-inferential' signals a description incorporating both causal and inferential elements.}

\begin{definition} \label{ontrepndefn}
A {\em \crealist representation} of an operational CI theory $\PI$, by a classical realist CI theory $\FI$, is a diagram-preserving map \colorbox{gray!30}{$\mathbf{\xi}:\PI \to \FI$}, depicted as
\beq
\InputIfFileExists{Diagrams/OntRep1.tikz}{}{\input{./figures/Diagrams/OntRep1.tikz}},
\eeq
satisfying {\bf (i)} the \emph{preservation of predictions},  namely that the diagram 
\beq\label{eq:EmpiricalAdequacy}
\InputIfFileExists{Diagrams/EmpiricalAdequacy.tikz}{}{\input{./figures/Diagrams/EmpiricalAdequacy.tikz}}
\eeq
 commutes,
where the double line between the two copies of $\Inf$ is an extended equals sign, and {\bf(ii)} the {\em preservation of  ignorability}
\beq\label{eq:CausalOntRep}
\InputIfFileExists{Diagrams/OntRep2.tikz}{}{\input{./figures/Diagrams/OntRep2.tikz}}\quad
=\quad
\InputIfFileExists{Diagrams/OntRep3.tikz}{}{\input{./figures/Diagrams/OntRep3.tikz}} \ .
\eeq
We will sometimes refer to the classical realist representation map $\mathbf{\xi}$ as an $\FI$-representation of $\PI$.
\end{definition}

Note that \colorbox{gray!30}{$\mathbf{\xi}:\PI \to \FI$} leaves inferential systems invariant. This can be derived from preservation of predictions, as follows.
 Start with some inferential system $X$ in the top left of the commuting square (Eq.~\eqref{eq:EmpiricalAdequacy}). There are two paths to the bottom right: in one direction, we map via the prediction map $\mathbf{p}$ to the same system $X$ in $\Inf$ and then we map via the equality to the same system $X$ in the other copy of $\Inf$; in the other direction, we map using the \crealist representation $\xi$ to $\xi_X$ in $\FI$ and then map to $\Inf$ via the prediction map $\mathbf{p^*}$ which leaves us with $\xi_X$. 
We then see that the only way that the diagram can commute is if $\xi_X=X$.

The fact that we take a \crealist representation to be diagram-preserving is an immediate consequence of our choice to take diagrams in an operational CI theory to represent one's hypothesis about the fundamental causal and inferential structure in the given scenario. Since an ontological representation is meant to be the most fundamental description of one's scenario, it should respect this hypothesis, with the only difference being that it will generally be a more fine-grained description (e.g., where laboratory procedures are replaced by functional dynamics).
We will leave to Appendix~\ref{manifestorcausal2} the reason behind our choice to have operational CI diagrams represent fundamental structure, and we also show therein that this choice does not limit the scope of possible classical realist representations in our framework.

A particularly natural class of classical realist representations are those which can be thought of simply as representing every state of knowledge about a procedure by a corresponding state of knowledge about the function that underlies it.\footnote{One might wonder whether this is sufficiently general given that for a procedure mapping system $A$ to system $B$, the variable $\Lambda_B$ might not be a function of $\Lambda_A$ alone but of $\Lambda_A$ together with some local auxiliary variable $\Lambda$ (whose value is drawn from some probability distribution).  Such worries are unfounded, however, since every value of $\Lambda$ defines a function from $\Lambda_A$ to $\Lambda_B$, and a probability distribution over this value induces a probability distribution over the latter function.}
Diagrammatically, these are represented as
\beq \label{natontrepn}
\InputIfFileExists{Diagrams/OntRep1.tikz}{}{\input{./figures/Diagrams/OntRep1.tikz}}\quad =\quad %
\InputIfFileExists{Diagrams/OntRep4.tikz}{}{\input{./figures/Diagrams/OntRep4.tikz}},
\eeq
where $\Xi_A^B$ is stochastic and satisfies a set of compositionality constraints in order for $\xi$ to be diagram-preserving, e.g.,
\beq
\InputIfFileExists{Diagrams/XiConstraints.tikz}{}{\input{./figures/Diagrams/XiConstraints.tikz}}\ = %
\InputIfFileExists{Diagrams/XiConstraints2.tikz}{}{\input{./figures/Diagrams/XiConstraints2.tikz}}.
\eeq
This class of classical realist representations is so natural, in fact, that one might even wish to demand that a classical realist representation be {\em defined} by such a constraint, although we have not done so here.

 Whether it is part of the definition or not, it is often sufficient to focus on this class alone because it turns out that every \crealist representation is inferentially equivalent to one in this class.

\begin{theorem}\label{thm:OntRepNF}
 Any \crealist representation $\xi$ satisfies
\beq
\InputIfFileExists{Diagrams/OntRep1.tikz}{}{\input{./figures/Diagrams/OntRep1.tikz}}\quad \sim_{{\bf p}^*}\quad %
\InputIfFileExists{Diagrams/OntRep4.tikz}{}{\input{./figures/Diagrams/OntRep4.tikz}} \ ,
\eeq
where $\Xi_A^B$ is a stochastic map taking states of knowledge about operational procedures to states of knowledge about functional dynamics.
\end{theorem}
\proof
The proof is given in Appendix~\ref{finalproof}.
\endproof

Nonetheless, it is not clear whether or not all \crealist representations, as defined in Definition~\ref{ontrepndefn}, are of the form of Eq.~\eqref{natontrepn}. For example, \crealist representations of the form
\beq
%
\InputIfFileExists{Diagrams/OntRep1.tikz}{}{\input{./figures/Diagrams/OntRep1.tikz}}\ = \ %
\InputIfFileExists{Diagrams/weirdOntModel3.tikz}{}{\input{./figures/Diagrams/weirdOntModel3.tikz}}
\eeq
may be consistent with Definition~\ref{ontrepndefn}, and appear to be more general than those of the form of Eq.~\eqref{natontrepn}.
 Such models, however, seem to fail to satisfy an assumption
 of autonomy---that the fundamental dynamics are independent of their inputs---and this may be grounds for dismissing them as candidates for a \crealist representation.
   It remains to be seen whether these can indeed be ruled out from our definition, or ruled out as a consequence of some formal notion of autonomy (which one might consider adding to Definition~\ref{ontrepndefn}).

 The question of the existence of a classical realist representation of an operational CI theory is closely connected to the pre-existing question of whether a given operational theory violates Bell-like inequalities.  We explore the connection in Section~\ref{causcompatineq}.

\subsection{\Crealist representations of quotiented operational CI theories}

It is also useful to define \crealist representations of quotiented operational CI theories.
\begin{definition} \label{quotontrepndefn}
 A {\em \crealist representation} of a quotiented operational CI theory, $\widetilde{\PI}$, by a quotiented \crealist CI theory, $\widetilde{\FI}$, is a diagram-preserving map \colorbox{gray!60}{$\mathbf{\widetilde{\xi}}:\widetilde{\PI} \to \widetilde{\FI}$}, depicted as 
\beq
\InputIfFileExists{Diagrams/QOntRep1.tikz}{}{\input{./figures/Diagrams/QOntRep1.tikz}} \ ,
\eeq
satisfying {\bf (i)} the \emph{preservation of predictions}, namely that the diagram
\beq\label{eq:EmpiricalAdequacy2}
\InputIfFileExists{Diagrams/EmpiricalAdequacy2.tikz}{}{\input{./figures/Diagrams/EmpiricalAdequacy2.tikz}},
\eeq
commutes,
where the double line between the two copies of $\Inf$ is an extended equals sign, and {\bf (ii)} the {\em preservation of  ignorability}
\beq\label{eq:CausalOntRep2}
\InputIfFileExists{Diagrams/QOntRep2.tikz}{}{\input{./figures/Diagrams/QOntRep2.tikz}}\quad
=\quad
\InputIfFileExists{Diagrams/QuotOntRep3.tikz}{}{\input{./figures/Diagrams/QuotOntRep3.tikz}} \ \ .
\eeq
 We will sometime refer to the classical realist representation $\mathbf{\widetilde{\xi}}:\widetilde{\PI} \to \widetilde{\FI}$ as an $\widetilde{\FI}$-representation of $\widetilde{\PI}$. 
\end{definition}

Representations of this sort are analogous to the simplex embedding maps introduced in Ref.~\cite{SchmidGPT} in the context of prepare-measure scenarios. (These mapped the states and effects of the GPT to probability distributions and response functions over the ontic state space in an ontological theory.) 

\begin{proposition}
The \crealist representation map $\widetilde{\xi}$ can be written as
\beq
\InputIfFileExists{Diagrams/QOntRep1New.tikz}{}{\input{./figures/Diagrams/QOntRep1New.tikz}}\ =  %
\InputIfFileExists{Diagrams/QOntRep4New.tikz}{}{\input{./figures/Diagrams/QOntRep4New.tikz}} \ ,
\eeq
where $\Xi_A^B$ is a stochastic map taking states of knowledge about operational procedures to states of knowledge about functional dynamics.
\end{proposition}
\proof
The proof is a direct adaptation of the proof of Theorem~\ref{thm:OntRepNF}, but where the starting point, Eq.~\eqref{normalformtoontrepn}, is modified
by replacing the inferential equivalence with equality and using the normal form for $\widetilde{\FI}$ as given by Corollary~\ref{normalformFI} , 
and where one uses the form of Lemma~\ref{lem:CausalProposition} which involves equality rather than inferential equivalence.
\endproof

\subsection{Leibnizianity (formalized)} \label{ncsec}

 A natural methodological principle to impose on candidate realist explanations of operational facts is the following~\cite{SpekLeibniz19}:
\begin{quote}
If an ontological theory implies the existence of two scenarios that are empirically indistinguishable in principle but ontologically distinct (where both indistinguishability and distinctness are evaluated by the lights of the theory in question), then the ontological theory should be rejected and replaced with one relative to which the two scenarios are ontologically identical.
\end{quote}
In Ref.~\cite{SpekLeibniz19}, it is argued that this methodological principle was proposed by Leibniz as a version of his principle of the identity of indiscernibles and that it was strongly endorsed (at least implicitly) by Einstein.
We shall refer to it here as {\em Leibniz's methodological principle.}

 Although this statement of Leibniz's principle has been argued to motivate the standard definition of generalized noncontextuality, it is necessary to consider an epistemological {\em generalization} of Leibniz's principle in order to motivate generalized noncontextuality in the rehabilitated form that we will endorse below.  
Specifically, instead of concerning pairs of scenarios that are empirically indistinguishable, the generalized principle concerns pairs of processes in a causal-inferential theory (such as states of knowledge) that are equivalent in the sense of allowing an agent to make precisely the same inferences.  
 We formalize the new version of the principle as a constraint on realist representations, which we term {\em Leibnizianity.}

\begin{definition}[Leibnizianity of a classical realist representation]\label{Leibnizianity}
A \crealist representation map \colorbox{gray!30}{$\mathbf{\xi}:\PI \to \FI$} is said to be {\em Leibnizian} if it preserves inferential equivalence relations. Otherwise, it is said to be {\em nonLeibnizian}. 
\end{definition}

\noindent More formally, a \crealist representation map $\xi$ is Leibnizian if, for any pair of inferentially equivalent processes $\mathcal{D},\mathcal{E}\in \PI$, one has
\beq\label{NCdiag1}
\InputIfFileExists{Diagrams/contextuality1.tikz}{}{\input{./figures/Diagrams/contextuality1.tikz}}\sim_\mathbf{p}%
\InputIfFileExists{Diagrams/contextuality2.tikz}{}{\input{./figures/Diagrams/contextuality2.tikz}} \ \ \implies \ \ %
\InputIfFileExists{Diagrams/contextuality4.tikz}{}{\input{./figures/Diagrams/contextuality4.tikz}}\sim_{\mathbf{p}^*}%
\InputIfFileExists{Diagrams/contextuality3.tikz}{}{\input{./figures/Diagrams/contextuality3.tikz}}.
\eeq

This means that if two processes lead one to make the same inferences when embedded into any diagram within the operational CI theory, then their representations within the \crealist CI theory must be such that they lead one to make all the same inferences when embedded into any diagram within the \crealist CI theory.
 As an example of an application of the principle, Leibnizianity stipulates that 
 {\em inferentially equivalent states of knowledge about experimental procedures must be represented by inferentially equivalent states of knowledge about functional dynamics. }

It is straightforward to verify from the definitions that an equivalent (process-theoretic) characterization of Leibnizianity is the following.
\begin{proposition} \label{equivcharLeib}
A classical realist representation $\xi:~\PI~\to~\FI$ of the unquotiented operational CI theory, is Leibnizian if and only if there exists a \crealist representation $\widetilde{\xi}:~\widetilde{\PI}~\to~\widetilde{\FI}$ of the quotiented operational CI theory, such that the following diagram commutes:
\beq\label{ncdiagcomm}
\begin{tikzpicture}
	\begin{pgfonlayer}{nodelayer}
		\node [style=none] (0) at (-1, 1.25) {$\PI$};
		\node [style=none] (2) at (1, 3) {$\widetilde{\PI}$};
		\node [style=none] (4) at (-1, -2) {$\FI$};
		\node [style=none] (5) at (1.25, 0) {$\widetilde{\FI}$};
		\node [style=none] (6) at (-1, 1.25) {};
		\node [style=none] (7) at (-1, -1.5) {};
		\node [style=none] (8) at (1.25, 0.5) {};
		\node [style=none] (9) at (1.25, 2.75) {};
		\node [style=none] (18) at (-0.75, 1.75) {};
		\node [style=none] (19) at (0.75, 2.75) {};
		\node [style=none] (20) at (0.75, -0.5) {};
		\node [style=none] (21) at (-0.75, -1.5) {};
		\node [style=none] (22) at (0.5, 2.75) {};
		\node [style=none] (23) at (0.5, -0.5) {};
		\node [style=none] (24) at (-1, -1.5) {};
		\node [style=none] (26) at (1, 2.75) {};
		\node [style=none] (27) at (1, -0.5) {};
		\node [style=right label] (33) at (-1, 0) {$\xi$};
		\node [style=right label] (34) at (1.25, 1.5) {$\widetilde{\xi}$};
		\node [style=up label] (35) at (0, 2.25) {$\sim_{\mathbf{p}}$};
		\node [style=up label] (36) at (0, -1) {$\sim_{\mathbf{p}^*}$};
	\end{pgfonlayer}
	\begin{pgfonlayer}{edgelayer}
		\draw [style=arrow plain] (6.center) to (7.center);
		\draw [style=arrow plain] (9.center) to (8.center);
		\draw [style=arrow plain] (18.center) to (19.center);
		\draw [style=arrow plain] (21.center) to (20.center);
	\end{pgfonlayer}
\end{tikzpicture} \ .
\eeq
\end{proposition}

We will ultimately be interested in contemplating the possibility of realist CI theories that are nonclassical alternatives to $\FI$ (see Sec.~\ref{outofnc}), so we will find it useful to sometimes refer to a representation in terms of $\FI$ as simply an $\FI$-representation. A given operational CI theory may admit of both Leibnizian and nonLeibnizian $\FI$-representations.  It will be termed {\em $\FI$-Leibnizian-representable} if it admits of at least {\em one} Leibnizian $\FI$-representation.

We have just described the implication $\exists\ \text{Leibnizian } \xi \implies  \exists \ \widetilde{\xi} $.
Whether the implication $\exists\ \text{Leibnizian } \xi \impliedby \exists\ \widetilde{\xi}$ holds remains an open question, but we conjecture that it does: 
\begin{conj}\label{conjecture1}
If a quotiented operational CI theory admits of a \crealist representation  (as a quotiented \crealist CI theory), then the unquotiented operational CI theory admits of a Leibnizian \crealist representation (as an unquotiented \crealist CI theory).  More formally, if there exists  a map \colorbox{gray!60}{$\mathbf{\widetilde{\xi}}:\widetilde{\PI} \to \widetilde{\FI}$} satisfying Definition~\ref{quotontrepndefn}, then there exists a map \colorbox{gray!30}{$\mathbf{\xi}:\PI \to \FI$} satisfying Definition~\ref{ontrepndefn} and which makes the diagram of Eq.~\eqref{ncdiagcomm} commute.
\end{conj}
Hence, we have
\beq \label{variouslemmas1}
\exists\ \text{Leibnizian }\xi\quad \begin{tikzpicture}
	\begin{pgfonlayer}{nodelayer}
		\node [style=none] (0) at (0, 0.5) {$\stackrel{?}{\impliedby}$};
		\node [style=none] (1) at (0, -0.25) {$\implies$};
	\end{pgfonlayer}
\end{tikzpicture}\quad \exists\ \widetilde{\xi}.
\eeq

Note that a \crealist representation $\widetilde{\xi}$ of a {\em quotiented} operational CI theory cannot itself be said to be either Leibnizian or nonLeibnizian. This is because in quotiented CI theories, there are no distinct but inferentially equivalent processes, for which one could ask whether their representations are inferentially equivalent or not.

In Refs.~\cite{SpekLeibniz19,Spe05}, it was argued that it is the plausibility of the Leibnizian methodological principle that accounts for the plausibility of the principle of generalized noncontextuality defined in Ref.~\cite{Spe05}.
Indeed, our notion of {Leibnizianity} turns out to have a close connection to the notion of generalized noncontextuality.
We discuss this connection in Section~\ref{sec:NCRehab}.

\subsection{Summary of the basic framework}

The full set of process theories and DP maps that we have introduced can be summarized by the following diagram (which is color-coded to match the intuitive schematic given further down):
\beq \label{fullfig}
\InputIfFileExists{Diagrams/CompleteDiagramColor.tikz}{}{\input{./figures/Diagrams/CompleteDiagramColor.tikz}} \quad.
\eeq
The top slice of this describes an operational CI theory (in blue) and its quotienting (in purple), while the bottom slice describes a \crealist CI theory (in red) and its quotienting (in magenta).
The {\em unquotiented} \crealist and operational  CI theories are constructed out of their respective causal theory ($\Func$ or $\Proc$) together with the classical theory of inference ($\Inf$), which is common to both.
The map $\xi$ (if it exists) constitutes a representation of an operational CI theory $\PI$ by a \crealist CI theory $\FI$, while the map $\widetilde{\xi}$ (if it exists) constitutes a representation of a quotiented operational CI theory $\widetilde{\PI}$ by a quotiented \crealist CI theory $\widetilde{\FI}$.
If both $\xi$ and $\widetilde{\xi}$ exist and the square Eq.~\eqref{ncdiagcomm} commutes, then we say that $\xi$ is a Leibnizian \crealist representation.

Fig.~\eqref{fullfig} is summarized by the following schematic:

\includegraphics[width=0.499\textwidth]{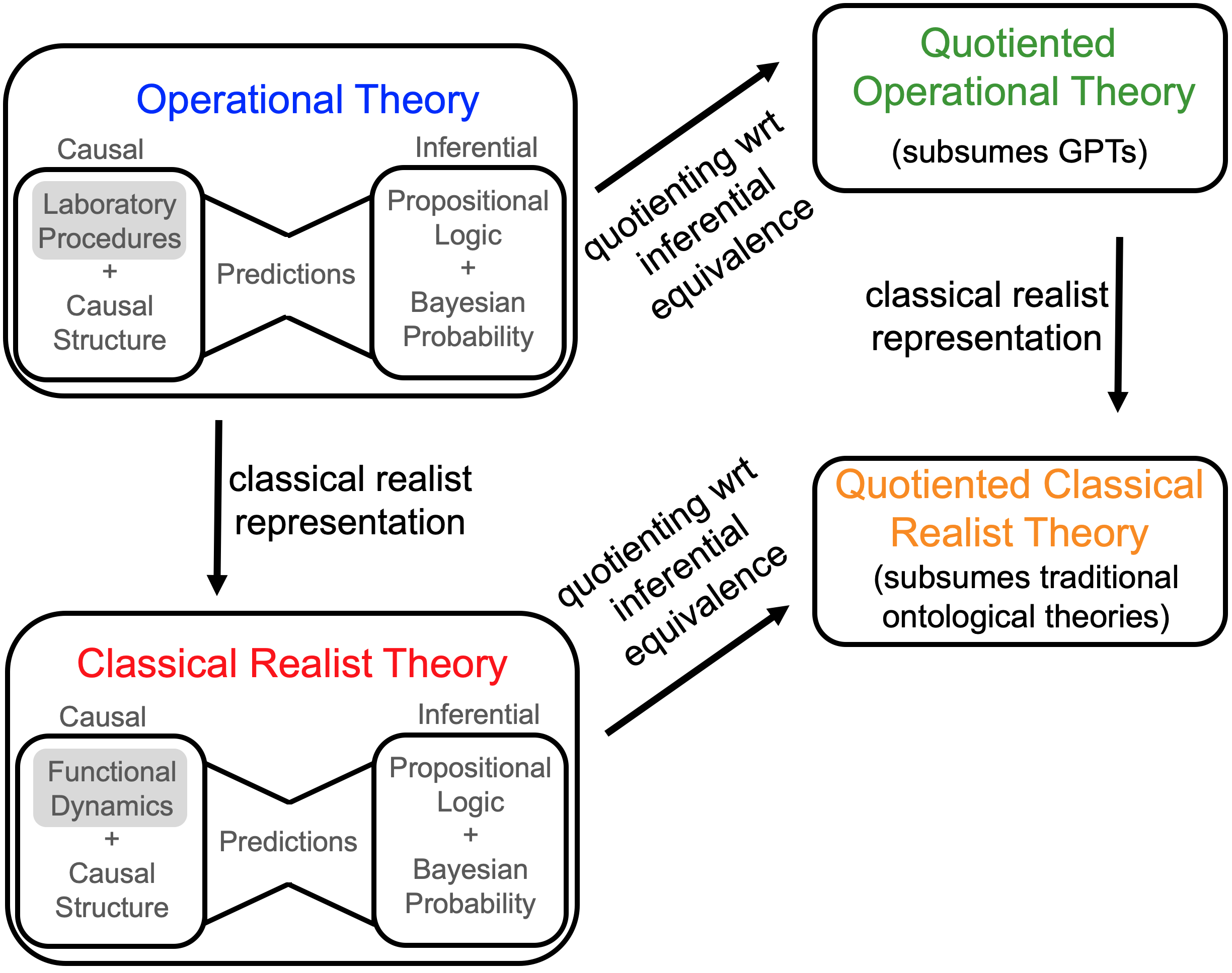}

 The form  of Diagram~\eqref{fullfig} might lead one to wonder about whether there ought to be a diagram-preserving map from $\Proc$ to $\Func$. One could define such a map, but it would state that every procedure could be associated to a unique function acting on the ontic states. As noted in Sec.~\ref{makingomelette}, this is not what we expect of a \crealist model of an operational theory, as laboratory procedures typically constitute a {\em coarse-grained} description, and as such are not associated with a
 unique function, but rather a distribution over these. 

We refer the reader to Appendix~\ref{relatedwork} for a discussion of prior work that is related to (or provided inspiration for) our framework.

\section{Bell-like no-go theorems}

\subsection{Bell-like inequalities as a consequence of assuming the existence of a classical realist representation
} \label{causcompatineq}

Consider a bipartite Bell experiment where $X$ and $Y$ denote the setting variables, and $A$ and $B$ denote the outcome variables.  Suppose one takes the causal structure of a Bell experiment to be given by the following  directed acyclic graph, or DAG,  where the triangle depicts an unobserved common cause:
\beq\label{Belldag}
\InputIfFileExists{Diagrams/bellDAG1.tikz}{}{\input{./figures/Diagrams/bellDAG1.tikz}}.
\eeq
This corresponds to the diagrams
\beq\label{Bellcircuit}
\InputIfFileExists{Diagrams/bellDAG2Rob.tikz}{}{\input{./figures/Diagrams/bellDAG2Rob.tikz}} \quad\text{and}\quad  %
\InputIfFileExists{Diagrams/bellDAG2RobFI.tikz}{}{\input{./figures/Diagrams/bellDAG2RobFI.tikz}},
\eeq
in $\Proc$ and $\Func$ respectively.

 A priori, this assumption about causal structure of the Bell experiment is the natural one.
 It is motivated by the idea that relativity implies no superluminal causation (not just a prohibition on superluminal signals).  We refer to it as the {\em common-cause hypothesis}~\cite{allen2017quantum,Wolfe2020quantifyingbell,2020arXiv200409194S} regarding the causal structure. In Section~\ref{waysoutofBell}, we will consider alternatives to it.

What about the full causal-inferential structure? In $\PI$ and $\FI$ these are
\beq \label{bellDiagr}
\InputIfFileExists{Diagrams/BellDAG3NewRob.tikz}{}{\input{./figures/Diagrams/BellDAG3NewRob.tikz}}\quad\text{and}\quad  %
\InputIfFileExists{Diagrams/BellDAG3NewRob2.tikz}{}{\input{./figures/Diagrams/BellDAG3NewRob2.tikz}},
\eeq
respectively, where we have allowed for arbitrary states of knowledge $\mu_A$, $\mu_B$, $\nu_A$ and $\nu_B$,  and $\sigma$ about the procedures  (respectively $\mu_A'$, $\mu_B'$, $\nu_A'$ and $\nu_B'$, and $\sigma'$ about the functions),
but we have not allowed for any statistical dependencies between the identities of the
procedures (nor, therefore, between the identities of the functions).
For instance, the factorization of $\nu_A$ and $\sigma$ (and hence of $\nu'_A$ and $\sigma'$) is motivated by the implausibility of nature conspiring to ensure that the mechanism that sets the value of the setting variable is related to the mechanism that fixes the value of the common cause. (This is related to superdeterminism, discussed further in the next section.)

We now discuss how the causal-inferential hypotheses embodied in Eq.~\eqref{bellDiagr} constrains the possible observations that can be made within the classical realist theory, as well as those that can be made within a given operational theory.

First, we consider the question of what joint distributions over $X, Y, A$ and $B$ can be generated in the causal-inferential structure of Eq.~\eqref{bellDiagr} within the classical realist CI theory $\FI$. These are given by the diagram
\beq \label{thecorrel}
\InputIfFileExists{Diagrams/BellDAG4NewRob.tikz}{}{\input{./figures/Diagrams/BellDAG4NewRob.tikz}} \ \ ,
\eeq
where one ranges over arbitrary sets $\Lambda$ and $\Lambda'$ and probability distributions $\sigma'$, $\mu_A'$, $\mu_B'$, $\nu_A'$, and $\nu_B'$.
We refer to any distribution that arises in this way as
{\em $\FI$-realizable}.
The constraints that pick out the set of $\FI$-realizable distributions generally come in the form of both equalities and inequalities, and we will term these {\em $\FI$-compatibility constraints}.  In particular, inequality constraints will be termed {\em $\FI$-compatibility inequalities}.

Within our framework, the Bell inequalities (e.g., the Clauser-Horne-Shimony-Holte inequalities for the case where $X,Y,A$ and $B$ are binary) are examples of $\FI$-compatibility inequalities for the causal-inferential structure of Eq.~\eqref{bellDiagr} (which, as noted earlier, is the natural choice for the Bell experiment). 

Next, we turn to the question of what joint distributions over $X, Y, A$ and $B$ can be generated in the causal-inferential structure of Eq.~\eqref{bellDiagr} within an operational CI theory $\PI$. These are given by the diagram
\beq \label{thecorrel2}
\InputIfFileExists{Diagrams/BellDAG4New.tikz}{}{\input{./figures/Diagrams/BellDAG4New.tikz}} \ \ ,
\eeq
where one ranges over arbitrary systems $S$ and $S'$ in $\PI$ and probability distributions $\sigma$, $\mu_A$, $\mu_B$, $\nu_A$, and $\nu_B$.
Note that the set of distributions that can be obtained in this way depends on the prediction map ${\bf p}$ of $\PI$, and so will vary from one operational CI theory to the next.  We refer to any distribution that can arise in this way within an operational CI theory $\PI$ as {\em $\PI$-realizable}.  The constraints that pick out the set of $\PI$-realizable distributions will be termed $\PI$-compatibility constraints.  In particular, any inequality constraints will be termed {\em $\PI$-compatibility inequalities}.

For the case of quantum theory, considered as the operational CI theory $\PQI$ introduced in Section~\ref{QTquaopCI}, the well-known Tsirelson inequalities~\cite{Tsirelson} are examples of $\PQI$-compatibility inequalities for the causal-inferential structure of Eq.~\eqref{bellDiagr}. 

If a classical realist representation as in Theorem~\ref{thm:OntRepNF} exists for a given operational theory $\PI$, it follows that every distribution that is $\PI$-realizable is also $\FI$-realizable.
Thus, for a given operational CI theory $\PI$ to admit of a classical realist representation, it must be the case that the set of $\PI$-realizable distributions for any possible causal-inferential structure is included in the $\FI$-realizable distributions for the same causal-inferential structure.


 For the case of quantum theory, there exist distributions~\cite{Bellreview} that satisfy the Tsirelson inequalities but violate the Bell inequalities, i.e., that are $\PQI$-compatible  but not $\FI$-compatible with the causal-inferential structure of Eq.~\eqref{bellDiagr}.  It follows that the set of $\PQI$-realizable distributions is not included in the set of $\FI$-realizable distributions, and consequently that $\PQI$ does not admit of a classical realist representation.

This is how Bell's theorem is conceptualized in our framework.

When Bell's theorem is conceptualized in this way, it is found to have counterparts in causal
 structures distinct from that of the Bell scenario.
   For instance, recent work has demonstrated that for the triangle scenario~\cite{fritz2012beyond,inflation,triangle,TriangleGisin}, whose DAG and associated diagram in $\Proc$ are
\beq\label{triangledag}
\InputIfFileExists{Diagrams/triangleDag.tikz}{}{\input{./figures/Diagrams/triangleDag.tikz}}\qquad %
\InputIfFileExists{Diagrams/triangleCircuit.tikz}{}{\input{./figures/Diagrams/triangleCircuit.tikz}} \ ,
\eeq
 there is a gap between what is realizable in a causal model where some common causes can be quantum and what is realizable in a causal model where all common causes are classical.
  A similar result has been shown for the instrumental scenario~\cite{Chaves2018,VanHimbeeck2019quantumviolationsin} , whose DAG and associated diagram in $\Proc$ are
\beq\label{instrumentaldag}
\InputIfFileExists{Diagrams/instrumentalDagNew.tikz}{}{\input{./figures/Diagrams/instrumentalDagNew.tikz}}
\qquad
\InputIfFileExists{Diagrams/instrumentalCircuit.tikz}{}{\input{./figures/Diagrams/instrumentalCircuit.tikz}}.
\eeq

These new Bell-like no-go theorems are subsumed in our framework as proofs of the impossibility of a classical realist CI representation
 based on the correlations predicted by quantum theory for these causal scenarios.
The counterpart within our framework to realizability by a {\em quantum} causal model~\cite{henson2014theory,fritz2012beyond,inflation,allen2017quantum,wolfe2019quantum,Barrett2019}
is $\PQI$-realizability,
while the counterpart within our framework to realizability by a {\em classical} causal model  is $\FI$-realizability.  Thus, the counterpart to these no-go theorems is that for each of these causal structures, one can find distributions over the observed variables that are $\PQI$-realizable,
  but not $\FI$-realizable.

\subsection{The conventional ways out of Bell-like no-go theorems}\label{waysoutofBell}
We now describe the standard responses to Bell-like no-go theorems, focusing on the specific case of Bell's theorem (rather than those based on, e.g., the instrumental or triangle scenarios).

A common attitude towards Bell's theorem is that it demonstrates that realism must be abandoned (at least in the quantum sphere) hence vindicating an operationalist philosophy of science. 
As discussed in the introduction, we take the key distinction between a realist and a purely operational account of statistical correlations to be whether or not these accounts provide a causal explanation of the correlations.  
We believe that any philosophy of science which is antirealist in this sense---namely, which denies the possibility of a causal explanation of statistical correlations---is unsatisfactory.  Furthermore, given the possibility (which we will introduce in Section~\ref{ncrealrepns}) for nonclassical realist representations that modify the notions of causation and inference used in one's causal explanation, we are certainly not persuaded by any claim that the standard no-go theorems {\em necessitate} a retreat from realism in the quantum sphere.  For these reasons, we are here interested in reviewing the standard ways of maintaining realism in the face of Bell's theorem, in order to contrast them with the way we propose to do so.

For those unwilling to compromise on realism, the conventional way out
 is  to deny the  natural 
  causal-inferential hypothesis of Eq.~\eqref{bellDiagr}.
 They assume therefore that a {\em radical} causal-inferential hypothesis underpins the correlations observed in a Bell scenario. In this case, the existence of a classical realist representation does not imply satisfaction of the Bell inequalities, and thus violations of Bell inequalities no longer imply a challenge to the possibility of such a representation. We now describe  the two most common positions about the nature of the radical causal hypotheses.

{\bf The hypothesis that a proponent of superluminal causation endorses.}
Those who see superluminal\footnote{Formally describing the distinction between sub- versus superluminality motivates a minor extension of our framework in which systems come equipped with spatiotemporal labels. We leave this for future work.} causation as the way out of Bell's theorem take the causal structure of a Bell experiment to be one that allows for causal influences between the wings, even when these are space-like separated.   As an example, this influence might be between the setting on the left wing and the outcome of the right, a causal hypothesis that is depicted by the following DAG, or equivalently, by the following circuit diagram:
\beq\label{superluminaldag}
\InputIfFileExists{Diagrams/superluminalDag.tikz}{}{\input{./figures/Diagrams/superluminalDag.tikz}}\qquad %
\InputIfFileExists{Diagrams/superluminalCircuit.tikz}{}{\input{./figures/Diagrams/superluminalCircuit.tikz}}.
\eeq

{\bf The hypothesis that a proponent of superdeterminism endorses.}
Those who see superdeterminism as the way out of Bell's theorem take there to be some statistical dependence between a setting variable, say $X$, and the common cause of the outcomes.
 This assumption could be encoded in either the causal or the inferential structure. We here opt to encode it as the assumption that $X$ and the common cause of the outcomes are not causally disconnected.  This assumption can be  depicted by the following DAG, or equivalently, by the following circuit diagram:
\beq\label{superdeterminismdag}
\InputIfFileExists{Diagrams/superdeterminismDag.tikz}{}{\input{./figures/Diagrams/superdeterminismDag.tikz}} \qquad %
\InputIfFileExists{Diagrams/superdeterminismCircuit.tikz}{}{\input{./figures/Diagrams/superdeterminismCircuit.tikz}}.
\eeq

 At this stage, we would like to head off a possible confusion.  Because it is customary within pre-existing operational  frameworks to use the standard quantum circuit of Eq.~\eqref{Bellcircuit} as the diagram representing the Bell experiment, it might seem that a proponent of a radical causal hypothesis must be contemplating a realist representation that fails to be diagram-preserving, and thus it might seem that our framework, which assumes diagram preservation, cannot do justice to their view.
However, as noted previously and elaborated on in Appendix~\ref{manifestorcausal2}, our framework stipulates that the diagram representing a given scenario in the causal subtheory of an operational CI theory is a representation of one's hypothesis about the fundamental causal structure, and hence {\em need not} correspond to the standard quantum circuit of Eq.~\eqref{Bellcircuit}.
As such, researchers with different realist worldviews, faced with the same experimental scenario and observed statistics, might model these {\em differently} within the framework of operational causal-inferential theories.
E.g., the representation of a Bell scenario within a quantum operational CI theory will be constrained by the causal structure of Eq.~\eqref{superluminaldag} by a proponent of superluminal causation, but will be constrained by the causal structure of Eq.~\eqref{superdeterminismdag} by a proponent of superdeterminism. (Furthermore, the manner by which such researchers would formalize quantum theory as an operational CI theory will not be the one 
 described in Section~\ref{QTquaopCI}.)

\subsubsection{Criticisms of the conventional ways out of Bell-like no-go theorems}
\label{criticismsbell}

We now discuss various reasons why we view these conventional responses to Bell-like no-go theorems as  unsatisfactory.

The following are grounds for rejecting the causal-inferential hypothesis of the proponent of superluminal causation:
\begin{itemize}
\item superluminal causal influences are in tension with the spirit of relativity theory (even if these influences are constrained in such a way as to be consistent with the impossibility of superluminal signals)
\item superluminal causal influences which are constrained in such a way as to be consistent with the impossibility of superluminal signals violate the principle of no fine-tuning (defined in Section~\ref{causmode}).~\cite{wood2015lesson}.
\end{itemize}

Meanwhile, the following are grounds for rejecting the causal-inferential hypothesis of the proponent of superdeterminism:
\begin{itemize}
\item for any of the various mechanisms that could determine the value of the setting variables (an agent's free choice, a random number generator, a hash of the day's stock prices, etcetera), it is implausible that it would be related to the mechanism which determines
the common cause of the outcome variables insofar as this would require a kind of conspiracy of causal determinations
\item the fact that any such nontrivial statistical associations between a setting variable and the common cause of the outcome variables must be constrained to be consistent with the observed {\em lack} of any statistical assocation between that setting variable and the outcome variable at the opposite wing of the experiment implies violation of the principle of no fine-tuning~\cite{wood2015lesson}.
\end{itemize}

In fact, in Ref.~\cite{wood2015lesson} it was shown that {\em every} causal hypothesis that can realize the distributions predicted by quantum theory (i.e., Bell-inequality-violating distributions) using a classical causal model
 (i.e., via $\FI$) implies a violation  of the principle of no fine-tuning.

We take these arguments to be good grounds for rejecting the idea of explaining correlations in Bell scenarios via a radical causal-inferential hypothesis together with a classical realist representation.

Even if one rejects the principle of no fine-tuning, standard tools of model selection (which are sensitive not just to underfitting but overfitting as well)
 can adjudicate between various hypotheses regarding the right way to operationally model the Bell experiment, and these also rule against a radical causal hypothesis~\cite{DaleySpekkensToBePublished}.

The no-fine-tuning arguments just given apply equally well to no-go theorems based on causal structures beyond Bell scenarios, e.g., the instrumental and triangle scenarios.  That is, one can attempt to resolve the contradiction in each of these cases by hypothesizing that the causal structure is in fact distinct from that depicted in Eq.~\eqref{triangledag} and Eq.~\eqref{instrumentaldag}. Such resolutions, however, also generally suffer from a fine-tuning objection insofar as the set of distributions that are $\PQI$-realizable for the original causal structure often have measure zero within the set of distributions that are $\FI$-realizable for the radical causal structure, and they will in some instances also require superluminal causation.

Insofar as these conventional ways out of Bell's theorem require radical causal-inferential assumptions with the aforementioned undesirable features, the natural question becomes:
{\em does there remain any recourse for achieving a realist causal-inferential representation of quantum theory without these unappealing features?}  In Section~\ref{ncrealrepns}, we outline a research program for achieving realism {\em while preserving the conservative causal-inferential hypothesis}, by allowing for intrinsically nonclassical notions of influence and inference.

Before coming to this, however, we discuss how no-go theorems based on the principle of noncontextuality appear within our framework.

\section{Noncontextuality no-go theorems}

\subsection{Generalized noncontextuality (rehabilitated)} \label{sec:NCRehab}

The pre-existing notion of generalized noncontextuality~\cite{Spe05} is framed as a constraint on ontological models of operational theories
  and can be summarized as ``{\em operationally equivalent procedures must be represented identically in the ontological theory}''. However, our framework has refined the traditional notions of ontological theories 
   and of operational theories, and the notion of operational equivalence of procedures has been replaced by inferential equivalence of states of knowledge about procedures.  The notion of generalized noncontextuality must therefore be refined accordingly.

Generalized noncontextuality is a principle that constrains {\em a} classical realist representation map, but it is {\em not} the map $\xi$ that we have been focused on so far. Rather, it is a constraint on a map  \colorbox{gray!60}{$\mathbf{\zeta}:\PI \to \widetilde{\FI}$} from the {\em unquotiented} operational CI theory $\PI$ to the {\em quotiented}  \crealist  CI theory $\widetilde{\FI}$. The fact that this is a map across the unquotiented-quotiented divide, means that this specific sort of \crealist representation is less fundamental than the map $\xi$. Nonetheless, it is useful to introduce this map as a formal tool, e.g., in order to make connections to existing literature.
\noindent Paralleling Definitions~\ref{ontrepndefn} and \ref{quotontrepndefn}, we define $\zeta$ as follows:
\begin{definition} \label{repnsthatmaybeNC}
  A classical realist representation of an unquotiented operational CI theory, $\PI$, by a quotiented classical realist CI theory, $\widetilde{\FI}$,
  is a diagram-preserving map \colorbox{gray!45}{$\mathbf{\zeta}:\PI \to \widetilde{\FI}$} 
satisfying {\bf (i)} the \emph{preservation of predictions}, namely that the diagram
\beq\label{eq:EmpiricalAdequacy2}
\InputIfFileExists{Diagrams/EmpiricalAdequacy3.tikz}{}{\input{./figures/Diagrams/EmpiricalAdequacy3.tikz}},
\eeq
commutes, 
where the double line between the two copies of $\Inf$ is an extended equals sign, and {\bf (ii)} the {\em preservation of  ignorability}
\beq\label{eq:CausalOntRep5}
\InputIfFileExists{Diagrams/QOntRep3.tikz}{}{\input{./figures/Diagrams/QOntRep3.tikz}}\quad
=\quad
\InputIfFileExists{Diagrams/QuotOntRep3.tikz}{}{\input{./figures/Diagrams/QuotOntRep3.tikz}}.
\eeq
 We will sometimes refer to the classical realist representation $\zeta:\PI\to\widetilde{\FI}$ as an $\widetilde{\FI}$-representation of $\PI$. 
\end{definition}

Adding the $\zeta$ map to the diagram relating $\PI$, $\FI$,  and their quotiented counterparts yields 
\beq \label{trianglecommutenc}
\begin{tikzpicture}
	\begin{pgfonlayer}{nodelayer}
		\node [style=none] (0) at (-1, 1.25) {$\PI$};
		\node [style=none] (2) at (1, 3) {$\widetilde{\PI}$};
		\node [style=none] (4) at (-1, -2) {$\FI$};
		\node [style=none] (5) at (1.25, 0) {$\widetilde{\FI}$};
		\node [style=none] (6) at (-1, 1.25) {};
		\node [style=none] (7) at (-1, -1.5) {};
		\node [style=none] (8) at (1.25, 0.5) {};
		\node [style=none] (9) at (1.25, 2.75) {};
		\node [style=none] (18) at (-0.75, 1.75) {};
		\node [style=none] (19) at (0.75, 2.75) {};
		\node [style=none] (20) at (0.75, -0.5) {};
		\node [style=none] (21) at (-0.75, -1.5) {};
		\node [style=none] (22) at (0.5, 2.75) {};
		\node [style=none] (23) at (0.5, -0.5) {};
		\node [style=none] (24) at (-1, -1.5) {};
		\node [style=none] (26) at (1, 2.75) {};
		\node [style=none] (27) at (1, -0.5) {};
		\node [style=right label] (33) at (-1, 0) {$\xi$};
		\node [style=right label] (34) at (1.25, 1.5) {$\widetilde{\xi}$};
		\node [style=up label] (35) at (0, 2.25) {$\sim_{\mathbf{p}}$};
		\node [style=up label] (36) at (0, -1) {$\sim_{\mathbf{p}^*}$};
		\node [style=none] (37) at (0.75, 0.25) {};
		\node [style=right label] (38) at (0.25, 0.75) {$\zeta$};
		\node [style=none] (39) at (-0.75, 1) {};
	\end{pgfonlayer}
	\begin{pgfonlayer}{edgelayer}
		\draw [style=arrow plain] (9.center) to (8.center);
		\draw [style=arrow plain] (39.center) to (37.center);
		\draw [style=arrow plain] (18.center) to (19.center);
		\draw [style=arrow plain] (21.center) to (20.center);
		\draw [style=arrow plain] (6.center) to (7.center);
	\end{pgfonlayer}
\end{tikzpicture}
 \quad .
\eeq

 Note that the $\zeta$ map is the closest counterpart in our framework to the pre-existing notion of an {\em ontological model of an operational theory}~\cite{Harrigan}.  

 It follows that the closest counterpart in our framework to the pre-existing notion of a {\em generalized-noncontextual} ontological model  is that of a generalized-noncontextual classical realist representation of this sort.  This can be defined in terms of the formal notion of Leibnizianity introduced in Definition~\ref{Leibnizianity}, but applied to $\zeta$ (rather than $\xi$):

\begin{definition}[Generalized-noncontextual classical realist representation]
 The classical realist representation map \colorbox{gray!45}{$\mathbf{\zeta}:\PI \to \widetilde{\FI}$}  is generalized-noncontextual if it preserves inferential equivalence relations.
\end{definition}
\noindent
 In the case of a single causally closed process, this can be summarized as
 ``{\em inferentially equivalent states of knowledge about experimental procedures must be represented by the same stochastic map.}''

We now give an equivalent (process-theoretic) characterization of a generalized-noncontextual classical realist representation, in analogy to Proposition~\ref{equivcharLeib}.
\begin{proposition}
A classical realist representation map  \colorbox{gray!45}{$\mathbf{\zeta}:\PI \to \widetilde{\FI}$} is generalized-noncontextual if and only if there exists a map $\widetilde{\xi}$ as defined in Definition~\ref{quotontrepndefn} such that the upper right triangle in Eq.~\eqref{trianglecommutenc} commutes, implying that the map $\zeta$ can be factored as $ \widetilde{\xi}\circ \sim_{\mathbf{p}}$.
\end{proposition}

 It follows that if there exists a map $\widetilde{\xi}$, then there exists a generalized-noncontextual map $\zeta$ (namely $\zeta = \widetilde{\xi} \circ  \sim_{\mathbf{p}}$).   In fact, the opposite implication holds as well because the claim that $\zeta$ is generalized-noncontextual means {\em by definition} that there exists a map $\widetilde{\xi}$ such that the upper right triangle in Eq.~\eqref{trianglecommutenc} commutes. We can therefore summarize the relationship as follows:
\beq \label{GNCzetaifftildexi}
\exists\ \text{generalized-noncontextual } \zeta\ \ \iff \ \ \exists\ \widetilde{\xi}.
\eeq
Note that insofar as a quotiented operational CI theory subsumes the notion of a GPT, this fact is the analogue, for the rehabilitated version of generalized noncontextuality 
of Proposition 1 of Ref.~\cite{schmid2020structure}, which asserts that an operational theory admits of a generalized-noncontextual ontological model if and only if the GPT defined by it admits of an ontological model.

To re-emphasize a point made in Ref.~\cite{schmid2020structure}, a {\em quotiented} operational CI theory is not the sort of thing that can be either generalized-noncontextual or generalized-contextual because information about context is precisely what is eliminated by the quotienting operation.  In other words, it is a category mistake to ask whether $\widetilde{\xi}$ is generalized-noncontextual or not.  Thus, for any experiment that does not support a Bell-like no-go theorem but does support a noncontextuality no-go theorem, its model within $\PI$ is {\em always} consistent with {\em some} classical realist representation $\xi$ given the representational freedom that is afforded by context-dependence. Its model within $\widetilde{\PI}$, however, might not admit of {\em any} classical realist representation.

Old proofs of the failure of generalized noncontextuality will imply proofs of the failure of the rehabilitated version of generalized noncontextuality, since all that has changed is how one conceptualizes the mathematics.  This is particularly clear given Eq.~\eqref{GNCzetaifftildexi} and the close connection between the question of the existence of a $\widetilde{\xi}$ map in our framework and the question of whether a given GPT model of an experiment admits of an ontological model.

 In previous work, generalized noncontextuality was defined case-by-case for various types of procedures (e.g. preparations, measurements, transformations~\cite{Spe05}, and instruments~\cite{PP2,AWV,schmid2020structure}), and it was then stipulated that the natural assumption of generalized noncontextuality is the {\em universal} version of this assumption, meaning for {\em all} types of procedures.  In contrast, our process-theoretic characterization of generalized noncontextuality 
  applies to {\em all} types of experimental procedures including more exotic processes such as combs and circuit fragments with arbitrary causal structure, 
 and therefore is a {\em universal} notion from the get-go.

\subsection{Failures of generalized noncontextuality imply failures of Leibnizianity}

We now discuss the relation between the rehabilitated notion of generalized noncontextuality and the notion of Leibnizianity.

While the rehabilitated notion of generalized noncontextuality is a principle that constrains the representation map $\zeta$, the notion of Leibnizianity is a principle that constrains the representation map $\xi$.  It is like
generalized noncontextuality insofar as it can be defined in terms of the commutation of the diagram, but it {\em cannot} be understood as a notion of independence on context, as we now explain.

We begin by defining the notion of context that is at play in the notion of generalized noncontextuality.  
For any process within a CI theory, its \emph{context}  
 is the information which determines which element of an inferential equivalence class of processes it is---that is, it is information about a process which is irrelevant for making predictions.
Note that both operational CI theories {\em and \crealist CI theories} have nontrivial contexts: in either case, a full specification of a state of knowledge over the relevant causal processes describes both the equivalence class and the context of a procedure. 
An explicit example of two processes in $\PI$ which are in the same inferential equivalence class but which differ by context is given by the two states of knowledge about operational procedures described in Eq.~\eqref{twoinfstatesofknow}. 
An explicit example of two processes in $\FI$  which are in the same inferential equivalence class but which differ by context is given by the two states of knowledge about functions described in  Eq.~\eqref{randommat}. 
 In this language, what Leibnizianity demands is that an operational process's context
can only determine the context of its image under the realist representation map, 
not the inferential equivalence class of its image under the realist representation map. 
Even when a given representation map $\xi$ satisfies this condition, the context of an operational process's image under $\xi$ \emph{can} depend on the context of the operational process, and hence, it would be inappropriate to call the map `context-independent' or `noncontextual'.
Thus, the inapplicability of the term `noncontextual'  for describing the relevant constraint on the representation map $\xi$ is seen to be a consequence of the fact 
 that the fundamental notion of a realist CI theory is an {\em unquotiented} one, which has contexts.

 In contrast, the realist representation map $\zeta$ has contexts in its domain but not its co-domain, so that what Leibnizianity demands in this case is that an operational process's context cannot determine anything about its image under $\zeta$. Thus, a map $\zeta$ can be either generalized-noncontextual or generalized-contextual, depending on whether or not the image of an operational process under this map is independent of the context of the operational process.

We turn now to the formal relationships between the two notions.
Combining Eq.~\eqref{variouslemmas1} and Eq.~\eqref{GNCzetaifftildexi}, we infer that
\beq \label{variouslemmas}
\exists\ \text{Leibnizian }\xi
\
 \begin{tikzpicture}
	\begin{pgfonlayer}{nodelayer}
		\node [style=none] (0) at (0, 0.35) {$\stackrel{?}{\impliedby}$};
		\node [style=none] (1) at (0, -0.25) {$\implies$};
	\end{pgfonlayer}
\end{tikzpicture}
\
 \exists\
\text{gen.-noncontextual }
 \zeta
  \iff
   \exists\ \widetilde{\xi}.
\eeq
\noindent Accordingly, Conjecture~\ref{conjecture1} can be reformulated as follows:
\begin{conjbis}{conjecture1}\label{conjecture2}
If an operational CI theory admits of a generalized-noncontextual \crealist representation as a quotiented \crealist CI theory,
  then it admits of a Leibnizian \crealist representation as an unquotiented \crealist CI theory.  More formally, if there exists  a map \colorbox{gray!45}{$\mathbf{\zeta}:\PI \to \widetilde{\FI}$} satisfying Definition~\ref{repnsthatmaybeNC} and which makes the upper right triangle in the diagram of~\eqref{trianglecommutenc} commute, then there exists a map \colorbox{gray!30}{$\mathbf{\xi}:\PI \to \FI$} satisfying Definition~\ref{ontrepndefn} and which makes the square in the diagram of\eqref{trianglecommutenc} commute.
\end{conjbis}

The existence of a Leibnizian \crealist representation $\xi$ implies the existence of a generalized-noncontextual \crealist representation $\zeta$. Contrapositively, the nonexistence of such a $\zeta$ implies the nonexistence of such a $\xi$, and therefore every no-go theorem for generalized-noncontextual classical realist representations yields a no-go theorem for Leibnizian classical realist representations.

What about 
 implications  in the other direction?  If Conjecture~\ref{conjecture2} is false, then there may be proofs of the impossibility of a Leibnizian classical realist representation that are not proofs of the impossibility of a generalized-noncontextual classical realist representation.  In this case, there may be novel no-go theorems for classical realist representations of operational quantum theory.  By contrast, if the conjecture is true, then every no-go theorem based on Leibnizianity yields a no-go theorem based on generalized noncontextuality.
 
  \subsection{Reframing the standard no-go theorem for generalized noncontextuality}

To illustrate how the notion of noncontextuality appears in our framework, we consider a simple prepare-measure scenario.  Let the setting variable for the measurement be denoted by $Y$ and the outcome be denoted  by $B$.  In many discussions of noncontextuality, the preparation device is imagined to have just a setting variable.  In order to achieve more symmetry between the preparation device and the measurement device, however, it is convenient to consider a preparation device that has not only a setting variable, but an outcome as well~\cite{kunjwal2018statistical}.  We denote the setting by $X$ and the outcome by $A$.  The causal structure of a prepare-measure is presumed to be the following
\beq\label{PMcircuitCaus}
\InputIfFileExists{Diagrams/conScen5.tikz}{}{\input{./figures/Diagrams/conScen5.tikz}} \quad \text{and} \quad %
\InputIfFileExists{Diagrams/conScen6.tikz}{}{\input{./figures/Diagrams/conScen6.tikz}}
\eeq
where we have shown the representations in $\Proc$ and $\Func$ respectively.

The full causal-inferential structure is represented as follows in $\PI$ and $\FI$ respectively:
\beq\label{PMcircuitCI}
\InputIfFileExists{Diagrams/conScen3.tikz}{}{\input{./figures/Diagrams/conScen3.tikz}} \quad \text{and} \quad %
\InputIfFileExists{Diagrams/conScen4.tikz}{}{\input{./figures/Diagrams/conScen4.tikz}}
\eeq
where we have allowed for arbitrary states of knowledge $\mu_P$, $\mu_M$, $\nu_P$, and $\nu_M$ about the procedures  (respectively $\mu_P'$, $\mu_M'$, $\nu_P'$, and $\nu_M'$ about the functions),

The $\PI$-realizable joint distributions over $X, Y, A$ and $B$ for the causal-inferential structure of Eq.~\eqref{PMcircuitCI} are those given by the diagram
\beq\label{PMCircuitPSRealizable}
\InputIfFileExists{Diagrams/conScen1.tikz}{}{\input{./figures/Diagrams/conScen1.tikz}}
\eeq
where one ranges over arbitrary system $S$ in $\PI$ and probability distributions $\mu_P$, $\mu_M$, $\nu_P$, and $\nu_M$.

The $\FI$-realizable distributions over $X, Y, A$ and $B$ are those given by
\beq\label{PMCircuitFSRealizable}
\InputIfFileExists{Diagrams/conScen2.tikz}{}{\input{./figures/Diagrams/conScen2.tikz}}
\eeq
where one ranges over an arbitrary set $\Lambda$ and probability distributions $\mu_P'$, $\mu_M'$, $\nu_P'$, and $\nu_M'$.

It turns out that for this causal-inferential structure, the only restrictions on $\FI$-realizable distributions over $X, Y, A$ and $B$ are that $Y$ must be independent of $A$ and $X$.   So there is no opportunity for a Bell-like no-go result for this causal-inferential structure. 

Nonetheless, one {\em can} prove a noncontextuality no-go result in such structures.  The reason is that proving a noncontextuality no-go is not about proving the impossibility of an $\FI$-representation of $\PI$, as  is the case for proving a  Bell no-go result.  Rather, it is about proving the impossibility of finding an $\FI$-representation of $\PI$ {\em that is Leibnizian}, as given in Definition~\ref{Leibnizianity}.

In the case of quantum theory, that is, $\PQI$, one can prove such a result using the causal-inferential structure of \eqref{PMCircuitFSRealizable}, by leveraging no-go theorems for generalized noncontextuality in prepare-measure scenarios.  
One can of course also consider the consequences of Leibnizianity for scenarios beyond the prepare-measure variety.

Whether an operational theory $\PI$ admits of a generalized-noncontextual $\widetilde{\FI}$-representation 
coincides with several pre-existing notions of classical explainability.  This follows from the fact that the existence of a generalized-noncontextual $\widetilde{\FI}$-representation
 of $\PI$ implies the existence of an $\widetilde{\FI}$-representation of $\widetilde{\PI}$ and the latter coincides with the following two notions of classical explainability: simplex-embeddability of the GPT describing prepare-measure experiments~\cite{SchmidGPT,shahandeh2019contextuality} and the existence of a positive quasiprobability representation of the GPT~\cite{schmid2020structure,negativity,ferrie2008frame}.  
 
 Additionally, an operational theory $\PI$ that fails to admit of a generalized-noncontextual classical realist representation provides advantages for information processing relative to those that do admit of such a representation~\cite{POM,RAC,RAC2,Saha_2019,saha2019preparation,schmid2018contextual,Lostaglio2020contextualadvantage,schmid2021only,magic,comp1,comp2,PhysRevLett.125.230603}.  This bolsters the notion that {\em failing} to admit of a generalized-noncontextual $\widetilde{\FI}$-representation is a good notion of {\em non}classicality.
 
 In light of the relationships proven above between Leibnizianity and generalized noncontextuality, each of these results concerning generalized noncontextuality can be repurposed as a motivation for assuming that the existence of a Leibnizian $\FI$-representation is 
 a good notion of classicality for an operational CI theory $\PI$.

To summarize, our framework yields a new perspective on the relationship between Bell-like and noncontextuality (or Leibnizianity) no-go theorems. Both types of no-go theorems concern the representability of an operational CI theory $\PI$ in terms of a classical realist CI theory $\FI$. A Bell-like no-go theorem is a demonstration that there does not exist {\em any} classical realist representation map $\xi : \PI \to \FI$ as in Definition~\ref{ontrepndefn}. A noncontextuality no-go theorem, on the other hand, is a demonstration that there does not exist such a map  {\em which is Leibnizian}---that is, one wherein the inferential equivalence relations are preserved, as in Definition~\ref{Leibnizianity}. Hence we see that Bell-like no-go theorems are based on a weaker assumption.  Nonetheless, in our view, the stronger assumption of Leibnizianity is just as plausible.  Furthermore, noncontextuality no-go theorems have greater breadth of applicability than 
their Bell-like counterparts since they can be proven for a broader set of causal-inferential structures---even those involving just a single causal system.

\subsection{The conventional way out of noncontextuality no-go theorems} \label{criticismsnc}

What is the conventional response to the lack of the existence of a generalized-noncontextual classical realist representation $\zeta$ of quantum theory, considered as the operational CI theory $\PQI$? (Or equivalently, to the lack of the existence of a classical realist representation $\widetilde{\xi}$ of quantum theory, considered as the quotiented operational CI theory $\widetilde{\PQI}$?)  For those who are unwilling to compromise on the standard notion of a realist representation, the typical response is to endorse a failure of generalized noncontextuality.

Endorsing such a failure requires a renouncement of Leibnizianity.  Thus, to anyone who is committed to 
 Leibniz's principle, this `way out' of the no-go results will be unappealing.

There is, however, the possibility of
  an alternative to the conventional response, one that aims to salvage Leibnizianity within a realist representation.
The idea is the one already noted at the end of Section~\ref{criticismsbell}, namely, to underlie operational quantum theory with a realist causal-inferential theory wherein the causal and inferential components are intrinsically nonclassical.  The next sections take up this research program.

\section{Beyond classical realism} \label{ncrealrepns}

In the conclusions of the last two sections, we criticized the conventional ways out of Bell-like and noncontextuality no-go theorems on the grounds that the price they must pay to salvage the standard notion of realism---violating the principle of no-superluminal causation, violating the principle of no fine-tuning, and abandoning the Leibnizian methodological principle---is too high.
 We also noted that this motivates a new type of research program, wherein one seeks to salvage these principles by considering novel notions of realist representations wherein the causal and inferential components thereof become intrinsically nonclassical.
The hope is that a realist causal-inferential theory of this type will have enough in common with its classical counterpart that a representation in terms of it can nonetheless be judged to provide satisfactory {\em explanations} of the operational phenomena.
(In previous work, this idea has been described as `achieving realism while going beyond the standard ontological models framework'~\cite{spekkens2016quasi}.) Up until now, the constraints that such a representation must satisfy have been articulated only vaguely, if at all.  The framework of causal-inferential theories allows us to say much more about the nature of such a representation and hence about how to further this research program.

Suppose that a nonclassical analogue of $\Func$ is denoted $\textsc{X}\Func$ and a nonclassical analogue of $\Inf$ is denoted $\textsc{X}\Inf$ and that the nonclassical realist causal-inferential theory defined by the interaction of these is denoted $\XFXI$. 
We have:
\beq
\begin{tikzpicture}
	\begin{pgfonlayer}{nodelayer}
		\node [style=none] (0) at (5.25, -0) {$\textsc{X}\Inf$};
		\node [style=none] (1) at (0.25, -0) {$\XFXI$};
		\node [style=none] (2) at (3.25, -0) {};
		\node [style=none] (3) at (1.25, -0) {};
		\node [style=none] (4) at (3.25, -0.25) {};
		\node [style=none] (5) at (1.25, -0.25) {};
		\node [style=none] (6) at (-0.75, -0) {};
		\node [style=none] (7) at (-2.75, -0) {};
		\node [style=none] (8) at (-3.75, -0) {$\textsc{X}\Func$};
		\node [style={up label}] (9) at (-1.75, -0) {$\mathbf{e'}$};
		\node [style={up label}] (10) at (2.25, -0) {$\mathbf{i'}$};
		\node [style={up label}] (11) at (2.25, -1) {$\mathbf{p^*}$};
	\end{pgfonlayer}
	\begin{pgfonlayer}{edgelayer}
		\draw [style={arrow plain}] (2.center) to (3.center);
		\draw [style={arrow dashed}] (5.center) to (4.center);
		\draw [style={arrow plain}] (7.center) to (6.center);
	\end{pgfonlayer}
\end{tikzpicture} \ ,
\eeq
where the maps $\mathbf{e'}$,  $\mathbf{i'}$ and $\mathbf{p^*}$ play the same roles that they do in $\FI$.

 The question becomes: what properties must the causal-inferential theory $\XFXI$ have in order to be considered realist, that is, such that representability in terms of such a theory can be considered to provide a {\em causal explanation} of the observed correlations? 
These properties will help to identify what alternatives there might be  to representing physical systems by sets, propositions about these by subsets, states of knowledge by distributions over a set and causal determination by functions from one set to another.


 Here, some readers might worry that all such alternatives should be ruled out by the very meanings of the terms in question: might it be that one {\em cannot even speak} about systems, logical propositions and states of knowledge about systems, and causal influences among systems, {\em unless} systems are described by sets, propositions by subsets, states of knowledge by distributions, and causal determination by functions?
 In short, some might argue that the definitions of notions of inference and causation are {\em analytic}, and therefore immune to revision. But this is not the case, a point we make by an analogy with nonEuclidean geometry.

The idea is that a putative nonclassical realist causal-inferential theory $\XFXI$ will stand to the classical realist causal-inferential theory $\FI$ as a nonEuclidean geometry stands to Euclidean geometry.\footnote{Hilary Putnam famously used this analogy to describe how a quantum logic ought to be conceptualized relative to classical logic~\cite{Putnam1969}.  We are simply extending the analogy to describe how probabilistic inference and causal influence ought to be conceptualized as well.} Just as the meanings of the terms `point' and `line'  in a nonEuclidean geometry are determined from the axioms of that geometry rather than corresponding to the common-sense notions, so too will the meaning of various causal and inferential concepts within a given nonclassical realist causal-inferential theory $\XFXI$ be determined by the specific axioms of that process theory (i.e., the diagrammatic rewrite rules)  rather than corresponding to the conventional ones.   In this sense, we are embracing the attitude towards mathematical structure that is characteristic of category theory  and that contrasts with the attitude of set theory wherein everything is built up from concepts concerning sets.   In particular, the fact that systems in a nonclassical realist causal-inferential theory are not associated with sets does not imply that such systems cannot be the locus of causal influences or the subject of propositions and states of knowledge.  Note that the attitude towards scientific realism that such a research program presumes is cognate with the philosophical position of structural realism~\cite{Ladyman2007}.

 Any attempt to provide a nonclassical generalization of the notions of causation and inference, however, is highly constrained insofar as it will need to preserve those features of these notions which 
one judges to be essential. 
  These constraints are the analogue of constraints on nonEuclidean geometries that arise from what one takes to be essential to the notions of `point' and `line'.   They are constraints on the process theory which, if violated, might well lead one to question whether the theory really is describing causal influences and inferences after all. Note that there is no way to be certain about the appropriateness of such constraints a priori, because without 
 concrete modification of the classical theory having been proposed and shown to be coherent and useful, it is difficult to know which of the features of the classical theory are essential.

\subsection{Constraints a causal-inferential theory must satisfy to be considered realist} \label{realconstr}
Given that we have distinguished the notions of operational and realist causal-inferential theories,
  it is clear that (in our view) a generic causal-inferential theory---in particular an {\em operational} CI theory---does not contain enough structure to be deemed worthy of the title `realist'.   While an  operational CI theory can predict observations, it does not itself provide a {\em realist explanation} of those predictions. Thus, we do not consider $\PI$ to be an instance of $\XFXI$. 
  In this section, we highlight the additional structure possessed by a classical realist CI theory $\FI$ over and above that possessed by an operational CI theory $\PI$. This structure helps to identify the properties one should demand of a causal-inferential theory $\textsc{X}\Func$ in order that it be deemed `realist'.

Because $\PI$ and $\FI$ are built out of the same inferential subtheory, $\Inf$, the contrast between them reduces to the contrast between the causal subtheories out of which they are built, $\Func$ and $\Proc$ respectively, and to differences in how these interact with $\Inf$.

In $\Func$, aspects of the causal structure are encoded not just in the shape of the circuit but also in the identities of the functions.   For example, if a process corresponds to a function that is independent of its argument, then there is no causal connection between the input and the output of that process.  In other words, the function associated to some process specifies the causal structure {\em internal} to the process.

In $\Proc$, on the other hand, the internal structure of a process is not specified. A given process does not necessarily even have a causal influence from its inputs to its outputs---it is only that there is {\em potential} for such a causal influence.
Some information about the internal causal structure may be inferred from the image of this process under the prediction map, but this does not generally provide a full specification of the internal causal structure of the processes (e.g., we saw in Section~\ref{opquotinf} that inferentially equivalent processes can correspond to different causal structures).

To summarize, $\Proc$ encodes {\em potential} causal influences, while $\Func$ encodes {\em actual} causal influences.  What we have termed a {\em causal theory} is meant as an umbrella for these two notions.\footnote{It is worth noting that the assumption of diagram-preservation for the classical realist representation $\xi$ ensures the preservation of the structure of {\em potential causal influences}, not that of actual causal influences. That is, if there is no potential causal influence between a pair of systems in some given diagram of $\Proc$, then the image of this diagram under $\xi$ must be such that there is no actual causal influence between these in $\Func$.  If, however, there {\em is} a potential causal influence between a pair of systems in some given diagram of $\Proc$, then there may or may not be an actual causal influence between these in  $\Func$.}

Formally, the issue is that the interpretation of processes in a process theory is derived primarily from their interactions with other processes---through the nontrivial equalities that involve them.
But $\Proc$ is a free process theory, with no nontrivial equalities; hence, processes in $\Proc$ have an interpretation that is impoverished relative to those in $\Func$.

Thus a first criterion for a CI theory to be deemed realist is the following:
\begin{itemize}
\item[1.]
The causal subtheory of the CI theory must have enough nontrivial equalities such that its processes represent actual causal influences rather than potential causal influences.
\end{itemize}
\noindent The exact formalization of this remains to be determined, but we give constraints on how to do so later in this section. A minimal requirement is that $\textsc{X}\Func$ is not a free process theory.

We now turn to a comparison of the interaction between $\Proc$ and $\Inf$ and the interaction between $\Func$ and $\Inf$.

$\Func$ and $\Inf$  exhibit strong forms of interaction. For example, one can define propositions about (or equivalently, directly gain information about) {\em any} causal system in $\FI$, since the generator in Eq.~\eqref{propont} is defined for all systems in a classical realist CI theory.
In contrast, the interaction between $\Proc$ and $\Inf$ is very limited, insofar as one cannot define propositions about (or directly learn information about) any systems which are nonclassical.
This leads to our next criterion for a CI theory to be deemed realist:
\begin{itemize}
\item[2.] \label{attachpropn} Propositions
 must be able to attach to all systems in the CI theory.
\end{itemize}

Formally, the fact that $\FI$ satisfies this criterion has important consequences.
 Chiefly, it is a prerequisite for one to introduce the equality  Eq.~\eqref{Axiom:PropKnowGenerators},
   which allows the translation of a proposition about the output of a causal mechanism into a proposition about its input and the identity of the mechanism.

By contrast, there is no equality analogous to Eq.~\eqref{Axiom:PropKnowGenerators} in operational CI theories. Specifying states of knowledge about the causal processes in an operational CI theory does not yield statistical predictions until one specifies a prediction map. Indeed, nearly all of the non-generic features of an operational CI theory are buried within the choice of prediction map.\footnote{Note that if an operational CI theory $\PI$ {\em does} admit of a classical realist representation in terms of $\FI$, then this representation serves to provide an explanation for the prediction map of $\PI$ in terms of the unique prediction map of $\FI$ and the
representation map $\xi$.}

We therefore elevate this feature into a criterion of its own:
\begin{itemize}
\item[3.]  It should be possible to propagate knowledge claims through any causal mechanism.  Formally, there must exist an analogue of Eq.~\eqref{Axiom:PropKnowGenerators}.
\end{itemize}

This last criterion is central to the idea of a realist causal-inferential theory.
A commitment to realism means that the systems can mediate causal influences,
 and that one can understand every valid inference as a consequence of knowledge propagation {\em through} these causal mediaries.  In particular, if it is the existence of a causal pathway between two variables that accounts for the inferences that can be made between these, then it must be possible to understand these inferences as decomposable into a sequence of inferences, stepping through systems along the causal pathway.
 For example, in a Bell scenario (as in Diagram~\eqref{thecorrel2}), updating one's knowledge of the outcome at the left wing, $A$ (which depends on background knowledge about $X$), leads to an updating of one's knowledge of the outcome on the right wing, $B$ (which depends on background knowledge about $Y$),  via the mediary of systems $S$ and $S'$.  Specifically, updating one's knowledge of $A$
 leads to an updating of one's knowledge of $S$, which
 in turn leads to an updating of one's knowledge of $S'$,
  which in turn leads to an updating of one's knowledge of $B$.
    The ability of systems to encode information and to be mediaries in a sequence of refinements of knowledge is key, we argue, for a given theory to be described as realist.

The equality in Eq.~\eqref{Axiom:PropKnowGenerators} leads to a great deal of the structure of $\FI$ and ultimately to the uniqueness of the prediction map, as in Theorem~\ref{uniquerule}.
 It seems an essential part of any {\em fundamental} theory of nature that the predictions one makes should be uniquely determined by a complete causal-inferential description of one's scenario within that fundamental theory.

Hence, we have another criterion for a CI theory to be deemed realist:
\begin{itemize}
\item[4.] The CI theory must have a unique prediction map.
\end{itemize}

Our final criterion for a CI theory to be deemed realist also relies on the fact that one can attach propositions to any system, as per Criterion~2. 
Recall that two processes are inferentially equivalent if no matter what causally-closed circuit they are embedded in, they lead one to make the same predictions concerning the classical variables that are the (inferential) inputs and outputs of that circuit.
 This is an `external' characterization of inferential equivalence. Within a classical realist theory, Lemma~\ref{prfinfstochastic} showed that there is a second equivalent characterization of inferential equivalence that is `intrinsically' defined. This intrinsic characterization hinges on a particular mapping, Eq.~\eqref{assocstoc}, from arbitrary processes in $\FI$ to processes in $\Inf$. This allowed us to characterize the inferential equivalence of processes in $\FI$ via the stochastic maps that naturally describe the intrinsic relationship between the causal and inferential inputs and outputs of the process, without reference to external scenarios involving the process.  

If Criterion~2 is satisfied within a nonclassical realist theory, then one can also define a mapping analogous to that in Eq.~\eqref{assocstoc}, where the mapping takes every process $\mathcal{D}$ in $\XFXI$ to a related process in $X\Inf$, such that the latter provides an intrinsic description of inferences from the causal and inferential inputs of $\mathcal{D}$ to the causal and inferential outputs of $\mathcal{D}$. These are analogous to the stochastic maps just described in the case of classical realist theories.
 The criterion, then, is that one can define this intrinsic characterization of inferential equivalence, and that it furthermore be equivalent to the external characterization: 
\begin{itemize}
\item[5.] Two arbitrary processes in a CI theory are inferentially equivalent if and only if the external and intrinsic characterizations of inferential equivalence coincide. Formally, there must exist an analogue of Eq.~\eqref{assocstoc} and Lemma~\ref{prfinfstochastic}. 
\end{itemize}

The challenge moving forward, then, is to find mathematical structures $\textsc{X}\Func$ and $\textsc{X}\Inf$ that respect all of the desiderata required for a CI theory to be deemed realist.

{\bf Concepts that should have analogues in any realist theory---}

One expects any putative nonclassical theory of causation to contain analogues of most, if not all, of the standard notions that arise in the framework of classical causal models: common causes, causal mediaries, d-separation, evaluation of counterfactuals, etcetera.
Preliminary work towards establishing how the evaluation of counterfactuals is formally achieved within $\FI$ is provided in Section~\ref{causmode}.

One also expects that any putative nonclassical theory of inference should contain analogues of most, if not all, of the standard notions that arise in Boolean logic and Bayesian probability theory: logical connectives, implication, conditional independence, sufficient statistics, etcetera.~
\footnote{Note that we have not yet incorporated all of these notions in $\Inf$ at a diagrammatic level, although it is clear how all of these can be determined nondiagrammatically. Incorporating these into the diagrammatic formalism (or demonstrating that they are dispensable without compromising the usefulness of the theory of inference) is therefore a topic for further research. Incorporating them explicitly would enable the study of how these inferential features interact with causal processes.}

\subsection{Nonclassical realist representations} \label{nonclassrealrepns}

 Having described what it means for a causal inferential theory to embody a satisfactory notion of realism, we can now describe the notion of a {\em nonclassical realist representation}: namely, a representation of an operational CI theory in terms of a nonclassical realist CI theory.
The definition is the analogue of Definition~\ref{ontrepndefn}, but where the image of the map is a nonclassical realist CI theory rather than a classical one. This definition (given below) is the sense in which we have now formalized the idea of `achieving realism while going beyond the standard ontological models framework'.

For an arbitrary operational CI theory $\PI$, we can seek to find a nonclassical realist CI theory $\XFXI$ in terms of which $\PI$ can be represented. That is, we can ask if there is a realist representation map \colorbox{gray!30}{$\mathbf{\xi}:\PI \to \XFXI$} which can be defined in an analogous way to Def.~\ref{ontrepndefn}.

Note, however, that whereas the inferential subtheories of $\FI$ and of $\PI$ were identical (namely, $\Inf$), the inferential subtheory of $\XFXI$ is allowed to be something more general than that of $\PI$, namely, what we have denoted $\textsc{X}\Inf$. Consequently, one  needs to modify the condition of preservation of empirical predictions (the commutation of Eq.~\eqref{eq:EmpiricalAdequacy}). Rather than the two inferential subtheories being equal, they will be related by some sort of map \colorbox{purple!30}{$\phi:\Inf \to \textsc{X}\Inf$} whose defining properties is a subject for future work \footnote{In the cases of primary interest to us, we expect that $\phi$ will be an inclusion map. Hence, it will still be possible and meaningful 
to make classical inferences within $\textsc{X}\Inf$. 
 }. Given a satisfactory definition of $\phi$, one can define:

\begin{definition} \label{defnnonclassrealrepn}
A {\em nonclassical realist representation} of an unquotiented operational CI theory $\PI$ by an unquotiented nonclassical realist CI theory $\XFXI$ is a diagram-preserving map \colorbox{gray!30}{$\mathbf{\xi}:\PI \to \XFXI$}
satisfying {\bf (i)} the \emph{preservation of predictions},  namely that the diagram
\beq
\begin{tikzpicture}
	\begin{pgfonlayer}{nodelayer}
		\node [style=none] (0) at (-2, 2) {$\PI$};
		\node [style=none] (1) at (3, 2) {$\Inf$};
		\node [style=none] (2) at (3, -2) {$\textsc{X}\Inf$};
		\node [style=none] (3) at (-2, -2) {$\XFXI$};
		\node [style=none] (4) at (-2, 1.5) {};
		\node [style=none] (5) at (-2, -1.5) {};
		\node [style=none] (6) at (3, -1.5) {};
		\node [style=none] (7) at (3, 1.5) {};
		\node [style=none] (12) at (-1.75, -1.5) {};
		\node [style=none] (13) at (-2, -1.5) {};
		\node [style=none] (14) at (2.75, -1.5) {};
		\node [style=none] (15) at (1, 2) {};
		\node [style=none] (16) at (-1, 2) {};
		\node [style=none] (17) at (1, -2) {};
		\node [style=none] (18) at (-1, -2) {};
		\node [style=right label] (19) at (-2, 0) {$\mathbf{\xi}$};
		\node [style=up label] (21) at (0, 2) {$\mathbf{p}$};
		\node [style=up label] (23) at (0, -2) {$\mathbf{p'}$};
		\node [style=right label] (24) at (3, 0) {$\phi$};
	\end{pgfonlayer}
	\begin{pgfonlayer}{edgelayer}
		\draw [style=arrow plain] (4.center) to (5.center);
		\draw [style=arrow plain] (7.center) to (6.center);
		\draw [style=arrow dashed] (16.center) to (15.center);
		\draw [style=arrow dashed] (18.center) to (17.center);
	\end{pgfonlayer}
\end{tikzpicture}
 \ .
\eeq
 commutes, and {\bf(ii)} the {\em preservation of  ignorability}
\beq\label{eq:CausalOntRep6}
\InputIfFileExists{Diagrams/OntRep2p.tikz}{}{\input{./figures/Diagrams/OntRep2p.tikz}}\quad
=\quad
\InputIfFileExists{Diagrams/OntRep3p.tikz}{}{\input{./figures/Diagrams/OntRep3p.tikz}}.
\eeq
We will sometimes refer to the nonclassical realist representation map $\mathbf{\xi}$ as an $\XFXI$-representation of $\PI$.
\end{definition}

Analogously, one can also extend the notion of a representation of a {\em quotiented} operational CI theory $\widetilde{\PI}$ in terms of a quotiented classical realist CI theory $\widetilde{\FI}$ to that of a representation in terms of a quotiented {\em nonclassical} CI theory $\widetilde{\XFXI}$. 
Again, the only non-trivial aspect of this generalization is in defining the map \colorbox{purple!30}{$\phi:\Inf \to \textsc{X}\Inf$}.

One can summarize the notions of nonclassical realist representations via the analogue of Eq.~\eqref{fullfig}: 
\beq \label{CompleteDiagramArbitraryThrys}
\InputIfFileExists{Diagrams/CompleteDiagramNew.tikz}{}{\input{./figures/Diagrams/CompleteDiagramNew.tikz}} \quad.
\eeq

\subsection{A new way out of Bell-like no-go theorems} \label{outofbell}

Recall that a Bell-like no-go theorem arises whenever one finds a causal structure in which the set of $\PI$-realizable probability distributions is not contained in the set of $\FI$-realizable probability distributions---that is, not contained among those that can be generated by a {\em classical} realist representation.

The possibility of {\em nonclassical} realist representations provides a novel way out of such no-go theorems.
Rather than asking if the observed experimental statistics are $\FI$-realizable, one can instead ask if they are $\XFXI$-realizable, that is, representable using a map $\mathbf{\xi}:\PI \to \XFXI$ into some {\em nonclassical} realist CI theory $\XFXI$.

If this can be done, it seems appropriate to claim that such a realist representation has salvaged locality.  More precisely, such a representation has provided a means of being conservative with respect to causal structure---in particular, not requiring superluminal influences---by being radical with respect to the nature of the realist CI theory.

We have not yet provided an explicit proposal for a realist CI theory $\XFXI$ which can reproduce the quantum predictions while providing a satisfactory realist explanation of them. However, our work in formalizing the notion of a nonclassical realist theory, e.g., via the formal criteria given in Section~\ref{realconstr}, constitutes a first concrete step in this direction.

\subsection{A new way out of noncontextuality no-go theorems} \label{outofnc}

A nonclassical realist CI theory $\XFXI$ necessarily includes a notion of inferential equivalence---because $\XFXI$ is assumed to provide a unique prediction map, one simply evaluates equivalences relative to it.
It follows that one can define Leibnizianity for nonclassical realist representations much as it was defined for classical realist representations (in Section~\ref{ncsec}):
\begin{definition}[Leibnizianity of a nonclassical realist representation]\label{Leibnizianity2}
A nonclassical realist representation map \colorbox{gray!30}{$\mathbf{\xi}:\PI \to \XFXI$} is said to be {\em Leibnizian} if it preserves inferential equivalence relations.
\end{definition}
Consequently, it makes just as much sense to ask whether a given operational CI theory admits of a {\em nonclassical} realist representation that is Leibnizian as it did to ask that question of a classical realist representation. This is a key benefit of our new process-theoretic definition of Leibnizianity.

As in the case of classical realist representations, we can give an equivalent characterization in terms of a commuting square. A nonclassical realist representation map $\xi: \PI \to \XFXI$ is Leibnizian if and only if there exists a map $\widetilde{\xi}: \widetilde{\PI} \to \widetilde{\XFXI}$ such that the following diagram commutes:
\beq\label{ncdiagcomm2}
\begin{tikzpicture}
	\begin{pgfonlayer}{nodelayer}
		\node [style=none] (0) at (-1, 1.25) {$\PI$};
		\node [style=none] (2) at (1, 3) {$\widetilde{\PI}$};
		\node [style=none] (4) at (-1, -2) {$\XFXI$};
		\node [style=none] (5) at (1.25, 0) {$\widetilde{\XFXI}$};
		\node [style=none] (6) at (-1, 1.25) {};
		\node [style=none] (7) at (-1, -1.5) {};
		\node [style=none] (8) at (1.25, 0.5) {};
		\node [style=none] (9) at (1.25, 2.75) {};
		\node [style=none] (18) at (-0.75, 1.75) {};
		\node [style=none] (19) at (0.75, 2.75) {};
		\node [style=none] (20) at (0.75, -0.5) {};
		\node [style=none] (21) at (-0.75, -1.5) {};
		\node [style=none] (22) at (0.5, 2.75) {};
		\node [style=none] (23) at (0.5, -0.5) {};
		\node [style=none] (24) at (-1, -1.5) {};
		\node [style=none] (26) at (1, 2.75) {};
		\node [style=none] (27) at (1, -0.5) {};
		\node [style=right label] (33) at (-1, 0) {$\xi$};
		\node [style=right label] (34) at (1.25, 1.5) {$\widetilde{\xi}$};
		\node [style=up label] (35) at (0, 2.25) {$\sim_{\mathbf{p}}$};
		\node [style=up label] (36) at (0, -1) {$\sim_{\mathbf{p}^{*}}$};
	\end{pgfonlayer}
	\begin{pgfonlayer}{edgelayer}
		\draw [style=arrow plain] (6.center) to (7.center);
		\draw [style=arrow plain] (9.center) to (8.center);
		\draw [style=arrow plain] (18.center) to (19.center);
		\draw [style=arrow plain] (21.center) to (20.center);
	\end{pgfonlayer}
\end{tikzpicture} \ \ .
\eeq
 where $\mathbf{p}^*$ is the unique prediction map in $\XFXI$.

One can extend the notion of generalized-noncontextuality to nonclassical realist representations in a similar fashion.  Specifically, the map $\zeta: \PI \to \widetilde{\XFXI}$ is defined to be generalized-noncontextual if the triangle in upper right of the following diagram commutes:
\beq
\begin{tikzpicture}
	\begin{pgfonlayer}{nodelayer}
		\node [style=none] (0) at (-1, 1.25) {$\PI$};
		\node [style=none] (2) at (1, 3) {$\widetilde{\PI}$};
		\node [style=none] (4) at (-1, -2) {$\XFXI$};
		\node [style=none] (5) at (1.25, 0) {$\widetilde{\XFXI}$};
		\node [style=none] (6) at (-1, 1.25) {};
		\node [style=none] (7) at (-1, -1.5) {};
		\node [style=none] (8) at (1.25, 0.5) {};
		\node [style=none] (9) at (1.25, 2.75) {};
		\node [style=none] (18) at (-0.75, 1.75) {};
		\node [style=none] (19) at (0.75, 2.75) {};
		\node [style=none] (20) at (0.75, -0.5) {};
		\node [style=none] (21) at (-0.75, -1.5) {};
		\node [style=none] (22) at (0.5, 2.75) {};
		\node [style=none] (23) at (0.5, -0.5) {};
		\node [style=none] (24) at (-1, -1.5) {};
		\node [style=none] (26) at (1, 2.75) {};
		\node [style=none] (27) at (1, -0.5) {};
		\node [style=right label] (33) at (-1, 0) {$\xi$};
		\node [style=right label] (34) at (1.25, 1.5) {$\widetilde{\xi}$};
		\node [style=up label] (35) at (0, 2.25) {$\sim_{\mathbf{p}}$};
		\node [style=up label] (36) at (0, -1) {$\sim_{\mathbf{p}^{*}}$};
		\node [style=none] (37) at (0.75, 0.25) {};
		\node [style=right label] (38) at (0.25, 0.75) {$\zeta$};
		\node [style=none] (39) at (-0.75, 1) {};
	\end{pgfonlayer}
	\begin{pgfonlayer}{edgelayer}
		\draw [style=arrow plain] (9.center) to (8.center);
		\draw [style=arrow plain] (39.center) to (37.center);
		\draw [style=arrow plain] (18.center) to (19.center);
		\draw [style=arrow plain] (21.center) to (20.center);
		\draw [style=arrow plain] (6.center) to (7.center);
	\end{pgfonlayer}
\end{tikzpicture}
 \ \  .
\eeq
The fact that the definitions are all process-theoretic implies that we have abstract notions of generalized noncontextuality and Leibnizianity that apply to nonclassical realist representations and that satisfy analogues of Eq.~\eqref{variouslemmas}, that is, 
\beq \label{variouslemmas2}
\exists\ \text{Leibnizian }\xi
\
 \begin{tikzpicture}
	\begin{pgfonlayer}{nodelayer}
		\node [style=none] (0) at (0, 0.35) {$\stackrel{?}{\impliedby}$};
		\node [style=none] (1) at (0, -0.25) {$\implies$};
	\end{pgfonlayer}
\end{tikzpicture}
\
 \exists\
\text{gen.-noncontextual }
 \zeta
  \iff
   \exists\ \widetilde{\xi}.
\eeq

The possibility of nonclassical realist representations therefore holds the potential for a novel way out of noncontextuality no-go theorems, a way out that does not compromise on the Leibnizian methodological principle.

\section{Discussion}

\subsection{Finding a satisfactory ontology and epistemology {\em for quantum theory}}\label{sec:QuantumOmelette}

The long-term aim of this work is to generate a compelling interpretation of quantum theory---one which satisfies the spirit of locality and Leibnizianity. We turn, therefore, to the special case of nonclassical realist representations {\em of quantum theory}. 

Recall from Section~\ref{QTquaopCI} that quantum mechanics can be cast as an operational CI theory $\PQI$ having the structure
\beq
\begin{tikzpicture}
	\begin{pgfonlayer}{nodelayer}
		\node [style=none] (0) at (0, 0) {$\PQI$};
		\node [style=none] (1) at (5, 0) {$\Inf$};
		\node [style=none] (6) at (0, -0.5) {};
		\node [style=none] (9) at (4, -0.5) {};
		\node [style=none] (10) at (3, 0) {};
		\node [style=none] (11) at (1, 0) {};
		\node [style=none] (23) at (3, -0.25) {};
		\node [style=none] (24) at (1, -0.25) {};
		\node [style=none] (33) at (4.25, 0.25) {};
		\node [style=none] (37) at (-1, 0) {};
		\node [style=none] (38) at (-3, 0) {};
		\node [style=none] (39) at (-4, 0) {$\Proc_Q$};
		\node [style=up label] (45) at (-2, 0) {$\mathbf{e}$};
		\node [style=up label] (46) at (2, 0) {$\mathbf{i}$};
		\node [style=up label] (47) at (2, -1) {$\mathbf{p}_Q$};
		\node [style=none] (59) at (-4, -0.5) {};
	\end{pgfonlayer}
	\begin{pgfonlayer}{edgelayer}
		\draw [style=arrow plain] (10.center) to (11.center);
		\draw [style=arrow dashed] (24.center) to (23.center);
		\draw [style=arrow plain] (38.center) to (37.center);
	\end{pgfonlayer}
\end{tikzpicture} \ .
\eeq
Our ultimate objective is to identify a {\em quantum} realist CI theory\footnote{Naturally, the `$Q$' in the notation where we previously put `X' refers to the fact that we are aiming specifically for a {\em quantum} generalization of $\Func$, $\Inf$, and $\FI$.} 
\beq
\begin{tikzpicture}
	\begin{pgfonlayer}{nodelayer}
		\node [style=none] (0) at (5.25, -0) {$\textsc{Q}\Inf$};
		\node [style=none] (1) at (0.25, -0) {$\QFQI$};
		\node [style=none] (2) at (3.25, -0) {};
		\node [style=none] (3) at (1.25, -0) {};
		\node [style=none] (4) at (3.25, -0.25) {};
		\node [style=none] (5) at (1.25, -0.25) {};
		\node [style=none] (6) at (-0.75, -0) {};
		\node [style=none] (7) at (-2.75, -0) {};
		\node [style=none] (8) at (-3.75, -0) {$\textsc{Q}\Func$};
		\node [style={up label}] (9) at (-1.75, -0) {$\mathbf{e'}$};
		\node [style={up label}] (10) at (2.25, -0) {$\mathbf{i'}$};
		\node [style={up label}] (11) at (2.25, -1) {$\mathbf{p^*}$};
	\end{pgfonlayer}
	\begin{pgfonlayer}{edgelayer}
		\draw [style={arrow plain}] (2.center) to (3.center);
		\draw [style={arrow dashed}] (5.center) to (4.center);
		\draw [style={arrow plain}] (7.center) to (6.center);
	\end{pgfonlayer}
\end{tikzpicture} \ ,
\eeq
which satisfies the constraints articulated in Section~\ref{realconstr} (so that it can meaningfully be said to be a realist theory).

It should also subsume the classical realist CI theory. As noted above, in a nonclassical realist CI theory, the systems need not be classical variables (i.e., they need not be associated with sets) and the states of knowledge of these systems need not be associated with probability distributions over these sets because such possibilities are what open up space for evading Bell-like and noncontextuality no-go theorems.   Nonetheless,  $\textsc{Q}\Func$ will need to {\em include} classical variables as an allowed type of system and functions between these as an allowed type of causal process, e.g., in order to describe the setting and outcome variables associated with experimental procedures.   Similarly, $\textsc{Q}\Inf$   needs to include $\Inf$ as a subtheory, so that it can make the correct statistical predictions for diagrams consisting entirely of such classical variables.  (We leave to future work the problem of articulating formally the constraint of subsuming the classical realist CI theory.)

Most importantly, the quantum realist CI theory $\QFQI$ should provide a Leibnizian representation of $\PQI$. The relation between $\QFQI$ and $\PQI$ that is implied by the existence of such a representation map can be summarized by the following diagram:
\beq \label{CompleteDiagramArbitraryThrys}
\InputIfFileExists{Diagrams/CompleteDiagramArbitraryThrysNew.tikz}{}{\input{./figures/Diagrams/CompleteDiagramArbitraryThrysNew.tikz}} \quad.
\eeq

If such a quantum realist representation map $\xi: \PQI \to \QFQI$ is found, then it follows that for any given causal-inferential structure, the set of $\QFQI$-realizable distributions includes the $\PQI$-realizable distributions.
Hence, one obtains a way out of Bell-like no-go theorems that is more satisfactory than the conventional ways out insofar as it need not involve any superluminal influences (thereby salvaging the spirit of locality), and insofar as it need not avail itself of any fine-tuning of parameters.
Since we have further required that the realist representation map $\xi: \PQI \to \QFQI$ must be
{\em Leibnizian} (as implied by the existence of the map $\tilde{\xi}$), one also obtains a way out of the noncontextuality no-go theorems that is more satisfactory than the conventional way out, insofar as it salvages the Leibnizian methodological principle (and thereby the spirit of generalized noncontextuality).

Although there has been some work on interpreting some of the formalism of quantum theory as a nonclassical generalization of Bayesian inference, the kind of classical theory that served as the target of this generalization was  an ontological model.
 As noted in Section~\ref{sec:subsumingOntological}, however, the notion of an ontological model corresponds to a {\em quotiented} classical realist CI theory.  And, as argued in Section~\ref{makingomelette}, such a theory involves a partial scrambling of causal and inferential concepts. This is problematic because it is likely that the constraints on putative quantum generalizations of classical theories are {\em only} clear if causal and inferential notions are cleanly separated in the latter, and hence only if these quantum generalizations are pursued at the level of the {\em unquotiented} theory.

 As an example, consider how the project of finding a nonclassical generalization of Bayesian inference was pursued in Ref.~\cite{LeiferSpekkens}, which built upon ideas proposed in Refs.~\cite{leifer2008quantum,leifer2007conditional}.  The focus was on
finding intrinsically quantum counterparts to the notions of joint, marginal and  conditional probability distributions, as well as counterparts to the relations that hold between these, such as the counterpart of marginalization, the law of total probability, and the formula for Bayesian inversion.   However, a conditional probability distribution, or equivalently, a stochastic map, represents an {\em inferential equivalence class} of states of knowledge about functional dynamics, and often involves a partial scrambling of causal and inferential concepts (as illustrated in Section~\ref{makingomelette}).
The fact that the focus of this earlier work was on a mathematical object that scrambled causal and inferential concepts may explain why there are outstanding problems with the approach, such as those described in Ref.~\cite{LeiferSpekkens} and in Ref.~\cite{horsman2017can}.
A state of knowledge about functional dynamics---unlike the inferential equivalence class of such objects---involves no such scrambling. It is consequently {\em this} object that is more appropriate to focus on and for which to seek an intrinsically quantum counterpart.

It is possible that the proposals for quantum generalizations of propositional logic which were pursued under the banner of `quantum logic'~\cite{hooker2012logico,hooker1979logico} also suffer from having mistaken inferential equivalence for identity.  Certainly, we believe that conventional approaches, such as the one that takes the counterpart of a Boolean lattice to be an orthomodular lattice, are unlikely to yield success in the research program described here.  This is because such approaches are informed solely by the structure of projectors on Hilbert space and this may well merely be describing aspects of the {\em quotiented} quantum realist CI theory, while it is only the {\em unquotiented} theory $\QFQI$ that one can hope to decompose into a causal subtheory $\textsc{Q}\Func$ and an inferential subtheory $\textsc{Q}\Inf$ (where the structure concerning propositional logic lives).

We now highlight some prior work that is likely to be useful in developing an intrinsically quantum notion of a realist CI theory.

On the causal side, recent work on developing an intrinsically quantum notion of a causal model~\cite{allen2017quantum,costa2016quantum,Barrett2019} is likely to provide a good starting point for finding the
correct quantum generalization of $\Func$.  In particular, the notion of decomposing a unitary gate into a more refined circuit that includes `dots' (isomorphisms wherein a Hilbert space is decomposed into a direct sum of tensor products), introduced in Ref.~\cite{allen2017quantum} and studied in depth in Ref.~\cite{Lorenz2020}, is likely to be incorporated in some way into $\textsc{Q}\Func$.

In pursuing the correct quantum generalization of $\Inf$, recent work developing a synthetic approach to probability theory (formalized as `Markov categories')~\cite{cho2017disintegration,jacobs2019causal,fritz2020synthetic,furber2013towards}
 is likely to be useful.  This is because if the $\textsc{Bayes}$ subtheory of $\Inf$ can be characterized more abstractly, the possibilities for quantum generalization should become more evident.  In particular, the work of Ref.~\cite{coecke2012picturing}, which is in the same spirit as Refs.~\cite{cho2017disintegration,fritz2020synthetic}, may provide an important piece of the puzzle (in spite of not having the benefit of a proper unscrambling of causal and inferential notions).  Specifically, the {\em logical broadcasting} map described therein may be the counterpart in $\textsc{Q}\Inf$ of the copy operation in $\Inf$.

Similar comments may well apply to prior work in the field of quantum logic, namely, that in spite of suffering from some causal-inferential scrambling, specific insights from that research program could prove useful in finding the counterpart within $\textsc{Q}\Inf$ to various notions within the subtheory $\textsc{Boole}$ of $\Inf$.

To close, we note that there has been a great deal of interest in whether certain mathematical objects in the quantum formalism---most notably quantum states---have an ontological or an epistemological status~\cite{caves2002quantum,fuchs2013quantum,fuchs2010qbism,Emersonthesis,spekkens2007evidence,Harrigan,pusey2012reality,leifer2014quantum,leifer2006quantum,leifer2007conditional}.
Although disentangling ontology and epistemology is certainly critical to the project of unscrambling Jaynes' omelette, it is worth noting that in some cases this question presumes a false dichotomy.
 To see this, note that even in a {\em classical} realist CI theory, certain mathematical objects play multiple roles---for example, functions appearing in $\Func$ describe the causal influence that one variable has on another, while the same functions in $\Inf$ (now represented as deterministic stochastic maps)
  describe how learning about one variable leads to updating one's knowledge of another.   It seems likely, therefore, that certain mathematical objects in a quantum realist CI theory  will also have counterparts in both the causal and inferential subtheories.  Indeed,
 a single-system unitary is likely to be such an object, sometimes describing the nature of  a causal influence in the causal subtheory and sometimes the nature of how one updates one's knowledge in the inferential subtheory.
The question about whether a given mathematical object in the quantum formalism  has an ontological or epistemological status, therefore, must sometimes be refined to take into account the context in which the mathematical object appears.

\subsection{Subsuming the framework of classical causal modeling} \label{causmode}

We have considered two distinct classes of causal theories, namely, $\Proc$ and \Func. The primary technical distinction between these two is that $\Func$ has equalities, while $\Proc$ does not. We have seen various consequences of this extra structure on $\Func$, e.g., the uniqueness of the prediction map 
 in $\FI$. Conceptually, the primitive type of process in $\Proc$ (a list of lab instructions) constitutes an extremely minimal description of the causal mechanism relating its inputs to its outputs. Meanwhile, the primitive type of process in $\Func$ (a functional dependence) constitutes a much more informative description.

In fact, causal dependences in a classical theory can be {\em defined} in terms of functional dependence of one variable on another. This is done, for instance, in {\em structural equation  models}~\cite{pearl2009causality}, and it is the notion of classical causation that we endorse here. Hence, one might expect that structural equation models could be subsumed in our framework within $\FI$, which allows for both a description of the functional (hence causal)
dependences among variables, as well as a specification of one's knowledge about exogenous variables.
Similarly, the notion of a probabilistic causal model (or `causal Bayesian network')~\cite{pearl2009causality}, wherein the functional dependences and the states of knowledge of the exogenous variables are not specified individually, but are folded together into a conditional probability distribution, is likely to be subsumed in our framework within the quotiented theory $\widetilde{\FI}$.
 In future work, we hope to explore the relationship between our framework and various notions of classical causal models, and to argue that in some regards, our framework is more general than the standard one.

 To accommodate all of the purposes to which classical causal models are put (in particular, considering the consequences of interventions and evaluating counterfactuals), it will
 be useful to introduce
a distinct type of causal theory of functional dynamics, embedded within $\Func$, which we will term $\textsc{PreFunc}$.
The systems and processes in $\textsc{PreFunc}$ are the same as in $\Func$, but the process theory is defined without equalities. In particular, the composition of two functions $f(\cdot)$ and $g(\cdot)$ in sequence in $\textsc{PreFunc}$ is {\em not} strictly equal to the function $f(g(\cdot))$. One can then define a DP map from $\textsc{PreFunc}$ to $\Func$ which induces an equivalence relation on $\textsc{PreFunc}$, namely, two diagrams in $\textsc{PreFunc}$ are equivalent if they define the same function when the component functions in the diagram are composed.

\begin{remark}
The transition from $\Func$ to $\textsc{PreFunc}$ can be viewed as an example of a very general construction on process theories. First, one defines a forgetful functor from the category $\textsc{ProcessTheory}$ to a new category (which we will call $\textsc{ProcessSet}$) where a particular process theory is mapped to its underlying set of processes thereby forgetting the compositional structure of the process theory.
 We can then define a free functor which is left adjoint to the forgetful functor. The composition of these two functors then defines a comonad on $\textsc{ProcTheory}$ which, in particular, takes $\Func$ to $\textsc{PreFunc}$. This is closely related to \cite[Ex. 4.2.2]{perrone2019notes}.
\end{remark}

To better understand the differences between $\textsc{PreFunc}$ and $\Func$, consider the following  
 pair of 
 diagrams, where the gate \resizebox{!}{3mm}{$%
\InputIfFileExists{Diagrams/cnot.tikz}{}{\input{./figures/Diagrams/cnot.tikz}}$} represents a classical controlled NOT operation:
\beq
a)\ %
\InputIfFileExists{Diagrams/cnotTriple1.tikz}{}{\input{./figures/Diagrams/cnotTriple1.tikz}}\qquad b)\ %
\InputIfFileExists{Diagrams/cnotTriple2.tikz}{}{\input{./figures/Diagrams/cnotTriple2.tikz}}.
\eeq

In $\Func$, the process described by diagram  (a) and that described by diagram  (b) are strictly equal. The two diagrams are merely distinct manners of specifying the overall input-output functionality of the effective function from $A_1$ and $A_2$ to $D_1$ and $D_2$.
In $\textsc{PreFunc}$, however, diagrams represent `histories' of processes, rather than merely representing input-output functionalities. These two diagrams viewed within $\textsc{PreFunc}$ are therefore not equal to one another, but rather are merely equivalent in the sense defined just above.

 Despite the fact that (a) and (b) are equal within $\Func$, it is clear that the interventions possible on each of them are distinct. 
To formally describe all possible interventions in a given scenario, it is essential that one works within $\textsc{PreFunc}$, wherein (a) and (b) are merely equivalent; e.g., this allows one to consider an intervention on $B_1$, $B_2$, $C_1$ or $C_2$.

In order to provide a fully formal diagrammatic treatment of the interventional aspects of the framework of classical causal models~\cite{pearl2009causality}, it will likely be useful to take the causal theory to be $\textsc{PreFunc}$ rather than $\Func$. We will address this project more explicitly in future work.

Insofar as our work reveals that stochastic matrices (equivalently, conditional probability distributions) relating a cause to its effect generically scramble together causal and inferential concepts, this is true even for the notion of a {\em do-conditional}, which is defined as the conditional probability distribution of an effect variable given a cause variable when the value of the cause is intervened upon, rather than being determined by its natural causal parents.
  From the perspective of our work, the only object which does not lose any information about what is known about the causal influence of one variable on another  is the probability distribution over the function that relates the one variable to the other,
   while the do-conditional merely describes an inferential equivalence class of such objects.
This and related ideas will also be explored in future work.

\subsection{More future directions}

There are many natural extensions of our work, and many ways in which it is likely to shed light on other research programs.
We now discuss some of these research directions, beyond those highlighted in the discussion sections or introduced throughout the paper.

We first note a straightforward supplementation to the notion of an operational CI theory.
Recall that an operational CI theory shares the inferential subtheory in common with the classical realist CI theory, but the causal subtheory $\Proc$ is distinct from $\Func$.  Nonetheless, because there is a distinction within $\Proc$ between classical systems (the settings and outcomes of procedures) and general systems, one can imagine a supplementation of $\Proc$ wherein the classical systems and all processes thereon have all of the structure of $\Func$.  Although the inferential consequences of this structure can in principle be obtained by encoding it in the prediction map, the framework is more useful if the additional equalities are present within $\Proc$ itself.  This supplementation is likely to be particularly useful for the study of computational complexity in general operational theories \cite{lee2015computation,barrett2019computational,barnum2018oracles,garner2018interferometric,krumm2019quantum}.

It should also be straightforward to formulate our framework using the language of category theory; category-theoretic tools might then provide guidance on which extensions of our framework are most easily formalized next, and might provide technical tools (e.g., for going beyond the finiteness assumptions that we have made). Making connections to the string diagrammatic representation of double-categories \cite{myers2016string} may be a useful first step.

The process-theoretic formulation of a CI theory makes it easy to incorporate extra structure into the systems. Of particular interest would be to equip systems with the action of particular groups in order to be able to represent symmetries explicitly in our formalism. This is essential for an understanding of unspeakable information~\cite{peres2002unspeakable,bartlett2007reference} and for leveraging this to prove new no-go theorems and find new types of nonclassicality. Tools from Refs.~\cite{selby2017leaks,selby2019compositional} provide a useful starting point for this project.

Additionally, it would be useful to complete the project begun in Section~\ref{causmode}, namely, that of determining how various results in the framework of classical causal models~\cite{pearl2009causality} can be recast using the formalism of classical realist causal-inferential theories, and to explore to what extent the additional causal-inferential unscrambling that our framework provides may be beneficial to the field of causal inference.
It will also be interesting to consider how the notion of an {\em operational} CI theory compares to the notion of a causal model with latent systems that can be quantum or GPT~\cite{henson2014theory,fritz2012beyond,fritz2016beyond} and whether
our framework offers some advantages relative to these.

At present, our framework describes only the reasoning of a single agent. It would be interesting to incorporate the reasoning of multiple agents. This project will require integrating insights such as pooling of states of knowledge~\cite{pooling2007,Leifer_2014}. Relatedly, it would be interesting to consider what insight our framework can add to puzzles regarding the fact that agents can themselves be considered as physical systems.  Such puzzles include the scenario of   Wigner's friend~\cite{Wigner1995} and variants thereof~\cite{Frauchiger2018}.
On a related note, it would be interesting to study how the available causal mechanisms in a CI theory determine the precise manner in which any agent, considered as a physical system, can gather information about its environment---and hence, what sort of theory of inference is most adaptive for it. One might expect that such considerations will constrain the interplay between the causal and inferential subtheories of any realist CI theory.

One could also seek to attack the problem of reconstructing quantum theory from novel axioms using our framework as an alternative to existing frameworks
for reconstructing quantum theory~\cite{Hardy,dakicbrukner,Masanes_2011,Chiribella2011,hardy2011reformulating,barnum2014higher,selby2018reconstructing,van2018effect,tull2018categorical,Chiribella2015QuantumTI,goyal2008information,budiyono2017quantum,hohn2017quantum,clifton2003characterizing}.
This would be particularly interesting if one could reconstruct quantum theory as a {\em realist} CI theory
rather than a generic operational CI theory.  Our framework may also may provide new insights into axioms that single out quantum correlations~\cite{brassardetal,Linden2007,Pawlowski2009,Fritz2013a,Henson2015,Gonda2018almostquantum} (including, e.g., the constraints articulated in Section~\ref{sec:QuantumOmelette}).

One can naturally define postquantumness~\cite{zyczkowski2008quartic,Cabello2012,sainz2015postquantum,barnum2016composites,hoban2018channel,lee2018no,schmid2020postquantum,hefford2020hyper} of correlations in our framework.
For a given causal structure, any distribution that is $\PI$-realizable by an operational theory $\PI$, but which is not $\PQI$-realizable, is said to be postquantum.
Our framework may also allow for new ways of studying postquantumness; e.g., if one were to develop a notion of a {\em quantum} realist causal-inferential theory, then, for a given operational CI theory, one could seek to determine which experimental scenarios manifest postquantumness in the sense of failing to admit of a quantum realist representation or failing to admit of a Leibnizian quantum realist representation. 

Epistemically restricted classical statistical theories, such as those described in Refs.~\cite{spekkens2007evidence,spekkens2016quasi}, if conceptualized as operational CI theories, are 
 theories that admit of a Leibnizian classical realist representation. In this sense, if the world were governed by such a theory, there would be no problem to providing satisfactory realist explanations of observations, and one would have no need to consider any departure from the classical notion of realism $\FI$.  Nonetheless, it might be interesting to try and cast such theories as examples of {\em nonclassical} realist CI theories themselves, that is, as defining a triple ($\textsc{X}\Func$, $\textsc{X}\Inf$, $\XFXI$) that differs from the classical triple ($\Func$, $\Inf$, $\FI$). Ideally, this would be done such that the epistemic restriction emerges as a consequence of assumptions about the underlying reality, as opposed to being a supplementary assumption.
  Even though such theories are classical insofar as they also admit a Leibnizian classical realist representation, this project might nonetheless constitute a useful warm-up for the project of characterizing $\QFQI$. On the one hand, by exploring other realist CI theories we will gain insight into how the causal and inferential subtheories constrain one another, and, on the other hand, there are many formal similarities between epistemically restricted statistical theories and quantum theory (indeed, they often constitute subtheories of quantum theory).

We also discussed in Section~\ref{ncsec} how,  if Conjecture~\ref{conjecture1} is false then there are no-go theorems for Leibnizian classical realist representations of operational quantum theory beyond the no-go theorems based on generalized noncontextuality.  It is therefore important to settle the question of the status of Conjecture~\ref{conjecture1}.

Although one should not demand (or even expect) Leibnizianity to be a strong enough principle to single out a unique nonclassical realist representation for operational quantum theory, results in Ref.~\cite{schmid2021only} are suggestive that this may in fact be possible. (In particular, Ref.~\cite{schmid2021only} proves that there is a unique classical realist representation of any odd-dimensional stabilizer subtheory, namely, that given by Refs.~\cite{spekkens2007evidence,gross2006hudson}.)

We have presented a partial development of a graphical calculus for Boolean propositional logic. We leave for future work the problems of developing this into a complete graphical calculus, extending it to incorporate predicate logic, and generalizing it to nonclassical logics.
Similarly, there are additional tools from Bayesian probability theory which would be useful to incorporate into $\Inf$ such as postselection and Bayesian inversion.
Both of these projects are likely to help with the eventual development of $\textsc{QSubStoch}$.

Another research direction concerns the development of a resource theory~\cite{coecke2016mathematical} of nonclassicality.  We have here argued that the distinction between classical and nonclassical is best understood as a distinction concerning the sort of realism required to provide an explanation of operational predictions.  Within any proposal for a nonclassical realist CI theory $\XFXI$ which subsumes the classical realist CI theory $\FI$, therefore, one can hope to formulate a resource theory of nonclassicality of processes.  In this way, the research program described here could clarify the notion of nonclassicality inherent in `common-cause boxes' (i.e., Bell scenarios), studied in Refs.~\cite{Wolfe2020quantifyingbell,Schmid2020typeindependent,rosset2019characterizing}, and the nonclassicality inherent in contextuality scenarios (i.e., scenarios that imply a noncontextuality no-go theorem, but not a Bell-like no-go theorem).

\section*{Acknowledgements}
We thank Tobias Fritz, Ana Bel\'en Sainz, Giulio Chiribella, Aleks Kissenger, Debashis Saha, and the participants of Quantum Physics and Logic 2020 for helpful discussions.  JHS thanks Nitica Sakharwade and Lucien Hardy for interesting discussions and in particular for suggesting the possibility for a unique prediction map arising from some quantum ontology. RWS acknowledges Bob Coecke for early discussions concerning how to provide a category-theoretic characterization of generalized noncontextuality. All of the diagrams were created using TikZit.
D.S. is supported by a Vanier Canada Graduate Scholarship. JHS acknowledges support by the Foundation for Polish Science (IRAP project, ICTQT, contract no.2018/MAB/5, co-financed by EU within Smart
Growth Operational Programme).
This research was supported by Perimeter Institute for Theoretical Physics. Research at Perimeter Institute is supported in part by the Government of Canada through the Department of Innovation, Science and Economic Development Canada and by the Province of Ontario through the Ministry of Colleges and Universities.

\bibliographystyle{apsrev4-1}
\bibliography{bibliography}

\appendix

\section{Related work}\label{relatedwork}

A number of previous works either inspired parts of our work, or would be interesting to relate to our work.

The basic diagrammatic notation underpinning this work can be traced back to the work of, for example, \cite{kelly1972many,joyal1991geometry}, which used string diagrams to represent particular types of categories. See \cite{selinger2010survey} for a clear survey of these notations, and see \cite{penrose1971applications} for a graphical representation of tensors. The two-directional diagrams which we used here were inspired by Hardy's duotensor notation \cite{hardy2011reformulating}. A seemingly related notation has also appeared in the context of double categories \cite{myers2016string}, and it would be interesting to see if there is a formal connection between these. The work of \cite{fritz2018bimonoidal} (which was itself based on \cite{mellies2006functorial}) first  introduced us to the graphical representation of diagram-preserving maps which we used in this work.

Diagrammatic notation was first used in the context of quantum theory within the research program of categorical quantum mechanics, which began in \cite{coecke2004logic}, was axiomatised in \cite{abramsky2004categorical}, and is now the basis of the textbook \cite{coecke2018picturing}. This sparked the quantum picturalism revolution \cite{coecke2010quantum,hardy2010formalism}, as well as use of similar notation for GPTs~\cite{hardy2011reformulating} and the operational probabilistic theories of the Pavia group \cite{chiribella2010probabilistic,Chiribella2011,d2017quantum}. Stronger connections between these notations have been developed in, for example, \cite{gogioso2017categorical,wilce2018shortcut,selby2018reconstructing}. Moreover, more categorical approaches to generalized theories have been studied extensively using diagrammatic notation, in particular by Gogioso in Refs.~\cite{gogioso2017fantastic,gogioso2017categorical,hefford2020hyper}, which also contains a formal treatment of the infinite dimensional case~\cite{gogioso2016infinite,gogioso2018quantum}.

There are many connections to the framework of operational probabilistic theories~\cite{chiribella2010probabilistic,Chiribella2011,d2017quantum}, which served as inspiration for multiple aspects of our framework. Developing a full understanding of the relations between the two is left for future work. Of particular note is the idea of a prediction map being used to define a notion of equivalence, with respect to which one can quotient.
This notion of quotienting also appeared in \cite{hardy2011reformulating} and \cite{schmid2020structure}. Moreover, the causality axiom of Ref.~\cite{chiribella2010probabilistic} is closely related to our ignorability assumption, and both of these are closely related to the notion of terminality of Refs.~\cite{coecke2014terminality,kissinger2017equivalence}.

Moreover, the rough idea of structure preservation in ontological models has appeared in various forms (e.g. as a diagram-preserving map, or equivalently as a functor between categories)
in Refs.~\cite{schmid2020structure,gheorghiu2019ontological,catani2020mathematical,mansfield2018quantum,karvonen2018categories,abramsky2019comonadic}.

\section{On the meaning of diagrams in our operational CI theories} \label{manifestorcausal2}

In Section~\ref{sec:reps}, we noted that diagram preservation is an immediate consequence of our choice to take diagrams in an operational CI theory to represent one's hypothesis about the fundamental causal and inferential structure in the given scenario. We now contrast this with the usual approach to operational theories, wherein one typically takes operational diagrams to be a representation of some kind of structure that is independent of one's interpretation.

For example, in quantum theory, any given scenario can be described as a circuit of completely-positive trace-preserving maps. The circuit assigned to a particular experiment (or to the idealized conception thereof) is essentially unique, and is a fact on which physicists of virtually all interpretational camps will agree upon. At a minimum, these camps agree on this circuit as the `correct' one in the sense of having maximal pragmatic utility as a mathematical representation of one's experiment. In the usual approach to operational theories, it is this circuit depicting the calculational structure that is typically taken as the diagram representing one's scenario.

To provide a realist representation of one's scenario, however, requires one (in our view)  to furthermore commit to an underlying causal structure. In general (depending on one's interpretational camp), this causal structure will not correspond to the calculational circuit just described. Hence, in such an approach, one's realist representation map would not be diagram-preserving, but must somehow map from the calculational circuit to one's hypothesized causal structure.

In contrast, in our framework, we do not represent the calculational circuit at all. Rather, we stipulate that the diagram one draws to describe a given scenario in an operational CI theory must be chosen to respect one's hypothesis about the fundamental causal-inferential structure. Hence, the classical realist representation map is diagram-preserving.

 The only real novelty here is that in our framework, the term `operational theory' no longer describes a description which is so bare-bones that all users of the framework will agree on it.

 We now note a key consequence of our choice to take operational diagrams to represent one's hypothesis about the fundamental causal-inferential structure: namely, that diagram-preservation does {\em not} constitute a limitation on the scope of realist representations within our framework.

To demonstrate this, consider the \crealist representation of a pair of independent causal systems in our framework. Diagram preservation implies that these are represented by a pair of independent systems in the \crealist CI theory. Since system composition in a \crealist CI theory is given by the cartesian product of the corresponding ontic state spaces, it appears as though our framework commits one to represent every pair of operational systems by a cartesian product of the corresponding state spaces, an assumption sometimes termed {\em ontic separability}~\cite{Harrigan}. If one is committed to the idea that the two systems in question fundamentally exhibit some holistic properties, then this assumption (and hence our assumption of diagram preservation) might appear overly restrictive. Such an impression is mistaken, however. In our framework, to posit such holistic properties is to grant
 that the actual causal situation is one in which the relevant degrees of freedom fundamentally {\em cannot} be divided into two independent subsystems---even if they are represented by a tensor product in the calculational diagram.
 Rather, they fundamentally behave as a single monolithic causal system. With this causal hypothesis, then, our framework demands that one represent the operational scenario using a single system rather than a pair of systems, and the \crealist representation of this single system is thereby allowed to be an arbitrary ontic state space, not necessarily a Cartesian product of ontic state spaces of two components. So we see that our framework does not limit the scope of \crealist representations.

Of course, given a commitment to a {\em particular} causal-inferential hypothesis, the assumption of diagram-preservation provides strong constraints on the scope of possible \crealist representations. These constraints take the form of causal compatibility constraints, as discussed in Section~\ref{causcompatineq}. Indeed, one can subsume a number of assumptions made in deriving no-go theorems on ontological representations (including those needed to derive Bell's theorem, a version of the preparation-independence postulate~\cite{pusey2012reality}, the Markovianity assumption used in Ref.~\cite{montina2008exponential}, lambda-screening~\cite{SpePIRSA}, and the assumptions used in Ref.~\cite{schmid2020structure}) under the assumption that the fundamental causal-inferential structure respects the standard (calculational) quantum circuit.

\section{Discussion of the first generator}\label{app:firstgenerator}

There is a somewhat awkward feature of our current rewrite rules for this generator, namely, that the way in which the causal structure can be encoded within the inferential structure is non-unique. For example, if we are given a situation such as
\beq
\InputIfFileExists{Diagrams/nonUniqueCausalToInferential.tikz}{}{\input{./figures/Diagrams/nonUniqueCausalToInferential.tikz}} 
\eeq
then there are many ways to use our rewrite rules to encode the causal structure of this diagram into the inferential structure, for example:
\beq
\InputIfFileExists{Diagrams/nonUniqueCausalToInferential.tikz}{}{\input{./figures/Diagrams/nonUniqueCausalToInferential.tikz}}\quad=\quad%
\InputIfFileExists{Diagrams/nonUniqueCausalToInferential1.tikz}{}{\input{./figures/Diagrams/nonUniqueCausalToInferential1.tikz}}\quad=\quad%
\InputIfFileExists{Diagrams/nonUniqueCausalToInferential2.tikz}{}{\input{./figures/Diagrams/nonUniqueCausalToInferential2.tikz}},
\eeq
or
\beq
\InputIfFileExists{Diagrams/nonUniqueCausalToInferential.tikz}{}{\input{./figures/Diagrams/nonUniqueCausalToInferential.tikz}}\quad=\quad%
\InputIfFileExists{Diagrams/nonUniqueCausalToInferential3.tikz}{}{\input{./figures/Diagrams/nonUniqueCausalToInferential3.tikz}}\quad=\quad%
\InputIfFileExists{Diagrams/nonUniqueCausalToInferential4.tikz}{}{\input{./figures/Diagrams/nonUniqueCausalToInferential4.tikz}}.
\eeq
This is not a problem, as the definitions of these two  stochastic maps ensure that they satisfy:
\beq\label{eq:bifunct}
\InputIfFileExists{Diagrams/bifunctorialToDistributivity.tikz}{}{\input{./figures/Diagrams/bifunctorialToDistributivity.tikz}} \quad = \quad %
\InputIfFileExists{Diagrams/bifunctorialToDistributivity1.tikz}{}{\input{./figures/Diagrams/bifunctorialToDistributivity1.tikz}}.
\eeq
This follows immediately by considering the action of the two of these on delta-function states of knowledge and then using their defintions together with bifunctoriality within the causal theory. This means that the important properties of the causal theory, which ensure that the diagrammatic representation is a faithful one, are naturally encoded within these stochastic maps which appear in our rewrite rules.

Whilst this may ensure consistency, one of the nice features of working with the diagrammatic representation of a process theory, is that we do not have to worry about  such conditions as they are made tautological by the diagrammatic representation. Clearly, by privileging sequential and parallel composition within our rewrite rules we have, to some extent, lost this nice feature. A solution to this comes from the work of \cite{patterson2021wiring} (see \cite{fong2018seven} for a clear conceptual introduction to operads and their algebras) which allows us to reconceptualise our causal process theory as an algebra for an acyclic wiring operad, $\mathcal{W}$. Then, to every diagram in the causal theory, such as,
\beq%
\InputIfFileExists{Diagrams/diagramSmall.tikz}{}{\input{./figures/Diagrams/diagramSmall.tikz}},\eeq
there is an associated acyclic wiring operation $w\in \mathcal{W}$. This tells us only how the processes in the diagram are wired together but nothing about the nature of the processes. We can then define an operad functor from $\mathcal{W} \to \Inf$ such that, for example:
\beq
\InputIfFileExists{Diagrams/operad1.tikz}{}{\input{./figures/Diagrams/operad1.tikz}} \quad = \quad %
\InputIfFileExists{Diagrams/operad2.tikz}{}{\input{./figures/Diagrams/operad2.tikz}}.
\eeq
where the $w$ in the white dot is the image of $w\in \mathcal{W}$ under the action of the operad functor.

The two stochastic maps that we have so far been using within our rewrite rules can then be seen simply as special cases of these $w$-dots. Moreover, the condition that we show in eq.~\eqref{eq:bifunct} is then an immediate consequence of operad-functoriality.

Exploring this operadic formulation further is another interesting research direction. On the one hand, it may give insight into the structure of causal-inferential theories themselves, and, on the other hand, allow for the study of theories with a less stringent notion of causation. For example, by following the  forthcoming work of \cite{timeneutral} we could explore causal-inferential theories without an intrinsic preferred time direction.

\section{Useful results in $\Inf$}\label{app:Inf}


We now list a number of useful equalities, some of which we will need for proofs in the next section.
Each can be verified immediately by composing the partial functions defining the relevant processes.

\beq
\InputIfFileExists{Diagrams/topQuestion1.tikz}{}{\input{./figures/Diagrams/topQuestion1.tikz}} \ = \ %
\InputIfFileExists{Diagrams/topQuestion2.tikz}{}{\input{./figures/Diagrams/topQuestion2.tikz}} \quad,\quad
\InputIfFileExists{Diagrams/bottomQuestion1.tikz}{}{\input{./figures/Diagrams/bottomQuestion1.tikz}} \ = \ %
\InputIfFileExists{Diagrams/bottomQuestion2.tikz}{}{\input{./figures/Diagrams/bottomQuestion2.tikz}}.
\eeq

\beq
\begin{tikzpicture}
	\begin{pgfonlayer}{nodelayer}
		\node [style=infcopoint] (0) at (0, -0) {$\top$};
		\node [style=none] (1) at (-1.25, -0) {};
		\node [style={up label}] (2) at (-1.25, -0) {$X$};
	\end{pgfonlayer}
	\begin{pgfonlayer}{edgelayer}
		\draw [cWire] (1.center) to (0);
	\end{pgfonlayer}
\end{tikzpicture}
\ \ = \
\begin{tikzpicture}
	\begin{pgfonlayer}{nodelayer}
		\node [style=infcopoint] (0) at (1.75, -0) {$\top$};
		\node [style={up label}] (1) at (1, -0) {$\textsc{b}$};
		\node [style={small box}] (2) at (0, -0) {$\pi$};
		\node [style=none] (3) at (-1.5, -0) {};
		\node [style={up label}] (4) at (-1, -0) {$X$};
	\end{pgfonlayer}
	\begin{pgfonlayer}{edgelayer}
		\draw [cWire] (3) to (2);
		\draw [cWire] (2) to (0);
	\end{pgfonlayer}
\end{tikzpicture}
\quad , \quad
\begin{tikzpicture}
	\begin{pgfonlayer}{nodelayer}
		\node [style=infcopoint] (0) at (0, -0) {$\perp$};
		\node [style=none] (1) at (-1.25, -0) {};
		\node [style={up label}] (2) at (-1.25, -0) {$X$};
	\end{pgfonlayer}
	\begin{pgfonlayer}{edgelayer}
		\draw [cWire] (1.center) to (0);
	\end{pgfonlayer}
\end{tikzpicture}
\ \ = \
\begin{tikzpicture}
	\begin{pgfonlayer}{nodelayer}
		\node [style=infcopoint] (0) at (1.75, -0) {$\perp$};
		\node [style={up label}] (1) at (1, -0) {$\textsc{b}$};
		\node [style={small box}] (2) at (0, -0) {$\pi$};
		\node [style=none] (3) at (-1.5, -0) {};
		\node [style={up label}] (4) at (-1, -0) {$X$};
	\end{pgfonlayer}
	\begin{pgfonlayer}{edgelayer}
		\draw [cWire] (3) to (2);
		\draw [cWire] (2) to (0);
	\end{pgfonlayer}
\end{tikzpicture}
\ \
\eeq

\beq
\InputIfFileExists{Diagrams/InfIdentity1.tikz}{}{\input{./figures/Diagrams/InfIdentity1.tikz}} \ = \ %
\begin{tikzpicture}
	\begin{pgfonlayer}{nodelayer}
		\node [style=none] (7) at (-7, 0) {};
		\node [style=none] (8) at (-4.75, 0) {};
	\end{pgfonlayer}
	\begin{pgfonlayer}{edgelayer}
		\draw [cWire] (7.center) to (8.center);
	\end{pgfonlayer}
\end{tikzpicture}
} \ = \
\InputIfFileExists{Diagrams/InfIdentity3.tikz}{}{\input{./figures/Diagrams/InfIdentity3.tikz}}
\eeq

\beq%
\InputIfFileExists{Diagrams/InfIdentity4New.tikz}{}{\input{./figures/Diagrams/InfIdentity4New.tikz}}\quad=\quad%
\InputIfFileExists{Diagrams/InfIdentity5New.tikz}{}{\input{./figures/Diagrams/InfIdentity5New.tikz}}\eeq

\begin{align}
\InputIfFileExists{Diagrams/tt1.tikz}{}{\input{./figures/Diagrams/tt1.tikz}}\ =\ %
\begin{tikzpicture}
	\begin{pgfonlayer}{nodelayer}
		\node [style=none] (0) at (0.7500002, -0) {};
		\node [style=none] (1) at (-0.5000001, -0) {};
		\node [style={up label}] (2) at (0, -0) {$\textsc{b}$};
	\end{pgfonlayer}
	\begin{pgfonlayer}{edgelayer}
		\draw [cWire] (1.center) to (0.center);
	\end{pgfonlayer}
\end{tikzpicture}} \quad &,\quad %
\InputIfFileExists{Diagrams/tt2.tikz}{}{\input{./figures/Diagrams/tt2.tikz}}\ =\ %
\InputIfFileExists{Diagrams/tt6.tikz}{}{\input{./figures/Diagrams/tt6.tikz}}\\
\InputIfFileExists{Diagrams/tt3.tikz}{}{\input{./figures/Diagrams/tt3.tikz}}\ =\ %
\begin{tikzpicture}
	\begin{pgfonlayer}{nodelayer}
		\node [style=none] (0) at (0.7500002, -0) {};
		\node [style=none] (1) at (-0.5000001, -0) {};
		\node [style={up label}] (2) at (0, -0) {$\textsc{b}$};
	\end{pgfonlayer}
	\begin{pgfonlayer}{edgelayer}
		\draw [cWire] (1.center) to (0.center);
	\end{pgfonlayer}
\end{tikzpicture}} \quad &, \quad %
\InputIfFileExists{Diagrams/tt4.tikz}{}{\input{./figures/Diagrams/tt4.tikz}}\ =\ %
\InputIfFileExists{Diagrams/tt7.tikz}{}{\input{./figures/Diagrams/tt7.tikz}}
\end{align}
\beq
\InputIfFileExists{Diagrams/copyY1.tikz}{}{\input{./figures/Diagrams/copyY1.tikz}}\ =\ %
\InputIfFileExists{Diagrams/copyY2.tikz}{}{\input{./figures/Diagrams/copyY2.tikz}}\quad,\quad
\InputIfFileExists{Diagrams/copyN1.tikz}{}{\input{./figures/Diagrams/copyN1.tikz}}\ =\ %
\InputIfFileExists{Diagrams/copyN2.tikz}{}{\input{./figures/Diagrams/copyN2.tikz}}
\eeq

\beq%
\InputIfFileExists{Diagrams/ii8New.tikz}{}{\input{./figures/Diagrams/ii8New.tikz}}\ =\ %
\InputIfFileExists{Diagrams/ii9New.tikz}{}{\input{./figures/Diagrams/ii9New.tikz}}\quad,\quad%
\InputIfFileExists{Diagrams/ii10New.tikz}{}{\input{./figures/Diagrams/ii10New.tikz}}\ =\ %
\InputIfFileExists{Diagrams/ii11New.tikz}{}{\input{./figures/Diagrams/ii11New.tikz}}\eeq


\begin{align}
\InputIfFileExists{Diagrams/copyCoUnit.tikz}{}{\input{./figures/Diagrams/copyCoUnit.tikz}}\quad =\quad %
\InputIfFileExists{Diagrams/copyCoUnitnew.tikz}{}{\input{./figures/Diagrams/copyCoUnitnew.tikz}}\quad \label{copycounitinf} &=\quad %
\begin{tikzpicture}
	\begin{pgfonlayer}{nodelayer}
		\node [style=none] (0) at (-1, -0) {};
		\node [style=none] (1) at (1, -0) {};
		\node [style={up label}] (2) at (0.7500001, -0) {$X$};
	\end{pgfonlayer}
	\begin{pgfonlayer}{edgelayer}
		\draw [cWire] (0.center) to (1.center);
	\end{pgfonlayer}
\end{tikzpicture}}\
\end{align}
\begin{align}
\InputIfFileExists{Diagrams/copyCoUnit4.tikz}{}{\input{./figures/Diagrams/copyCoUnit4.tikz}}\quad &=\quad %
\InputIfFileExists{Diagrams/copyCoUnit3.tikz}{}{\input{./figures/Diagrams/copyCoUnit3.tikz}}
\end{align}

\beq \label{usedforabsorp}
\InputIfFileExists{Diagrams/absorption2.tikz}{}{\input{./figures/Diagrams/absorption2.tikz}}\quad=\quad%
\InputIfFileExists{Diagrams/absorption3.tikz}{}{\input{./figures/Diagrams/absorption3.tikz}};
\eeq

\beq
\label{eq:tripleCopy}
\InputIfFileExists{Diagrams/ii12New.tikz}{}{\input{./figures/Diagrams/ii12New.tikz}} \quad = \quad %
\InputIfFileExists{Diagrams/ii13New.tikz}{}{\input{./figures/Diagrams/ii13New.tikz}}\eeq

\beq  \label{eq:distributivity1}
\InputIfFileExists{Diagrams/ii14New.tikz}{}{\input{./figures/Diagrams/ii14New.tikz}} \quad = \quad %
\InputIfFileExists{Diagrams/ii15New.tikz}{}{\input{./figures/Diagrams/ii15New.tikz}},\eeq

\beq  \label{eq:distributivity2}
\InputIfFileExists{Diagrams/distNew1.tikz}{}{\input{./figures/Diagrams/distNew1.tikz}} \quad = \quad %
\InputIfFileExists{Diagrams/distNew2.tikz}{}{\input{./figures/Diagrams/distNew2.tikz}},\eeq


\subsection{Properties of a Boolean algebra (proved diagrammatically)} \label{Boolpropproof}

A Boolean algebra satisfies the following properties (which are not all independent). For simplicity, we will here use $\alpha$, $\beta$, and $\gamma$ to denote propositions.
\begin{itemize}
\item associativity:
$\alpha \lor (\beta \lor \gamma) = (\alpha \lor \beta) \lor \gamma$ and $\alpha \land  (\beta \land \gamma) = (\alpha \land \beta) \land \gamma$
\item commutativity:
$\alpha \lor \beta = \beta \lor \alpha$ and $\alpha \land \beta = \beta \land \alpha$
\item identity:
$\alpha \lor \perp = \alpha$ and $\alpha \land \top = \alpha$
\item complements:
$\alpha \lor \neg \alpha = \top$ and $\alpha \land \neg \alpha = \perp$
\item distributivity: $\alpha \lor (\beta \land \gamma) = (\alpha \lor \beta) \land (\alpha \lor \gamma)$ and $\alpha \land (\beta \lor \gamma) = (\alpha \land \beta) \lor (\alpha \land \gamma)$
\item idempotence:
$\alpha \lor \alpha = \alpha$ and $\alpha \land \alpha = \alpha$
\item annihilation:
$\alpha \lor \top = \top$ and $\alpha \land \perp = \perp$
\item absorption:
$\alpha \lor (\alpha \land \beta) = \alpha$ and $\alpha \land (\alpha \lor \beta) = \alpha$.
\end{itemize}

 We now prove that each of these expressions holds in our diagrammatic representations. We prove only the first of each of these expressions; in each case, the proof of the second is similar. The associativity and commutativity axiom follow immediately by the symmetry of Eq.~\eqref{logicalconnect1}. The identity axiom follows simply from
 \beq
\InputIfFileExists{Diagrams/identityAxiomProof1.tikz}{}{\input{./figures/Diagrams/identityAxiomProof1.tikz}}\ = \ %
\InputIfFileExists{Diagrams/identityAxiomProof2.tikz}{}{\input{./figures/Diagrams/identityAxiomProof2.tikz}} \ = \ %
\InputIfFileExists{Diagrams/identityAxiomProof3.tikz}{}{\input{./figures/Diagrams/identityAxiomProof3.tikz}}.
 \eeq
To prove the complements axiom, one has
\begin{align}
\InputIfFileExists{Diagrams/complementsAxiomProof1.tikz}{}{\input{./figures/Diagrams/complementsAxiomProof1.tikz}}\ &= \ %
\InputIfFileExists{Diagrams/complementsAxiomProof2.tikz}{}{\input{./figures/Diagrams/complementsAxiomProof2.tikz}} \\
 &= \ %
\InputIfFileExists{Diagrams/complementsAxiomProof3.tikz}{}{\input{./figures/Diagrams/complementsAxiomProof3.tikz}} \\
 &= \ %
\InputIfFileExists{Diagrams/complementsAxiomProof4.tikz}{}{\input{./figures/Diagrams/complementsAxiomProof4.tikz}}\\
 &=\ %
\InputIfFileExists{Diagrams/complementsAxiomProof5.tikz}{}{\input{./figures/Diagrams/complementsAxiomProof5.tikz}}
 \end{align}

\noindent The proof of distributivity is as follows:
\begin{align}\label{eq:idempotence}
\InputIfFileExists{Diagrams/ii16.tikz}{}{\input{./figures/Diagrams/ii16.tikz}} \quad &= \quad %
\InputIfFileExists{Diagrams/ii17.tikz}{}{\input{./figures/Diagrams/ii17.tikz}} \\
&= %
\InputIfFileExists{Diagrams/ii18.tikz}{}{\input{./figures/Diagrams/ii18.tikz}} \\
&= %
\InputIfFileExists{Diagrams/ii19.tikz}{}{\input{./figures/Diagrams/ii19.tikz}} \label{ii19} \\
&= %
\InputIfFileExists{Diagrams/ii20.tikz}{}{\input{./figures/Diagrams/ii20.tikz}}
\end{align}
where Eq.~\eqref{ii19} follows from Eq.~\eqref{eq:tripleCopy} and Eq.~\eqref{eq:idempotence} follows from Eq.~\eqref{eq:distributivity1}.

\noindent The proof of idempotence is as follows:
\beq%
\InputIfFileExists{Diagrams/ii21.tikz}{}{\input{./figures/Diagrams/ii21.tikz}}=%
\InputIfFileExists{Diagrams/ii22.tikz}{}{\input{./figures/Diagrams/ii22.tikz}}=
\begin{tikzpicture}
	\begin{pgfonlayer}{nodelayer}
		\node [style=none] (0) at (-1.25, 0) {};
		\node [style=small box] (1) at (-2.25, 0) {$\alpha$};
		\node [style=none] (6) at (-3.25, 0) {};
	\end{pgfonlayer}
	\begin{pgfonlayer}{edgelayer}
		\draw [cWire](1) to (0.center);
		\draw [cWire](6.center) to (1);
	\end{pgfonlayer}
\end{tikzpicture}
}\eeq

\noindent The proof of annihilation is as follows:
\begin{align}
\InputIfFileExists{Diagrams/ii24.tikz}{}{\input{./figures/Diagrams/ii24.tikz}}&=%
\InputIfFileExists{Diagrams/ii25.tikz}{}{\input{./figures/Diagrams/ii25.tikz}}\\
&=%
\InputIfFileExists{Diagrams/ii26.tikz}{}{\input{./figures/Diagrams/ii26.tikz}}=%
\begin{tikzpicture}
	\begin{pgfonlayer}{nodelayer}
		\node [style=small box] (0) at (0, 0) {$\top$};
		\node [style=none] (1) at (-1, 0) {};
		\node [style=none] (2) at (1, 0) {};
	\end{pgfonlayer}
	\begin{pgfonlayer}{edgelayer}
		\draw [cWire](1.center) to (0);
		\draw [cWire](0) to (2.center);
	\end{pgfonlayer}
\end{tikzpicture}
}.
\end{align}

\noindent The proof of absorption is as follows:
\begin{align}
\InputIfFileExists{Diagrams/absorption1.tikz}{}{\input{./figures/Diagrams/absorption1.tikz}} &= %
\InputIfFileExists{Diagrams/absorption4.tikz}{}{\input{./figures/Diagrams/absorption4.tikz}}\\
&= %
\InputIfFileExists{Diagrams/absorption6.tikz}{}{\input{./figures/Diagrams/absorption6.tikz}}\\
&= %
\InputIfFileExists{Diagrams/absorption7.tikz}{}{\input{./figures/Diagrams/absorption7.tikz}}\label{absorption7}\\
&= %
\begin{tikzpicture}
	\begin{pgfonlayer}{nodelayer}
		\node [style=none] (0) at (-1, -0) {};
		\node [style={small box}] (1) at (0, -0) {$\alpha$};
		\node [style=none] (2) at (1, -0) {};
	\end{pgfonlayer}
	\begin{pgfonlayer}{edgelayer}
		\draw [cWire] (1) to (2.center);
		\draw [cWire] (0.center) to (1);
	\end{pgfonlayer}
\end{tikzpicture}},
\end{align}
where Eq.~\eqref{absorption7} follows from Eq.~\eqref{usedforabsorp}.

\subsection{Partial functions in $\textsc{Boole}$ } \label{partfuncapp}

Recall from Eq.~\eqref{anypartfunc1} that any partial function $\hat{f}$ can be written as
\beq
\begin{tikzpicture}
	\begin{pgfonlayer}{nodelayer}
		\node [style=small box] (0) at (0, 0) {$\hat{f}$};
		\node [style=none] (1) at (-1, 0) {};
		\node [style=none] (2) at (1, 0) {};
	\end{pgfonlayer}
	\begin{pgfonlayer}{edgelayer}
		\draw [cWire] (1.center) to (0);
		\draw [cWire] (0) to (2.center);
	\end{pgfonlayer}
\end{tikzpicture}
} = %
\InputIfFileExists{Diagrams/arbitraryPartial2.tikz}{}{\input{./figures/Diagrams/arbitraryPartial2.tikz}},
\eeq
where $\chi_{\hat{f}}$ specifies the domain of $\hat{f}$ and $F$ is a propositional map.

Consider now the action of $\hat{f}$ on an arbitrary propositional effect defined by $\pi \in \mathcal{B}(Y)$, namely
\beq
\begin{tikzpicture}
	\begin{pgfonlayer}{nodelayer}
		\node [style=infcopoint] (0) at (2.25, -0) {$\pi$};
		\node [style={up label}] (1) at (1, -0) {$Y$};
		\node [style={small box}] (2) at (0, -0) {$\hat{f}$};
		\node [style=none] (3) at (-2, -0) {};
		\node [style={up label}] (4) at (-1, -0) {$X$};
	\end{pgfonlayer}
	\begin{pgfonlayer}{edgelayer}
		\draw [cWire] (3) to (2);
		\draw [cWire] (2) to (0);
	\end{pgfonlayer}
\end{tikzpicture}\ \ =:\
\begin{tikzpicture}
	\begin{pgfonlayer}{nodelayer}
		\node [style=infcopoint] (0) at (0, -0) {$\hat{f}(\pi)$};
		\node [style=none] (1) at (-1.25, -0) {};
		\node [style={up label}] (2) at (-1.25, -0) {$X$};
	\end{pgfonlayer}
	\begin{pgfonlayer}{edgelayer}
		\draw [cWire] (1.center) to (0);
	\end{pgfonlayer}
\end{tikzpicture}.
\eeq
This defines a map from $\mathcal{B}(Y)$ to $\mathcal{B}(X)$, but it remains to see what structure of the Boolean algebra of propositional effects this map preserves. We now show that the action of partial functions on propositional effects preserves $\perp$, $\lor$ and $\land$, but not $\top$ and $\neg$.

First, note that we can reexpress $\hat{f}(\pi)$ in terms of the propositional effect $\chi_{\hat{f}} \in \mathcal{B}(X)$ and the total function $F$ as follows:
\begin{align}
\begin{tikzpicture}
	\begin{pgfonlayer}{nodelayer}
		\node [style=infcopoint] (0) at (2.25, -0) {$\pi$};
		\node [style={up label}] (1) at (1, -0) {$Y$};
		\node [style={small box}] (2) at (0, -0) {$\hat{f}$};
		\node [style=none] (3) at (-2, -0) {};
		\node [style={up label}] (4) at (-1, -0) {$X$};
	\end{pgfonlayer}
	\begin{pgfonlayer}{edgelayer}
		\draw [cWire] (3) to (2);
		\draw [cWire] (2) to (0);
	\end{pgfonlayer}
\end{tikzpicture}  &= %
\InputIfFileExists{Diagrams/partBool1.tikz}{}{\input{./figures/Diagrams/partBool1.tikz}} \\
 &= %
\InputIfFileExists{Diagrams/partBool2.tikz}{}{\input{./figures/Diagrams/partBool2.tikz}} \\
 &= %
\InputIfFileExists{Diagrams/partBool3.tikz}{}{\input{./figures/Diagrams/partBool3.tikz}} \\
 &= %
\InputIfFileExists{Diagrams/partBool4.tikz}{}{\input{./figures/Diagrams/partBool4.tikz}} \\
 &= %
\InputIfFileExists{Diagrams/partBool5.tikz}{}{\input{./figures/Diagrams/partBool5.tikz}} \\
 &= %
\InputIfFileExists{Diagrams/partBool6.tikz}{}{\input{./figures/Diagrams/partBool6.tikz}}.
\end{align}
This is a very natural expression, stating that $\hat{f}(\pi)$ is equivalent to a propositional effect defined by the subset of $X$ which is both in the domain of $\hat{f}$ and in the image of $\pi$ under $F$.

At this point it is easy to verify that $\perp$ is preserved:
\begin{align}
\InputIfFileExists{Diagrams/parBool8.tikz}{}{\input{./figures/Diagrams/parBool8.tikz}} &= %
\InputIfFileExists{Diagrams/parBool7.tikz}{}{\input{./figures/Diagrams/parBool7.tikz}}\\
&= %
\InputIfFileExists{Diagrams/parBool9.tikz}{}{\input{./figures/Diagrams/parBool9.tikz}}\\
&= %
\begin{tikzpicture}
	\begin{pgfonlayer}{nodelayer}
		\node [style=none] (0) at (-0.25, -0) {};
		\node [style=none] (1) at (-2, -0) {};
		\node [style={up label}] (2) at (-2, -0) {$X$};
		\node [style=infcopoint] (3) at (0, -0) {$\perp$};
	\end{pgfonlayer}
	\begin{pgfonlayer}{edgelayer}
		\draw [CcWire] (1.center) to (0.center);
	\end{pgfonlayer}
\end{tikzpicture}},
\end{align}
but that $\top$ is not preserved if $\chi_f\subsetneq X$:
\begin{align}
\InputIfFileExists{Diagrams/parBool11.tikz}{}{\input{./figures/Diagrams/parBool11.tikz}} &= %
\InputIfFileExists{Diagrams/parBool12.tikz}{}{\input{./figures/Diagrams/parBool12.tikz}}\\
&= %
\InputIfFileExists{Diagrams/parBool13.tikz}{}{\input{./figures/Diagrams/parBool13.tikz}}\\
&= %
\begin{tikzpicture}
	\begin{pgfonlayer}{nodelayer}
		\node [style=none] (0) at (-0.25, -0) {};
		\node [style=none] (1) at (-2, -0) {};
		\node [style={up label}] (2) at (-2, -0) {$X$};
		\node [style=infcopoint] (3) at (0, -0) {$\chi_{\hat{f}}$};
	\end{pgfonlayer}
	\begin{pgfonlayer}{edgelayer}
		\draw [CcWire] (1.center) to (0.center);
	\end{pgfonlayer}
\end{tikzpicture}}.
\end{align}
Hence, $\hat{f}$ does not define a Boolean algebra homomorphism (as these preserve $\top$).

However, it does preserve $\lor$ and $\land$, as we now show. Preservation of $\lor$ can be derived as
\begin{align}
\InputIfFileExists{Diagrams/parBool15.tikz}{}{\input{./figures/Diagrams/parBool15.tikz}} &= %
\InputIfFileExists{Diagrams/parBool16.tikz}{}{\input{./figures/Diagrams/parBool16.tikz}}\\
&= %
\InputIfFileExists{Diagrams/parBool17.tikz}{}{\input{./figures/Diagrams/parBool17.tikz}}\\
&= %
\InputIfFileExists{Diagrams/parBool18.tikz}{}{\input{./figures/Diagrams/parBool18.tikz}}\\
&= %
\InputIfFileExists{Diagrams/parBool19.tikz}{}{\input{./figures/Diagrams/parBool19.tikz}} \\
&= %
\InputIfFileExists{Diagrams/parBool20.tikz}{}{\input{./figures/Diagrams/parBool20.tikz}}\label{parbool20}\\
&= %
\InputIfFileExists{Diagrams/parBool21.tikz}{}{\input{./figures/Diagrams/parBool21.tikz}}\\
&= %
\InputIfFileExists{Diagrams/parBool22.tikz}{}{\input{./figures/Diagrams/parBool22.tikz}}\label{parbool22}\\
&= %
\InputIfFileExists{Diagrams/parBool23.tikz}{}{\input{./figures/Diagrams/parBool23.tikz}}\\
&= %
\InputIfFileExists{Diagrams/parBool24.tikz}{}{\input{./figures/Diagrams/parBool24.tikz}},
\end{align}
 where Eq.~\eqref{parbool20} follows from Eq.~\eqref{eq:distributivity2} and Eq.~\eqref{parbool22} follows from Eq.~\eqref{eq:tripleCopy}.
Preservation of $\land$ can be derived as
\begin{align}
\InputIfFileExists{Diagrams/parBool25.tikz}{}{\input{./figures/Diagrams/parBool25.tikz}} &= %
\InputIfFileExists{Diagrams/parBool26.tikz}{}{\input{./figures/Diagrams/parBool26.tikz}}\\
&= %
\InputIfFileExists{Diagrams/parBool27.tikz}{}{\input{./figures/Diagrams/parBool27.tikz}}\\
&= %
\InputIfFileExists{Diagrams/parBool28.tikz}{}{\input{./figures/Diagrams/parBool28.tikz}}\label{parbool28}\\
&= %
\InputIfFileExists{Diagrams/parBool29.tikz}{}{\input{./figures/Diagrams/parBool29.tikz}}\label{parbool29}\\
&= %
\InputIfFileExists{Diagrams/parBool30.tikz}{}{\input{./figures/Diagrams/parBool30.tikz}}\\
&= %
\InputIfFileExists{Diagrams/parBool31.tikz}{}{\input{./figures/Diagrams/parBool31.tikz}},
\end{align}
where Eq.~\eqref{parbool28} follows from Eq.~\eqref{eq:idempotence} and Eq.~\eqref{parbool29} follows from Eq.~\eqref{eq:tripleCopy}.

In summary, we see that $\hat{f}$ is a Boolean algebra homomorphism from $\mathcal{B}(Y)$ to $\mathcal{B}(\chi_{\hat{f}})$.

\subsection{Propositions about composite systems} \label{compositepropns}

In Eq.~\eqref{compprop}, we noted that one can express propositional questions about composite systems, e.g. as
\beq
\InputIfFileExists{Diagrams/compositePropositionalQuestion.tikz}{}{\input{./figures/Diagrams/compositePropositionalQuestion.tikz}}
\eeq
where $\pi \in \mathcal{B}(X\times Y)$.
However, suppose that we have some propositional question $\pi$ about $X$ and some propositional question $\pi'$ about $Y$; then, how should these be composed to give a propositional question about $X\times Y$? One's first guess might be to simply compose these in parallel within $\textsc{Boole}$; however, this would give
\beq
\InputIfFileExists{Diagrams/composingQuestions1.tikz}{}{\input{./figures/Diagrams/composingQuestions1.tikz}},
\eeq
which is not a propositional question as its a function to $\textsc{b}\times \textsc{b}$ rather than simply $\textsc{b}$. The resolution comes from examining how we expect these to compose as Boolean Algebras. Note that the sets compose via the cartesian product, hence Boolean Algebras compose via the following rule: $\mathcal{B}(X)\otimes \mathcal{B}(Y) := \mathcal{B}(X\times Y)$ and moreover that $\pi\otimes \pi' \in \mathcal{B}(X)\otimes\mathcal{B}(Y)$ can be defined as a subset of $X\times Y$ by $\pi\times \pi'$ viewed as subsets of $X$ and $Y$ respectively. In our diagrammatic language this is represented by:
\beq
\InputIfFileExists{Diagrams/composingQuestions3.tikz}{}{\input{./figures/Diagrams/composingQuestions3.tikz}}
\quad\quad=
\InputIfFileExists{Diagrams/composingQuestions2.tikz}{}{\input{./figures/Diagrams/composingQuestions2.tikz}}
\eeq
This composite therefore can be clearly interpreted as the situation in which we are interested in both $\pi$ and $\pi'$ being true about their respective systems.

Note, however, that we then clearly have other ways that we could compose these, for instance via:
\beq
\InputIfFileExists{Diagrams/composingQuestions4.tikz}{}{\input{./figures/Diagrams/composingQuestions4.tikz}}
\eeq
to see what this means in terms of the Boolean Algebra consider the following rewrites:
\begin{align}
\InputIfFileExists{Diagrams/composingQuestions4.tikz}{}{\input{./figures/Diagrams/composingQuestions4.tikz}} \quad &= \quad
\InputIfFileExists{Diagrams/composingQuestions5.tikz}{}{\input{./figures/Diagrams/composingQuestions5.tikz}} \\
&= \quad
\InputIfFileExists{Diagrams/composingQuestions6.tikz}{}{\input{./figures/Diagrams/composingQuestions6.tikz}} \\
&= \quad
\InputIfFileExists{Diagrams/composingQuestions7.tikz}{}{\input{./figures/Diagrams/composingQuestions7.tikz}}
\end{align}
(Here we have included or omitted subscripts labeling the systems about which propositions are being made, as convenient.)

This corresponds, intuitively, to what we would mean for the logical disjunction of two propositions about distinct systems.

\subsection{Useful relations between stochastic maps}

There are various relationships between the stochastic maps which we have defined which are useful in the proofs in this paper. We list them here for reference.
\beq
\InputIfFileExists{Diagrams/starIso1.tikz}{}{\input{./figures/Diagrams/starIso1.tikz}}\quad=\quad%
\begin{tikzpicture}
	\begin{pgfonlayer}{nodelayer}
		\node [style=none] (0) at (1.25, -0) {};
		\node [style=none] (1) at (-0.25, -0) {};
		\node [style={up label}] (2) at (0, -0) {$\Lambda$};
	\end{pgfonlayer}
	\begin{pgfonlayer}{edgelayer}
		\draw[cWire] (1.center) to (0.center);
	\end{pgfonlayer}
\end{tikzpicture}
}
\eeq
\beq
\InputIfFileExists{Diagrams/starIso3.tikz}{}{\input{./figures/Diagrams/starIso3.tikz}}\quad=\quad%
\begin{tikzpicture}
	\begin{pgfonlayer}{nodelayer}
		\node [style=none] (0) at (1, -0) {};
		\node [style=none] (1) at (-0.5, -0) {};
		\node [style={up label}] (2) at (.25, -0) {$\morph{\star}{\Lambda}$};
	\end{pgfonlayer}
	\begin{pgfonlayer}{edgelayer}
		\draw[cWire] (1.center) to (0.center);
	\end{pgfonlayer}
\end{tikzpicture}
}
\eeq

\beq
\InputIfFileExists{Diagrams/grayTriangleIso1.tikz}{}{\input{./figures/Diagrams/grayTriangleIso1.tikz}}\quad=\quad%
\InputIfFileExists{Diagrams/grayTriangleIso2.tikz}{}{\input{./figures/Diagrams/grayTriangleIso2.tikz}}
\eeq
\beq
\InputIfFileExists{Diagrams/grayTriangleIso3.tikz}{}{\input{./figures/Diagrams/grayTriangleIso3.tikz}}\quad=\quad%
\begin{tikzpicture}
	\begin{pgfonlayer}{nodelayer}
		\node [style=none] (0) at (-1, -0) {};
		\node [style={up label}] (1) at (0, -0) {$X\times Y$};
		\node [style=none] (2) at (1.5, -0) {};
	\end{pgfonlayer}
	\begin{pgfonlayer}{edgelayer}
		\draw [cWire] (2.center) to (0.center);
	\end{pgfonlayer}
\end{tikzpicture}}
\eeq

\beq \label{ruleinf1}
\InputIfFileExists{Diagrams/ii28New.tikz}{}{\input{./figures/Diagrams/ii28New.tikz}}\quad=\quad%
\InputIfFileExists{Diagrams/ii29New.tikz}{}{\input{./figures/Diagrams/ii29New.tikz}}.
\eeq

\beq\label{ruleinf2}
\InputIfFileExists{Diagrams/dotdotblackdot.tikz}{}{\input{./figures/Diagrams/dotdotblackdot.tikz}}\quad=\quad%
\InputIfFileExists{Diagrams/dotdotblackdot2.tikz}{}{\input{./figures/Diagrams/dotdotblackdot2.tikz}}.
\eeq

\section{Useful results in FI} \label{usefulidentitiesapp}

\subsection{A useful lemma about (sub)stochastic maps} \label{lemmaaboutpropns}

We now state and prove a lemma which was needed to justify Eq.~\eqref{usesproplemma}. 
\begin{lemma}\label{lem:CausalProposition}
\begin{align*}
\InputIfFileExists{Diagrams/CausalPropositionMap.tikz}{}{\input{./figures/Diagrams/CausalPropositionMap.tikz}}\ \  =\ \ %
\InputIfFileExists{Diagrams/CausalPropositionMap1.tikz}{}{\input{./figures/Diagrams/CausalPropositionMap1.tikz}}\\ \qquad  \iff \\
\qquad \qquad %
\InputIfFileExists{Diagrams/CausalPropositionMap3.tikz}{}{\input{./figures/Diagrams/CausalPropositionMap3.tikz}}\ \ =\ \ %
\InputIfFileExists{Diagrams/CausalPropositionMap2.tikz}{}{\input{./figures/Diagrams/CausalPropositionMap2.tikz}}
\end{align*}
where $\sigma$ is an arbitrary substochastic map. This result and proof also hold if one replaces all of the equalities with inferential equivalences.
\end{lemma}

\proof
The $\Leftarrow$ direction trivially follows from 
Eq.~\eqref{trueprop}. To prove the $\Rightarrow$ direction, we begin by assuming that
\beq \label{causpropmap1}
\InputIfFileExists{Diagrams/CausalPropositionMap.tikz}{}{\input{./figures/Diagrams/CausalPropositionMap.tikz}}\quad =\quad %
\InputIfFileExists{Diagrams/CausalPropositionMap1.tikz}{}{\input{./figures/Diagrams/CausalPropositionMap1.tikz}}
\eeq
Composing this with a state preparation generator and the star isomorphism gives
\beq \label{causpropmap2}
\InputIfFileExists{Diagrams/CausalPropProof1.tikz}{}{\input{./figures/Diagrams/CausalPropProof1.tikz}}\quad =\quad %
\InputIfFileExists{Diagrams/CausalPropProof4.tikz}{}{\input{./figures/Diagrams/CausalPropProof4.tikz}}
\eeq
which can be rewritten using Eq.~\eqref{blackdotstar} to
\beq
\InputIfFileExists{Diagrams/CausalPropProof2.tikz}{}{\input{./figures/Diagrams/CausalPropProof2.tikz}}\quad =\quad %
\InputIfFileExists{Diagrams/CausalPropProof4.tikz}{}{\input{./figures/Diagrams/CausalPropProof4.tikz}}
\eeq
and then, using Eq.~\eqref{eq:CausalOntRep} and Eq.~\eqref{copycounitinf}, to%
\beq\label{eq:CPPKeyStep}
\InputIfFileExists{Diagrams/CausalPropProof3.tikz}{}{\input{./figures/Diagrams/CausalPropProof3.tikz}}\quad =\quad %
\InputIfFileExists{Diagrams/CausalPropProof5.tikz}{}{\input{./figures/Diagrams/CausalPropProof5.tikz}}
\eeq
Finally, we use (on the LHS) the fact that the star is an isomorphism and (on the RHS) the fact that the star is stochastic to obtain the result:
\beq
\InputIfFileExists{Diagrams/CausalPropositionMap3.tikz}{}{\input{./figures/Diagrams/CausalPropositionMap3.tikz}}
 \quad = \quad  %
\InputIfFileExists{Diagrams/CausalPropositionMap2.tikz}{}{\input{./figures/Diagrams/CausalPropositionMap2.tikz}}.
\eeq

The proof proceeds in the same way if one replaces all equalities with inferential equivalences. If one assumes Eq.~\eqref{causpropmap1} but with the equality replaced by inferential equivalence, then Eq.~\eqref{causpropmap2} follows by the fact that inferential equivalence is preserved by composition. The remainder of the proof then follows by the same rewrite rules.

\endproof

\subsection{Copy function and complete common causes}
First, let us define the {\em copy function} in $\Func$, denoted
\beq \label{copyFn}
\InputIfFileExists{Diagrams/copyFunctionNew.tikz}{}{\input{./figures/Diagrams/copyFunctionNew.tikz}},
\eeq
by $\bullet(\lambda)=(\lambda,\lambda)$ for all $\lambda\in\Lambda$.

We now show how some useful properties of the copy function can be lifted to define a corresponding process in $\FI$, via
\beq \label{copyfuncincl}
\InputIfFileExists{Diagrams/copyFunctionInclusion.tikz}{}{\input{./figures/Diagrams/copyFunctionInclusion.tikz}}.
\eeq
Specifically, that this acts as a suitable copy operation for $\FI$, that is, it is symmetric
\beq
\InputIfFileExists{Diagrams/CopySymmetry.tikz}{}{\input{./figures/Diagrams/CopySymmetry.tikz}} = %
\InputIfFileExists{Diagrams/CopySymmetry1.tikz}{}{\input{./figures/Diagrams/CopySymmetry1.tikz}} = %
\InputIfFileExists{Diagrams/CopySymmetry2.tikz}{}{\input{./figures/Diagrams/CopySymmetry2.tikz}}
\eeq
and associative,
\beq
\InputIfFileExists{Diagrams/CopyAssociativity1.tikz}{}{\input{./figures/Diagrams/CopyAssociativity1.tikz}}=%
\InputIfFileExists{Diagrams/CopyAssociativity2.tikz}{}{\input{./figures/Diagrams/CopyAssociativity2.tikz}}=%
\InputIfFileExists{Diagrams/CopyAssociativity3.tikz}{}{\input{./figures/Diagrams/CopyAssociativity3.tikz}}=%
\InputIfFileExists{Diagrams/CopyAssociativity4.tikz}{}{\input{./figures/Diagrams/CopyAssociativity4.tikz}},
\eeq
as follows immediately from diagram preservation of $\mathbf{e'}$ and the associativity of the underlying function.
By Eq.~\eqref{ignoreempty},
the embedding of the unique function $u$ to the trivial system $\star$ is equal to the ignoring map:
\beq \label{eqcopyunit2}
\InputIfFileExists{Diagrams/CopyUnit1.tikz}{}{\input{./figures/Diagrams/CopyUnit1.tikz}}\ = \ %
\begin{tikzpicture}
	\begin{pgfonlayer}{nodelayer}
		\node [style=epiCopoint] (0) at (0, -0) {};
		\node [style=none] (1) at (0, -1.75) {};
		\node [style=infpoint] (2) at (-1, -0) {$[u]$};
	\end{pgfonlayer}
	\begin{pgfonlayer}{edgelayer}
		\draw [CcWire] (2) to (0);
		\draw [oWire] (1.center) to (0);
	\end{pgfonlayer}
\end{tikzpicture}} \ = \ %
\InputIfFileExists{Diagrams/ignoreConstraintNewFI2.tikz}{}{\input{./figures/Diagrams/ignoreConstraintNewFI2.tikz}}\ = \ %
\begin{tikzpicture}
	\begin{pgfonlayer}{nodelayer}
		\node [style=ignore] (0) at (0, -0) {};
		\node [style={small black dot}] (1) at (0, -0) {};
		\node [style=none] (2) at (0, -1) {};
		\node [style=none] (3) at (0.5, -0) {};
	\end{pgfonlayer}
	\begin{pgfonlayer}{edgelayer}
		\draw [oWire] (1) to (2.center);
	\end{pgfonlayer}
\end{tikzpicture}
},
\eeq
one can also see that it is a counit for the copy:
\beq
\InputIfFileExists{Diagrams/CopyUnit8.tikz}{}{\input{./figures/Diagrams/CopyUnit8.tikz}}=%
\InputIfFileExists{Diagrams/CopyUnit3.tikz}{}{\input{./figures/Diagrams/CopyUnit3.tikz}}=%
\InputIfFileExists{Diagrams/CopyUnit6.tikz}{}{\input{./figures/Diagrams/CopyUnit6.tikz}}=%
\begin{tikzpicture}
	\begin{pgfonlayer}{nodelayer}
		\node [style=none] (0) at (0, 2.5) {};
		\node [style=none] (1) at (0, -1.5) {};
	\end{pgfonlayer}
	\begin{pgfonlayer}{edgelayer}
		\draw [oWire] (0.center) to (1.center);
	\end{pgfonlayer}
\end{tikzpicture}
}=%
\InputIfFileExists{Diagrams/CopyUnit5.tikz}{}{\input{./figures/Diagrams/CopyUnit5.tikz}}=%
\InputIfFileExists{Diagrams/CopyUnit4.tikz}{}{\input{./figures/Diagrams/CopyUnit4.tikz}}=%
\InputIfFileExists{Diagrams/CopyUnit9.tikz}{}{\input{./figures/Diagrams/CopyUnit9.tikz}}.
\eeq

We now show that processes of the form
\beq
\InputIfFileExists{Diagrams/CCC.tikz}{}{\input{./figures/Diagrams/CCC.tikz}}
\eeq
describe situations in which $\Lambda$ is the {\em complete} common cause~\cite{allen2017quantum} of $\Lambda'$ and $\Lambda''$.
This is a consequence of the fact that the outputs of a process in $\Func$ are (by construction) deterministic functions of the inputs, and hence in this diagram, $\Lambda$ is the only possible cause of $\Lambda'$ and $\Lambda''$. (To represent a scenario in which the two have more than one common cause, one would represent these explicitly as inputs to the process. Note that our framework does not incorporate a diagrammatic distinction between latent and observed variables, although introducing such a distinction might be useful in future work.)
\begin{lemma} \label{lem:CCC}
\beq%
\InputIfFileExists{Diagrams/CCC.tikz}{}{\input{./figures/Diagrams/CCC.tikz}} \quad = \quad %
\InputIfFileExists{Diagrams/CCC1.tikz}{}{\input{./figures/Diagrams/CCC1.tikz}}\eeq
where the black triangle is a stochastic map defined by linearity and
\beq
\InputIfFileExists{Diagrams/BlackTriangleStochasticMap.tikz}{}{\input{./figures/Diagrams/BlackTriangleStochasticMap.tikz}} :: \left[ %
\InputIfFileExists{Diagrams/FunctionOneTwo.tikz}{}{\input{./figures/Diagrams/FunctionOneTwo.tikz}}\right] \mapsto \left[%
\InputIfFileExists{Diagrams/FunctionOneTwoDecomposition2.tikz}{}{\input{./figures/Diagrams/FunctionOneTwoDecomposition2.tikz}} \right]\otimes \left[%
\InputIfFileExists{Diagrams/FunctionOneTwoDecomposition4.tikz}{}{\input{./figures/Diagrams/FunctionOneTwoDecomposition4.tikz}}\right]
\eeq
where $u$ is the unique function to $\star$.
\end{lemma}
\proof
First, note that we have a similar decomposition in $\Func$, that is for all functions $F$ there exist $f_l$ and $f_r$ such that
\beq \label{eqfunctiononetwo}
\InputIfFileExists{Diagrams/FunctionOneTwo.tikz}{}{\input{./figures/Diagrams/FunctionOneTwo.tikz}} \quad = \quad %
\InputIfFileExists{Diagrams/FunctionOneTwoDecomposition.tikz}{}{\input{./figures/Diagrams/FunctionOneTwoDecomposition.tikz}}
\eeq
where
\beq
\InputIfFileExists{Diagrams/FunctionOneTwoDecomposition1.tikz}{}{\input{./figures/Diagrams/FunctionOneTwoDecomposition1.tikz}} \ = \ %
\InputIfFileExists{Diagrams/FunctionOneTwoDecomposition2.tikz}{}{\input{./figures/Diagrams/FunctionOneTwoDecomposition2.tikz}}
\quad\text{and}\quad %
\InputIfFileExists{Diagrams/FunctionOneTwoDecomposition3.tikz}{}{\input{./figures/Diagrams/FunctionOneTwoDecomposition3.tikz}} \ = \ %
\InputIfFileExists{Diagrams/FunctionOneTwoDecomposition4.tikz}{}{\input{./figures/Diagrams/FunctionOneTwoDecomposition4.tikz}}.
\eeq
We will now show that this result can be lifted to $\FI$.
Note that the definition of the black-triangle stochastic map, together with the above choices for $f_r$ and $f_l$, imply that
\beq
\InputIfFileExists{Diagrams/BlackTriangleStochasticMap1.tikz}{}{\input{./figures/Diagrams/BlackTriangleStochasticMap1.tikz}}\quad = \quad %
\InputIfFileExists{Diagrams/BlackTriangleStochasticMap2.tikz}{}{\input{./figures/Diagrams/BlackTriangleStochasticMap2.tikz}}.
\eeq
Now consider the following set of rewrites, where $[\bullet]$ is the state of certain knowledge that the copy operation, $\bullet:\Lambda\to\Lambda\times\Lambda$, of Eq.~\eqref{copyFn} has occurred.

\begin{align}
\InputIfFileExists{Diagrams/CCC2.tikz}{}{\input{./figures/Diagrams/CCC2.tikz}} \quad &= \quad %
\InputIfFileExists{Diagrams/CCC3.tikz}{}{\input{./figures/Diagrams/CCC3.tikz}} \\
&= \quad %
\InputIfFileExists{Diagrams/CCC4.tikz}{}{\input{./figures/Diagrams/CCC4.tikz}} \\
&= \quad %
\InputIfFileExists{Diagrams/CCC5.tikz}{}{\input{./figures/Diagrams/CCC5.tikz}}. \label{eqhalfoffunny}
\end{align}
Next, note that
\beq \label{weirdidentity}
\InputIfFileExists{Diagrams/BlackTriangleInverse.tikz}{}{\input{./figures/Diagrams/BlackTriangleInverse.tikz}} \quad = \quad %
\InputIfFileExists{Diagrams/BlackTriangleInverse2.tikz}{}{\input{./figures/Diagrams/BlackTriangleInverse2.tikz}},
\eeq
as can be verified by computing its action on an arbitrary delta function state of knowledge $[F]$, namely
\begin{align}
\left[%
\InputIfFileExists{Diagrams/FunctionOneTwo.tikz}{}{\input{./figures/Diagrams/FunctionOneTwo.tikz}} \right] & \mapsto \left[%
\InputIfFileExists{Diagrams/FunctionOneTwoDecomposition2.tikz}{}{\input{./figures/Diagrams/FunctionOneTwoDecomposition2.tikz}} \right]\otimes \left[%
\InputIfFileExists{Diagrams/FunctionOneTwoDecomposition4.tikz}{}{\input{./figures/Diagrams/FunctionOneTwoDecomposition4.tikz}}\right] \\
& \mapsto \left[%
\InputIfFileExists{Diagrams/FunctionOneTwoDecomposition2.tikz}{}{\input{./figures/Diagrams/FunctionOneTwoDecomposition2.tikz}}%
\InputIfFileExists{Diagrams/FunctionOneTwoDecomposition4.tikz}{}{\input{./figures/Diagrams/FunctionOneTwoDecomposition4.tikz}}\right] \\
& \mapsto \left[%
\InputIfFileExists{Diagrams/FunctionOneTwoDecomposition.tikz}{}{\input{./figures/Diagrams/FunctionOneTwoDecomposition.tikz}} \right].
\end{align}
Hence, by Eq.~\eqref{eqfunctiononetwo}, we see that
\beq
\InputIfFileExists{Diagrams/BlackTriangleInverse.tikz}{}{\input{./figures/Diagrams/BlackTriangleInverse.tikz}} ::  \left[%
\InputIfFileExists{Diagrams/FunctionOneTwo.tikz}{}{\input{./figures/Diagrams/FunctionOneTwo.tikz}} \right]  \mapsto \left[%
\InputIfFileExists{Diagrams/FunctionOneTwo.tikz}{}{\input{./figures/Diagrams/FunctionOneTwo.tikz}} \right],
\eeq
justifying Eq.~\eqref{weirdidentity}.
The conjunction of Eq.~\eqref{weirdidentity} and Eq.~\eqref{eqhalfoffunny} immediately establishes the lemma.
\endproof

Next, we show that learning about an ontological system is the same as first copying that system and then learning about the copy:
\beq \label{eqB16}
\InputIfFileExists{Diagrams/learnLambda.tikz}{}{\input{./figures/Diagrams/learnLambda.tikz}}\quad=\quad%
\InputIfFileExists{Diagrams/learnCopy.tikz}{}{\input{./figures/Diagrams/learnCopy.tikz}}.
\eeq
\proof
We start with the RHS and will rewrite it into the LHS. In the following equalities, Eq.~\eqref{copyproof3} follows from Eq.~\eqref{Axiom:PropKnowGenerators}, Eq.~\eqref{copyproof4} follows from the fact that $[\bullet]$ is a point distribution, Eq.~\eqref{copyproof5} follows from Eq.~\eqref{copycounitinf}, and Eq.~\eqref{copyproof6} follows from Eq.~\eqref{eqcopyunit2}.
\begin{align}
\InputIfFileExists{Diagrams/learnCopy.tikz}{}{\input{./figures/Diagrams/learnCopy.tikz}}\quad &=\quad %
\InputIfFileExists{Diagrams/learnCopyProof.tikz}{}{\input{./figures/Diagrams/learnCopyProof.tikz}}
\end{align}
\begin{align}
\qquad &=\quad%
\InputIfFileExists{Diagrams/learnCopyProof2.tikz}{}{\input{./figures/Diagrams/learnCopyProof2.tikz}}\\
&=\quad%
\InputIfFileExists{Diagrams/learnCopyProof3.tikz}{}{\input{./figures/Diagrams/learnCopyProof3.tikz}}\label{copyproof3} \\
&=\quad%
\InputIfFileExists{Diagrams/learnCopyProof4.tikz}{}{\input{./figures/Diagrams/learnCopyProof4.tikz}} \label{copyproof4} \\
&=\quad %
\InputIfFileExists{Diagrams/learnCopyProof45.tikz}{}{\input{./figures/Diagrams/learnCopyProof45.tikz}}\\
&=\quad%
\InputIfFileExists{Diagrams/learnCopyProof5.tikz}{}{\input{./figures/Diagrams/learnCopyProof5.tikz}} \label{copyproof5} \\
&=\quad%
\InputIfFileExists{Diagrams/learnCopyProof6.tikz}{}{\input{./figures/Diagrams/learnCopyProof6.tikz}} \label{copyproof6} \\
&=\quad%
\InputIfFileExists{Diagrams/learnCopyProof7.tikz}{}{\input{./figures/Diagrams/learnCopyProof7.tikz}}\\
&=\quad%
\InputIfFileExists{Diagrams/learnCopyProof8.tikz}{}{\input{./figures/Diagrams/learnCopyProof8.tikz}}\\
&=\quad %
\begin{tikzpicture}
	\begin{pgfonlayer}{nodelayer}
		\node [style=none] (1) at (0, -1) {};
		\node [style=none] (10) at (0, 1) {};
		\node [style=clear dot] (12) at (0, 0) {};
		\node [style=none] (13) at (1, 0) {};
	\end{pgfonlayer}
	\begin{pgfonlayer}{edgelayer}
		\draw [oWire] (10.center) to (1.center);
		\draw [cWire] (12) to (13.center);
	\end{pgfonlayer}
\end{tikzpicture}
}
\end{align}
\endproof

\subsection{Other equalities}

A special case of Eq.~\eqref{Axiom:PropKnowGenerators} is
\beq \label{infidentcausal}
\InputIfFileExists{Diagrams/inferentialIdentityNew.tikz}{}{\input{./figures/Diagrams/inferentialIdentityNew.tikz}}\quad = \quad %
\InputIfFileExists{Diagrams/inferentialIdentityCausal.tikz}{}{\input{./figures/Diagrams/inferentialIdentityCausal.tikz}}
\eeq
since
\begin{align}
\InputIfFileExists{Diagrams/inferentialIdentityCausal.tikz}{}{\input{./figures/Diagrams/inferentialIdentityCausal.tikz}}\ &= \ %
\InputIfFileExists{Diagrams/infIdentCaus1.tikz}{}{\input{./figures/Diagrams/infIdentCaus1.tikz}} \\
 \ &= \ %
\InputIfFileExists{Diagrams/infIdentCaus2.tikz}{}{\input{./figures/Diagrams/infIdentCaus2.tikz}} \\
  \ &= \ %
\InputIfFileExists{Diagrams/inferentialIdentityNew.tikz}{}{\input{./figures/Diagrams/inferentialIdentityNew.tikz}}.
\end{align}

First, we show that one can always find at least one possible causal explanation in FI for any inference.
A simple example of this is
\beq \label{infidentcausal}
\begin{tikzpicture}
	\begin{pgfonlayer}{nodelayer}
		\node [style=none] (0) at (-0.75, 0) {};
		\node [style=none] (1) at (0.75, 0) {};
	\end{pgfonlayer}
	\begin{pgfonlayer}{edgelayer}
		\draw [cWire] (0.center) to (1.center);
	\end{pgfonlayer}
\end{tikzpicture}
}\quad = \quad %
\InputIfFileExists{Diagrams/inferentialIdentityCausalNew.tikz}{}{\input{./figures/Diagrams/inferentialIdentityCausalNew.tikz}};
\eeq
here, the inference described by the identity function is seen to have a possible causal explanation as the statement that a causal system has not evolved.
As another simple example, inferences described by functions can always arise from a causal system evolving under that function as its dynamics, e.g. as
\beq
\begin{tikzpicture}
	\begin{pgfonlayer}{nodelayer}
		\node [style=small box] (0) at (0, 0) {$f$};
		\node [style=none] (1) at (-1.5, 0) {};
		\node [style=none] (2) at (1.5, 0) {};
	\end{pgfonlayer}
	\begin{pgfonlayer}{edgelayer}
		\draw [cWire] (1.center) to (0);
		\draw [cWire] (0) to (2.center);
	\end{pgfonlayer}
\end{tikzpicture}
}\quad = \quad %
\InputIfFileExists{Diagrams/funcInferenceCaus.tikz}{}{\input{./figures/Diagrams/funcInferenceCaus.tikz}}.
\eeq
Most generally, an inference described by a general substochastic map can always be viewed as having partial knowledge about the input to some functional causal dynamics and considering a proposition about part of the output of the dynamics. That is, an arbitrary process $S\in \Inf$ satisfies
\begin{align}
\InputIfFileExists{Diagrams/substochInference.tikz}{}{\input{./figures/Diagrams/substochInference.tikz}}\quad &= \quad  %
\InputIfFileExists{Diagrams/substochDilation.tikz}{}{\input{./figures/Diagrams/substochDilation.tikz}}\\ \label{eqb8}
&= \quad  %
\InputIfFileExists{Diagrams/substochDilationCaus.tikz}{}{\input{./figures/Diagrams/substochDilationCaus.tikz}}
\end{align}
where $\sigma$ is a probability distribution, $f$ is a function, and $\pi$ is a propositional effect.

\proof
The proof is as follows, where  Eq.~\eqref{trickybit1} follows from Eq.~\eqref{ruleinf1} and Eq.~\eqref{trickybit2} follows from Eq.~\eqref{ruleinf2} (and the remaining equalities follow from the rewrite rules in $\FI$ that we have introduced):
\begin{align}
\InputIfFileExists{Diagrams/substochDilationCaus.tikz}{}{\input{./figures/Diagrams/substochDilationCaus.tikz}}\quad &= \quad %
\InputIfFileExists{Diagrams/substochDilationCaus1.tikz}{}{\input{./figures/Diagrams/substochDilationCaus1.tikz}}
\end{align}
\begin{align}
\qquad &= \quad %
\InputIfFileExists{Diagrams/substochDilationCaus2.tikz}{}{\input{./figures/Diagrams/substochDilationCaus2.tikz}} \\
&= \quad %
\InputIfFileExists{Diagrams/substochDilationCaus3.tikz}{}{\input{./figures/Diagrams/substochDilationCaus3.tikz}} \\
&= \quad %
\InputIfFileExists{Diagrams/substochDilationCaus4.tikz}{}{\input{./figures/Diagrams/substochDilationCaus4.tikz}} \\
&= \quad %
\InputIfFileExists{Diagrams/substochDilationCaus5.tikz}{}{\input{./figures/Diagrams/substochDilationCaus5.tikz}} \\
&= \quad %
\InputIfFileExists{Diagrams/substochDilationCaus6.tikz}{}{\input{./figures/Diagrams/substochDilationCaus6.tikz}} \\
&= \quad %
\InputIfFileExists{Diagrams/substochDilationCaus7.tikz}{}{\input{./figures/Diagrams/substochDilationCaus7.tikz}} \label{trickybit1}\\
&= \quad %
\InputIfFileExists{Diagrams/substochDilationCaus8.tikz}{}{\input{./figures/Diagrams/substochDilationCaus8.tikz}} \label{trickybit2}\\
&= \quad %
\InputIfFileExists{Diagrams/substochDilationCaus9.tikz}{}{\input{./figures/Diagrams/substochDilationCaus9.tikz}} \\
&= \quad %
\InputIfFileExists{Diagrams/substochDilation.tikz}{}{\input{./figures/Diagrams/substochDilation.tikz}}
\end{align}
 \endproof

Next, we show that one can always replace what we know about a transformation with what we know about a variable that controls the transformation. First let us describe a ``universal control'' function $%
\begin{tikzpicture}
	\begin{pgfonlayer}{nodelayer}
		\node [style=uControl] (1000) at (0, -0) {};
	\end{pgfonlayer}
\end{tikzpicture}
}\in \textsc{Func}$, as follows:
\beq
\forall f\in \morph{\Lambda}{\Lambda'} \qquad %
\InputIfFileExists{Diagrams/universalControl.tikz}{}{\input{./figures/Diagrams/universalControl.tikz}}\quad=\quad%
\InputIfFileExists{Diagrams/universalControl1.tikz}{}{\input{./figures/Diagrams/universalControl1.tikz}}
\eeq
where the black diamond represents the universal control transformation $%
\begin{tikzpicture}
	\begin{pgfonlayer}{nodelayer}
		\node [style=uControl] (1000) at (0, -0) {};
	\end{pgfonlayer}
\end{tikzpicture}
}$, and where we have introduced a causal system which ranges over the (finite) set of functions from $\Lambda$ to $\Lambda'$.
Then, one can show
\beq\label{eq:uniControl}
\InputIfFileExists{Diagrams/uniControl1.tikz}{}{\input{./figures/Diagrams/uniControl1.tikz}}\quad=\quad%
\InputIfFileExists{Diagrams/uniControl.tikz}{}{\input{./figures/Diagrams/uniControl.tikz}}
\eeq
\proof One can rewrite the RHS into the LHS, as follows:
\begin{align}
\InputIfFileExists{Diagrams/uniControl.tikz}{}{\input{./figures/Diagrams/uniControl.tikz}}\quad &=\quad %
\InputIfFileExists{Diagrams/uniControl2.tikz}{}{\input{./figures/Diagrams/uniControl2.tikz}}  \\
&=\quad%
\InputIfFileExists{Diagrams/uniControl3.tikz}{}{\input{./figures/Diagrams/uniControl3.tikz}}  \\
&=\quad%
\InputIfFileExists{Diagrams/uniControl4.tikz}{}{\input{./figures/Diagrams/uniControl4.tikz}}  \\
&=\quad%
\InputIfFileExists{Diagrams/uniControl1.tikz}{}{\input{./figures/Diagrams/uniControl1.tikz}}  \\
\end{align}
where the last step follows from the fact that
\beq
\begin{tikzpicture}
	\begin{pgfonlayer}{nodelayer}
		\node [style=parComp] (0) at (0.25, -0.75) {};
		\node [style=infpoint] (1) at (0.25, 1.25) {$[%
\begin{tikzpicture}
	\begin{pgfonlayer}{nodelayer}
		\node [style=uControl] (1000) at (0, -0) {};
	\end{pgfonlayer}
\end{tikzpicture}
}]$};
		\node [style=infpoint] (2) at (-1, -1.5) {$[\mathds{1}]$};
		\node [style=none] (3) at (-1.25, 0) {};
		\node [style=none] (4) at (2.25, 0.25) {};
		\node [style=seqComp] (5) at (1.5, 0.25) {};
		\node [style=none] (6) at (-2, 0) {};
		\node [style=addStar] (7) at (-1.25, 0) {};
	\end{pgfonlayer}
	\begin{pgfonlayer}{edgelayer}
		\draw [cWire, in=135, out=0, looseness=1.00] (3.center) to (0);
		\draw [cWire, in=0, out=-135, looseness=1.00] (0) to (2);
		\draw [cWire, in=120, out=0, looseness=1.00] (1) to (5);
		\draw [cWire, in=-120, out=0, looseness=1.00] (0) to (5);
		\draw [cWire] (5) to (4.center);
		\draw [cWire] (6.center) to (3.center);
	\end{pgfonlayer}
\end{tikzpicture}
\quad =\quad
\InputIfFileExists{Diagrams/Identity3.tikz}{}{\input{./figures/Diagrams/Identity3.tikz}}
\eeq
as can be verified by computing its action on an arbitrary delta function state of knowledge $[f]$. Namely,
\begin{align}
\left[%
\InputIfFileExists{Diagrams/stochCalcProof1.tikz}{}{\input{./figures/Diagrams/stochCalcProof1.tikz}}\right]\quad &\mapsto\quad \left[%
\begin{tikzpicture}
	\begin{pgfonlayer}{nodelayer}
		\node [style=point] (0) at (0, -0.25) {$f$};
		\node [style=none] (1) at (0, .75) {};
		\node [style={right label}] (2) at (0, .75) {$\morph{\Lambda}{\Lambda'}$};
	\end{pgfonlayer}
	\begin{pgfonlayer}{edgelayer}
		\draw [oWire] (1.center) to (0);
	\end{pgfonlayer}
\end{tikzpicture}
}\right] \\
&\mapsto\quad \left[%
\InputIfFileExists{Diagrams/stochCalcProof3.tikz}{}{\input{./figures/Diagrams/stochCalcProof3.tikz}}\right] \\
&\mapsto\quad \left[%
\InputIfFileExists{Diagrams/stochCalcProof4.tikz}{}{\input{./figures/Diagrams/stochCalcProof4.tikz}}\right] \\
&=\quad \left[%
\InputIfFileExists{Diagrams/stochCalcProof1.tikz}{}{\input{./figures/Diagrams/stochCalcProof1.tikz}}\right],
\end{align}
where the final equality is given by Eq.~\eqref{eq:uniControl}.
\endproof

\onecolumngrid
\subsection{Proof of normal form for $\FI$} \label{App:NF}

We now prove Theorem~\ref{thmnormalform}; namely, the normal form
\beq
\InputIfFileExists{Diagrams/NF1.tikz}{}{\input{./figures/Diagrams/NF1.tikz}}
\eeq
for $\FI$, where $S$ is a substochastic map.

\proof

We will prove this by induction. First, we show (Step i) that every generator can be written into normal form. Second, we prove (Step ii) that the composite of two normal form diagrams can be rewritten into normal form. Given these, it is clear that one can write any diagram into normal form by first rewriting all of the generators involved into normal form using Step i, composing these according to Step ii, and iterating until the entire diagram is in normal form.

{\bf Step i}---The fact that each generator is in normal form can be seen by inspection. For example, stochastic maps are generators in our theory, and are already in normal form---namely, the special case that arises when one takes all the causal systems in the normal form to be trivial. The other three generators (describing interactions between the causal and inferential systems) can be written in normal form as follows:
\beq
\begin{tikzpicture}
	\begin{pgfonlayer}{nodelayer}
		\node [style=none] (0) at (0, -0.5000001) {};
		\node [style=small black dot] (1) at (0, 0.7500002) {};
		\node [style=ignore] (3) at (0, 0.7500002) {};
	\end{pgfonlayer}
	\begin{pgfonlayer}{edgelayer}
		\draw [oWire] (0.center) to (1.center);
	\end{pgfonlayer}
\end{tikzpicture}
}\ =\ %
\InputIfFileExists{Diagrams/NFG3.tikz}{}{\input{./figures/Diagrams/NFG3.tikz}}\ ,
\quad
\InputIfFileExists{Diagrams/interactionEpistemicFunctional.tikz}{}{\input{./figures/Diagrams/interactionEpistemicFunctional.tikz}}\ =\ %
\InputIfFileExists{Diagrams/NFG1.tikz}{}{\input{./figures/Diagrams/NFG1.tikz}}
\ \text{and}\quad
\InputIfFileExists{Diagrams/FuncPropInteractionNew.tikz}{}{\input{./figures/Diagrams/FuncPropInteractionNew.tikz}}\ =\ %
\InputIfFileExists{Diagrams/NFG2.tikz}{}{\input{./figures/Diagrams/NFG2.tikz}}.
\eeq

{\bf Step ii}---First, we write down the most general way to compose two diagrams and then expand each of these in terms of the conjectured normal form:
\beq
\begin{tikzpicture}
	\begin{pgfonlayer}{nodelayer}
		\node [style=none] (0) at (-0.5, 0.5) {};
		\node [style=none] (1) at (-0.5, -0.5) {};
		\node [style=none] (2) at (0.5, -0.5) {};
		\node [style=none] (3) at (0.5, 0.5) {};
		\node [style=none] (4) at (0, 0) {$\mathcal{D}_1$};
		\node [style=none] (5) at (0.25, 0.5) {};
		\node [style=none] (6) at (2.25, 1.75) {};
		\node [style=none] (7) at (0, -0.5) {};
		\node [style=none] (8) at (0, -0.75) {};
		\node [style=none] (9) at (0.5, -0.25) {};
		\node [style=none] (10) at (3.25, -0.25) {};
		\node [style=none] (15) at (-0.5, 0) {};
		\node [style=none] (16) at (-0.75, 0) {};
		\node [style=none] (18) at (2.5, 3) {};
		\node [style=none] (19) at (2, 1.75) {};
		\node [style=none] (20) at (2.5, 2.25) {$\mathcal{D}_2$};
		\node [style=none] (21) at (3, 2.25) {};
		\node [style=none] (23) at (3, 2.75) {};
		\node [style=none] (24) at (2, 2.5) {};
		\node [style=none] (25) at (3.25, 2.25) {};
		\node [style=none] (26) at (2, 2.75) {};
		\node [style=none] (27) at (3, 1.75) {};
		\node [style=none] (28) at (-0.75, 2.5) {};
		\node [style=none] (29) at (2.5, 2.75) {};
		\node [style=none] (30) at (-0.25, 2.75) {};
		\node [style=none] (31) at (-0.25, 0.5) {};
		\node [style=none] (32) at (2.75, 1.75) {};
		\node [style=none] (33) at (2.75, -0.5) {};
		\node [style=none] (34) at (2, 2) {};
		\node [style=none] (35) at (0.5, 0.25) {};
		\node [style=none] (37) at (0.5, -0.25) {};
		\node [style=none] (42) at (-2, 0.25) {};
		\node [style=none] (43) at (-0.75, 2.5) {};
		\node [style=none] (45) at (-2, 0) {};
		\node [style=none] (47) at (-0.75, 0) {};
		\node [style=none] (48) at (4.5, 2.25) {};
		\node [style=none] (50) at (3.25, 2.25) {};
		\node [style=none] (52) at (3.25, -0.25) {};
		\node [style=none] (53) at (4.5, 2) {};
		\node [style=none] (56) at (0, -0.75) {};
		\node [style=none] (57) at (0, -2.25) {};
		\node [style=none] (58) at (2.75, -0.5) {};
		\node [style=none] (59) at (0.25, -2.25) {};
		\node [style=none] (60) at (2.5, 4.5) {};
		\node [style=none] (61) at (2.5, 3) {};
		\node [style=none] (62) at (2.25, 4.5) {};
		\node [style=none] (63) at (-0.25, 2.75) {};
	\end{pgfonlayer}
	\begin{pgfonlayer}{edgelayer}
		\draw [oWire, in=-90, out=90] (5.center) to (6.center);
		\draw [oWire] (7.center) to (8.center);
		\draw [RcWire] (9.center) to (10.center);
		\draw [RcWire] (15.center) to (16.center);
		\draw (0.center) to (3.center);
		\draw (3.center) to (2.center);
		\draw (2.center) to (1.center);
		\draw (1.center) to (0.center);
		\draw [oWire] (29.center) to (18.center);
		\draw [RcWire] (21.center) to (25.center);
		\draw [RcWire] (24.center) to (28.center);
		\draw (26.center) to (23.center);
		\draw (23.center) to (27.center);
		\draw (27.center) to (19.center);
		\draw (19.center) to (26.center);
		\draw [oWire, in=-90, out=90] (31.center) to (30.center);
		\draw [oWire, in=-90, out=90] (33.center) to (32.center);
		\draw [RcWire, in=180, out=0] (35.center) to (34.center);
		\draw [RcWire, in=0, out=180, looseness=0.75] (43.center) to (42.center);
		\draw [RcWire] (47.center) to (45.center);
		\draw [RcWire] (50.center) to (48.center);
		\draw [RcWire, in=180, out=0, looseness=0.75] (52.center) to (53.center);
		\draw [oWire, in=-90, out=90] (57.center) to (56.center);
		\draw [oWire, in=-90, out=90] (59.center) to (58.center);
		\draw [oWire] (61.center) to (60.center);
		\draw [oWire, in=-90, out=90, looseness=0.75] (63.center) to (62.center);
	\end{pgfonlayer}
\end{tikzpicture}.
\quad=\quad %
\InputIfFileExists{Diagrams/NF2.tikz}{}{\input{./figures/Diagrams/NF2.tikz}}.
\eeq
Removing the dashed gray lines and simply moving the wires around gives
\beq
=\quad %
\InputIfFileExists{Diagrams/NF3.tikz}{}{\input{./figures/Diagrams/NF3.tikz}}
\eeq
Next, we use Eq.~\eqref{trueprop} and Eq.~\eqref{identityembedding} to add in extra processes to obtain
\beq
=\quad %
\InputIfFileExists{Diagrams/NF4.tikz}{}{\input{./figures/Diagrams/NF4.tikz}}
\eeq
Merging some of these together, one obtains
\beq
=\quad %
\InputIfFileExists{Diagrams/NF5.tikz}{}{\input{./figures/Diagrams/NF5.tikz}}
\eeq
Next, simply moving wires around yields
\beq
=\quad %
\InputIfFileExists{Diagrams/NF6.tikz}{}{\input{./figures/Diagrams/NF6.tikz}},
\eeq
at which point one can identify the two gray dashed boxes as stochastic maps (since these contain only normalized inferential processes). Denoting these $S_1'$ and $S_2'$ one obtains
\beq
=:\quad %
\InputIfFileExists{Diagrams/NF7.tikz}{}{\input{./figures/Diagrams/NF7.tikz}},
\eeq
We can use Eq.~\eqref{Axiom:PropKnowGenerators} to rewrite this as
\beq
=\quad %
\InputIfFileExists{Diagrams/NF11.tikz}{}{\input{./figures/Diagrams/NF11.tikz}} =\quad %
\InputIfFileExists{Diagrams/NF12.tikz}{}{\input{./figures/Diagrams/NF12.tikz}}.
\eeq
Rewriting to express compositional structure within one's inferences, one gets
\beq
=\quad %
\InputIfFileExists{Diagrams/NF14.tikz}{}{\input{./figures/Diagrams/NF14.tikz}}\quad =: \quad %
\InputIfFileExists{Diagrams/NF15.tikz}{}{\input{./figures/Diagrams/NF15.tikz}}
\eeq
where one has identified the process in the dashed box as a stochastic map $S$.
Noting then that each pair of wires can be considered as a single composite wire and that $S$ is in the image of $\mathbf{i'}$, this is indeed seen to be in the claimed normal form.
\endproof

\twocolumngrid

\section{Useful results in $\widetilde{\FI}$}
\subsection{Statement and Proof of Lemma~\ref{eq:newrewrite}} \label{prflemnewrewrite}

We now state and prove a lemma used in the main text.
\begin{lemma}\label{eq:newrewrite}
A causal identity is inferentially equivalent to a process which factors through an inferential system as
\beq \label{identityOnticInf}
\begin{tikzpicture}
	\begin{pgfonlayer}{nodelayer}
		\node [style=none] (0) at (0, 1) {};
		\node [style=none] (1) at (0, -1) {};
	\end{pgfonlayer}
	\begin{pgfonlayer}{edgelayer}
		\draw [oWire] (0.center) to (1.center);
	\end{pgfonlayer}
\end{tikzpicture}}\quad \sim_{\mathbf{p^*}} \quad %
\InputIfFileExists{Diagrams/identityOnticInf.tikz}{}{\input{./figures/Diagrams/identityOnticInf.tikz}}.
\eeq
 \end{lemma}

 \proof

By Lemma~\ref{simplerinfeq}, we can establish the inferential equivalence by showing the following:
\beq
  \forall \tau \quad
 \begin{tikzpicture}
	\begin{pgfonlayer}{nodelayer}
		\node [style=none] (0) at (-0.5, -0.5) {};
		\node [style=none] (1) at (0.5, -0.5) {};
		\node [style=none] (2) at (0.5, 0.5) {};
		\node [style=none] (3) at (-1, -1.25) {};
		\node [style=none] (4) at (-1, -2.25) {};
		\node [style=none] (5) at (1.5, -2.25) {};
		\node [style=none] (6) at (1.5, 1.5) {};
		\node [style=none] (7) at (-1.5, 1.5) {};
		\node [style=none] (8) at (-1.5, -0.5) {};
		\node [style=none] (9) at (-1.25, -0.5) {};
		\node [style=none] (10) at (-1.25, 1.25) {};
		\node [style=none] (11) at (1.25, 1.25) {};
		\node [style=none] (12) at (1.25, -1.25) {};
		\node [style=none] (13) at (-2.25, 2.5) {};
		\node [style=none] (14) at (2.25, 2.5) {};
		\node [style=none] (15) at (2.25, -3) {};
		\node [style=none] (16) at (-2.25, -3) {};
		\node [style=none] (17) at (0, -1.75) {$\tau$};
		\node [style=none] (18) at (0, 0.5) {};
		\node [style=none] (19) at (0, 1.25) {};
		\node [style=none] (20) at (0, 0.5) {};
		\node [style=none] (21) at (0, -1.25) {};
		\node [style=none] (22) at (1.25, 0) {};
		\node [style=none] (23) at (3, 0) {};
		\node [style=none] (24) at (1.5, 0) {};
		\node [style=none] (25) at (-3, 0) {};
		\node [style=none] (26) at (-1.5, 0) {};
		\node [style=none] (27) at (2, -2.75) {\footnotesize $\mathbf{p^*}$};
	\end{pgfonlayer}
	\begin{pgfonlayer}{edgelayer}
\filldraw [fill=darkred!30,draw=darkred!60] (14.center) to (15.center) to (16.center) to (13.center) to cycle;
		\filldraw [fill=white,draw] (4.center) to (5.center) to (6.center) to (7.center) to (8.center)to (9.center)to (10.center)to (11.center)to (12.center)to (3.center) to cycle;
		\draw [oWire] (18.center) to (19.center);
		\draw [oWire] (20.center) to (21.center);
	\end{pgfonlayer}
\end{tikzpicture}
\quad = \quad
\begin{tikzpicture}
	\begin{pgfonlayer}{nodelayer}
		\node [style=none] (0) at (-0.5, -1) {};
		\node [style=none] (1) at (0.75, -0.5) {};
		\node [style=none] (2) at (0.75, 1) {};
		\node [style=none] (3) at (-1, -1.25) {};
		\node [style=none] (4) at (-1, -2.25) {};
		\node [style=none] (5) at (1.5, -2.25) {};
		\node [style=none] (6) at (1.5, 1.5) {};
		\node [style=none] (7) at (-1.5, 1.5) {};
		\node [style=none] (8) at (-1.5, -0.5) {};
		\node [style=none] (9) at (-1.25, -0.5) {};
		\node [style=none] (10) at (-1.25, 1.25) {};
		\node [style=none] (11) at (1.25, 1.25) {};
		\node [style=none] (12) at (1.25, -1.25) {};
		\node [style=none] (13) at (-2.25, 2.5) {};
		\node [style=none] (14) at (2.25, 2.5) {};
		\node [style=none] (15) at (2.25, -3) {};
		\node [style=none] (16) at (-2.25, -3) {};
		\node [style=none] (17) at (0, -1.75) {$\tau$};
		\node [style=none] (18) at (0.75, 0.5) {};
		\node [style=none] (19) at (0, 1.25) {};
		\node [style=none] (20) at (-0.75, -0.5) {};
		\node [style=none] (21) at (0, -1.25) {};
		\node [style=none] (22) at (1.25, 0) {};
		\node [style=none] (23) at (3, 0) {};
		\node [style=none] (24) at (1.5, 0) {};
		\node [style=none] (25) at (-3, 0) {};
		\node [style=none] (26) at (-1.5, 0) {};
		\node [style=none] (27) at (2, -2.75) {\footnotesize $\mathbf{p^*}$};
		\node [style=none] (28) at (-0.75, 0.5) {};
		\node [style={small black dot}] (29) at (-0.75, 0.5) {};
		\node [style=ignore] (30) at (-0.75, 0.5) {};
		\node [style=epiPoint] (31) at (0.75, -0) {};
		\node [style={clear dot}] (32) at (-0.75, -0) {};
		\node [style=addStar] (33) at (0, 0) {};
	\end{pgfonlayer}
	\begin{pgfonlayer}{edgelayer}
\filldraw [fill=darkred!30,draw=darkred!60] (14.center) to (15.center) to (16.center) to (13.center) to cycle;
		\filldraw [fill=white,draw] (4.center) to (5.center) to (6.center) to (7.center) to (8.center)to (9.center)to (10.center)to (11.center)to (12.center)to (3.center) to cycle;
		\draw [oWire, in=-90, out=90, looseness=1.00] (18.center) to (19.center);
		\draw [oWire, in=90, out=-90, looseness=1.00] (20.center) to (21.center);
		\draw [oWire] (20.center) to (28.center);
		\draw [oWire] (18.center) to (31);
		\draw [cWire] (32) to (31);
	\end{pgfonlayer}
\end{tikzpicture}
\eeq
 Now, as follows from Section~\ref{sec:PTs}, these testers are shorthand notation for a diagram of the form:
 \beq
 \begin{tikzpicture}
	\begin{pgfonlayer}{nodelayer}
		\node [style=none] (0) at (-0.5, -1) {};
		\node [style=none] (1) at (0.5, -0.5) {};
		\node [style=none] (2) at (0.5, 1) {};
		\node [style=none] (3) at (-1, -1.25) {};
		\node [style=none] (4) at (-1, -2.25) {};
		\node [style=none] (5) at (1.5, -2.25) {};
		\node [style=none] (6) at (1.5, 1.5) {};
		\node [style=none] (7) at (-1.5, 1.5) {};
		\node [style=none] (8) at (-1.5, -0.5) {};
		\node [style=none] (9) at (-1.25, -0.5) {};
		\node [style=none] (10) at (-1.25, 1.25) {};
		\node [style=none] (11) at (1.25, 1.25) {};
		\node [style=none] (12) at (1.25, -1.25) {};
		\node [style=none] (13) at (0, -1.75) {$\tau$};
		\node [style=none] (14) at (0, 0.5) {};
		\node [style=none] (15) at (0, 1.25) {};
		\node [style=none] (16) at (0, -0.5) {};
		\node [style=none] (17) at (0, -1.25) {};
		\node [style=none] (18) at (1.25, 0) {};
	\end{pgfonlayer}
	\begin{pgfonlayer}{edgelayer}
		\draw [oWire, in=-90, out=90, looseness=1.00] (14.center) to (15.center);
		\draw [oWire, in=90, out=-90, looseness=1.00] (16.center) to (17.center);
		\draw (8.center) to (9.center);
		\draw (9.center) to (10.center);
		\draw (10.center) to (11.center);
		\draw (11.center) to (12.center);
		\draw (12.center) to (3.center);
		\draw (3.center) to (4.center);
		\draw (4.center) to (5.center);
		\draw (5.center) to (6.center);
		\draw (6.center) to (7.center);
		\draw (7.center) to (8.center);
	\end{pgfonlayer}
\end{tikzpicture}\quad=\quad
\begin{tikzpicture}
	\begin{pgfonlayer}{nodelayer}
		\node [style=none] (0) at (-3.5, -1.25) {};
		\node [style=none] (1) at (-3.5, -2.25) {};
		\node [style=none] (2) at (-1.5, -2.25) {};
		\node [style=none] (3) at (-1.5, -1.25) {};
		\node [style=none] (4) at (-1.5, -1.25) {};
		\node [style=none] (5) at (-2.5, -1.75) {$x_\tau$};
		\node [style=none] (6) at (-3, 0.25) {};
		\node [style=none] (7) at (0, 1.25) {};
		\node [style=none] (8) at (-3, -0.75) {};
		\node [style=none] (9) at (-3, -1.25) {};
		\node [style=none] (10) at (0.5, 1.75) {$y_\tau$};
		\node [style=none] (11) at (1.5, 1.25) {};
		\node [style=none] (12) at (1.5, 2.25) {};
		\node [style=none] (13) at (1.5, 2.25) {};
		\node [style=none] (14) at (-0.5, 2.25) {};
		\node [style=none] (15) at (-0.5, 1.25) {};
		\node [style=none] (16) at (-1.5, -1.75) {};
		\node [style=none] (17) at (-0.5, 1.75) {};
		\node [style=none] (18) at (-2, -1.25) {};
		\node [style=none] (19) at (1, 1.25) {};
	\end{pgfonlayer}
	\begin{pgfonlayer}{edgelayer}
		\draw [oWire, in=-90, out=90, looseness=1.00] (6.center) to (7.center);
		\draw [oWire, in=90, out=-90, looseness=0.50] (8.center) to (9.center);
		\draw (4.center) to (0.center);
		\draw (0.center) to (1.center);
		\draw (1.center) to (2.center);
		\draw (2.center) to (3.center);
		\draw (13.center) to (14.center);
		\draw (14.center) to (15.center);
		\draw (15.center) to (11.center);
		\draw (11.center) to (12.center);
		\draw [cWire, in=180, out=0, looseness=0.50] (16.center) to (17.center);
		\draw [oWire, in=-90, out=90, looseness=0.75] (18.center) to (19.center);
	\end{pgfonlayer}
\end{tikzpicture}
\eeq
By applying the normal form of Theorem~\ref{thmnormalform} to the special case of processes of the form of $x_{\tau}$ and of the form of $y_{\tau}$, we can write this tester explicitly as
   \beq
\InputIfFileExists{Diagrams/causalFITester.tikz}{}{\input{./figures/Diagrams/causalFITester.tikz}},
  \eeq
  and hence our goal is to prove the following equality:
 \beq
 \begin{tikzpicture}
	\begin{pgfonlayer}{nodelayer}
		\node [style=none] (0) at (-1.5, 1) {};
		\node [style=none] (1) at (-2.75, -0.5) {};
		\node [style=none] (2) at (-1.5, -2) {};
		\node [style=none] (3) at (-2, -0.5) {$s_{\tau}$};
		\node [style=none] (4) at (1.5, 2.5) {};
		\node [style=none] (5) at (2.75, 1.25) {};
		\node [style=none] (6) at (1.5, -0) {};
		\node [style=none] (7) at (2, 1.25) {$e_{\tau}$};
		\node [style=none] (8) at (-1.5, 0.5) {};
		\node [style=none] (9) at (1.5, 2) {};
		\node [style=epiPointWide] (10) at (0.25, -2.75) {};
		\node [style=none] (11) at (-1.5, -1.5) {};
		\node [style=none] (12) at (0, -2.5) {};
		\node [style=none] (13) at (0, -0.75) {};
		\node [style=none] (14) at (0, -0.75) {};
		\node [style=none] (15) at (0, 0.5) {};
		\node [style=none] (16) at (1, -0.75) {};
		\node [style=none] (17) at (0.5, -0) {};
		\node [style={wide clear dot}] (18) at (0.25, -0) {};
		\node [style=none] (19) at (1.5, 0.5) {};
		\node [style={small black dot}] (20) at (0, 0.5) {};
		\node [style={small black dot}] (21) at (0.5, 1) {};
		\node [style=ignore] (22) at (0, 0.5) {};
		\node [style=ignore] (23) at (0.5, 1) {};
		\node [style=none] (24) at (-0.5, -0.75) {};
		\node [style=none] (25) at (0.5, -0.75) {};
		\node [style=none] (26) at (-0.5, -1.75) {};
		\node [style=none] (27) at (1, -1.75) {};
		\node [style=none] (28) at (0.5, -2.5) {};
		\node [style=none] (29) at (0.5, -1.75) {};
		\node [style=none] (30) at (0.5, 1) {};
		\node [style=none] (31) at (-3, 1.25) {};
		\node [style=none] (32) at (-1, 1.25) {};
		\node [style=none] (33) at (-1, -2.5) {};
		\node [style=none] (34) at (-3, -2.5) {};
		\node [style=none] (35) at (1.25, 2.75) {};
		\node [style=none] (36) at (1.25, -0.25) {};
		\node [style=none] (37) at (3, -0.25) {};
		\node [style=none] (38) at (3, 2.75) {};
		\node [style=none] (39) at (-1.25, -2.25) {\footnotesize$\mathbf{i'}$};
		\node [style=none] (40) at (2.75, -0) {\footnotesize $\mathbf{i'}$};
		\node [style=none] (41) at (-3.25, 3) {};
		\node [style=none] (42) at (3.25, 3) {};
		\node [style=none] (43) at (3.25, -3.25) {};
		\node [style=none] (44) at (-3.25, -3.25) {};
		\node [style=none] (45) at (3, -3) {\footnotesize$\mathbf{p^*}$};
	\end{pgfonlayer}
	\begin{pgfonlayer}{edgelayer}
		\filldraw [fill=darkred!30,draw=darkred!60] (44.center) to (43.center) to (42.center) to (41.center) to cycle;
		\filldraw [fill=darkgreen!30,draw=darkgreen!60] (38.center) to (37.center) to (36.center) to (35.center) to cycle;
		\filldraw [fill=darkgreen!30,draw=darkgreen!60] (34.center) to (33.center) to (32.center) to (31.center) to cycle;
		\filldraw[fill=white,draw=black] (1.center) to (0.center) to (2.center) to cycle;
		\filldraw[fill=white,draw=black] (5.center) to (4.center) to (6.center) to cycle;
		\draw [cWire, in=180, out=0, looseness=1.00] (8.center) to (9.center);
		\draw [cWire, in=180, out=0, looseness=1.00] (11.center) to (10);
		\draw [oWire] (13.center) to (12.center);
		\draw [oWire] (15.center) to (14.center);
		\draw [oWire, in=90, out=-90, looseness=1.00] (17.center) to (16.center);
		\draw [style=cWire, in=180, out=0, looseness=1.00] (18) to (19.center);
		\draw [oWire] (16.center) to (27.center);
		\draw [oWire, in=90, out=-90, looseness=1.00] (27.center) to (28.center);
		\draw [thick gray dashed edge] (24.center) to (26.center);
		\draw [thick gray dashed edge] (26.center) to (29.center);
		\draw [thick gray dashed edge] (29.center) to (25.center);
		\draw [thick gray dashed edge] (25.center) to (24.center);
		\draw [oWire] (30.center) to (17.center);
	\end{pgfonlayer}
\end{tikzpicture}
\quad=\quad
\begin{tikzpicture}
	\begin{pgfonlayer}{nodelayer}
		\node [style=none] (0) at (-1.5, 1) {};
		\node [style=none] (1) at (-2.75, -0.5) {};
		\node [style=none] (2) at (-1.5, -2) {};
		\node [style=none] (3) at (-2, -0.5) {$s_{\tau}$};
		\node [style=none] (4) at (1.5, 2.5) {};
		\node [style=none] (5) at (2.75, 1.25) {};
		\node [style=none] (6) at (1.5, -0) {};
		\node [style=none] (7) at (2, 1.25) {$e_{\tau}$};
		\node [style=none] (8) at (-1.5, 0.5) {};
		\node [style=none] (9) at (1.5, 2) {};
		\node [style=epiPointWide] (10) at (0.7499998, -2.75) {};
		\node [style=none] (11) at (-0.25, -2.75) {};
		\node [style=none] (12) at (0.5000001, -2.75) {};
		\node [style=none] (13) at (-0.5000001, -2) {};
		\node [style=none] (14) at (0, -0.5) {};
		\node [style=none] (15) at (0, 0.75) {};
		\node [style=none] (16) at (1.5, -0.5) {};
		\node [style=none] (17) at (0.5, 0.25) {};
		\node [style={wide clear dot}] (18) at (0.25, 0.25) {};
		\node [style=none] (19) at (1.5, 0.5) {};
		\node [style={small black dot}] (20) at (0, 0.75) {};
		\node [style={small black dot}] (21) at (0.5, 1.25) {};
		\node [style=ignore] (22) at (0, 0.75) {};
		\node [style=ignore] (23) at (0.5, 1.25) {};
		\node [style=none] (24) at (-0.75, -0.5) {};
		\node [style=none] (25) at (1.25, -0.5) {};
		\node [style=none] (26) at (-0.75, -2) {};
		\node [style=none] (27) at (1.5, -2) {};
		\node [style=none] (28) at (1, -2.75) {};
		\node [style=none] (29) at (1.25, -2) {};
		\node [style=none] (30) at (0.5, 1.25) {};
		\node [style=none] (31) at (-3, 1.25) {};
		\node [style=none] (32) at (-1, 1.25) {};
		\node [style=none] (33) at (-1, -2.5) {};
		\node [style=none] (34) at (-3, -2.5) {};
		\node [style=none] (35) at (1.25, 2.75) {};
		\node [style=none] (36) at (1.25, -0.25) {};
		\node [style=none] (37) at (3, -0.25) {};
		\node [style=none] (38) at (3, 2.75) {};
		\node [style=none] (39) at (-1.25, -2.25) {\footnotesize$\mathbf{i'}$};
		\node [style=none] (40) at (2.75, -0) {\footnotesize $\mathbf{i'}$};
		\node [style=none] (41) at (-3.25, 3) {};
		\node [style=none] (42) at (3.25, 3) {};
		\node [style=none] (43) at (3.25, -3.25) {};
		\node [style=none] (44) at (-3.25, -3.25) {};
		\node [style=none] (45) at (3, -3) {\footnotesize$\mathbf{p^*}$};
		\node [style=none] (46) at (-0.5, -1) {};
		\node [style={small black dot}] (47) at (-0.5, -1) {};
		\node [style=ignore] (48) at (-0.5, -1) {};
		\node [style=epiPoint] (49) at (1, -1.5) {};
		\node [style={clear dot}] (50) at (-0.5000001, -1.5) {};
		\node [style=addStar] (51) at (0.2499997, -1.5) {};
		\node [style=none] (52) at (-1.5, -1.5) {};
	\end{pgfonlayer}
	\begin{pgfonlayer}{edgelayer}
		\filldraw [fill=darkred!30,draw=darkred!60] (44.center) to (43.center) to (42.center) to (41.center) to cycle;
		\filldraw [fill=darkgreen!30,draw=darkgreen!60] (38.center) to (37.center) to (36.center) to (35.center) to cycle;
		\filldraw [fill=darkgreen!30,draw=darkgreen!60] (34.center) to (33.center) to (32.center) to (31.center) to cycle;
		\filldraw[fill=white,draw=black] (1.center) to (0.center) to (2.center) to cycle;
		\filldraw[fill=white,draw=black] (5.center) to (4.center) to (6.center) to cycle;
		\draw [cWire, in=180, out=0, looseness=1.00] (8.center) to (9.center);
		\draw [cWire, in=180, out=0, looseness=1.00] (11.center) to (10);
		\draw [oWire, in=90, out=-90, looseness=1.00] (13.center) to (12.center);
		\draw [oWire] (15.center) to (14.center);
		\draw [oWire, in=90, out=-90, looseness=1.00] (17.center) to (16.center);
		\draw [style=cWire, in=180, out=0, looseness=1.00] (18) to (19.center);
		\draw [oWire] (16.center) to (27.center);
		\draw [oWire, in=90, out=-90, looseness=1.00] (27.center) to (28.center);
		\draw [thick gray dashed edge] (24.center) to (26.center);
		\draw [thick gray dashed edge] (26.center) to (29.center);
		\draw [thick gray dashed edge] (29.center) to (25.center);
		\draw [thick gray dashed edge] (25.center) to (24.center);
		\draw [oWire] (30.center) to (17.center);
		\draw [oWire, in=90, out=-90, looseness=1.00] (46.center) to (13.center);
		\draw [oWire, in=90, out=-90, looseness=1.25] (14.center) to (49);
		\draw [cWire] (50) to (49);
		\draw [cWire, in=180, out=0, looseness=1.00] (52.center) to (11.center);
	\end{pgfonlayer}
\end{tikzpicture}
 \eeq
 This equality follows immediately from the following set of rewrites (where the first and last equality follow from the special case of Lemma~\ref{lem:CCC} where $\Lambda$ is trivial):
  \begin{align}
\InputIfFileExists{Diagrams/identityOnticInfProof.tikz}{}{\input{./figures/Diagrams/identityOnticInfProof.tikz}}&=  %
\InputIfFileExists{Diagrams/identityOnticInfProof2.tikz}{}{\input{./figures/Diagrams/identityOnticInfProof2.tikz}}\\
  &=  %
\InputIfFileExists{Diagrams/identityOnticInfProof3.tikz}{}{\input{./figures/Diagrams/identityOnticInfProof3.tikz}}\\
  &=  %
\InputIfFileExists{Diagrams/identityOnticInfProof4.tikz}{}{\input{./figures/Diagrams/identityOnticInfProof4.tikz}}\\
  &=  %
\InputIfFileExists{Diagrams/identityOnticInfProof5.tikz}{}{\input{./figures/Diagrams/identityOnticInfProof5.tikz}}\\
  &=  %
\InputIfFileExists{Diagrams/identityOnticInfProof6.tikz}{}{\input{./figures/Diagrams/identityOnticInfProof6.tikz}}
 \end{align}
\endproof


\subsection{Proof of Lemma~\ref{prfinfstochastic}} \label{prfinfstoch}
We now prove Lemma~\ref{prfinfstochastic}, that two processes in $\widetilde{\FI}$ are inferentially equivalent if and only if they are associated with the same substochastic map.
\proof
The $\implies$ direction follows immediately from the definition of inferential equivalence and the fact that the following diagram is a tester:
\beq
\InputIfFileExists{Diagrams/exampleTester.tikz}{}{\input{./figures/Diagrams/exampleTester.tikz}}.
\eeq
To prove the $\Leftarrow$ direction,
one can apply the fact that $\mathbf{p^*}=\tilde{\mathbf{p}}^* \circ \sim_{\mathbf{p^*}}$ and then apply Lemma~\ref{eq:newrewrite} to show that
\begin{align}
\InputIfFileExists{Diagrams/sufficientTesterProof2.tikz}{}{\input{./figures/Diagrams/sufficientTesterProof2.tikz}} &= %
\InputIfFileExists{Diagrams/sufficientTesterProof3.tikz}{}{\input{./figures/Diagrams/sufficientTesterProof3.tikz}} \\
 &= %
\InputIfFileExists{Diagrams/sufficientTesterProof4.tikz}{}{\input{./figures/Diagrams/sufficientTesterProof4.tikz}} \\
 &= %
\InputIfFileExists{Diagrams/sufficientTesterProof5.tikz}{}{\input{./figures/Diagrams/sufficientTesterProof5.tikz}}\label{proof:sufficientTester1}
\end{align}
and similarly for $\mathcal{E}$.

Now, starting with the RHS of the implication in Eq.~\eqref{infstochimp},
\beq \label{RHSofimpl}
\InputIfFileExists{Diagrams/sufficientTester1.tikz}{}{\input{./figures/Diagrams/sufficientTester1.tikz}}=%
\InputIfFileExists{Diagrams/sufficientTester2.tikz}{}{\input{./figures/Diagrams/sufficientTester2.tikz}},
\eeq
we will derive the LHS using the result we just proved.
First, note that this equality implies that
\beq
\forall \mathcal{T} \quad %
\InputIfFileExists{Diagrams/sufficientTesterProof6.tikz}{}{\input{./figures/Diagrams/sufficientTesterProof6.tikz}} = %
\InputIfFileExists{Diagrams/sufficientTesterProof7.tikz}{}{\input{./figures/Diagrams/sufficientTesterProof7.tikz}}.
\eeq
 Diagram preservation then allows us to write this as:
\beq\label{proof:sufficientTester2}
\forall \mathcal{T} \quad %
\InputIfFileExists{Diagrams/sufficientTesterProof8.tikz}{}{\input{./figures/Diagrams/sufficientTesterProof8.tikz}} = %
\InputIfFileExists{Diagrams/sufficientTesterProof9.tikz}{}{\input{./figures/Diagrams/sufficientTesterProof9.tikz}}
\eeq
Now, consider a special class of testers, namely, those of the form:
\beq
\InputIfFileExists{Diagrams/sufficientTesterProof11.tikz}{}{\input{./figures/Diagrams/sufficientTesterProof11.tikz}} = %
\InputIfFileExists{Diagrams/sufficientTesterProof10.tikz}{}{\input{./figures/Diagrams/sufficientTesterProof10.tikz}}
\eeq
for any $\mathcal{T}'$. Condition \ref{proof:sufficientTester2} therefore implies that
\beq \forall \mathcal{T}'  \quad  %
\InputIfFileExists{Diagrams/sufficientTesterProof12.tikz}{}{\input{./figures/Diagrams/sufficientTesterProof12.tikz}} = %
\InputIfFileExists{Diagrams/sufficientTesterProof13.tikz}{}{\input{./figures/Diagrams/sufficientTesterProof13.tikz}}.
\eeq
Using Eq.~\eqref{proof:sufficientTester1}, this is equivalent to
\beq \forall \mathcal{T}' \quad %
\InputIfFileExists{Diagrams/sufficientTesterProof14.tikz}{}{\input{./figures/Diagrams/sufficientTesterProof14.tikz}} = %
\InputIfFileExists{Diagrams/sufficientTesterProof15.tikz}{}{\input{./figures/Diagrams/sufficientTesterProof15.tikz}}
\eeq
But this is just the definition of inferential equivalence:
\beq
\begin{tikzpicture}
	\begin{pgfonlayer}{nodelayer}
		\node [style=none] (0) at (-0.5, 0.5) {};
		\node [style=none] (1) at (-0.5, -0.5) {};
		\node [style=none] (2) at (0.5, -0.5) {};
		\node [style=none] (3) at (0.5, 0.5) {};
		\node [style=none] (4) at (0, 0) {$\mathcal{D}$};
		\node [style=none] (5) at (1, -0) {};
		\node [style=none] (6) at (0.5, 0) {};
		\node [style=none] (7) at (-0.5, 0) {};
		\node [style=none] (8) at (-1, -0) {};
		\node [style=none] (9) at (-1.25, 1.25) {};
		\node [style=none] (10) at (-1.25, -1.25) {};
		\node [style=none] (11) at (1.25, -1.25) {};
		\node [style=none] (12) at (1.25, 1.25) {};
		\node [style=none] (14) at (0, 0.5000001) {};
		\node [style=none] (15) at (0, 1) {};
		\node [style=none] (16) at (0, -0.5000001) {};
		\node [style=none] (17) at (0, -1) {};
	\end{pgfonlayer}
	\begin{pgfonlayer}{edgelayer}
\filldraw [fill=white,draw] (0.center) to (1.center) to (2.center) to (3.center) to cycle;
		\draw [CcWire] (5.center) to (6.center);
		\draw [CcWire] (7.center) to (8.center);
		\draw[oWire] (14.center) to (15.center);
		\draw[oWire] (17.center) to (16.center);
	\end{pgfonlayer}
\end{tikzpicture}
\sim_{\mathbf{p^*}}
\begin{tikzpicture}
	\begin{pgfonlayer}{nodelayer}
		\node [style=none] (0) at (-0.5, 0.5) {};
		\node [style=none] (1) at (-0.5, -0.5) {};
		\node [style=none] (2) at (0.5, -0.5) {};
		\node [style=none] (3) at (0.5, 0.5) {};
		\node [style=none] (4) at (0, 0) {$\mathcal{E}$};
		\node [style=none] (5) at (1, -0) {};
		\node [style=none] (6) at (0.5, 0) {};
		\node [style=none] (7) at (-0.5, 0) {};
		\node [style=none] (8) at (-1, -0) {};
		\node [style=none] (9) at (-1.25, 1.25) {};
		\node [style=none] (10) at (-1.25, -1.25) {};
		\node [style=none] (11) at (1.25, -1.25) {};
		\node [style=none] (12) at (1.25, 1.25) {};
		\node [style=none] (14) at (0, 0.5000001) {};
		\node [style=none] (15) at (0, 1) {};
		\node [style=none] (16) at (0, -0.5000001) {};
		\node [style=none] (17) at (0, -1) {};
	\end{pgfonlayer}
	\begin{pgfonlayer}{edgelayer}
\filldraw [fill=white,draw] (0.center) to (1.center) to (2.center) to (3.center) to cycle;
		\draw [CcWire] (5.center) to (6.center);
		\draw [CcWire] (7.center) to (8.center);
		\draw[oWire] (14.center) to (15.center);
		\draw[oWire] (17.center) to (16.center);
	\end{pgfonlayer}
\end{tikzpicture}
\eeq

\endproof

\subsection{Proof of Theorem~\ref{thm:QNF}} \label{appqnf}
We now prove Theorem~\ref{thm:QNF}, which immediately led to the normal form for $\widetilde{\FI}$ given in Corollary~\ref{normalformFI}.
In the equalities that follow, Eq.~\eqref{prfdetails2} follows from Eq.~\eqref{eqb8},  Eq.~\eqref{prfdetails3} follows from Eq.~\eqref{eqB16} and Eq.~\eqref{eq:uniControl}, Eq.~\eqref{prfdetails4} follows from two applications of Eq.~\eqref{identityOnticInf}, Eq.~\eqref{prfdetails9} follows from  Eq.~\eqref{Axiom:PropKnowGenerators}, Eq.~\eqref{prfdetails10} follows from Lemma~\eqref{lem:CCC}.

\proof
\begin{align}
\InputIfFileExists{Diagrams/QNF1.tikz}{}{\input{./figures/Diagrams/QNF1.tikz}}\quad &=\quad %
\InputIfFileExists{Diagrams/QNF2.tikz}{}{\input{./figures/Diagrams/QNF2.tikz}} \label{prfdetails1}
\end{align}
\begin{align}
\qquad &=\quad %
\InputIfFileExists{Diagrams/QNF3.tikz}{}{\input{./figures/Diagrams/QNF3.tikz}} \label{prfdetails2} \\
&=\quad %
\InputIfFileExists{Diagrams/QNF4.tikz}{}{\input{./figures/Diagrams/QNF4.tikz}} \label{prfdetails3} \\
&\sim_{\mathbf{p^*}}\quad %
\InputIfFileExists{Diagrams/QNF5.tikz}{}{\input{./figures/Diagrams/QNF5.tikz}} \label{prfdetails4} \\
&=\quad %
\InputIfFileExists{Diagrams/QNF6.tikz}{}{\input{./figures/Diagrams/QNF6.tikz}}
\end{align}
where
\beq
D\quad =\quad %
\InputIfFileExists{Diagrams/QNF7.tikz}{}{\input{./figures/Diagrams/QNF7.tikz}}
\eeq
We can then further rewrite this as:
\begin{align}
\qquad &= \quad %
\InputIfFileExists{Diagrams/QNF8.tikz}{}{\input{./figures/Diagrams/QNF8.tikz}}\label{prfdetails6} \\
&= \quad %
\InputIfFileExists{Diagrams/QNF9.tikz}{}{\input{./figures/Diagrams/QNF9.tikz}}\label{prfdetails7} \\
&= \quad %
\InputIfFileExists{Diagrams/QNF10.tikz}{}{\input{./figures/Diagrams/QNF10.tikz}} \label{prfdetails8} \\
&= \quad %
\InputIfFileExists{Diagrams/QNF11.tikz}{}{\input{./figures/Diagrams/QNF11.tikz}} \label{prfdetails9} \\
&= \quad %
\InputIfFileExists{Diagrams/QNF12.tikz}{}{\input{./figures/Diagrams/QNF12.tikz}} \label{prfdetails10} \\
&= \quad %
\InputIfFileExists{Diagrams/QNF13.tikz}{}{\input{./figures/Diagrams/QNF13.tikz}}\label{prfdetails11} \\
&= \quad %
\InputIfFileExists{Diagrams/QNF14.tikz}{}{\input{./figures/Diagrams/QNF14.tikz}} \label{prfdetails12} \\
&= \quad %
\InputIfFileExists{Diagrams/QNF15.tikz}{}{\input{./figures/Diagrams/QNF15.tikz}} \label{prfdetails13}
\end{align}
By their construction, one can see that $\Sigma$ is a stochastic map and $\Pi$ is a propositional map.
\endproof

\section{Useful results for \crealist representations}
\subsection{Proof of Theorem~\ref{thm:OntRepNF}} \label{finalproof}
We now prove Theorem~\ref{thm:OntRepNF}.
\proof
First, note that Eq.~\eqref{eq:CausalOntRep}, diagram preservation of $\xi$, and the constraint of ignorability, Eq.~\eqref{constraint6}, imply that
\begin{align}
\InputIfFileExists{Diagrams/OntRepProof20.tikz}{}{\input{./figures/Diagrams/OntRepProof20.tikz}}\quad &= \quad %
\InputIfFileExists{Diagrams/OntRepProof21.tikz}{}{\input{./figures/Diagrams/OntRepProof21.tikz}} \label{LHSontrepn} \\
 &= \quad%
\InputIfFileExists{Diagrams/OntRepProof22.tikz}{}{\input{./figures/Diagrams/OntRepProof22.tikz}}.
\end{align}

Now,  Theorem~\ref{thm:QNF} gives that
\beq \label{normalformtoontrepn}
\InputIfFileExists{Diagrams/OntRep1.tikz}{}{\input{./figures/Diagrams/OntRep1.tikz}}\quad \sim_{{\bf p}^*}\quad %
\InputIfFileExists{Diagrams/QNF16.tikz}{}{\input{./figures/Diagrams/QNF16.tikz}},
\eeq
for some substochastic map $\Sigma$ and some propositional effect $\Pi$. Applying this to decompose the process on the LHS of Eq.~\eqref{LHSontrepn}, one gets
\beq
\InputIfFileExists{Diagrams/QNF17.tikz}{}{\input{./figures/Diagrams/QNF17.tikz}}\quad\sim_{\mathbf{p^*}}\quad  \begin{tikzpicture}
	\begin{pgfonlayer}{nodelayer}
		\node [style=none] (0) at (-1.5, 0.75) {};
		\node [style=none] (1) at (0, -0.75) {};
		\node [style=none] (2) at (0, -1.75) {};
		\node [style=infupground] (3) at (-0.5, 0.75) {};
		\node [style=none] (4) at (0, -0.75) {};
		\node [style=none] (5) at (0, -0.25) {};
		\node [style=none] (6) at (-1.25, 1.25) {};
		\node [style={small black dot}] (7) at (0, -0.25) {};
		\node [style=ignore] (8) at (0, -0.25) {};
	\end{pgfonlayer}
	\begin{pgfonlayer}{edgelayer}
		\draw [style=CcWire, in=180, out=0, looseness=0.75] (0.center) to (3);
		\draw [oWire] (4.center) to (5.center);
		\draw [oWire] (1.center) to (2.center);
	\end{pgfonlayer}
\end{tikzpicture}.
\eeq

Rewriting the LHS of this we obtain
\beq
\InputIfFileExists{Diagrams/QNF18.tikz}{}{\input{./figures/Diagrams/QNF18.tikz}}\quad\sim_{\mathbf{p^*}}\quad  \begin{tikzpicture}
	\begin{pgfonlayer}{nodelayer}
		\node [style=none] (0) at (-1.5, 0.75) {};
		\node [style=none] (1) at (0, -0.75) {};
		\node [style=none] (2) at (0, -1.75) {};
		\node [style=infupground] (3) at (-0.5, 0.75) {};
		\node [style=none] (4) at (0, -0.75) {};
		\node [style=none] (5) at (0, -0.25) {};
		\node [style=none] (6) at (-1.25, 1.25) {};
		\node [style={small black dot}] (7) at (0, -0.25) {};
		\node [style=ignore] (8) at (0, -0.25) {};
	\end{pgfonlayer}
	\begin{pgfonlayer}{edgelayer}
		\draw [style=CcWire, in=180, out=0, looseness=0.75] (0.center) to (3);
		\draw [oWire] (4.center) to (5.center);
		\draw [oWire] (1.center) to (2.center);
	\end{pgfonlayer}
\end{tikzpicture}
\eeq
Using Lemma~\ref{lem:CausalProposition} (stated and proved in Appendix~\ref{lemmaaboutpropns}), we obtain
\beq \label{usesproplemma}
\begin{tikzpicture}
	\begin{pgfonlayer}{nodelayer}
		\node [style=none] (0) at (-1.25, 0.5) {};
		\node [style=none] (1) at (-1.25, -0.5) {};
		\node [style=none] (2) at (0.5, -0) {$\Pi$};
		\node [style=none] (3) at (1.25, -0) {};
		\node [style=none] (4) at (1.25, -0) {};
		\node [style=none] (5) at (0, 1) {};
		\node [style=none] (6) at (0, -1) {};
		\node [style=none] (7) at (0, 0.5) {};
		\node [style=none] (8) at (0, -0.5) {};
	\end{pgfonlayer}
	\begin{pgfonlayer}{edgelayer}
		\draw (4.center) to (5.center);
		\draw (5.center) to (6.center);
		\draw (6.center) to (3.center);
		\draw (3.center) to (4.center);
		\draw [cWire] (0.center) to (7.center);
		\draw [cWire, in=180, out=0, looseness=1.00] (1) to (8.center);
	\end{pgfonlayer}
\end{tikzpicture}
\quad \sim_{\mathbf{p}^*} \quad \begin{tikzpicture}
	\begin{pgfonlayer}{nodelayer}
		\node [style=none] (0) at (-1.25, 0.75) {};
		\node [style=none] (1) at (-1.25, -0.75) {};
		\node [style=none] (2) at (0, 0.75) {};
		\node [style=none] (3) at (0, -0.75) {};
		\node [style=infcopoint] (4) at (0.25, 0.75) {$\chi_\Pi$};
		\node [style=infcopoint] (5) at (0.25, -0.75) {$\top$};
	\end{pgfonlayer}
	\begin{pgfonlayer}{edgelayer}
		\draw [cWire] (0.center) to (2.center);
		\draw [cWire, in=180, out=0, looseness=1.00] (1.center) to (3.center);
	\end{pgfonlayer}
\end{tikzpicture}
\eeq
Substituting this in, we obtain
\beq
\InputIfFileExists{Diagrams/QNF19.tikz}{}{\input{./figures/Diagrams/QNF19.tikz}}\quad\sim_{\mathbf{p^*}}\quad \begin{tikzpicture}
	\begin{pgfonlayer}{nodelayer}
		\node [style=none] (0) at (-1.5, 0.75) {};
		\node [style=none] (1) at (0, -0.75) {};
		\node [style=none] (2) at (0, -1.75) {};
		\node [style=infupground] (3) at (-0.5, 0.75) {};
		\node [style=none] (4) at (0, -0.75) {};
		\node [style=none] (5) at (0, -0.25) {};
		\node [style=none] (6) at (-1.25, 1.25) {};
		\node [style={small black dot}] (7) at (0, -0.25) {};
		\node [style=ignore] (8) at (0, -0.25) {};
	\end{pgfonlayer}
	\begin{pgfonlayer}{edgelayer}
		\draw [style=CcWire, in=180, out=0, looseness=0.75] (0.center) to (3);
		\draw [oWire] (4.center) to (5.center);
		\draw [oWire] (1.center) to (2.center);
	\end{pgfonlayer}
\end{tikzpicture}
\eeq
and so
\beq
\InputIfFileExists{Diagrams/QNF20.tikz}{}{\input{./figures/Diagrams/QNF20.tikz}}\quad\sim_{\mathbf{p^*}}\quad\begin{tikzpicture}
	\begin{pgfonlayer}{nodelayer}
		\node [style=none] (0) at (-1.5, 0.75) {};
		\node [style=none] (1) at (0, -0.75) {};
		\node [style=none] (2) at (0, -1.75) {};
		\node [style=infupground] (3) at (-0.5, 0.75) {};
		\node [style=none] (4) at (0, -0.75) {};
		\node [style=none] (5) at (0, -0.25) {};
		\node [style=none] (6) at (-1.25, 1.25) {};
		\node [style={small black dot}] (7) at (0, -0.25) {};
		\node [style=ignore] (8) at (0, -0.25) {};
	\end{pgfonlayer}
	\begin{pgfonlayer}{edgelayer}
		\draw [style=CcWire, in=180, out=0, looseness=0.75] (0.center) to (3);
		\draw [oWire] (4.center) to (5.center);
		\draw [oWire] (1.center) to (2.center);
	\end{pgfonlayer}
\end{tikzpicture}.
\eeq
Hence, it must be that
\beq
\begin{tikzpicture}
	\begin{pgfonlayer}{nodelayer}
		\node [style=none] (0) at (-2.25, 0.5000001) {};
		\node [style=none] (1) at (-3, 0.5000001) {};
		\node [style=none] (2) at (-1.25, -0) {};
		\node [style=none] (3) at (-2.25, 1) {};
		\node [style=none] (4) at (-1.25, 1.5) {};
		\node [style=none] (5) at (-1.25, -0.5000001) {};
		\node [style=none] (6) at (-2.25, -0) {};
		\node [style=none] (7) at (-1.75, 0.5000001) {$\Sigma$};
		\node [style=infcopoint] (8) at (-0.5, -0) {$\chi_\Pi$};
		\node [style=none] (9) at (-1.25, 1) {};
		\node [style=none] (10) at (0.5, 1) {};
	\end{pgfonlayer}
	\begin{pgfonlayer}{edgelayer}
		\draw [cWire] (1.center) to (0.center);
		\draw (3.center) to (4.center);
		\draw (4.center) to (5.center);
		\draw (5.center) to (6.center);
		\draw (6.center) to (3.center);
		\draw [cWire] (2.center) to (8);
		\draw [cWire] (9.center) to (10.center);
	\end{pgfonlayer}
\end{tikzpicture} \quad =: \quad \begin{tikzpicture}
	\begin{pgfonlayer}{nodelayer}
		\node [style=none] (0) at (-2, -0) {};
		\node [style=none] (1) at (-2.75, -0) {};
		\node [style=none] (2) at (-0.9999998, -0.5000003) {};
		\node [style=none] (3) at (-0.9999998, -0) {};
		\node [style=none] (4) at (0, -0) {};
		\node [style=none] (5) at (0, -0) {};
		\node [style=none] (6) at (0, 0.25) {};
		\node [style=none] (7) at (-2, 0.5000003) {};
		\node [style=none] (8) at (-0.9999998, 0.5000003) {};
		\node [style=none] (9) at (-0.9999998, -0.5000003) {};
		\node [style=none] (10) at (-2, -0.5000003) {};
		\node [style=none] (11) at (-1.5, -0) {$\Xi_A^B$};
	\end{pgfonlayer}
	\begin{pgfonlayer}{edgelayer}
		\draw [cWire] (1.center) to (0.center);
		\draw [style=CcWire, in=180, out=0, looseness=0.75] (3.center) to (5.center);
		\draw (7.center) to (8.center);
		\draw (8.center) to (9.center);
		\draw (9.center) to (10.center);
		\draw (10.center) to (7.center);
	\end{pgfonlayer}
\end{tikzpicture}
\eeq
is a stochastic map.

Finally, substituting the decomposition of $\Pi$ into Eq.~\eqref{normalformtoontrepn} and then using the definition of $\Xi_A^B$, one obtains
\begin{align}
\begin{tikzpicture}
	\begin{pgfonlayer}{nodelayer}
		\node [style=epiBox] (0) at (0, -0) {};
		\node [style=none] (1) at (0, -1) {};
		\node [style=none] (2) at (0, 1) {};
		\node [style=none] (3) at (-3, -0) {};
		\node [style=none] (4) at (0.75, -0.75) {\footnotesize $\xi$};
		\node [style={up label}] (5) at (-1, -0) {$\morph{\op{A}}{\op{B}}$};
		\node [style={up label}] (6) at (-2.5, -0) {$\morph{\op{A}}{\op{B}}$};
		\node [style={right label}] (7) at (0, 0.5) {$\op{B}$};
		\node [style={right label}] (8) at (0, -0.75) {$\op{A}$};
		\node [style={right label}] (9) at (0, 1.5) {$\Lambda_{\op{B}}$};
		\node [style={right label}] (10) at (0, -1.5) {$\Lambda_{\op{A}}$};
		\node [style=none] (11) at (-1.75, 1) {};
		\node [style=none] (12) at (1, 1) {};
		\node [style=none] (13) at (1, -1) {};
		\node [style=none] (14) at (-1.75, -1) {};
		\node [style=none] (15) at (0, -2) {};
		\node [style=none] (16) at (0, 2) {};
	\end{pgfonlayer}
	\begin{pgfonlayer}{edgelayer}
		\filldraw [fill=gray!30,draw=gray!60] (11.center) to (12.center) to (13.center) to (14.center) to cycle;
		\draw [qWire] (0) to (2.center);
		\draw [qWire] (0) to (1.center);
		\draw [oWire] (16) to (2.center);
		\draw [oWire] (15) to (1.center);
		\draw [CcWire] (3.center) to (0);
	\end{pgfonlayer}
\end{tikzpicture}
\quad
&\sim_\mathbf{p^*}\quad
\begin{tikzpicture}
	\begin{pgfonlayer}{nodelayer}
		\node [style=none] (0) at (-2.25, 0.5000001) {};
		\node [style=none] (1) at (-3, 0.5000001) {};
		\node [style=none] (2) at (-1.25, -0) {};
		\node [style={clear dot}] (3) at (0, -1.5) {};
		\node [style=none] (4) at (-1.25, 1) {};
		\node [style=none] (5) at (0, 0.9999998) {};
		\node [style=none] (6) at (0, -2.25) {};
		\node [style=none] (7) at (0, 0.9999998) {};
		\node [style=epiBox] (8) at (0, 0.9999998) {};
		\node [style=none] (9) at (0, 2.25) {};
		\node [style=none] (10) at (-2.25, 1) {};
		\node [style=none] (11) at (-1.25, 1.5) {};
		\node [style=none] (12) at (-1.25, -0.5000001) {};
		\node [style=none] (13) at (-2.25, -0) {};
		\node [style=none] (14) at (-1.75, 0.5000001) {$\Sigma$};
		\node [style=none] (15) at (1.25, -0) {};
		\node [style=none] (16) at (1.25, -1.5) {};
		\node [style=infcopoint] (17) at (1.5, -0) {$\chi_\pi$};
		\node [style=infcopoint] (18) at (1.5, -1.5) {$\top$};
	\end{pgfonlayer}
	\begin{pgfonlayer}{edgelayer}
		\draw [cWire] (1.center) to (0.center);
		\draw [style=CcWire, in=180, out=0, looseness=0.75] (4.center) to (7.center);
		\draw [oWire] (8) to (9.center);
		\draw [oWire] (5.center) to (6.center);
		\draw (10.center) to (11.center);
		\draw (11.center) to (12.center);
		\draw (12.center) to (13.center);
		\draw (13.center) to (10.center);
		\draw [cWire] (2.center) to (15.center);
		\draw [cWire, in=180, out=0, looseness=1.00] (3) to (16.center);
	\end{pgfonlayer}
\end{tikzpicture}\\
\quad
&=\quad%
\InputIfFileExists{Diagrams/QNF21.tikz}{}{\input{./figures/Diagrams/QNF21.tikz}}
\end{align}
That is, every \crealist representation is inferentially equivalent to updating one's knowledge about the operational procedure to knowledge about functional dynamics.
\endproof

\end{document}